\documentclass[11pt,english, singlespacing, headsepline]{MastersDoctoralThesis}
\usepackage[utf8]{inputenc} 
\usepackage[T1]{fontenc} 
\usepackage[backend=bibtex,style=numeric-comp ,natbib=true, sorting=none]{biblatex} 
\usepackage[autostyle=true]{csquotes} 
\usepackage{graphicx}
\usepackage{amsthm}
\usepackage{color,soul}
\usepackage{xcolor}
\usepackage{lineno,hyperref}
\usepackage[]{hyperref}
\usepackage{epsfig}
\usepackage{subfigure}
\usepackage{amssymb}
\usepackage{siunitx}
\usepackage{latexsym}
\usepackage{soul}
\usepackage{amsmath}
\usepackage{dirtytalk}
\usepackage{xfrac}
\usepackage{textcomp}
\usepackage{xcolor,colortbl}
\definecolor{Gray}{gray}{0.9}
\usepackage{amsfonts}
\hypersetup{urlcolor=blue}
\hypersetup{citecolor=blue}
\hypersetup{
    colorlinks=true,
    linkcolor=blue,
    citecolor=blue,
    }
\usepackage{blindtext}
\newtheorem{rem}{Remark}
\newtheorem{lem}{Lemma}
\newtheorem{thm}{Theorem}
\newtheorem{assump}{Assumption}
\usepackage{algorithm}
\usepackage[noend]{algpseudocode}
\usepackage{scrextend}
\usepackage{pbox}
\usepackage{xcolor}
\usepackage{soul}

\modulolinenumbers[5]

\thesistitle{Optimal and Robust Fault Tolerant Control of Wind Turbines Working under Sensor, Actuator, and System Faults} 
\supervisor{\href{https://researchonline.gcu.ac.uk/en/persons/ibrahim-k\%C3\%BC\%C3\%A7\%C3\%BCkdemiral}{Dr. Ibrahim Beklan Kucukdemiral}  \newline \href{https://researchonline.gcu.ac.uk/en/persons/geraint-bevan}{Dr. Geraint Bevan}} 
\degree{Doctor of Philosophy} 
\author{Seyedyashar Mousavi} 
\addresses{} 

\subject{} 
\university{\href{https://www.gcu.ac.uk/}{Glasgow Caledonian University}} 
\department{\href{https://www.gcu.ac.uk/cebe/}{School of Computing, Engineering and Built Environment}} 
\group{\href{https://www.gcu.ac.uk/cebe/research/researchgroups/powerandrenewableenergysystemspres/}{Power and Renewable Energy Systems (PRES)}} 
\faculty{\href{https://www.gcu.ac.uk/cebe/aboutus/departments/appliedscience/}{Department of Applied Science}} 

\AtBeginDocument{
\hypersetup{pdftitle=\ttitle} 
\hypersetup{pdfauthor=\authorname} 
\hypersetup{pdfkeywords=\keywordnames} 
}

\begin{document}

\frontmatter 

\pagestyle{plain} 


\begin{titlepage}
\begin{center}
\includegraphics[width=2 in]{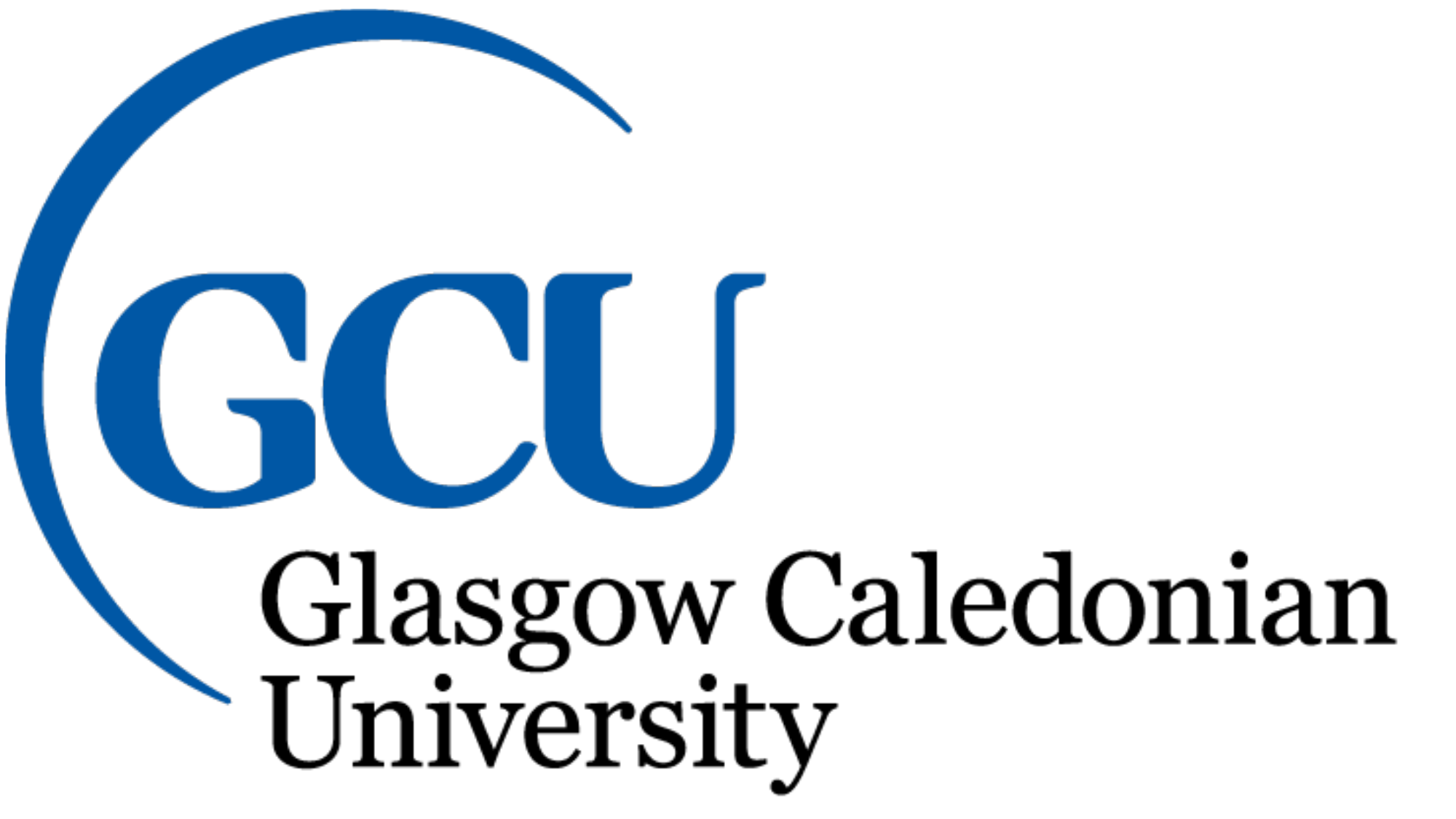} 

\vspace*{.06\textheight}
{\scshape\LARGE \univname\par}\vspace{1.5cm} 
\textsc{\Large Doctoral Thesis}\\[0.5cm] 

\HRule \\[0.4cm] 
{\huge \bfseries \ttitle\par}\vspace{0.4cm} 
\HRule \\[1.5cm] 

\begin{minipage}[t]{0.4\textwidth}
\emph{Author:}\\
\href{https://scholar.google.com/citations?hl=en&user=y_WAYDQAAAAJ&view_op=list_works&sortby=pubdate}{\authorname} 
\end{minipage}
\begin{minipage}[t]{0.5\textwidth}
\emph{Supervisor:} \\
{\supname} 
\end{minipage}\\[3cm]

\vfill

\large \textit{Doctoral Thesis}\\[0.3cm] 
\groupname\\\deptname\\[2cm] 

\vfill

{\large \today}\\[4cm] 

\vfill
\end{center}
\end{titlepage}

\section*{Abstract}
\addchaptertocentry{\abstractname} 

Motivated by the increasing concerns over environmental challenges such as global warming and the exhaustion of fossil-fuel reserves, the renewable energy industry has become the most demanded electrical energy production source worldwide. In this context, wind energy conversion systems (WECSs) are the most dominant and fastest-growing alternative energy production technologies, playing an increasingly vital role in renewable power generation. To meet this growing demand and considering the vulnerability of WECSs' performance to various sets of internal/external-caused faults, the cost-effectiveness and efficient power production of WECSs must be ensured, highlighting the critical role of the control system. This topic has been intensively studied in the literature, and many control approaches have been developed to cope with the simultaneous efficiency and reliability enhancement of WECSs. However, among all control strategies, sliding mode control (SMC) has shown its reliable and superior performance due to its robustness against uncertainties and disturbances along with the relative simplicity of implementation.

One of the most prominent deficiencies of conventional SMC controllers (SMCs) is the chattering phenomenon, which is a high-frequency oscillation caused by the discontinuous switching control.  \textcolor{black}{This phenomenon makes the state trajectories quickly oscillate around the sliding surface, leading to degradation of the control system's performance. These oscillations make the controller consume more power to deliver the desired performance, which leads to unwanted larger control signals. Also, undesirable heat losses for electrical power circuits or higher wear of moving mechanical parts can result in practice. This issue manifests itself when the system is faulty, a semi-common situation for WECSs that are constantly exposed to environment changes and characteristic changes of mechanical parts.} As one solution for the chattering problem, higher-order SMCs are developed that consider higher-order derivatives with respect to time. Thus, the sliding variable and its consecutive derivatives tend to zero with higher stabilization accuracy in finite time. Consequently, the chattering phenomenon's undesired effects are reduced; however, the problem is not entirely resolved and needs more effort and improvements. On the other hand, fractional-order derivatives add distinctive memory features to the control system, resulting in faster convergence of the state variables to the equilibrium point and mitigating the conventional SMCs' chattering problem. However, \textcolor{black}{due to the highly nonlinear behavior of WECSs stemming from its electro-mechanical components}, very few studies in the literature have dealt with the control problem of WECSs using fractional-calculus-based SMCs.

To tackle the pitch control problem of WECSs, this thesis proposes an optimal fault-tolerant fractional-order pitch control strategy for pitch angle regulation of WT blades subjected to sensor, actuator, and system faults. \textcolor{black}{To showcast the effects of faults, changes in the system parameters are considered as a result of sensor, actuator and system faults with various levels of severity.} Furthermore, taking advantage of the favourable merits of higher-order SMCs and fractional calculus, this thesis develops a fault-tolerant fractional-calculus-based higher-order sliding mode controller for optimum rotor speed tracking and power production maximization of WECSs. \textcolor{black}{The partial loss of the generator output torque is considered as an actuator failure, leading to loss of partial actuation power.} \textcolor{black}{Moreover, active fault-tolerant fractional-order higher-order SMC strategies are developed for rotor current regulation and speed trajectory tracking of doubly-fed induction generator (DFIG) -driven WECSs subjected to model uncertainties and rotor current sensor faults.} The developed controllers are augmented with two state observers, an algebraic state estimator and a sliding mode observer, to estimate the rotor current dynamics during sensors' faults. \textcolor{black}{The developed control schemes demonstrate robustness against model uncertainties and sensor faults.} In addition to tackling the control problems of WECSs, the proposed SMC-based controllers aim to alleviate the chattering problem, guarantee a fast finite-time convergence of the system states, provide robustness against external disturbances and model uncertainties, and deliver higher control precision.

\clearpage
\section*{List of Publications (Achieved from the PhD thesis)}
\subsection*{Published}
\begin{enumerate}
\item \textbf{Y. Mousavi}, G. Bevan, I. B. Kucukdemiral, A. Fekih, "Fault-tolerant observer-based Higher-order Sliding Mode Control of DFIG-based Wind Turbines with Sensor Faults," \textit{IEEE Transactions on Industry Applications}, 2023.
  \item \href{https://www.sciencedirect.com/science/article/pii/S1364032122006232}{\textbf{Y. Mousavi}, G. Bevan, I. B. Kucukdemiral, A. Fekih, "Sliding mode control of wind energy conversion systems: trends and applications," \textit{Renewable and Sustainable Energy Reviews}, 167, 2022.}
  \item \textbf{Y. Mousavi}, G. Bevan, I. B. Kucukdemiral, A. Fekih, "Active fault-tolerant fractional-order terminal sliding mode control for DFIG-based wind turbines subjected to sensor faults," \textit{2022 IEEE IAS Global Conference on Emerging Technologies (GlobConET), May 2022}.
  \item \href{https://www.sciencedirect.com/science/article/pii/S0019057821005383}{\textbf{Y. Mousavi}, G. Bevan, I. B. Kucukdemiral, "Fault-tolerant optimal pitch control of wind turbines using dynamic weighted parallel firefly algorithm," \textit{ISA Transactions}, 128, 301-317 2022.}
  \item \href{https://www.mdpi.com/1996-1073/14/18/5887}{\textbf{Y. Mousavi}, G. Bevan, I. B. Kucukdemiral, A. Fekih, "Maximum power extraction from wind turbines using a fault-tolerant fractional-order nonsingular terminal sliding mode controller," \textit{Energies}, 14(18), 2021.}
\end{enumerate}

\section*{List of Publications (partly related to the PhD thesis)}
\subsection*{Published}
\begin{enumerate}
\item \href{https://ieeexplore.ieee.org/abstract/document/10194923}{A. Mehrzad, M. Darmiani, \textbf{Y. Mousavi}, M. Shafie-khah, M. Aghamohammadi, "A Review on Data-driven Security Assessment of Power Systems: Trends and Applications of Artificial Intelligence," \textit{IEEE Access}, 2023.}
\item \href{https://www.sciencedirect.com/science/article/pii/S0016003223004787?casa_token=L9J9P3buiMgAAAAA:bRaWqkECUdPhIcc9pCH315qamuokcY9GJNjCkEtLzZN3Ajlu6Bo1uhXIi_yQav2L766nkOjpJg}{Z. S. Aghayan, A. Alfi, \textbf{Y. Mousavi}, A. Fekih, "Stability analysis of a class of variable fractional-order uncertain neutral-type systems with time-varying delay," \textit{Journal of the Franklin Institute}, 2023.}
  \item \textbf{Y. Mousavi}, G. Bevan, I. B. Kucukdemiral, A. Fekih, "Disturbance observer and tube-based model reference adaptive control for active suspension systems with non-ideal actuators," \textit{22nd World Congress of the International Federation of Automatic Control (IFAC2023), 2023}.
\item \href{https://ieeexplore.ieee.org/document/10068797}{Z. S. Aghayan, A. Alfi, \textbf{Y. Mousavi}, A. Fekih, "Criteria for stability and stabilization of variable fractional-order uncertain neutral systems with time-varying delay: Delay-dependent analysis," \textit{IEEE Transactions on Circuits and Systems II: Express Briefs}, 2023.}
  \item \href{https://www.sciencedirect.com/science/article/pii/S0960077922007238}{Z. S. Aghayan, A. Alfi, \textbf{Y. Mousavi}, I. B. Kucukdemiral, A. Fekih, "Guaranteed cost robust output feedback control design for fractional-order uncertain neutral delay systems," \textit{Chaos Solitons \& Fractals}, 2022.}
  \item \href{https://www.ieee-jas.net/en/article/doi/10.1109/JAS.2022.105470}{\textbf{Y. Mousavi}, A. Alfi, I. B. Kucukdemiral, A. Fekih, "Tube-based model reference adaptive control for vibration suppression of active suspension systems," \textit{IEEE-CAA Journal of Automatica Sinica}, 9(4), 728-731, 2022.}
  \item \href{https://www.sciencedirect.com/science/article/pii/S0019057821002883}{\textbf{Y. Mousavi}, A. Zarei, Z. Sane Jahromi, "Robust adaptive fractional-order nonsingular terminal sliding mode stabilization of three-axis gimbal platforms," \textit{ISA Transactions}, 123, 98-109, 2022.}
  \item \href{https://ieeexplore.ieee.org/document/9777992}{A. Mehrzad, M. Darmiani, \textbf{Y. Mousavi}, M. Shafie-khah, M. Aghamohammadi, "An efficient rapid method for generators coherency identification in large power systems," \textit{IEEE Open Access Journal of Power and Energy}, 9, 151-160, 2022.}
  \item \href{https://link.springer.com/article/10.1007/s11633-021-1282-3}{\textbf{Y. Mousavi}, A. Zarei, A. Mousavi, M. Biari, "Robust optimal higher-order-observer-based dynamic sliding mode control for VTOL unmanned aerial vehicles," \textit{International Journal of Automation and Computing}, 18(5), 802-813, 2021.}
  \item \href{https://ieeexplore.ieee.org/abstract/document/9151887}{\textbf{Y. Mousavi}, A. Alfi, I. B. Kucukdemiral, "Enhanced fractional chaotic whale optimization algorithm for parameter identification of isolated wind-diesel power systems," \textit{IEEE Access}, 8, 140862-140875, 2020.}
\end{enumerate}
\subsection*{Submitted/Under Review}
\begin{enumerate}
  \item Z. S. Aghayan, A. Alfi, \textbf{Y. Mousavi}, A. Fekih, "Robust delay-dependent output-feedback PD controller design for variable fractional-order uncertain neutral systems with time-varying delay," \textit{IEEE Transactions on Systems, Man, and Cybernetics: Systems}, under review, 2023. Manuscript ID: SMCA-22-05-1364.
\end{enumerate}
\clearpage
\hypersetup{linkcolor=black}
\tableofcontents 

\listoffigures 

\listoftables 

\begin{abbreviations}{ll} 

ANFIS & Adaptive Neuro-Fuzzy Inference System  \\
APC & Active Power Control  \\
ASMC & Adaptive Sliding Mode Control  \\
CART & Controls Advanced Research Turbine \\
CCGBFO & Enhanced Bacterial Foraging Optimization \\
CEC & Congress on Evolutionary Computation \\
CR & Chattering Reduction \\
DCVR & DC-link Voltage Regulation \\
DDSG & Direct Driven Synchronous Generator   \\
DFIG & Doubly-Fed Induction Generator  \\
DOIG & Double Output Induction Generator  \\
DSIG & Dual Stator Induction Generator  \\
EA & Evolutionary Algorithm \\
EGWO & Enhanced Grey Wolf Optimization \\
EM-FOPID & Extended Memory FOPID \\
F-SMC & Fuzzy Sliding Mode Control  \\
FA & Firefly Algorithm \\
FIS & Fuzzy Inference System  \\
FLC & Fuzzy Logic Control  \\
FNTSMC & Fractional-order Nonsingular Terminal Sliding Mode Control \\
FO & Fractional Order\\
FO-SMC & Fractional-Order Sliding Mode Control \\
FOFA & Fractional-order FA \\
FOPID & Fractional-order PID \\
FPSOMA & Fractional PSO-based memetic algorithm \\
FTC & Fault Tolerant Control \\
FTPC & Fault Tolerant Pitch Control \\
G-L & Grunwald--Letnikov \\
GA & Genetic Algorithm \\
GSA & Gravitational Search Algorithm  \\
GSC & Grid-Side Converter  \\
GSCV & Grid-Side Converter Voltage regulation \\
GWO & Grey Wolf Optimizer \\
HAWT & Horizontal Axis Wind Turbine \\
HO-SMC & Higher-Order Sliding Mode Control   \\
ISCO & Integral of Squared Control\\
ISMC & Integral Sliding Mode Control \\
ITSE & Integral of Time Multiplied Squared Error \\
LFC & Load Frequency Control  \\
LQR & Linear Quadratic Regulator  \\
MIMO & Multi-Input Multi-Output   \\
ML & Mittag--Leffler \\
MPC & Model Predictive Control \\
MPE & Maximum Power Extraction  \\
MPPT & Maximum Power Point Tracking \\
MSM & Mechanical Stress Minimization\\
NFE & Number of Maximum Function Evaluation\\
NN-SMC & Neural Network Sliding Mode Control\\
OFOPID & Optimal FOPID \\
PAC & Pitch Angle Control\\
PI & Proportional Integral\\
PID & Proportional Integral Derivative\\
PMSG & Permanent-Magnet Synchronous Generator\\
PMSM & Permanent-Magnet Synchronous Motor\\
PSO & Particle Swarm Optimization\\
PWM & Pulse Width Modulation \\
R-L & Riemann--Liouville \\
RBFNN & Radial Basis Function Neural Network \\
RPC & Reactive Power Control\\
RSC & Rotor-Side Converter\\
RSCV & Rotor-Side Converter Voltage Regulation\\
SA & Simulated Annealing \\
SCIG & Squirrel-Cage Induction Generator \\
SEIG & Self-Excited Induction Generator\\
SFG & Simple First-order Generator\\
SMC & Sliding Mode Control\\
SMO & Sliding Mode Observer\\
SO-SMC & Second-Order Sliding Mode Control\\
SOFTSMC & Second-order Fast Terminal SMC \\
SSG & Simple Second-order Generator\\
ST-SMC & Super-Twisting Sliding Mode Control\\
TSMC & Terminal Sliding Mode Control\\
WECS & Wind Energy Conversion System\\
WT & Wind Turbine\\
\end{abbreviations}


\begin{symbols}{lll} 

\textbf{\scalebox{1.2}{Aerodynamics}} & & \\
\textbf{Symbol} & \textbf{Parameter} & \textbf{Unit}\\
$c$&Scale factor & --\\
$C_P, C_q, C_t$&Power, torque, and thrust coefficients & -- \\
$C_{p,\max}$&Maximum power coefficient & --\\
$F_t$&Thrust force & \si{N}\\
$k$&Shape factor & --\\
$P_a, T_a$&Aerodynamic power and torque & \si{W}, \si{N.m}  \\
$P_{g}$ & Generated power by the generator & \si{W}\\
$P_{ref}$ & Reference power & \si{W}\\
$R$&Blade radius & \si{m}\\
$\beta$&Blade pitch angle & \si{\circ}\\
$\beta_{opt}$&Optimum blade pitch angle & \si{\circ}\\
$\beta _{r} $ & Command signal for the pitch angle & \si{\circ}\\
$\eta _{g}$ & Generator's efficiency & --\\
$\lambda$&Tip-speed ratio & --\\
$\lambda_{opt}$&Optimum tip-speed ratio & --\\
$\xi $ & Damping factor & --\\
$\xi_0$ & Damping factor nominal value & --\\
$\xi_f$ & Damping factor low pressure value & --\\
$\rho$&Air density & \si{kg/m^3}\\
$\upsilon_{w}$&Effective wind speed & \si{m/s}\\
$\upsilon _{m} \left(t\right)$ & Slow wind variations & \si{m/s}\\
$\upsilon _{s} \left(t\right)$ & Stochastic wind behavior & \si{m/s}\\
$\upsilon _{ws} \left(t\right)$ & Wind shear effects & \si{m/s}\\
$\upsilon _{ts} \left(t\right)$ & Tower shadow effects & \si{m/s}\\
$\omega _{g} $ & Generator speed & \si{rad/s} \\
$\omega _{n} $ & Natural frequency & \si{rad/s}\\
$\omega _{n,0}$ & Natural frequency nominal value & \si{rad/s}\\
$\omega _{n,f}$ & Natural frequency low pressure value & \si{rad/s}\\
$\omega_r$&Rotor speed & \si{rad/s} \\
$\omega _{r,opt}$ & Optimal rotor speed & \si{rad/s}\\
\addlinespace 
\textbf{\scalebox{1.2}{Drivetrain}} & & \\
\textbf{Symbol} & \textbf{Parameter} & \textbf{Unit}\\
$B_{G}$ & Viscous friction of the high-speed shaft & \si{N.m.s/rad}\\
$D_B, D_H, D_{BH}$&Blade, hub, and turbine damping coefficients & \si{N.s/rad}\\
$D_{LS}, D_{HS}$&Low and high speed shaft damping coefficients & \si{N.s/rad}\\
$D_G, D_{GB}$&Generator and gearbox damping coefficients & \si{N.s/rad}\\
$D_{GBG}$&Gearbox and generator damping coefficient & \si{N.s/rad}\\
$D_R$ & Rotor external damping coefficient & \si{N.s/rad}\\
$D_t$ & Induced total external damping on the rotor side & \si{N.s/rad} \\
$D_{WT}$&Damping coefficient of the 1-Mass system & \si{N.s/rad}\\
$J_B, J_G, J_{GB}, J_H$&Blade, generator, gearbox, and hub inertias & \si{kg.m^2}\\
$J_{GBG}$&Summation of gearbox and generator inertia & \si{kg.m^2}\\
$J_{BH}$&Summation of blade and hub inertia & \si{kg.m^2}\\
$J_{R} $ & Rotor inertia & \si{kg.m^2}\\
$J_t$ & Induced total inertia & \si{kg.m^2}\\
$J_{WT}$&Summation of all rotating components' masses & \si{kg.m^2}\\
$K_{LS}, K_{HS}$&Low and high speed shaft stiffness constants & \si{N.m/rad}\\
$K_{BH}$&Blade stiffness constant & \si{N.m/rad}\\
$K_{HLS}$&Parallel shaft stiffness & \si{N.m/rad}\\
$N_{GB}$&Gearbox ratio & -- \\
$T_G, T_W$&Generator and aerodynamic torques & \si{N.m}\\
$T_{G,ref} $ & Torque reference to the generator & \si{N.m}\\
$T_{LS}$ & Low speed shaft torque & \si{N.m}\\
$T_{HS}$ & High speed shaft torque & \si{N.m}\\
$\alpha _{gc}$ & Generator and converter unit coefficient & -- \\
$\zeta_f\left(t\right)$ & Actuator efficiency factor & --\\
$\eta _{dt} $ & Drivetrain efficiency & -- \\
$\theta_{BH}$&Angle between blade and hub & \si{\circ}\\
$\theta_{LS}$&Angle between hub and gearbox & \si{\circ}\\
$\theta_{HS}$&Angle between gearbox and generator rotor & \si{\circ}\\
$\theta_{HLS}$&Angle between $J_{BH}$ and $J_{GBG}$ & \si{\circ}\\
$\theta _{\Delta } $ & Drivetrain torsion angle & \si{\circ}\\
$\theta_{R}$, $\theta_{LS}$ & Rotation angle of the rotor and generator shafts & \si{\circ}\\
$\tau_{gc}$ & Generator and converter time constant & \si{s}\\
$\omega_B, \omega_H, \omega_GB, \omega_G$&Blade, hub, gearbox, and generator speeds & \si{rad/s}\\
$\omega_0$&Synchronous speed of electrical system & \si{rad/s}\\
$\omega_{BH}$&Turbine rotational speed & \si{rad/s}\\
$\omega_{GBG}$&Summation of gearbox and generator speeds & \si{rad/s}\\
$\omega_{WT}$&Rotational speed on the 1-Mass system & \si{rad/s}\\
$\omega_{LS}$ & Low shaft speed & \si{rad/s}\\
\addlinespace 
\textbf{\scalebox{1.2}{Generators}} & & \\
\textbf{Symbol} & \textbf{Parameter} & \textbf{Unit}\\
\textbf{DFIG} & & \\
$f_r$ & Friction coefficient & --\\
$I_{ds}$,$I_{dr}$, $I_{qs}$, $I_{qr}$&d-axis, q-axis stator and rotor currents & \si{A}\\
$I_{dr-ref}$, $I_{qr-ref}$ & Reference d-axis, q-axis stator and rotor currents & \si{A}\\
$L_s$, $L_r$, $L_m$&Stator, rotor, and magnetizing inductances & \si{H}\\
$N_P$ &Number of pole pairs & -- \\
$p_i, q_i, \phi, \alpha_i, \beta_i, \kappa_i, \eta_i$ & Developed FTSMC controllers' parameters & --\\
$P_s$, $Q_s$, $P_r$, $Q_r$&Stator and rotor active and reactive powers & \si{W}, \si{VAR}\\
$R_s$, $R_r$&Stator and rotor resistances & \si{\Omega}\\
$T_{em}$&Electromagnetic torque & \si{N.m}\\
$V_{ds}$, $V_{dr}$, $V_{qs}$, $V_{qr}$&d-axis, q-axis stator and rotor voltages & \si{V}\\
$\psi_{ds}$, $\psi_{dr}$, $\psi_{qs}$, $\psi_{qr}$&d-axis, q-axis stator and rotor flux & \si{Wb}\\
$\omega_s$ & Stator speed & \si{rad/s}\\
\textbf{PMSG} & & \\
$I_{sd}$,$V_{sd}$&d-axis stator current and voltage & \si{A}, \si{V}\\
$I_{sq}$,$V_{sq}$&q-axis stator current and voltage & \si{A}, \si{V}\\
$L_d$, $L_q$&d-q axis mutual inductances & \si{H}\\
$\psi_s$&Flux linkage & \si{Wb}\\
$\omega_e$&Rotor electrical angular speed & \si{rad/s}\\
\addlinespace 
\textbf{\scalebox{1.2}{DWPFA Algorithm}} & & \\
\textbf{Symbol} & \textbf{Parameter} & \textbf{Unit}\\
$\mathfrak I$ &Light intensity & --\\
$\mathfrak I_{0} $ &Initial light intensity & --\\
$\mathfrak I_{GL}$ &Subgroup leader's light intensity & --\\
$n_{it} $, $n_{it,\max } $ &Current iteration and max iteration & --\\
$n_{p}$ &Total members of the populations & --\\
$n_{sg}$ &Number of predefined groups & --\\
$\mathfrak r=0$ &Distance between two fireflies & --\\
$R_f$ & DWPFA damping coefficient & --\\
$\mathfrak X_{GL} $ &Subgroup leader's position & --\\
$w_{sg}$ &Weighting coefficient of subgroups & --\\
$\gamma_f $ &Light absorption coefficient & --\\
$\eta_f $, $\psi_f $ & Random numbers within the interval [0,1] & --\\
$\kappa_f $ &Switching coefficient sensitivity parameter & --\\
$\chi_f$ &Firefly's attractiveness & --\\
$\chi_{f,0} $ &Firefly's attractiveness at $\mathfrak r=0$ & --\\
\addlinespace 
\textbf{\scalebox{1.2}{Fractional-order Calculus}} & &\\
\textbf{Symbol} & \textbf{Parameter} & \textbf{Unit}\\
$\mathfrak D$ & Fractional differentiator/integrator & --\\
$J$ &Objective function of the controller parameters & --\\
$K$ & Oustaloup recursive approximation adaptive gain & -- \\
$K_{p}, K_{i}, K_{d}$ & FOPID parameters & --\\
$L\{\cdot\}$ & Laplace transformation & --\\
$\mathfrak{p}$ & Laplace operator & -- \\
$R_{CH}$ &Radius of convergence & --\\
$T$ &Sampling time & \si{s}\\
$V_c$ &Constraints violation coefficient & --\\
$\gamma$ &Fractional derivation order & --\\
$\Gamma \left(\cdot\right)$ & Euler's gamma function & --\\
$\sigma_{ML} $ &Order of ML function & --\\
$\Upsilon$ &Truncation order & --\\
$\omega _{b} $& Lower constraint of the approximation frequency & --\\
$\omega _{h} $ & Upper constraint of the approximation frequency& --\\
\end{symbols}


\mainmatter 

\pagestyle{thesis} 



\chapter{Introduction} 

\label{Chapter1} 


\newcommand{\keyword}[1]{\textbf{#1}}
\newcommand{\tabhead}[1]{\textbf{#1}}
\newcommand{\code}[1]{\texttt{#1}}
\newcommand{\file}[1]{\texttt{\bfseries#1}}
\newcommand{\option}[1]{\texttt{\itshape#1}}

\section{Brief Introduction}
\label{sec:1.01}

This chapter outlines the motivation for the research undertaken on optimal and robust fault tolerant control of wind turbines working under sensor, actuator, and system faults, the research aims and objectives, and then conclude by outlining the structure of the remainder of the thesis.

\section{Motivation}
\label{sec:1.1}
Wind energy technology has seen considerable growth in the past two decades by continuing to push the boundaries in energy efficiency and reliability of supply \cite{adedeji2020wind, fu2020fatigue, pape2019offshore}. Taking into consideration the global concerns on climate change, global warming, and exhaustion of fossil energy resources, wind energy conversion systems (WECSs) have established themselves as the world's leading economically and ecologically power production technologies, overtaking other renewable green resources such as solar, hydro, and ocean waves.

WECSs are complex systems involving nonlinear dynamics with strongly coupled internal variables, parameter uncertainties, and external disturbances, which are expected to maximize the electric power generated from the intermittent stochastic wind and minimize the operational costs \cite{yang2016survey}. In accordance with the wind speed and to meet the power generation capacity and safe operation of WECSs, four operational regions are considered. \textcolor{black}{Figure \ref{fig:11-11} demonstrates the operating regions, where $V_{cut-in}$ and $V_{cut-out}$ stand for the wind speed at which the power generation start and stops, respectively. Also, $V_{rated}$ denotes the wind speed that the blades reach the speed, where the maximum power ($P_{g,rated}$) is generated. In regions I and IV, the wind speed is respectively, too low or too high, so the rotor is detached from the blades}. In region II (\textit{i.e.} the partial-load region), the wind speed is lower than the rated wind speed and a generator torque controller is usually used to maximize the power capture. At the rated wind speed, the turbine is capable of generating electricity at its optimal capacity. In region III (\textit{i.e.} the full-load region), the wind speed exceeds the rated value. \textcolor{black}{Accordingly, the WECS control objectives considering the generator types can be categorized into three main classes:} (a) maximum power extraction (MPE), which is pursued in region II by controlling the generator torque, (b) active and reactive power regulation, which is broadly pursued in region III by performing the generator DC-link voltage regulation (DCVR), grid-side converter (GSC) and rotor-side converter (RSC) control as well as pitch angle control (PAC), and (c) load mitigation, which is to ensure a smooth transition from region II to III.
\begin{figure}
\centering
\includegraphics[width=3.5 in]{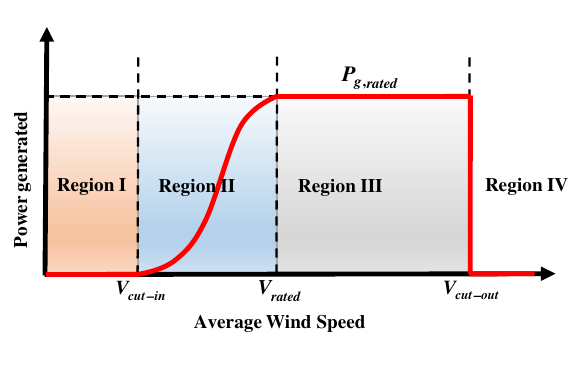}
\caption{Wind turbine operational regions.}
\label{fig:11-11}
\end{figure}

A critical issue to be considered is that similar to other power generation systems, the WECSs are prone to various possible faults such as sensors malfunctioning, unsymmetrical loads and grid voltage dips, actuator faults, and different types of system faults, while the gearbox faults have found to be the most common cause of \textcolor{black}{WECS failure. Thus, it has always been a significant challenge for engineers to overcome the control complexities of guaranteeing the MPE and the desired active and reactive power especially under faulty situations.} As a result, a great deal of research effort in academia and industry has been devoted to the development of reliable control approaches that guarantee the WECS power generation \cite{yin2016turbine, ghanbarpour2020dependable, kamal2013fuzzy, shahbazi2012fpga, merabet2018power}, power conversion efficiency and quality increment \cite{sharmila2019fuzzy, nian2014independent, kuhne2018fault, jlassi2019fault, shaker2014active}, mechanical damage alleviation \cite{cetrini2019line, brekken2005control, jain2018fault, camblong2014wind}, and operational and maintenance costs reduction \cite{jain2018fault, carroll2017availability, jimmy2019energy} both with and without considering faults effects.

Classical proportional-integral-derivative (PID) controllers are the standard conventional control methods that are traditionally being used for WECS control \cite{ren2016nonlinear}; however, due to lack of robustness in system parameters and the operating point changing conditions, and also nonlinear behavior of WTs and the existence of faults, they cannot be counted as reliable control methods. In this regard, researchers have proposed various innovative advanced control strategies to enhance the performance of the control system. Robust control methods are widely investigated in the literature to mitigate undesirable effects of wind and enhance the power quality, where $H_2$ and $H_{\infty}$ methods are the most used ones \cite{ surinkaew2014robust, das2018h}. Soft computing-based methods such as fuzzy logic control, artificial neural networks, and metaheuristic algorithms-based controllers are other investigated advanced approaches that offer an efficient and quick response to overcome uncertainties in WECSs \cite{kuhne2018fault, aguilar2020multi, dewangan2021performance}. 

Among all the existing control strategies, SMC has been shown to be a suitable and effective method for various nonlinear control problems due to fast dynamic response, good transient performance, stability, and robustness to matched parameter variations and external disturbances \cite{bakhshande2020robust,banza2020integral, incremona2016sliding, caballero2018sliding, zhang2020adaptive}. However, regardless of the satisfactory tracking performance of the conventional SMC technique in practical applications, it still has some shortcomings such as the vulnerability to measurement noise, the hardness in asymptotical stability achievement when the system is affected by mismatched perturbations, producing unnecessarily large control signals to overcome the parametric uncertainties, and the most serious one which is associated with the high-frequency oscillations caused by the discontinuous switching control; the chattering phenomenon \cite{dursun2020second, kelkoul2020stability, morshed2020design, song2019dynamic}. Due to the non-existence of a switching component element capable of shifting to an infinite frequency, the effects of the chattering phenomenon on the real-application systems cannot be precisely estimated \cite{liu2020event}. Hence, the main effort is usually to decrease/alleviate this phenomenon to prevent its consequences \cite{ni2016fast}. Besides, since the conventional SMC design methods rely on a linear sliding surface, the system's states may fail to converge to the equilibrium point in finite time \cite{incremona2016sliding}. Thus, in order to achieve high convergence speed in the conventional SMC, the reaching gains should be increased to be large enough when the SMC surface is around the equilibrium point, which is not desirable.

Consequently, although conventional SMCs have been found to be reliable control schemes, many efforts have been made to enhance their performance and maximize their reliability in dealing with nonlinear systems \cite{jiang2019takagi, liu2014finite, zhang2019semiglobal}. As one alternative solution, higher-order SMC approaches have been developed to fill these gaps by considering higher-order derivatives and introducing nonlinear functions in the SMC's surface in order to achieve finite-time convergence of the system dynamics and reduce the chattering problem \cite{levant2003higher, laghrouche2021barrier, edwards2016adaptive, tripathi2020finite}. However, despite faster convergence in finite-time \cite{laghrouche2021barrier}, better chattering reduction \cite{edwards2016adaptive}, and higher control precision \cite{tripathi2020finite} of higher-order SMCs in comparison with the conventional SMCs, they still need more improvements in performance and reliability. Fractional-order calculus has recently demonstrated its promising performance in tackling the deficiencies of SMCs, which has led to different integrations of SMC methods and fractional-order calculus in the literature \cite{sun2017discrete, wang2019adaptive, sun2017practical, ren2019fractional}. Compared with the conventional SMC and higher-order SMC approaches, the integration of fractional-order calculus with SMC offers more degrees of freedom due to adding more design parameters to the system. The fractional-order calculus adds a memory to the controller that allows it to consider the whole history of input signals, which effectively reduces the chattering phenomenon and the tracking errors of conventional SMC \cite{wang2018fractional, kumar2019fractional, hui2021chattering}.

Motivated by the above-discussed literature, although higher-order and fractional-order SMCs have successfully established themselves as desirable alternatives with outstanding performance and superiorities over other SMC-based approaches in terms of ensured fast finite-time convergence, mitigated chattering, robustness against external disturbances and model uncertainties, and higher control precision; more developments are yet to be done to achieve better WT control performances.


\section{Research Aims and Objectives}
\label{sec:1.2}
The aim of this thesis is to investigate reliable advanced control schemes to deal with the control problems of wind turbine systems in the presence of various sets of fault scenarios. This thesis is divided into three parts: a) blade pitch control in Chapter \ref{Chapter4}, b) rotor speed control and maximum power extraction in Chapter \ref{Chapter5}, and c) DFIG-driven WECSs power control through current estimation and rotor-side converter control in Chapter \ref{Chapter6}. Fractional-calculus-based higher-order sliding mode controllers are proposed to tackle various control problems of wind turbine systems in the presence of different faults. The proposed modifications in the conventional SMC aim to alleviate the chattering problem associated with the conventional SMCs, ensure fast finite-time convergence with higher control precision, and achieve higher robustness against external disturbances and faults. This main contributions of this research are outlined as follows:

\begin{itemize}

  \item Developing a fault-tolerant optimal fractional-order PID controller to adjust the pitch angle of WT blades subjected to sensor, actuator, and system faults. Benefitting from the fractional calculus, a new concept of pitch angle control with extended memory will be introduced and incorporated with the proposed optimal fractional-order PID control scheme (Chapter \ref{Chapter4}).
  \item Proposing a fractional-order nonsingular terminal sliding mode controller with enhanced finite-time convergence speed of system states and alleviated chattering to track the optimum rotor speed and maximize the power production of WECSs (Chapter \ref{Chapter5}).
  \item Proposing active fault-tolerant nonlinear control strategies for the rotor-side converter control of DFIG-driven WECSs subjected to model uncertainties and rotor current sensor faults. In this regard, two fractional-order nonsingular terminal sliding mode controllers with guaranteed fast finite-time convergence of system states and suppressed chattering will be proposed for rotor current regulation and speed trajectory tracking. Furthermore, the control scheme will be incorporated with two state observers, an algebraic state estimator and a sliding mode observer, to estimate the rotor current dynamics during sensors' faults (Chapter \ref{Chapter6}).
\end{itemize}

\section{Thesis Layout}
\label{sec:1.3}
This thesis consists of seven chapters structured as follows:

In Chapter \ref{Chapter2}, the model of wind energy conversion systems is presented in detail, that includes aerodynamics, pitch actuator system, drivetrain, and generators. This is then followed by a brief introduction to sliding mode control and some preliminary background in fractional calculus.

In Chapter \ref{Chapter3}, a comprehensive review of the state-of-the-art literature on the application of sliding mode control-based approaches with application to different control problems of wind energy conversion systems is provided.

In Chapter \ref{Chapter4}, a fault-tolerant pitch control strategy with extended memory of pitch angles is proposed that improves the power generation of wind turbines. A novel optimization algorithm, namely the dynamic weighted parallel firefly algorithm, is also proposed to tune the controller parameters optimally. Comparative simulations are provided that reveal the remarkable performance of the proposed pitch control strategy with respect to other approaches.

Chapter \ref{Chapter5} investigates the maximum power extraction problem of wind energy conversion systems operating below their rated wind speeds in the presence of actuator faults. To this end, a fractional-order nonsingular terminal sliding mode controller with alleviated chattering and enhanced finite-time convergence speed of system states is proposed to track the optimum rotor speed and maximize power production. Simulation results and analysis are provided and demonstrate the notable optimal rotor speed tracking and satisfactory power extraction performance of the proposed control scheme with fewer fluctuations and faster transient response with respect to other approaches.

In Chapter \ref{Chapter6}, aiming at rotor current regulation and speed trajectory tracking of doubly-fed induction generators-based wind turbines, two active fault-tolerant control schemes based on fractional-order nonsingular terminal sliding mode controllers are proposed. Benefitting from the proposed sliding surfaces, fast finite-time convergence of system states is guaranteed, and the chattering is effectively suppressed. The first control scheme is augmented with a well-performed algebraic state estimator to estimate the rotor current dynamics during sensors' faults. In contrast, a sliding mode observer is developed for the second control scheme to deal with the same problem. Comparative performance assessments are provided and reveal the promising performance of the proposed control scheme.

Finally, Chapter \ref{Chapter7} summarizes the contributions of the thesis and provides the conclusions and potential future works. 

\chapter{Preliminaries} 

\label{Chapter2} 



\section{Wind Turbine Modelling}
\label{sec:2.1}
In this section, a brief summary of the wind-to-electric energy conversion concept is presented. Figure \ref{fig:1} depicts the block diagram of WECS, where the aerodynamic wind energy is converted into useful mechanical energy through the blades driving the rotor. WECS can be classified into three main parts as follows \cite{bianchi2006wind}:
\begin{enumerate}
\item The mechanical part consists of the turbine blades, aerodynamic torque, and the drivetrain. This section's input is the wind kinetic energy, which is converted into rotational mechanical energy, and then transmitted to the wind generator through the drivetrain.
\item \textcolor{black}{The electrical part comprises the generator, rotor- and grid-side converters (RSC and GSC in DFIG-based WECSs) alongside their harmonic filters, and the transformer.} The rotational mechanical energy produced by the mechanical section is converted into electrical energy by the generator.
\item The control part includes various controllers to perform the desired tasks such as: (a) maximize the wind power capture, (b) provide the safe-mode operating situation for the turbine by controlling the power, current, voltage, rotational speed, and torque considering the limitations, and (c) minimize the mechanical stress on the drivetrain.
\end{enumerate}

In order to calculate the wind energy, a knowledge of the wind speed distribution at the given site is necessary. Various wind speed models are used in the literature, while the Weibull statistical distribution has found to be the most commonly used method for analyzing wind speed measurements and determining wind energy potential \cite{wais2017review}. The Weibull distribution can be expressed by the following probability density function.
\begin{equation}
f\left(\omega \right)=\frac{k}{c} \left(\frac{\omega }{c} \right)^{k-1} e^{-\left(\omega /c\right)^{k} },
\label{Equation_W1_}
\end{equation}
where the dimensionless parameter $k$ is the shape factor, $c$ denotes the scale factor, and $\omega $ is the wind speed.

The average expected wind speed is calculated as
\begin{equation}
\bar{\omega }=\int _{0}^{\infty }\omega f\left(\omega \right)d\omega =\frac{c}{k} \Gamma \left(\frac{1}{k} \right)
\label{Equation_W2_}
\end{equation}
where $\Gamma \left(t\right)=\int _{0}^{\infty }e^{-u} u^{t-1} du $ denotes Euler's Gamma function \cite{sabatier2007advances}.

The Weibull distribution is known as the Rayleigh distribution if $k=2$. Hence, for the Rayleigh distribution and considering the average wind speed, the scale factor is calculated as $c=\frac{2}{\sqrt{\pi } } \bar{\omega }$. The wind speed probability density function of the Rayleigh distribution with average wind speeds are illustrated in Figure \ref{fig:W1}. More information regarding Weibull distributions can be found in \cite{wais2017review}.

\begin{figure*}
\centering
\includegraphics[width=4.5 in]{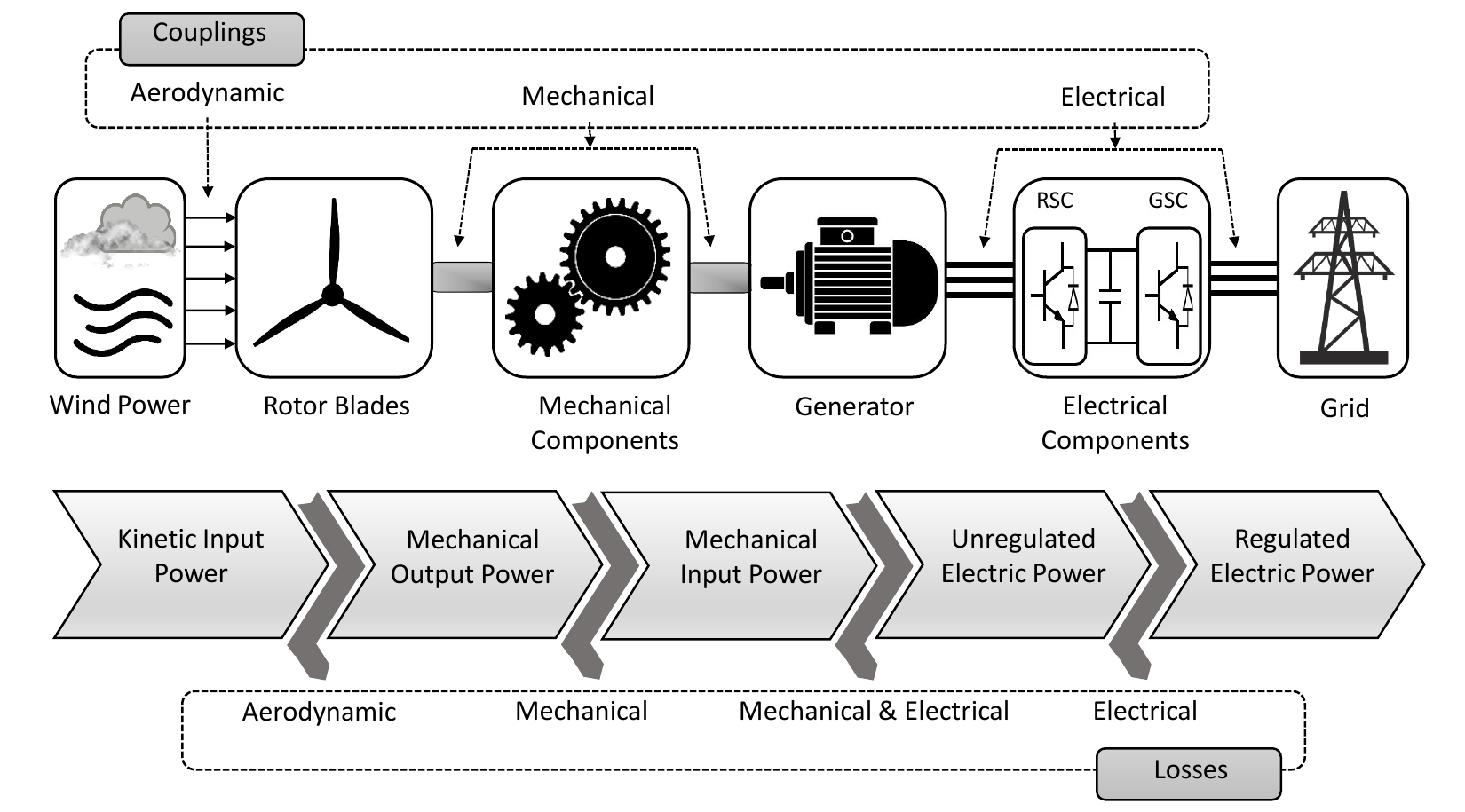}
\caption{\textcolor{black}{General block diagram of wind energy conversion systems, along with their coupling types \cite{yaramasu2016model}.}}
\label{fig:1}
\end{figure*}

\subsection{Aerodynamics}
\label{sec:2.1.1}
This section represents the conversion of wind power to rotational energy, where the aerodynamic performance of WECS is affected by the wind speed, pitch angles of the blades, and the rotor speed. The aerodynamic power ($P_{a} $), aerodynamic torque ($T_{a} $), and thrust force ($F_{t} $) of the turbine can be expressed as follows:
\begin{equation}
P_{a} =\frac{1}{2} \rho \pi R^{2} \upsilon _{w} ^{3} C_{P} \left(\lambda ,\beta \right),
\label{Equation_1_}
\end{equation}
\begin{equation}
\textcolor{black}{T_{a} =\frac{1}{2} \rho \pi R^{3} \upsilon _{w} ^{2} C_{q} \left(\lambda ,\beta \right),}
\label{Equation_2_}
\end{equation}
\begin{equation}
F_{t} =\frac{1}{2} \rho \pi R^{2} \upsilon _{w} ^{2} C_{t} \left(\lambda ,\beta \right),
\label{Equation_3_}
\end{equation}
where $\rho=1.225$ [\si{kg/m^{3}}] denotes the air density at sea level for the International Standard Atmosphere (ISO 2533:1975), $R$ [\si{m}] is the rotor radius, and $\upsilon _{w} $ [\si{m/s}] is the effective wind model at the rotor plane. $C_{p} $, $C_{q} $, and $C_{t} $ represent the power, torque, and thrust coefficients, respectively, where $C_{q} \left(\lambda ,\beta \right)=C_{p} \left(\lambda ,\beta \right)/\lambda $.

\begin{figure*}
\centering
\includegraphics[width=3.2 in]{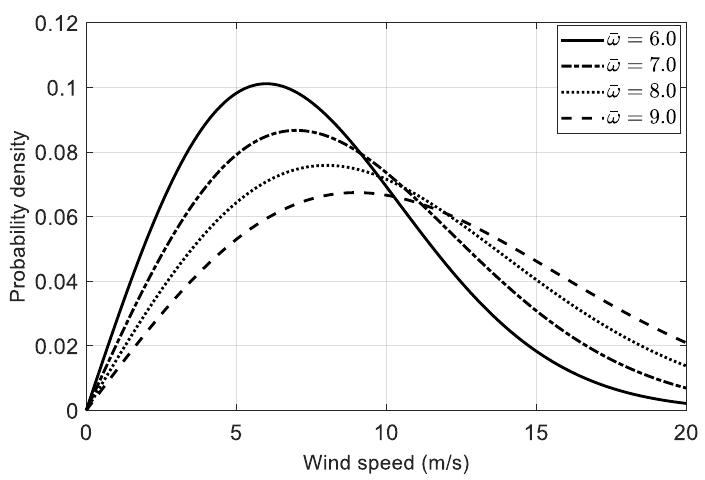}
\caption{Probability density of the Rayleigh distribution.}
\label{fig:W1}
\end{figure*}

The aerodynamic torque produced by the WECS can be expressed by $T_{a} =P_{a} /\omega _{r} $, where $\omega _{r}$ [\si{rad/s}] represents the rotational speed of the rotor. The dimensionless power coefficient $C_{p} $ is a function of the blade parameters, the tip-speed ratio ($\lambda =R\omega _{r} /\upsilon _{w} $) and the pitch angle $\beta $; an experimental coefficient that is generally adopted from a look-up table in terms of ($\lambda $,$\beta $). However, using nonlinear curve-fitting techniques \cite{abolvafaei2020maximum, deraz2013new, sharmila2019fuzzy} it can be approximated as mathematical functions such as \eqref{Equation_4_} and \eqref{Equation_5_} as follows.
\begin{equation}
C_{P} \left(\lambda ,\beta \right)=C_{1} \left(\frac{C_{2} }{\lambda _{i} } -C_{3} \beta -C_{4} \right)e^{-C_{5} /\lambda _{i} } +C_{6} \lambda,
\label{Equation_4_}
\end{equation}
\begin{equation}
C_{P} \left(\lambda ,\beta \right)=C_{1} -C_{2} \beta \sin \left(\frac{\pi \left(\lambda _{i} +C_{3} \right)}{C_{4} +C_{5} \beta } \right)-C_{6} \left(\lambda _{i} -3\right)\beta,
\label{Equation_5_}
\end{equation}
where
\begin{equation}
\frac{1}{\lambda _{i} } =\frac{1}{\lambda +C_{7} \beta } -\frac{C_{8} }{\beta ^{3} +1},
\label{Equation_6_}
\end{equation}
In order to capture the maximum power from the wind, the power coefficient $C_{p} $ should be obtained based on the optimum pitch angle $\beta _{opt} $ and the optimum tip speed ratio $\lambda _{opt} $, i.e. $C_{p,\max } \triangleq C_{p} \left(\beta _{opt},\lambda _{opt} \right)$, where $C_{p,\max } $ denotes the maximum value of the power coefficient for the turbine. Setting $ C_{1}=0.5176$, $ C_{2}=116$, $ C_{3}=0.4$, $ C_{4}=5$, $ C_{5}=21$, $ C_{6}=0.0068$ \cite{mousavi2021fault}, the variations of power coefficient $C_{P}$ \eqref{Equation_4_} for different values of $\lambda$ and $\beta$ can be demonstrated as Figure \ref{fig:1-4}. According to Figure \ref{fig:1-4} it can be observed that the WT has a maximum efficiency of approximately $C_{P,\max }=0.4797$ for a tip speed ratio $\lambda _{opt}=8.2$.
\begin{figure}
\centering
\includegraphics[width=3.5 in]{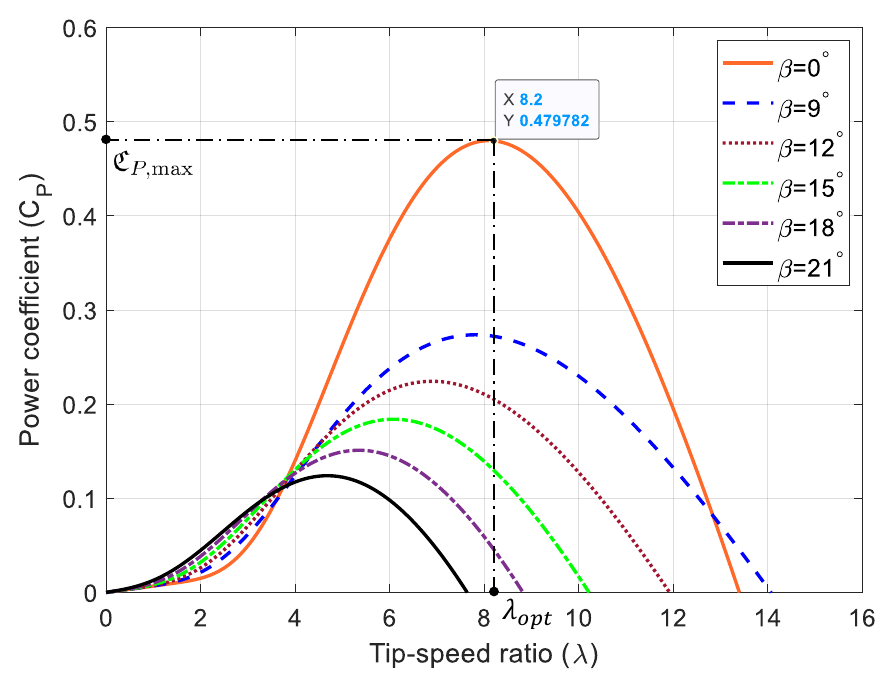}
\caption{The power coefficient curve based on the optimal tip-speed ratio for different pitch angles.}
\label{fig:1-4}
\end{figure}

The theoretical value of $C_{p,\max } $ cannot exceed $0.593$, which is known as the Betz limit \cite{betz1920maximum}, \textcolor{black}{and means that the turbine can extract maximum $59.3\%$ of the available wind power;} however, this value cannot be reached in practice \cite{ragheb2011wind}.

\subsection{Pitch Actuator System}
\label{sec:2.1.2}
The main objective in region III is to maintain the generated power around the rated values while limiting structural loads to avoid mechanical and electrical constraints violations. The pitch system consists of three identical pitch actuators that are assumed to have the same dynamic structure, where each one can be individually controlled using its associated internal controller. The pitch actuator provides the rotational movement of each blade around its longitudinal axis and adjusts the pitch angle, the angle between the blade string and the blade rotation plane. The hydraulic pitch system is generally modelled as a closed-loop transfer function between the measured pitch angle $\beta$ [$^{\si{\circ}}$] and its reference $\beta_{ref}$ [$^{\si{\circ}}$]. Considering only one pitch actuator, the pitch system is usually modelled by the following first-order and second-order systems \cite{anaya2011wind,colombo2020pitch}.
\begin{subequations}
\label{Equation_611}
\begin{align}
\dot{\beta}\left(t\right)&=\tau^{-1}\Big(\beta_{ref}\left(t\right)-\beta\left(t\right)\Big) \,\,\, \left[^{\circ}/\si{s}\right], \\
\ddot{\beta}\left(t\right)&=-2\zeta\omega_n \dot{\beta}\left(t\right)-\omega_n^2 \beta\left(t\right)+\omega_n^2\beta_{ref}\left(t-t_d\right) \,\,\, \left[^{\circ}/\si{s^2}\right],
\end{align}
\end{subequations}
where $\zeta$ and $\omega _{n}$ [\si{rad/s}] denote the damping ratio and the natural frequency of the pitch actuator model, respectively, $t_d$ represents the communication delay to the pitch actuator in [\si{s}], and $\tau$ is the pitch actuator time constant.

According to traditional pitch control schemes, the foregoing objectives are accomplished by maintaining the generator torque constant while the collective pitch control regulates the generator speed to the rated value. As a baseline controller, the standard collective gain scheduling PI controller has been broadly employed in the literature to deal with the pitch control problem and cope with the nonlinearities caused by the pitch actuation mechanism. However, other advanced blade pitch control approaches have been widely investigated in the literature to overcome the deficiencies associated with the standard PI approach and deliver more desirable control performance, such as fractional-order PID with extended pitch memory \cite{mousavi2021fault}, high-order SMC \cite{aghaeinezhad2021individual}, MPC \cite{wakui2021stabilization}, fuzzy-based control \cite{van2015advanced}, and adaptive neural network control \cite{jiao2019adaptive}.

\subsection{Drivetrain}
\label{sec:2.1.3}
The mechanical part of WT, namely the drivetrain, is an intricate part comprising the low-speed shaft, high-speed generator shaft, gearbox, mechanical coupling sections, bearings, and mechanical brakes. The drivetrain provides the generator's required rotational speed by converting the high torque on the low-speed shaft to a low torque on the high-speed shaft to be transferred to the generator unit. Various $n$-mass drivetrain models have been studied in the literature \cite{muyeen2007comparative, mahmoud2019adaptive, xu2011influence, mandic2012active}, $n=1,2,...,6$. The one-mass model is straightforward, where all the drivetrain components are lumped together and work as a single rotating mass. Accordingly, when the focus is on the study of the interactions between wind farms and AC grids, it can be used for the sake of time efficiency \cite{yin2007modeling}; however, due to neglecting the dynamics of the flexible shaft, the one-mass model cannot represent the mechanical oscillations properly and thus, may result in unstable mechanical modes \cite{muyeen2007comparative}. In the two-mass model, the generator and gearbox masses are lumped together. In this model, the turbine's low-speed mass is connected to the high-speed mass of the generator through a flexible shaft. Thus, this model has been found to be acceptably accurate for transient stability studies of wind energy systems \cite{shubhanga2017rotor,prajapat2017stability,schechner2018scaling}. On the other hand, the three-mass model is more accurate with a higher transient stability than that of the two-mass model, but it requires more computational effort. In this model, the blades and shaft flexibilities are considered, where the blades are considered as a single mass element separated from the hub by an elastic component \cite{zhang2014small}. More accurate dynamic behavior of WT drivetrain can be achieved by increasing the number of masses \cite{muyeen2008stability, mahmoud2019adaptive, seixas2014offshore}. Although higher-mass models have been found to be more appropriate for transient stability analysis, especially during fault conditions, they slow the simulations due to complex and lengthy mathematical computation. Thus, for the sake of simplicity, the lower-mass drivetrain models are mostly taken into consideration by researchers for analyzing the transient behavior of WTs. A brief review on dynamic models of four-, three-, two-, and one-mass are presented in this section. The four-mass drivetrain dynamic model shown in Figure \ref{fig:501} can be represented by \cite{mahmoud2019adaptive}:
\begin{figure}[!b]
\centering
\includegraphics[width=4 in]{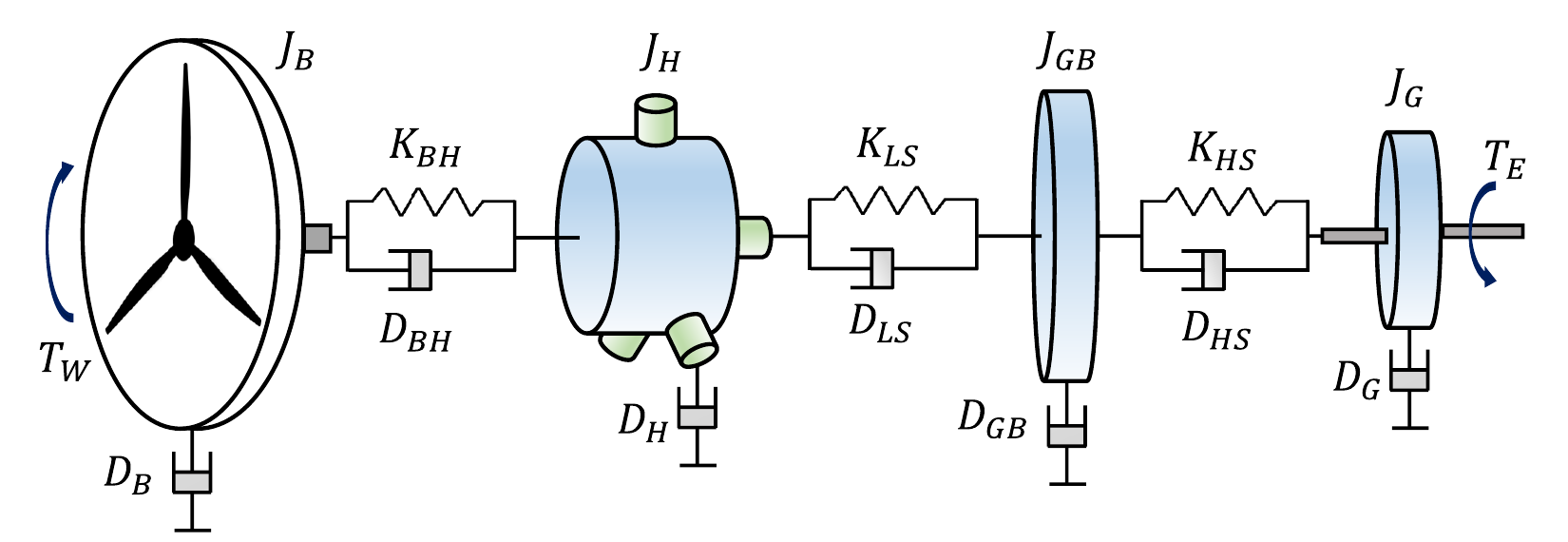}
\caption{\textcolor{black}{Schematic diagram of a four-mass drivetrain.}}
\label{fig:501}
\end{figure}
\begin{subequations}
\label{Equation_7_}
\allowdisplaybreaks
\begin{align}
2J_{B} \dot{\omega }_{B} &= T_{W} -K_{BH} \theta _{BH} -D_{BH} \left(\omega _{B} -\omega _{H} \right)-D_{B} \omega _{B}, \\
2J_{H} \dot{\omega }_{H} &= K_{BH} \theta _{BH} +D_{BH} \left(\omega _{B} -\omega _{H} \right) \nonumber \\
&-K_{LS} \theta _{LS} -D_{LS} \left(\omega _{H} -\omega _{GB} \right)-D_{H} \omega _{H}, \\
2J_{GB} \dot{\omega }_{GB} &= K_{LS} \theta _{LS} +D_{LS} \left(\omega _{H} -\omega _{GB} \right) \nonumber \\
&-\left(K_{HS} \theta _{HS} +D_{HS} \left(\omega _{GB} -\omega _{G} \right)\right)D_{GB} \omega _{GB}, \\
2J_{G} \dot{\omega }_{G} &= K_{HS} \theta _{HS} +D_{HS} \left(\omega _{GB} -\omega _{G} \right)-D_{G} \omega _{G} -T_{G}, \\
\dot{\theta }_{BH} &= \omega _{0} \left(\omega _{B} -\omega _{H} \right), \\
\dot{\theta }_{LS} &= \omega _{0} \left(\omega _{H} -\omega _{GB} \right), \\
\dot{\theta }_{HS} &= \omega _{0} \left(\omega _{GB} -\omega _{G} \right),
\end{align}
\end{subequations}
where $J_{B} $, $J_{G} $, $J_{GB} $, and $J_{H} $ denote the blade, generator, gearbox, and hub inertia in [\si{kg.m^{2}}], respectively; $K_{LS} $ and $K_{HS} $ are the low- and high-speed shaft spring constants in [\si{N.m/rad}], respectively, and $K_{BH} $ [\si{N.m/rad}] represents the blade stiffness constant. $\theta _{BH} $ [\si{^{\circ}}] denotes the angle between the blade disk and the hub, $\theta _{LS} $[\si{^{\circ}}] is the angle between the hub and gearbox, and $\theta _{HS} $ [\si{^{\circ}}] is the angle between the gearbox and generator rotor. $D_{B} $ [\si{N.s/rad}] is the blade damping coefficient, and $D_{H} $, $D_{BH} $, $D_{LS} $, $D_{HS} $, $D_{G} $, and $D_{GB} $ represent the damping coefficients of hub, turbine, low-speed shaft, high-speed shaft, generator, and gearbox in [\si{N.m.s/rad}], respectively. $\omega _{B} $, $\omega _{H} $, $\omega _{GB} $, and $\omega _{G} $ respectively represent the rotational speeds of the flexible blade disk, hub, gearbox, and generator in [\si{rad/s}]. $\omega _{0} $ [\si{rad/s}] is the synchronous speed of the electrical system, and $T_{G} $ and $T_{W} $ represent the generator and aerodynamic torques in [\si{N.m}], respectively.

The three-mass drivetrain dynamic model illustrated in Figure \ref{fig:502} can be represented by \cite{xu2011influence}:
\begin{figure}[!b]
\centering
\includegraphics[width=3.8 in]{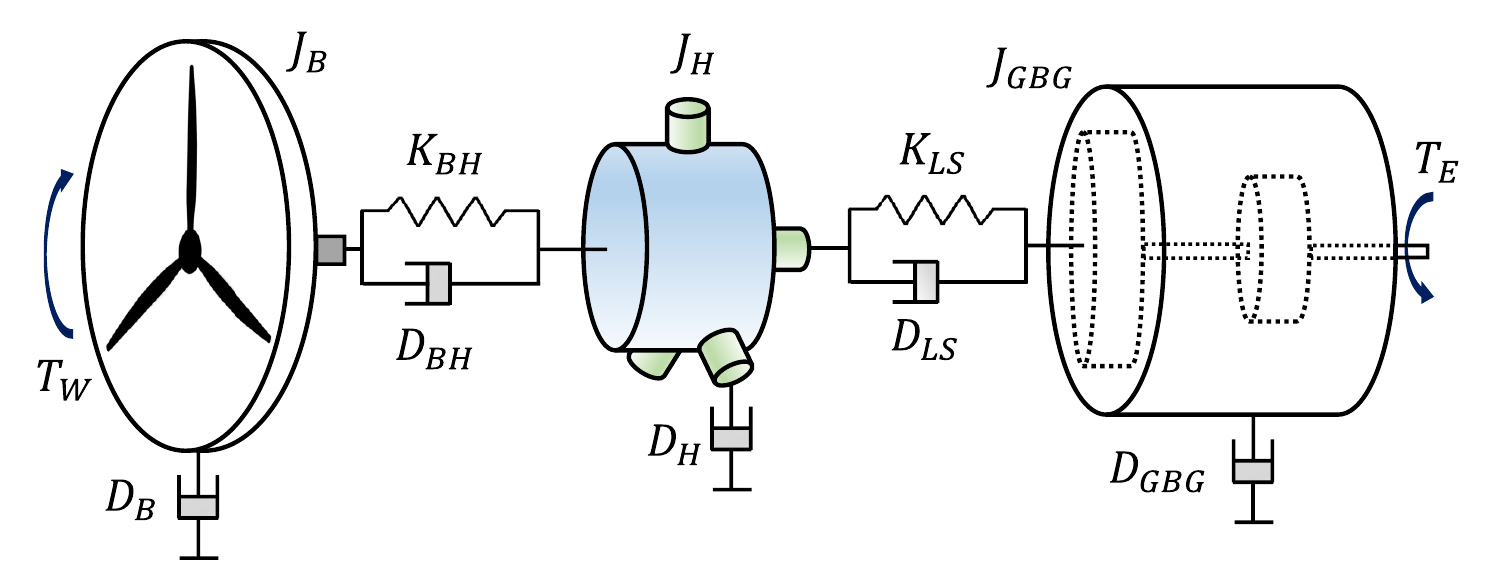}
\caption{\textcolor{black}{Schematic diagram of a three-mass drivetrain.}}
\label{fig:502}
\end{figure}
\begin{subequations}
\label{Equation_8_}
\begin{align}
2J_{B} \dot{\omega }_{B} &=T_{W} -K_{BH} \theta _{BH} -D_{BH} \left(\omega _{B} -\omega _{H} \right)-D_{B} \omega _{B}, \\
2J_{H} \dot{\omega }_{H} &=K_{BH} \theta _{BH} +D_{BH} \left(\omega _{B} -\omega _{H} \right)-K_{LS} \theta _{LS} \nonumber \\
&-D_{LS} \left(\omega _{H} -\omega _{GBG} \right)-D_{H} \omega _{H}, \\
2J_{GBG} \dot{\omega }_{GBG} &=K_{LS} \theta _{LS}+D_{LS} \left(\omega _{H} -\omega _{GBG} \right)-D_{GBG} \omega _{GBG} -T_{G}, \\
\dot{\theta }_{BH} &=\omega _{0} \left(\omega _{B} -\omega _{BH} \right), \\
\dot{\theta }_{LS} &=\omega _{0} \left(\omega _{H} -\omega _{GBG} \right),
\end{align}
\end{subequations}
where $J_{BH}$ is the summation of blade and hub inertia, $J_{GBG} $ [\si{kg.m^{2}}] represents the summation of gearbox and generator inertia, $D_{GBG} $ [\si{N.m.s/rad}] is the gearbox and generator damping coefficient, $\omega _{BH} $ and $\omega _{GBG} $ are the turbine rotational speed, and the summation of gearbox and generator rotational speeds in [\si{rad/s}], respectively.

The two-mass drivetrain dynamic model depicted in Figure \ref{fig:503} can be represented by:
\begin{figure}[!hb]
\centering
\includegraphics[width=2.7 in]{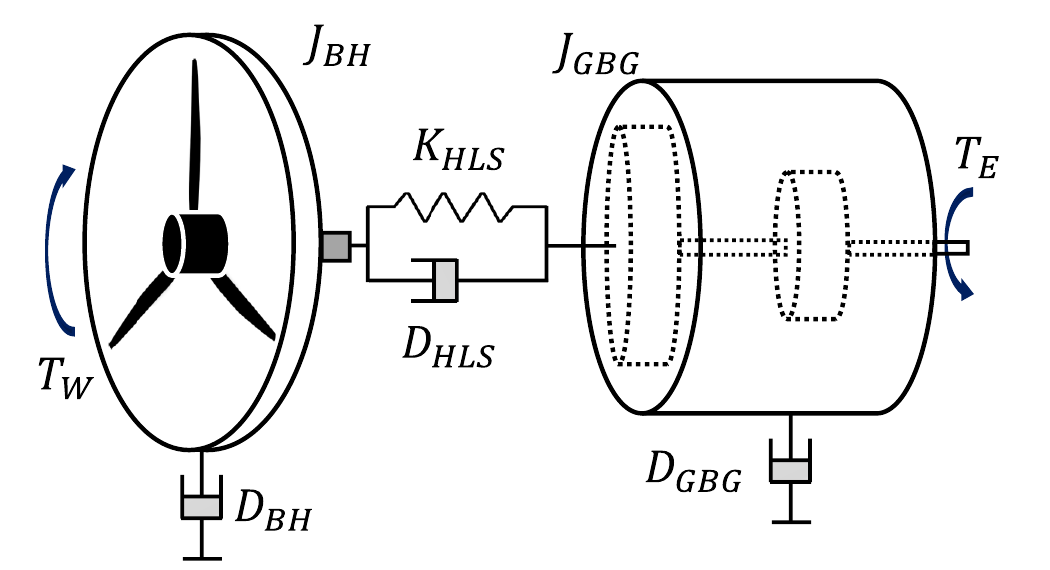}
\caption{\textcolor{black}{Schematic diagram of a two-mass drivetrain.}}
\label{fig:503}
\end{figure}
\begin{subequations}
\label{Equation_9_}
\begin{align}
2J_{BH} \dot{\omega }_{BH} &=T_{W} -K_{HLS} \theta _{HLS} -D_{HLS} \left(\omega _{BH} -\omega _{GBG} \right)-D_{BH} \omega _{BH}, \\
2J_{GBG} \dot{\omega }_{GBG} &=K_{HLS} \theta _{HLS} +D_{HLS} \left(\omega _{BH} -\omega _{GBG} \right)-D_{GBG} \omega _{GBG} -T_{G}, \\
\dot{\theta }_{HLS} &=\omega _{0} \left(\omega _{BH} -\omega _{GBG} \right),
\end{align}
\end{subequations}
where $\theta _{HLS} $ [\si{^{\circ}}] is the angle between $J_{BH} $ and $J_{GBG} $. $K_{HLS} $ [\si{N.m/rad}] represents the parallel shaft stiffness, expressed as
\begin{equation}
K_{HLS} =\left(\frac{N_{GB}^{2} }{K_{HS} } +\frac{1}{K_{LS} } \right)^{-1},
\label{Equation_10_}
\end{equation}
where $N_{GB} $ denotes the gearbox ratio.

The one-mass drivetrain dynamic model shown in Figure \ref{fig:504} can be represented by:
\begin{figure}[!h]
\centering
\includegraphics[width=2 in]{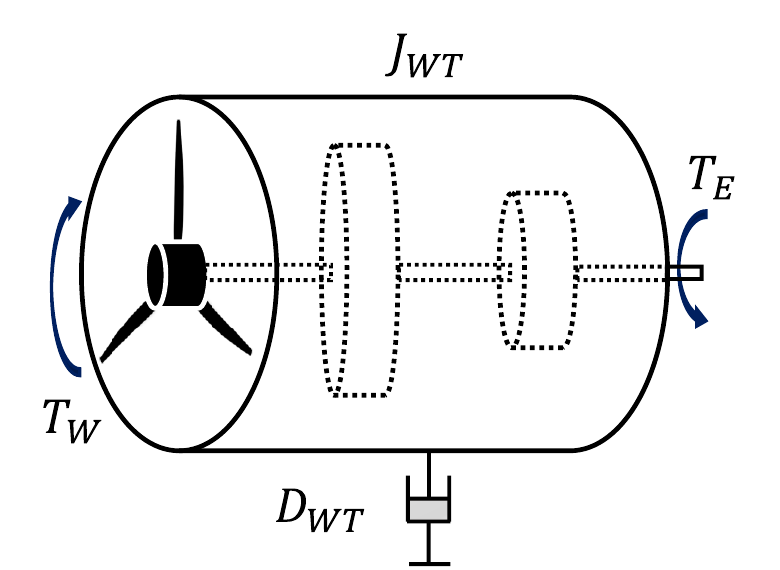}
\caption{Schematic diagram of a one-mass drivetrain.}
\label{fig:504}
\end{figure}
\begin{equation}
2J_{WT} \dot{\omega }_{WT} =T_{W} -T_{G} -D_{WT} \omega _{WT},
\label{Equation_11_}
\end{equation}
where $J_{WT} $ [\si{kg.m^{2}}] represents all the inertia of rotating components, $\omega _{WT} $ [\si{rad/s}] is the rotational speed, and $D_{WT} $ [\si{N.m.s/rad}] denotes the damping coefficient of the single mass.

\subsection{Generators}
\label{sec:2.1.4}
The generator converts the mechanical wind energy into electrical energy. Various generators have been developed for WECS, where the most used ones are doubly fed induction generator (DFIG), permanent-magnet synchronous generator (PMSG), and squirrel-cage induction generator (SCIG) \cite{anaya2011wind, alnasir2013analytical}. The DFIG has established itself as the most commonly used generator thanks to advantages such as independent control of active and reactive power using partial capacity converters, low installation costs, and the flexibility to operate in both sub- and super-synchronous speeds with the state-of-the-art improvements in power electronic devices and corresponding controllers \cite{morshed2019sliding}. However, since the stator is directly connected to the grid, DFIGs are highly sensitive to voltage fluctuations and grid disturbances \cite{li2019sliding}, in the sense that any amount of voltage dip can lead to a sharp increase in stator and rotor currents and damage the converters, which can result in deterioration of WT's energy conversion process \cite{yao2012enhanced}. Besides, due to the gearbox requirement, the WT's overall size is increased, and regular maintenance is also necessary. On the other hand, PMSG is preferred in offshore applications due to the self-excitation system and the elimination of the gearbox, which results in reduced mechanical losses, enhanced reliability, and lower maintenance costs. In addition, being capable of producing higher torque at low speeds, they can provide the maximum electric power to the grid with improved quality \cite{yang2018passivity}. However, PMSGs can also be impacted by grid side voltage sag, which can result in ripples with rising magnitude in the DC-link voltage due to unbalanced power flow between the generator and the GSC and also rising magnitude of the injected currents beyond the safe operating limit of the inverter \cite{thakur2018control}. Besides, the magnets' characteristics tend to change when subjected to high currents or temperatures over time \cite{errouissi2017novel}. Similar to PMSGs, SCIGs are counted as good solutions in isolated power systems, as they are relatively inexpensive and require minimal maintenance \cite{chen2011analysis}.

\subsubsection{DFIG}
\label{sec:2.1.4.1}
\textcolor{black}{The dynamic rotor and stator currents equivalent circuit equations of DFIG in the synchronous $dq$ reference frame shown in Figure \ref{fig:2_222} can be expressed as follows \cite{shi2019perturbation}:}
\begin{subequations}
\label{Equation_12}
{\color{black}\begin{align}
V_{ds}&=R_s I_{ds}+\frac{d}{dt}\psi_{ds}-\omega_s\psi_{qs}, \\
V_{qs}&=R_s I_{qs}+\frac{d}{dt}\psi_{qs}-\omega_s\psi_{ds}, \\
V_{dr}&=R_r I_{ds}+\frac{d}{dt}\psi_{dr}-\left(\omega_s-\omega_r\right)\psi_{qr}, \\
V_{dr}&=R_r I_{qs}+\frac{d}{dt}\psi_{qr}+\left(\omega_s-\omega_r\right)\psi_{dr},
\end{align}}
\end{subequations}

The flux equations of stator and rotor are given as
\begin{subequations}
\label{Equation_13_}
\begin{align}
\psi _{ds} &=L_{s} I_{ds} +L_{m} I_{dr}, \\
\psi _{qs} &=L_{s} I_{qs} +L_{m} I_{qr}, \\
\psi _{dr} &=L_{r} I_{dr} +L_{m} I_{ds}, \\
\psi _{qr} &=L_{r} I_{qr} +L_{m} I_{qs},
\end{align}
\end{subequations}
where $\left\{I_{ds} ,I_{qs} \right\}$ and $\left\{I_{dr} ,I_{qr} \right\}$ are the d-axis and q-axis stator and rotor current components in [A], respectively, and $\left\{V _{ds} ,V _{qs} \right\}$ and $\left\{V _{dr} ,V _{qr} \right\}$ represent the d-axis and q-axis stator and rotor voltage components in [V]. $R_{s} $ and $R_{r} $ are the stator and rotor resistances in [$\Omega $], $L_{s} $, $L_{r} $, and $L_{m} $ are stator, rotor, and magnetizing inductances in [H]. $\left\{\psi _{ds} ,\psi _{qs} \right\}$ and $\left\{\psi _{dr} ,\psi _{qr} \right\}$ are the stator and rotor flux components in [Wb], respectively. $\omega _{r} =P\varsigma _{r} $ and $\omega _{s} $ represent the rotor electrical speed and synchronous frequency, respectively, where $\varsigma _{r} $  [\si{rad/s}] is the rotor speed and $N_P$ is the number of pole pairs.

\begin{figure}
\centering
\includegraphics[width=4.2 in]{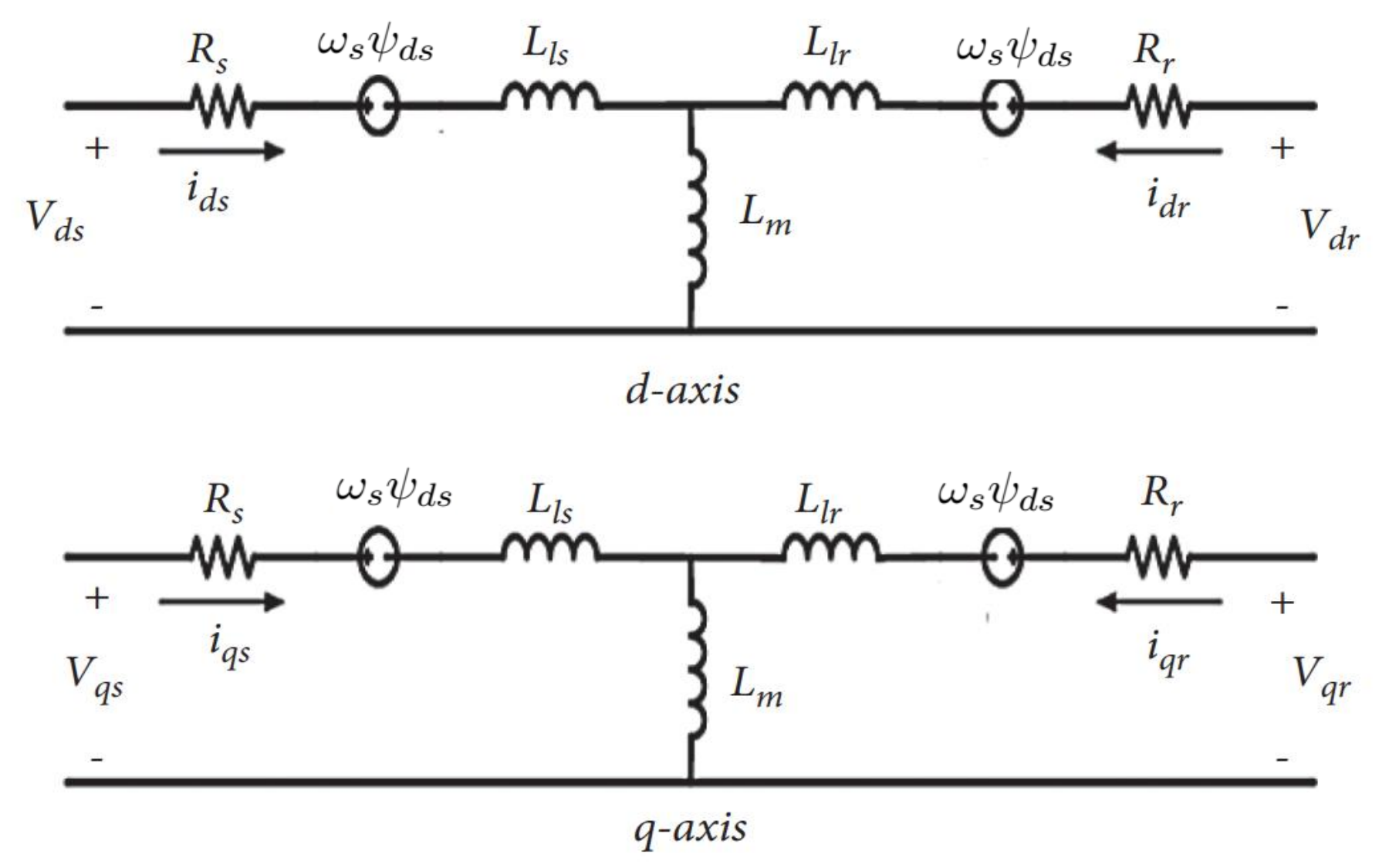}
\caption{\textcolor{black}{The DFIG equivalent circuit}.}
\label{fig:2_222}
\end{figure}

The electromagnetic torque $T_{em} $ can be expressed with stator fluxes and currents as
\begin{equation}
T_{em} =\frac{3P}{2} \left(\psi _{ds} I_{qs} -\psi _{qs} I_{ds} \right).
\label{Equation_14_}
\end{equation}

Neglecting the power losses associated with the rotor, the stator and rotor active and reactive power can be expressed as
\begin{subequations}
\label{Equation_15_}
\begin{align}
P_{s} =-V _{ds} I_{ds} -V _{qs} I_{qs}, \\
Q_{s} =-V _{qs} I_{ds} +V _{ds} I_{qs}, \\
P_{r} =-V _{dr} I_{dr} -V _{qr} I_{qr}, \\
Q_{r} =-V _{qr} I_{dr} +V _{dr} I_{qr}.
\end{align}
\end{subequations}

\begin{rem}
\label{rmk:1}
Setting $V _{dr} =V _{qr} =0$ in the DFIG dynamic model leads to the dynamic model of the SCIG \cite{mahmoud2019adaptive}.
\end{rem}

\subsubsection{PMSG}
\label{sec:2.1.4.2}
The dynamic currents equations of PMSG in the synchronous $dq$ reference frame is expressed as follows \cite{zhang2015discrete}:
\begin{subequations}
\label{Equation_16_}
\begin{align}
\frac{dI_{ds} }{dt} &=-\frac{R_{s} }{L_{d} } I_{ds} +\omega _{e} \frac{L_{q} }{L_{d} } I_{qs} +\frac{1}{L_{d} } V _{ds}, \\
\frac{dI_{qs} }{dt} &=-\frac{R_{s} }{L_{q} } I_{qs} -\omega _{e} \frac{L_{d} }{L_{q} } I_{ds} -\omega _{e} \psi _{s} +\frac{1}{L_{q} } V _{qs},
\end{align}
\end{subequations}
where, $\left\{I_{ds} ,I_{qs} \right\}$ [\si{A}] and $\left\{V _{ds} ,V _{qs} \right\}$ [\si{V}] denote the stator currents and voltages, respectively, $R_{s} $ [$\Omega $] is the stator resistance, $\omega _{e} $ [\si{rad/s}] is the rotor electrical angular speed, $L_{d} $ and $L_{q} $ represent the \textit{d--q} axis mutual inductance [\si{H}], and $\varphi _{s} $ [\si{Wb}] is the flux linkage generated by the permanent magnets.

The stator flux equations are given as
\begin{subequations}
\label{Equation_17_}
\begin{align}
\psi _{ds} &=L_{d} I_{ds} +\psi _{m}, \\
\psi _{qs} &=L_{q} I_{qs}.
\end{align}
\end{subequations}

The electromagnetic torque $T_{em} $ generated by the PMSG can be calculated by
\begin{equation}
\label{Equation_18_}
\left\{\begin{array}{ll} {T_{em} =\frac{3P}{2} \psi _{m} I_{qs} +\frac{3P}{2} \left(L_{d} -L_{q} \right)I_{ds} I_{qs} } & {\textcolor{black}{\mathrm{if}}\,\, L_{d} \ne L_{q} } \\ {T_{em} =\frac{3P}{2} \psi _{m} I_{qs} } & {\textcolor{black}{\mathrm{if}}\,\, L_{d} =L_{q} } \end{array}\right.
\end{equation}

\section{Sliding Mode Control}
\label{sec:2.2}

Sliding mode control is one of the most effective control methods under high uncertainty conditions \cite{shtessel2014sliding, edwards1998sliding}. The sliding mode strategy is based on the discontinuous control signal design driving the system states toward predesigned surfaces in state space. SMC design comprises two steps: (i) the design of a stable sliding surface to obtain the desired dynamics and (ii) the design of a control law that ensures the reaching of the states to the chosen sliding surface in finite time and staying on it. Essentially, in a system controlled by an SMC, state dynamics consist of two modes: reaching mode and sliding mode \cite{edwards1998sliding}. In the reaching mode, the controller forces the states to reach the sliding surface, and as all the states reach the sliding surface, the sliding mode occurs \cite{shtessel2014sliding}. Let us consider the nonlinear system represented by the following state equation:

\begin{equation}
\mathbf{x}^{(n)} =\mathbf{f}\left(x,t\right)+\mathbf{g}\left(x,t\right)u\left(t\right),
\label{Equation_S1_}
\end{equation}
where $\mathbf{x}$ is an \textit{n}-dimensional column state vector, $\mathbf{f}$ and $\mathbf{g}$ are \textit{n}-dimensional continuous functions in $x$ and $t$ vector fields, and $u$ is the control input.

Generally, the time-varying sliding surface of the \textit{n}-order system is selected as follows:
\begin{equation}
\mathbf{S}\left(\mathbf{x},t\right) = \mathbf{C}^T\mathbf{x}=\sum_{i-1}^{n}c_i x_i=\sum_{i=1}^{n-1}c_i x_i+x_n,
\label{Equation_S2_}
\end{equation}
where $\mathbf{C}=\left[c_1, c_2,...,c_{n-1},1\right]^T$ should be selected so that the polynomial $\mathfrak{p}^{n-1}+c_{n-1}\mathfrak{p}^{n-2}+...+c_2\mathfrak{p}+c_{1}$ is Hurwitz, where $\mathfrak{p}$ represents the Laplace operator.

The control law should meet the sufficiency conditions for a sliding mode's existence and reachability such that $\mathbf{S}\dot{\mathbf{S}}\leq {0}$. The existence of a sliding mode on the sliding surface implies stability of the system \cite{choi2007lmi,roh1999robust}. For the considered system, the control input consists of two components; a continuous component $u_{eq}$ and a discontinuous one $u_{n}$ \cite{slotine1984sliding}.
\begin{equation}
u\left(t\right) = u_{eq}\left(t\right)+u_{n}\left(t\right),
\label{Equation_S3_}
\end{equation}

The continuous component ensures the system's movement on the sliding surface whenever the system is on the surface. This component maintains the sliding mode and satisfies the condition $\dot{\mathbf{S}}=0$, resulting in the convergence of the output to $\mathbf{x}_d$. On the other hand, the discontinuous component is responsible for forcing the states with arbitrary initial values toward the sliding surface in a finite time, so that the state trajectory is attracted to the sliding surface $\mathbf{S}\left(x,t\right)=0$. It is worth mentioning that, the discontinuous component $u_n$ normally includes the $sgn\left(\cdot\right)$ function, which causes undesirable chattering. In this context, many studies use the boundary layer methods and replace the $sgn\left(\cdot\right)$ function  with the saturation function $sat\left(\cdot\right)$ or $tanh\left(\cdot\right)$ function to restrain the chattering phenomenon. To ensure an attractive and an invariant sliding surface, a positive definite Lyapunov function $V\left(x\right)$ is defined, where in order to provide the asymptotic stability about the equilibrium point, the derivative of $V\left(x\right)$ with respect to time must be negative definite $\dot {V}\left(x\right)<0$.

\section{Fractional-order Calculus}
\label{sec:2.3}

This section presents a brief overview of fractional-order calculus. To this end, first, some preliminary definitions are expressed and addressed. Then, some discussions on the stability and convergence of fractional-order systems are provided to address its effects on the chattering problem in sliding mode control.

\subsection{Preliminaries on Fractional-order Calculus}
\label{sec:2.3.1}
Fractional-order calculus extends the integer-order calculus concept, demonstrating the hereditary characteristics of processes and describing an infinite memory of the terms \cite{sabatier2007advances}. The general representation of the fractional-order integrator and differentiator can be expressed as ${}_{t_{0} } \mathfrak D_{t}^{\gamma } $, where $t_{0} $ is the initial time, and $\gamma $ is the fractional order, while $\gamma >0$  and $\gamma <0$  denote the integration and differentiation characteristics, respectively. Several definitions for fractional differential operators are investigated in the literature such as Grunwald--Letnikov (G-L), Riemann--Liouville (R-L), Caputo, Atangana–Baleanu, etc \cite{machado2011recent}. As an example, the Riemann-Liouville definition can be expressed as follows:
\begin{equation}
\label{GrindEQ__F1_}
{}_{t_{0} } \mathfrak D_{t}^{\gamma } f\left(t\right)=\frac{d^{\gamma } f\left(t\right)}{dt^{\gamma } } =\frac{1}{\Gamma \left(n-\gamma \right)} \frac{d^{n} }{dt^{n} } \int _{t_{0} }^{t}\frac{f\left(\tau \right)}{\left(t-\tau \right)^{\gamma -n+1} } d\tau  ,\, \, \, t>t_{0}
\end{equation}
where $n\in {\mathbb N}$\textit{ }is the first integer, and $n-1\le \gamma <n$. It is worth noting that since the R-L derivative \eqref{GrindEQ__F1_} performs the derivation operation after the integration, the fractional derivation of a constant number is not zero.

The Laplace transformation of the fractional-order derivation is given as
\begin{equation}
\label{GrindEQ__F2_}
\int _{0}^{\infty }{}_{0} \mathfrak D_{t}^{\gamma }  f\left(t\right)e^{-st} dt=S^{\gamma } L\left\{f\left(t\right)\right\}-\sum _{\kappa =0}^{n-1}S^{\kappa } {}_{0} \mathfrak D_{t}^{\gamma -\kappa -1} f\left(t\right)\mathop{|}\nolimits_{t=0}
\end{equation}

The Oustaloup recursive approximation algorithm is applied to approximate the fractional orders by integer-order transfer function and synthesize fractional operators in the frequency domain by a recursive distribution of zeros and poles as follows:
\begin{equation}
\label{GrindEQ__F3_}
S^{\gamma } \approx K\prod _{n=-N}^{N}\frac{1+\left(S/\omega _{z,n} \right)}{1+\left(S/\omega _{p,n} \right)}, \, \,  {\gamma >0}
\end{equation}
\begin{subequations} \label{GrindEQ__F4_}
\begin{align}
\omega _{z,n} =\omega _{b} \left(\frac{\omega _{h} }{\omega _{b} } \right)^{\left(n+N+\frac{1-\gamma }{2}/{2N+1}\right)} \\
\omega _{p,n} =\omega _{b} \left(\frac{\omega _{h} }{\omega _{b} } \right)^{\left(n+N+\frac{1+\gamma }{2}/{2N+1}\right)}
\end{align}
\end{subequations}
where the number of poles and zeros is $2N+1$, $K=\omega _{h}^{\gamma } $ is the adaptive gain, and $\omega _{b} $\textit{ }and $\omega _{h} $\textit{ }represent the lower and upper constraints of the approximation frequency, respectively. More information regarding fractional-order calculus can be found in {\cite{sabatier2007advances}}.

\subsection{Convergence of Fractional-order Systems}
\label{sec:2.3.2}
In order to investigate the characteristics of fractional-order systems along with their effects on the chattering phenomenon in SMCs, the stability of these systems should be first studied. In this regard, let us consider a fractional-order system defined as follows:
\begin{equation} \label{GrindEQ__F5_}
{}_{t_{0} } \mathfrak D_{t}^{\gamma } x\left(t\right)=f\left(x\left(t\right)\right),\, \, \, \, x\left(0\right)=x_{0}
\end{equation}

By applying the Laplace transformation and according to the two-parametric Mittag-Leffler (ML) definition \cite{machado2011recent}, the solution of \eqref{GrindEQ__F5_} can be expressed as follows:
\begin{equation} \label{GrindEQ__F6_}
\mathcal E_{\gamma ,\sigma_{ML} } \left(t\right)=\sum _{\kappa =0}^{\infty }\frac{t^{\kappa } }{\Gamma \left(\gamma \kappa +\sigma_{ML} \right)}  ,\, \, \, \, 0<\gamma \le 1,\, \, \sigma_{ML} >0
\end{equation}
where $\sigma_{ML} $ is the order of ML function.

\begin{lem}
\label{lem:1}
\sloppy \cite{machado2011recent} Defining the coefficient $c_{\kappa } :=\frac{1}{\Gamma \left(\gamma \kappa +1\right)} $ in \eqref{GrindEQ__F6_} and according to the Cauchy-Hadamard formula \cite{sabatier2007advances} for the radius of convergence $R_{CH}=\lim \sup _{\kappa \to \infty } \frac{\left|c_{\kappa } \right|}{\left|c_{\kappa +1} \right|} $, one can observe that \eqref{GrindEQ__F6_} converges in the whole complex plane for all $\gamma >0$.
\end{lem}

\begin{lem}
\label{lem:2}
\sloppy \cite{li2010stability} The state $x\left(t\right)$ satisfies the ML stability if $\left\| x\left(t\right)\right\| \le \left\{m\left[x\left(0\right)\right] \mathcal E_{\gamma } \left(-\phi _{1} t^{\gamma } \right)\right\}^{\phi _{2} } $, where $\mathcal E_{\gamma } \left(.\right)$ denotes the ML function defined as \eqref{GrindEQ__F6_}, $\phi _{1} ,\phi _{2} >0$, $m\left(0\right)=0$, $m\left(x\right)\ge 0$, and $m\left(x\right)$ is locally Lipschitz on $x\in {\mathbb R}^{n} $.
\end{lem}

\begin{lem}
\label{lem:3}
\sloppy \cite{li2010stability} The fractional-order system \eqref{GrindEQ__F5_} is ML stable at the equilibrium point $x\left(t\right)=0$, if there exists a continuous differentiable Lyapunov function $V\left(x\left(t\right),t\right)$ satisfying $\varsigma _{1} \left\| x\right\| ^{\varsigma _{2} } \le V\left(x,t\right)\le \varsigma _{3} \left\| x\right\| ^{\varsigma _{2} \varsigma _{4} } $ and ${}_{t_{0} } \mathfrak D_{t}^{\gamma } V\left(x,t\right)\le -\varsigma _{5} \left\| x\right\| ^{\varsigma _{2} \varsigma _{4} } $, where $\varsigma _{i} ,\, i=1,...,5$ are arbitrary positive constants and $\left\| \cdot \right\| $ denotes the $\ell _{2} $ norm. If these assumptions hold globally on ${\mathbb R}^{n} $, then $x\left(t\right)=0$ is globally ML stable.
\end{lem}

Since $\sigma_{ML} =1$ in {\eqref{GrindEQ__F6_}}, the ML is the same as the one-parametric ML, thus $\mathcal E_{\gamma ,1} \left(t\right)=\mathcal E_{\gamma } \left(t\right)$. According to \eqref{GrindEQ__F6_}, the integer-order case with $\gamma =1$ can be expressed as
\begin{equation} \label{GrindEQ__F8_}
\mathcal E_{1} \left(t\right)=\sum _{\kappa =0}^{\infty }\frac{t^{\kappa } }{\kappa !}  =e^{t}
\end{equation}

On the other hand, the ML function with a fractional-order $0<\gamma <1$ holds two situations; (a) $\gamma \ne 0.5$, and (b) $\gamma =0.5$. The ML function follows the Cauchy inequality \cite{aghayan2022guaranteed,aghayan2023criteria} and $\Gamma \left(.\right)$ characteristics, hence, for the first situation it is assumed that there exists a $\kappa \ge 0$ and a positive scalar $r\left(\kappa \right)$ such that,
\begin{equation} \label{GrindEQ__F9_}
M_{E_{\gamma } } \left(r\right):=\max _{\left|t\right|=r} \left| \mathcal E_{\gamma } \left(t\right)\right|<e^{r^{\kappa } } ,\, \, \, \forall r>r\left(\kappa \right)
\end{equation}
while for the second situation, which is the particular case of ML function we have,
\begin{equation} \label{GrindEQ__F10_}
\mathcal E_{0.5} \left(t\right)=\sum _{\kappa =0}^{\infty }\frac{t^{\kappa } }{\Gamma \left(0.5\kappa +1\right)} =e^{t^{2} } erfc\left(-t\right)
\end{equation}
where $erfc\left(.\right)$ is complementary to the well-known error function $erf\left(.\right)$:
\begin{equation} \label{GrindEQ__F11_}
erfc\left(t\right):=\frac{2}{\sqrt{\pi } } \int _{t}^{\infty }e^{-\tau ^{2} } d\tau =1-erf\left(t\right),\, \, t\in {\mathbb C}
\end{equation}

\sloppy According to \eqref{GrindEQ__F8_}-\eqref{GrindEQ__F11_}, it is observed that as $t\to \infty $, the ML distribution exhibits an exponential decay of $e^{-t} $ in the integer-order case, meaning that the state $x\left(t\right)$ converges to zero with an exponential decay law of $e^{-t} $, while the fractional-order cases exhibit faster convergence with exponential decay laws of $e^{-r^{\kappa } } $ and $e^{-t^{2} } erfc\left(-t\right)$ for $\gamma \ne 0.5$ and $\gamma =0.5$, respectively. Figure \ref{fig:111} shows this difference in the chattering phenomenon occurrence on sliding surface of fractional-order SMC compared to the integer-order SMC, where it can be seen that the characteristics of fractional-order calculus can desirably reduce the chattering phenomenon.

\begin{figure}[!t]
\centering
\includegraphics[width=2.8in]{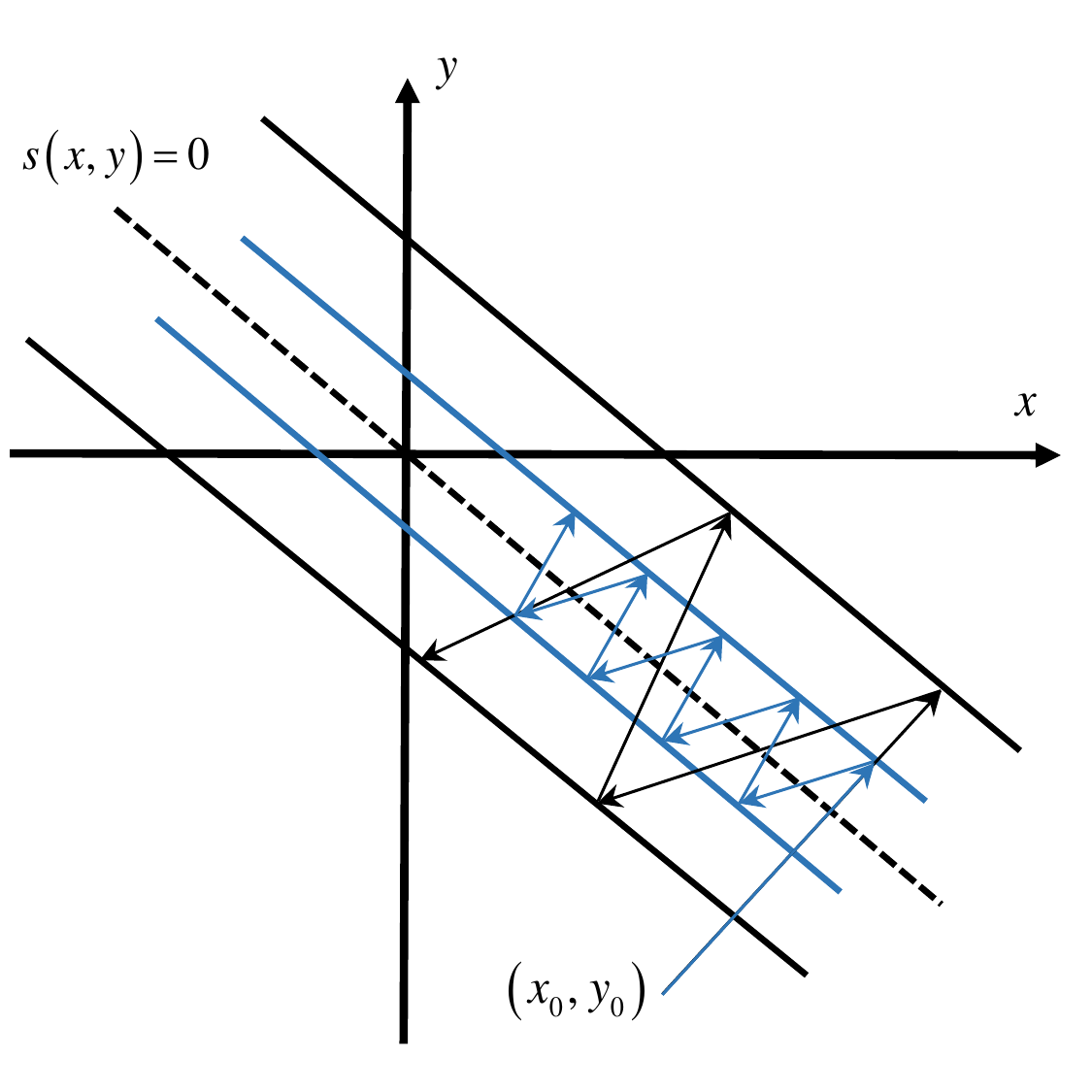}
\caption{Sliding mode movement under fractional-order (blue) and integer-order (black) control system}
\label{fig:111}
\end{figure}

\section{Conclusions}
\label{sec:2.4}
\textcolor{black}{This chapter investigated the model of wind energy conversion systems in detail. Accordingly, the main components of the WECS consisting of aerodynamics, pitch actuator system, drivetrain, and generators were presented. In addition, since the main control schemes being developed in Chapters 5 and 6 are on the basis of sliding mode control method, a brief introduction to sliding mode control was delivered. Moreover, some preliminary definitions of fractional calculus as well as discussions on the stability and convergence of fractional-order systems were provided. Having introduced the principles of WECSs, SMC and fractional-order calculus, the next chapter will review the application of SMC controllers on WECSs.}

\chapter{Literature Review} 

\label{Chapter3} 




\section{Introduction}
\label{sec:3.1}

This chapter presents an extensive review of the existing literature on the application of conventional SMC and its modifications in dealing with different control problems of WECSs. \textcolor{black}{Accordingly, separate sections are provided as conventional SMC (Section \ref{sec:3.2}), adaptive SMC (Section \ref{sec:3.3}), fractional-order SMC (Section \ref{sec:3.4}), higher-order SMC (Section \ref{sec:3.5}), fuzzy SMC (Section \ref{sec:3.6}), and neural network SMC (Section \ref{sec:3.7}). Besides, in each section, studies are categorized based on generator types consisting of DFIG and PMSG as the most used ones, along with other generators which include SCIG, direct driven synchronous generator (DDSG), double output induction generator (DOIG), dual stator induction generator (DSIG), simple first-order generator (SFG), and simple second-order generator (SSG).}

Depending on the combinations of fixed/variable speed with fixed/variable pitch angle of the blades, different classifications can be considered for WTs. In fixed-pitch turbines, the pitch angle of the blades cannot be adjusted/controlled. Thus, they are unable to provide a reliable power production performance due to their inability to mitigate structural loads \cite{njiri2016state}. On the other hand, in variable-pitch turbines, the captured aerodynamic power is limited by controlling the pitch angle at above rated wind speed situations, which keeps the turbine operating at its rated power. Variable-speed WTs are capable of producing power under variable wind speed conditions. However, they require converters to guarantee the desired generated power performance as well as coupling them to the grid \cite{soliman2020novel}. On the contrary, fixed-speed WTs are designed to operate at a predetermined wind speed (rated wind speed) or designed for optimum operation at a rated wind speed; however, in practice, they have been found to operate in variable winds and are designed accordingly. The generator of fixed-speed WTs is directly coupled to the grid, which makes the generator speed dependent on the grid frequency. Furthermore, although compared to the variable-speed WTs, the fixed-speed ones have simpler structures, they have been found to be less effective in terms of wind energy extraction and induction of mechanical stress under variable wind speed conditions \cite{wang2007survey}. Accordingly, the most commonly used large-scale WT structure has been the variable-speed variable-pitch WT structure in the past decade.

\section{Conventional SMC for WECS}
\label{sec:3.2}
Numerous studies with various control objectives have been carried out using conventional SMC for WECS with different generator types. Some of the studies are discussed in this section; however, the full summary can be found in Table \ref{tab:2}, where the index ``ex.'' denotes that experimental investigations were also carried out in the study. Accordingly, studies are categorized based on generator types consisting of doubly fed induction generator (DFIG), permanent-magnet synchronous generator (PMSG), and squirrel-cage induction generator (SCIG), permanent magnet synchronous motor (PMSM), direct driven synchronous generator (DDSG), double output induction generator (DOIG), dual stator induction generator (DSIG), simple first-order generator (SFG), and simple second-order generator (SSG). Also, the main control objectives pursued in each study are categorized as MPE, active power control (APC) and reactive power control (RPC) denoting the power setpoints tracking rather than the maximum power points, chattering reduction (CR), rotor-side (RSCV) and grid-side (GSCV) converter voltage regulation, DC-link voltage regulation (DCVR), fault-tolerant control (FTC), pitch angle control (PAC), load frequency control (LFC), and mechanical stress minimization (MSM). Figures \ref{fig:2} and \ref{fig:3} demonstrate the percentage-based and year-based distribution of studies devoted to SMC-based control of WECS during the past two decades, respectively. As it is observed, compared to other SMC-based methods, the conventional SMC has gained the most attention owing to its simplicities in design and implementation procedure. However, the simplest common modification on conventional SMC has been the substitution of $sgn\left(\cdot\right)$ function in the conventional SMC's sliding surface with $sat\left(\cdot\right)$ or $tanh\left(\cdot\right)$ to reduce the undesirable chattering effects. In this respect, although $tanh\left(\cdot\right)$ has reportedly demonstrated superior performance over $sat\left(\cdot\right)$ and $sgn\left(\cdot\right)$, respectively, higher-order SMC strategies' preferabilities among scholars have shown a growing trend during the past few years.

\begin{figure}
\centering
\includegraphics[width=3 in]{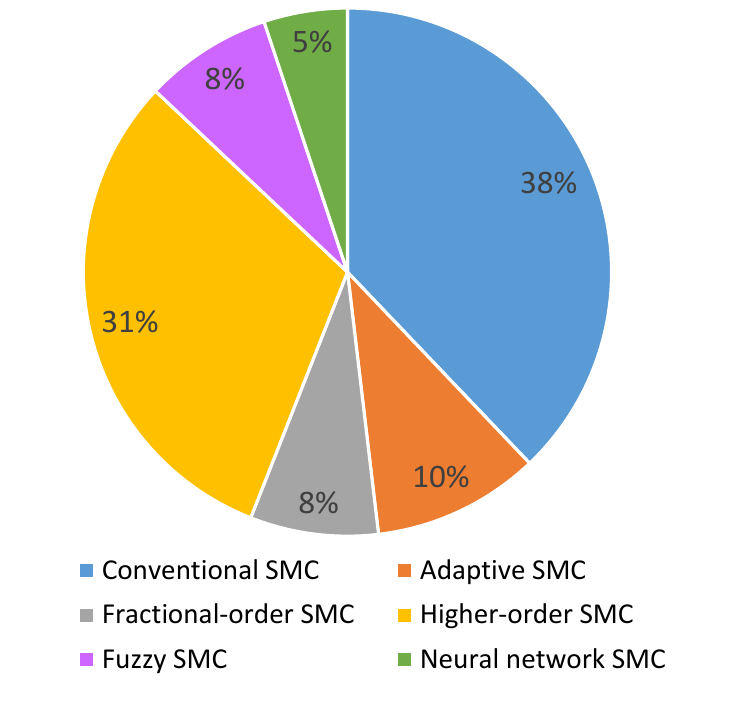}
\caption{Distribution of the total research studies devoted to SMC-based control of WECS.}
\label{fig:2}
\end{figure}

\begin{figure}
\centering
\includegraphics[width=4.2 in]{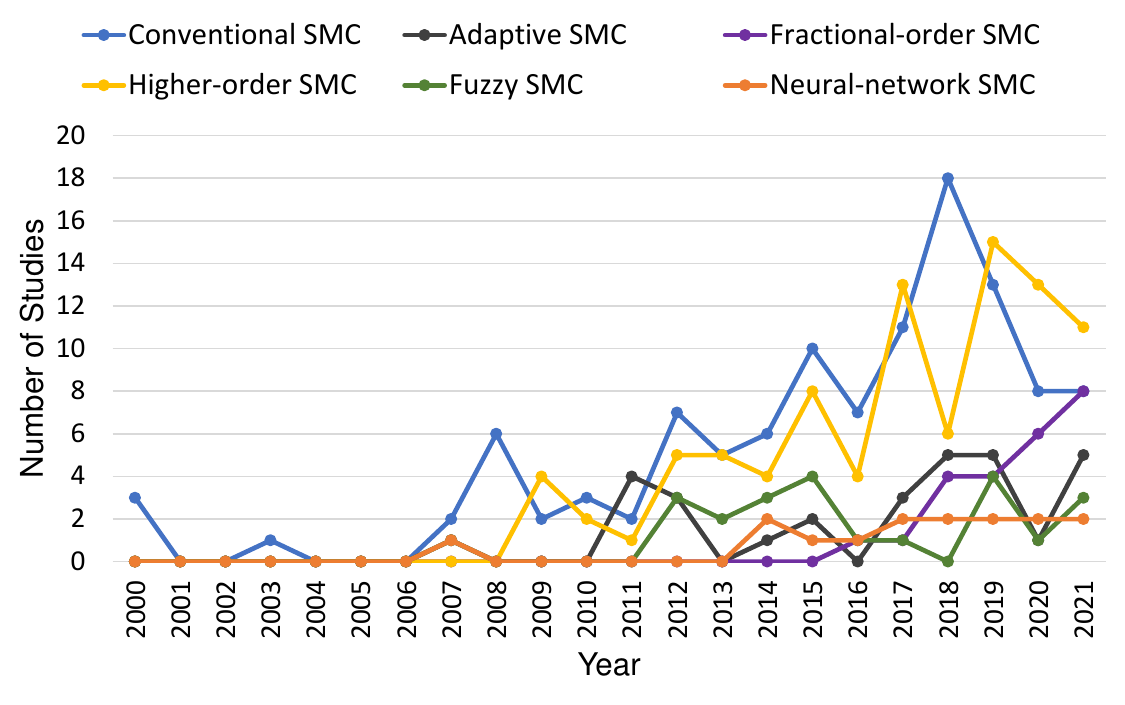}
\caption{The number of research studies devoted to SMC-based control of WECS over the past two decades.}
\label{fig:3}
\end{figure}

\begin{table}[hbt!]
\caption{Summary of conventional SMC methods applied to WECS. The index ``ex.'' indicates experimental studies.}
\centering
\label{tab:2}
\resizebox{\textwidth}{!}{
\begin{tabular}{l c c c c c c c c c c c c}
\hline\hline \\[-3mm]
\multicolumn{1}{l}{Reference} & \multicolumn{11}{c}{Main control objectives} & \multicolumn{1}{c}{Generator} \\ \hline
 & \multicolumn{1}{c}{MPE} & \multicolumn{1}{c}{APC}  & \multicolumn{1}{c}{RPC} & \multicolumn{1}{c}{CR} & \multicolumn{1}{c}{RSCV} & \multicolumn{1}{c}{GSCV} & \multicolumn{1}{c}{DCVR} & \multicolumn{1}{c}{FTC} & \multicolumn{1}{c}{PAC} & \multicolumn{1}{c}{LFC} & \multicolumn{1}{c}{MSM}  \\[1.2ex] \hline
\cite{pan2008maximum, weng2014sliding, ihedrane2017power, zheng2009sliding, amira2020sliding} & * &  &  &  &  &  &  &  &  &  &  & DFIG \\
\rowcolor{lightgray}\cite{jaladi2018dc}${}_{ex.}$ & * &  & * &  &  &  & * & * &  &  &  & DFIG \\
\cite{valenciaga2009geometric,liu2018dfig}, \cite{xiong2021event}${}_{ex.}$ & * &  &  & * &  &  &  &  &  &  &  & DFIG \\
\rowcolor{lightgray}\cite{barambones2019robust} & * &  &  &  &  &  & * &  &  &  &  & DFIG \\
\cite{noussi2021adaptive} & * &  & * & * &  &  & * &  &  &  &  & DFIG \\
\rowcolor{lightgray}\cite{munteanu2008energy}${}_{ex.}$ & * &  &  &  &  &  &  &  &  &  & * & DFIG \\
\cite{el2016comparative, bekakra2014dfig, hu2010direct, pande2013discrete, mahboub2017sliding, hamane2014control, hagh2015direct, tohidi2013multivariable,hamid2020improved}, \cite{jeong2008sliding}${}_{ex.}$ &  & * & * &  &  &  &  &  &  &  &  & DFIG \\
\rowcolor{lightgray}\cite{nayeh2020multivariable, pan2021integral}, \cite{djoudi2018sliding}${}_{ex.}$ &  & * & * & * &  &  &  &  &  &  &  & DFIG \\
\cite{aghatehrani2011sliding} &  & * & * &  &  &  & * &  &  &  &  & DFIG \\
\rowcolor{lightgray}\cite{saad2015low}, \cite{shang2012sliding}${}_{ex.}$, \cite{villanueva2018grid}${}_{ex.}$ &  & * & * &  &  &  &  & * &  &  &  & DFIG \\
\cite{jafari2017analysis} &  & * &  & * &  &  &  &  &  &  &  & DFIG \\
\rowcolor{lightgray}\cite{dahiya2019optimal} &  &  &  & * &  &  &  &  &  & * &  & DFIG \\
\cite{martinez2011sliding,chojaa2021integral} &  & * &  &  & * & * &  &  &  &  &  & DFIG \\
\rowcolor{lightgray}\cite{shehata2015sliding, martinez2013sliding, taher2018new, li2019sliding, yang2018robust}, \cite{merabet2018power}${}_{ex.}$, \cite{djilali2018real}${}_{ex.}$,  &  & * &  &  & * & * & * & * &  &  &  & DFIG \\
\cite{dahiya2017hybridized} &  &  &  &  &  &  &  &  &  & * &  & DFIG \\ \hline
\rowcolor{lightgray}\cite{tang2019non, hu2019sliding, jingfeng2015maximum, errami2013maximum, kusumawardana2019simple, suleimenov2020disturbance}, \cite{yang2018passivity}${}_{ex.}$,\cite{yin2015sliding}${}_{ex.}$ & * &  &  &  &  &  &  &  &  &  &  & PMSG \\
\cite{pan2020wind} & * &  &  & * &  &  &  &  &  &  &  & PMSG \\
\rowcolor{lightgray}\cite{benadja2018hardware}${}_{ex.}$, \cite{errouissi2017novel}${}_{ex.}$ & * &  &  &  &  &  &  & * &  &  &  & PMSG \\
\cite{ayadi2015sliding, xin2014sliding, lee2010sliding} & * &  &  &  &  &  &  &  & * &  &  & PMSG \\
\rowcolor{lightgray}\cite{valenciaga2003power} &  & * &  &  &  &  &  &  &  &  &  & PMSG \\
\cite{jena2017novel, gajewski2017analysis} &  & * &  &  & * & * &  &  &  &  &  & PMSG \\
\rowcolor{lightgray}\cite{mozayan2016sliding}${}_{ex.}$ &  &  &  & * & * & * &  &  &  &  &  & PMSG \\
\cite{thakur2018control}${}_{ex.}$ &  &  &  &  & * & * & * & * &  &  &  & PMSG \\
\rowcolor{lightgray}\cite{merzoug2012sliding}, \cite{huang2019dc}${}_{ex.}$,  &  &  &  &  &  &  & * &  &  &  &  & PMSG \\ \hline
\cite{soufi2016particle} & * &  &  &  &  &  &  &  &  &  &  & SCIG \\
\rowcolor{lightgray}\cite{mi2014sliding} & * &  &  &  &  &  &  &  & * &  &  & SCIG \\
\cite{de2000dynamical, pati2013sliding, pati2012performance} &  & * &  &  &  &  &  &  &  &  &  & SCIG \\ \hline
\rowcolor{lightgray}\cite{matas2008feedback} &  &  &  &  & * & * & * &  &  &  &  & DDSG \\ \hline
\cite{de2000sliding, puleston2000sliding} & * &  &  &  &  &  &  &  &  &  &  & DOIG \\ \hline
\rowcolor{lightgray}\cite{amimeur2012sliding} &  & * &  &  & * & * &  &  &  &  &  & DSIG \\ \hline
\cite{ur2019disturbance, faskhodi2019output, beltran2008sliding, saravanakumar2015validation} & * &  &  &  &  &  &  &  &  &  &  & SFG \\
\rowcolor{lightgray}\cite{sami2012fault} & * &  &  &  &  &  &  & * &  &  &  & SFG \\
\cite{cui2017observer} & * &  &  &  &  &  &  & * &  & * &  & SFG \\
\rowcolor{lightgray}\cite{gui2020complementary, larrea2018sliding, corradini2017sliding, yinzhu2016study, colombo2020pitch, yuan2018coordinated, qingmei2019novel} & * &  &  &  &  &  &  &  & * &  &  & SFG \\
\cite{zaafouri2018uncertain, torchani2016variable} & * &  &  &  &  &  &  &  &  &  & * & SFG \\
\rowcolor{lightgray}\cite{berrada2020new} & * &  &  & * &  &  &  &  &  &  & * & SFG \\
\cite{zhang2013sliding} &  & * &  &  &  &  &  &  &  &  &  & SFG \\ \hline
\rowcolor{lightgray}\cite{beltran2007sliding} & * &  &  &  &  &  &  &  &  &  &  & SSG \\
\cite{agarwala2019design} &  & * &  &  &  &  &  &  & * &  &  & SSG \\
\rowcolor{lightgray}\cite{prasad2019non} &  &  &  &  &  &  &  &  & * & * &  & SSG \\
[1ex]
\hline\hline
\end{tabular}}
\end{table}

\subsubsection{DFIG (SMC)}
\label{sec:3.2.1}
Authors in \cite{taher2018new, li2019sliding} developed SMC controllers to control the DC-link voltage, rotor-, and grid-side converters to enhance the fault ride-through performance of DFIG-based WT, while \cite{yang2018robust} utilized a nonlinear perturbation observer-based SMC to deal with the same problem. According to the results reported, the presented control schemes demonstrated superior performance compared to PI, conventional SMC, conventional vector control, and feedback linearization control. In another work \cite{villanueva2018grid}, the authors investigated SMC's application in FTC of DFIG-based WECS under low voltage grid faults. As the authors reported, the developed control scheme was able to effectively withstand the balanced and unbalanced voltage dips and tracks the torque and reactive power references.

An improved SMC controller with reduced chattering was proposed in \cite{valenciaga2009geometric} for wind power capture maximization of a grid-connected DFIG-based WECS under bounded uncertainties and disturbances. In order to mitigate the chattering effects, the authors developed a control component that delivered the minimum discontinuous action required for effective disturbance rejection. In contrast, \cite{liu2018dfig} proposed an exponential reaching law to deal with the chattering phenomenon. As reported, the proposed approaches effectively reduced the chattering problem and minimized the machine losses. In \cite{nayeh2020multivariable}, the APC and RPC of a DFIG-based WT in the presence of various uncertainties was investigated using the SMC approach. As the authors reported, although some minor chattering still exist in the developed SMC, they can be counted acceptable due to less tracking error and better performance of SMC in the transient time response in terms of overshoot and settling time compared to the $H_{\infty }$ robust control method.

\subsubsection{PMSG (SMC)}
\label{sec:3.2.2}
The MPE of a PMSG-based offshore WT was studied in \cite{benadja2018hardware}, taking advantage of an SMC controller to ensure the voltage source converter's stable operation in a high voltage direct current station during DC faults. Experimental investigations were carried out, and an improved dynamic response of the system in terms of signal quality, stability, and response time of the system was reported using the proposed strategy in comparison with conventional PI. In another study \cite{suleimenov2020disturbance}, the authors incorporated a generalized high-order disturbance observer with integral SMC (ISMC) to deal with the same problem. In the presented work, the fast-changing uncertainties were estimated using a disturbance observer, while the ISMC was responsible for the rotor speed control. Comparative simulation results were provided, and as reported, by guaranteeing less steady-state error and smaller response time, the proposed scheme was shown to have better performance over the linear quadratic regulator (LQR) approach. In order to recover the transient performance of a PMSG-based WT subjected to external disturbances, a combination of a PI-type ISMC and feedback linearization methods was proposed in \cite{errouissi2017novel}. In order to reduce the chattering, the authors used the $sgn\left(\cdot\right)\rightarrow sat\left(\cdot\right)$ substitution in the switching control law, where two design parameters were implemented to make a trade-off between desired control and chattering reduction objectives. As reported through the experimental investigations, the developed scheme successfully eliminated the steady-state error and retained the feedback linearization technique's nominal transient performance.

\subsubsection{Other Generators (SMC)}
\label{sec:3.2.3}
SMC-based PAC strategies were proposed in \cite{colombo2020pitch, yuan2018coordinated} to maximize power extraction of a WT with an SFG model. Authors in \cite{colombo2020pitch} enhanced the SMC's performance by applying the boundary layer method \cite{utkin1992scope} to the control law and reducing the chattering, while authors in \cite{yuan2018coordinated} utilized an adaptive disturbance observer to estimate the load disturbance. As reported, compared to the standard PI-based baseline controller, the developed control schemes demonstrated superior performance in terms of tracking error. A robust FTC approach based on SMC was proposed in \cite{sami2012fault} to optimize the wind energy captured by a WT in the low wind speed range of operation. According to the authors, the need for a state observer could be removed by using the proposed strategy. Besides, the SMC switching gain could be adapted so that the unmeasurable signals are compensated. Authors in \cite{saravanakumar2015validation} proposed an augmentation of Newton Raphson wind speed estimator and an ISMC to maximize the power capture in region II and power regulation in region III of a variable speed variable pitch WT. The authors dealt with the chattering problem using the tangent hyperbolic function $tanh\left(\cdot\right)$. According to the authors, in comparison with conventional SMC, the proposed control strategy demonstrated superior performance in terms of optimum power extraction with reduced transient load on the drivetrain. In another study \cite{cui2017observer}, an observer was utilized to estimate the wind power system states under perturbation of nonlinear load, and an ISMC-based load frequency controller was also developed to reduce the frequency deviations of the overall power system.

\section{Adaptive SMC for WECS}
\label{sec:3.3}
The parameters of the WECS are not always constant over time, and the system is not exempt from parameter variations, perturbations, and uncertainties \cite{yang2016survey,mousavi2020enhanced}. Since many of these uncertainty and perturbation bounds cannot be directly estimated, their negative impact on the SMC controller's robust performance is inevitable \cite{jing2019adaptive}. In such situations, dynamic gain adaptation laws can be utilized to deal with uncertain disturbances and avoid the chattering phenomenon's incrementation due to large switching control signals \cite{liu2020adaptive}.

\subsection{DFIG (ASMC)}
\label{sec:3.3.1}
Authors in \cite{ayyarao2019modified} developed an adaptive SMC (ASMC) controller for maximum power point tracking (MPPT) of DFIG-based WTs in the presence of external disturbances. The authors diminished the chattering effects and achieved finite-time convergence by selecting the control gains using the positive semi-definite barrier function. An ASMC scheme was developed to deal with the load mitigation problem of WT systems \cite{hu2017active} in the presence of system uncertainties and external disturbances. The authors used the $sgn\left(\cdot\right)\rightarrow sat\left(\cdot\right)$ substitution to suppress the chattering effects. As the authors reported, the reachability of the sliding surface and stability of the sliding motion was guaranteed. In \cite{yin2015adaptive}, a variable-displacement pump-controlled pitch system was developed to deal with flap-wise load fluctuations and generator power mitigation of WTs. An adaptive sliding mode pump displacement controller and a backstepping stroke piston controller were also proposed to control the pitch system. Later, an electro-hydraulic servo pitch system was proposed by employing a hydraulic pump-controlled hydraulic motor to enhance the pitch control efficiency \cite{yin2019adaptive}, where, an ASMC controller was developed to track the desired pitch angle trajectory. According to the authors, compared with the conventional pitch systems, the developed control scheme has the advantages of larger power/torque to volume ratio, better pitch angle tracking accuracy, and higher overall efficiency due to valve-less control.

\subsection{PMSG (ASMC)}
\label{sec:3.3.2}
In \cite{ameli2019adaptive, merabet2011adaptive, barambones2015wind}, ASMC controllers were developed to deal with the maximum power capture problem of variable speed WTs performing in region III. In order to reduce the chattering phenomenon and achieve a smooth pitching action, an adaption mechanism for the upper bounds of the norm of the uncertainties was proposed in \cite{ameli2019adaptive}, while \cite{merabet2011adaptive, barambones2015wind} introduced modified sliding surfaces. The proposed schemes demonstrated less mechanical stress and fluctuation of generator torque in comparison with the conventional SMC. Taking advantage of the super-twisting algorithm, a robust aerodynamic torque observer incorporated with an ASMC control scheme was presented in \cite{barambones2016adaptive} to deal with the WECS power control problem in the presence of system uncertainties and external disturbances.

\subsection{Other Generators (ASMC)}
\label{sec:3.3.3}
In \cite{azizi2019fault}, an active FTC strategy was presented to control the rotor speed and power of a WT in the presence of actuator faults and uncertainties. In this regard, the authors designed a full-order compensator for fault and disturbance attenuation, and an adaptive output feedback SMC with an integral surface to perform the FTC. Compared with PID and disturbance accommodation controller, it was shown that the proposed strategy demonstrated better fault-tolerant capability and more robust behavior with fewer fluctuations and less fatigue on the rotor speed, generator speed, output power, and drivetrain torsion angle. Authors in \cite{sami2012wind} developed a robust fault-tolerant ASMC controller to optimize the wind energy captured by a variable speed WT operating at low wind speeds. The proposed method involved a robust descriptor observer that provided a robust simultaneous estimation of states and the unknown sensor faults and noise. In another study \cite{lan2018fault}, authors proposed a strategy for compensating the pitch actuator faults to recover the nominal pitch dynamics of a WT by utilizing a PI controller as a baseline system to achieve nominal pitch performance, along with an adaptive step-by-step sliding mode observer (SMO) as a fault compensator to eliminate the actuator fault effects. According to the authors, the proposed control design was able to recover the nominal pitch actuation in the presence of low-pressure actuator faults.

\textcolor{black}{Table \ref{tab:3} summarizes the objectives, features, advantages, and disadvantages of the discussed studies on ASMC approaches for WECS control.}
\begin{table*}[th!]
\caption{\textcolor{black}{Summary of adaptive SMC approaches for WECS control.}}
\centering
\label{tab:3}
\resizebox{\textwidth}{!}{
\begin{tabular}{l l l l l l}
\hline\hline \\[-3mm]
\begin{minipage}[t]{0.1\columnwidth} Work \end{minipage} & \begin{minipage}[t]{0.15\columnwidth} Objectives \end{minipage} & \begin{minipage}[t]{0.15\columnwidth} Operating region \end{minipage} & \begin{minipage}[t]{0.15\columnwidth} Generator \end{minipage} & \begin{minipage}[t]{0.5\columnwidth} Advantages \end{minipage}	& \begin{minipage}[t]{0.5\columnwidth}Disadvantages \end{minipage}\\ \hline
\begin{minipage}[t]{0.1\columnwidth} \cite{ayyarao2019modified}  \end{minipage} & \begin{minipage}[t]{0.15\columnwidth} MPPT  \end{minipage} &\begin{minipage}[t]{0.15\columnwidth} Partial-load  \end{minipage} &\begin{minipage}[t]{0.15\columnwidth} DFIG  \end{minipage} & \begin{minipage}[t]{0.5\columnwidth} \begin{itemize}
  \item Chattering reduced using barrier functions
  \item Grid disturbances considered
\end{itemize} \end{minipage} & \begin{minipage}[t]{0.5\columnwidth} \begin{itemize}
\item Comparisons only with PI
\end{itemize} \end{minipage} \\ [12mm]
&&&&& \\
\rowcolor{lightgray}\begin{minipage}[t]{0.1\columnwidth} \cite{hu2017active}  \end{minipage} & \begin{minipage}[t]{0.15\columnwidth} Load mitigation  \end{minipage} &\begin{minipage}[t]{0.15\columnwidth} Full-load  \end{minipage} &\begin{minipage}[t]{0.15\columnwidth} DFIG  \end{minipage} & \begin{minipage}[t]{0.5\columnwidth} \begin{itemize}
  \item Chattering reduced by $sgn\left(\cdot\right)\rightarrow sat\left(\cdot\right)$ substitution
  \item External disturbances considered
\end{itemize} \end{minipage} & \begin{minipage}[t]{0.5\columnwidth} -- \end{minipage} \\ [12mm]
&&&&& \\
\begin{minipage}[t]{0.1\columnwidth} \cite{yin2015adaptive}  \end{minipage} & \begin{minipage}[t]{0.15\columnwidth} PAC  \end{minipage} &\begin{minipage}[t]{0.15\columnwidth} Full-load  \end{minipage} &\begin{minipage}[t]{0.15\columnwidth} DFIG  \end{minipage} & \begin{minipage}[t]{0.5\columnwidth} \begin{itemize}
  \item Load fluctuations considered
\end{itemize} \end{minipage} & \begin{minipage}[t]{0.5\columnwidth} \begin{itemize}
\item Chattering problem not tackled
\item No comparisons with other controllers
\end{itemize} \end{minipage} \\ [12mm]
\rowcolor{lightgray}\begin{minipage}[t]{0.1\columnwidth} \cite{yin2019adaptive}  \end{minipage} & \begin{minipage}[t]{0.15\columnwidth} PAC  \end{minipage} &\begin{minipage}[t]{0.15\columnwidth} Full-load  \end{minipage} &\begin{minipage}[t]{0.15\columnwidth} DFIG  \end{minipage} & \begin{minipage}[t]{0.5\columnwidth} \begin{itemize}
  \item External disturbances and model uncertainties considered
\end{itemize} \end{minipage} & \begin{minipage}[t]{0.5\columnwidth} \begin{itemize}
\item Chattering problem not tackled
\item Comparisons only with PI
\end{itemize} \end{minipage} \\  [12mm] \hline
\begin{minipage}[t]{0.1\columnwidth} \cite{ameli2019adaptive}  \end{minipage} & \begin{minipage}[t]{0.15\columnwidth} MPE, PAC  \end{minipage} &\begin{minipage}[t]{0.15\columnwidth} Full-load  \end{minipage} &\begin{minipage}[t]{0.15\columnwidth} PMSG  \end{minipage} & \begin{minipage}[t]{0.5\columnwidth} \begin{itemize}
  \item Chattering reduced using an adaption mechanism for the upper bounds of the uncertainties
  \item Uncertainties considered
  \item Simulation on FAST model
\end{itemize} \end{minipage} & \begin{minipage}[t]{0.5\columnwidth} -- \end{minipage} \\ [15mm]
&&&&& \\
\rowcolor{lightgray}\begin{minipage}[t]{0.1\columnwidth} \cite{merabet2011adaptive}  \end{minipage} & \begin{minipage}[t]{0.15\columnwidth} MPE  \end{minipage} &\begin{minipage}[t]{0.15\columnwidth} Full-load  \end{minipage} &\begin{minipage}[t]{0.15\columnwidth} PMSG  \end{minipage} & \begin{minipage}[t]{0.5\columnwidth} \begin{itemize}
  \item External disturbances and parametric uncertainties considered
\end{itemize} \end{minipage} & \begin{minipage}[t]{0.5\columnwidth} \begin{itemize}
\item No comparisons provided
\item Chattering problem not properly tackled
\end{itemize} \end{minipage} \\ [12mm]
\begin{minipage}[t]{0.1\columnwidth} \cite{barambones2015wind}  \end{minipage} & \begin{minipage}[t]{0.15\columnwidth} MPE  \end{minipage} &\begin{minipage}[t]{0.15\columnwidth} Full-load  \end{minipage} &\begin{minipage}[t]{0.1\columnwidth} PMSG  \end{minipage} & \begin{minipage}[t]{0.5\columnwidth} \begin{itemize}
  \item Wind speed estimation using an observer
  \item External disturbances and parametric uncertainties considered
\end{itemize} \end{minipage} & \begin{minipage}[t]{0.5\columnwidth} \begin{itemize}
\item No comparisons provided
\item Chattering problem not tackled
\end{itemize} \end{minipage} \\ [15mm]
\rowcolor{lightgray}\begin{minipage}[t]{0.1\columnwidth} \cite{barambones2016adaptive}  \end{minipage} & \begin{minipage}[t]{0.15\columnwidth} MPE  \end{minipage} &\begin{minipage}[t]{0.15\columnwidth} Partial-load  \end{minipage} &\begin{minipage}[t]{0.15\columnwidth} PMSG  \end{minipage} & \begin{minipage}[t]{0.5\columnwidth} \begin{itemize}
  \item Wind speed estimation using higher-order observer
  \item External disturbances and parametric uncertainties considered
  \end{itemize} \end{minipage} & \begin{minipage}[t]{0.5\columnwidth} \begin{itemize}
\item No comparisons provided
\item Chattering problem not tackled
\end{itemize} \end{minipage} \\  [12mm]
\hline
&&&&& \\
\begin{minipage}[t]{0.1\columnwidth} \cite{azizi2019fault}  \end{minipage} & \begin{minipage}[t]{0.15\columnwidth} MPE, FTC, MSM  \end{minipage} &\begin{minipage}[t]{0.15\columnwidth} Full-load  \end{minipage} &\begin{minipage}[t]{0.15\columnwidth} SSG  \end{minipage} & \begin{minipage}[t]{0.5\columnwidth} \begin{itemize}
  \item Actuator faults and model uncertainties considered
  \item Chattering reduced by $sgn\left(\cdot\right)\rightarrow sat\left(\cdot\right)$ substitution
\end{itemize} \end{minipage} & \begin{minipage}[t]{0.5\columnwidth} \begin{itemize}
\item No comparison with SMC-based approaches
\end{itemize} \end{minipage} \\ [12mm]
&&&&& \\
\rowcolor{lightgray}\begin{minipage}[t]{0.1\columnwidth} \cite{sami2012wind}  \end{minipage} & \begin{minipage}[t]{0.15\columnwidth} MPE, FTC  \end{minipage} &\begin{minipage}[t]{0.15\columnwidth} Partial-load  \end{minipage} &\begin{minipage}[t]{0.15\columnwidth} SSG  \end{minipage} & \begin{minipage}[t]{0.5\columnwidth} \begin{itemize}
  \item Sensor faults estimated using an observer
  \item Chattering reduced by approximating $sgn\left(\cdot\right)$
\end{itemize} \end{minipage} & \begin{minipage}[t]{0.5\columnwidth} \begin{itemize}
\item Chattering problem not properly tackled
\item No comparisons provided
\end{itemize} \end{minipage} \\ [12mm]
&&&&& \\
\begin{minipage}[t]{0.1\columnwidth} \cite{lan2018fault}  \end{minipage} & \begin{minipage}[t]{0.15\columnwidth} PAC, MPE, FTC  \end{minipage} &\begin{minipage}[t]{0.15\columnwidth} Full-load  \end{minipage} &\begin{minipage}[t]{0.15\columnwidth} SSG  \end{minipage} & \begin{minipage}[t]{0.5\columnwidth} \begin{itemize}
  \item Actuator faults estimated using SMO
\end{itemize} \end{minipage} & \begin{minipage}[t]{0.5\columnwidth} \begin{itemize}
\item Chattering problem not tackled
\item No external disturbances considered
\item No comparisons provided
\end{itemize} \end{minipage} \\
[20mm]
\hline\hline
\end{tabular}}
\end{table*}

\section{Fractional-order SMC for WECS}
\label{sec:3.4}
Fractional calculus-based SMC controllers have shown their superior control performance compared to their integer counterparts thanks to two main characteristics they add to the control system \cite{sun2017discrete, sun2018practical, yin2014fractional}; (a) more tunable parameters which lead to more degrees of freedom to the control algorithm, and (b) the memory effect introduced by the fractional-order (FO) differential operators \cite{mousavi2018fractional}. Due to the distinctive memory features of FO derivatives \cite{mousavi2015memetic}, the status on the FO sliding surface will reach the equilibrium with faster convergence \cite{mousavi2021fault}. These features have been demonstrated as reliable solutions to mitigate the chattering problem of conventional SMCs in many practical applications \cite{mousavi2021robustG,xie2021coupled,musarrat2021fractional,mousavi2021maximum,li2020mitigating}. \textcolor{black}{A general representation of an adaptive FO-SMC for typical systems is depicted in Figure \ref{fig:511}, where $0<\alpha<1$ corresponds to the FO differentiation/integration of the error signal and $c_i>0$, $i=1,2,…,5$. It should be noted that, the differentiation/integration terms may change as per the designer's preference.}

\begin{figure}
\centering
\includegraphics[width=5.2 in]{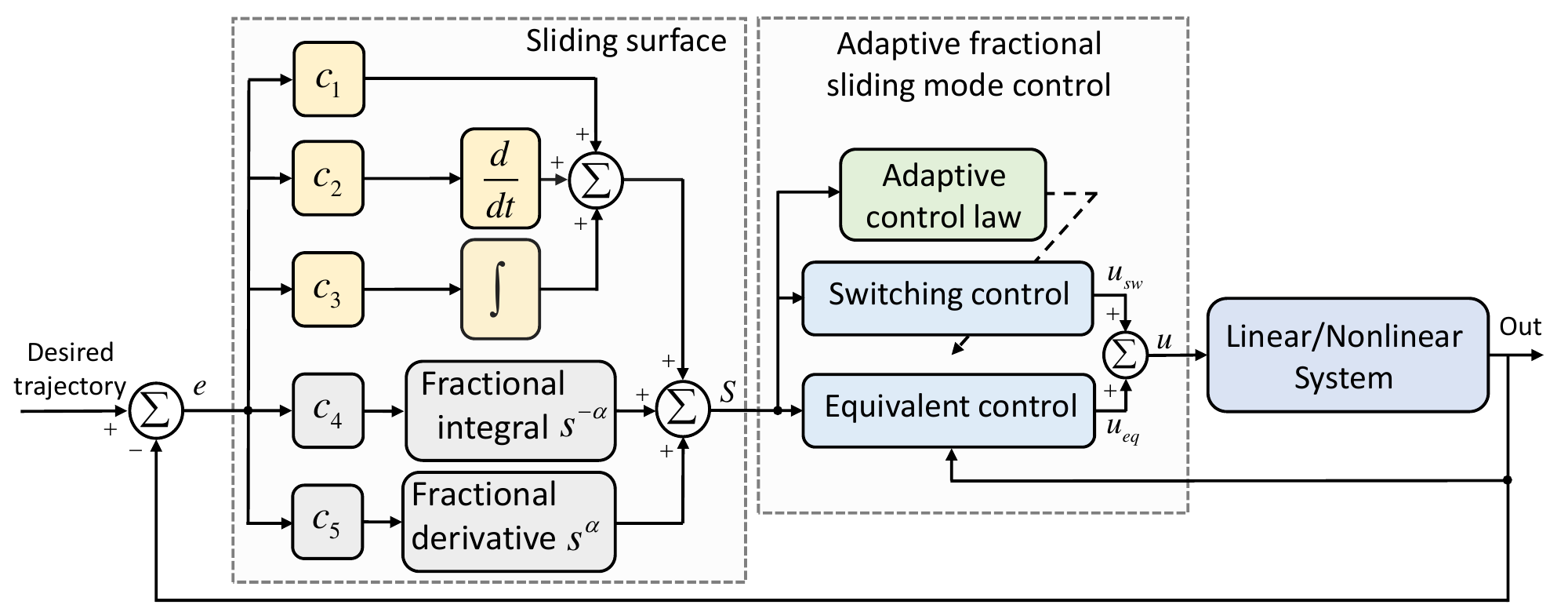}
\caption{\textcolor{black}{The general block diagram of an adaptive FO-SMC.}}
\label{fig:511}
\end{figure}

\subsection{DFIG (FO-SMC)}
\label{sec:3.4.1}
Authors in \cite{ullah2017adaptive} designed a fractional-order adaptive terminal SMC for both the RSC and GSC of the DFIG system. The authors developed a FO sliding surface to mitigate the chattering problem. The designed controller was compared with conventional SMC and PI, and it was deduced that it exhibited the fastest transient response in terms of settling time and convergence of tracking errors at all test conditions. In a similar study \cite{dash2018adaptive}, the authors proposed an adaptive fractional integral terminal SMC to control the $dq$ axis series injected voltage of the unified power flow controller in order to improve the transient stability of a DFIG wind farm penetrated multi-machine power system. Taking advantage of a fractional-order terminal sliding surface, less chattering and more robust performance of the developed control strategy was reported in comparison with the conventional PI control approach. Authors in \cite{mousavi2022active} proposed an active fault-tolerant nonlinear control scheme for the RSC control of a DFIG-driven WECS subjected to model uncertainties and rotor current sensor faults. They proposed two fractional-order nonsingular TSMC controllers for rotor current regulation and speed trajectory tracking. A control scheme was incorporated with a state observer to estimate the rotor current dynamics during sensors' faults. The chattering problem was successfully attenuated, and fast finite-time convergence of system states was guaranteed. As reported, the developed active FTC desirable outperformed the FO-SMC technique in terms of speed and power tracking performance under faulty situations.

Using robust FO-SMC controllers with continuous control laws and active state observers to estimate and compensate the uncertainties and disturbances, authors in \cite{ebrahimkhani2016robust} investigated the MPPT control problem of DFIG-based WECS. Later, a similar methodology was developed for DC-link voltage regulation and transient stability enhancement of DFIG-based WECS \cite{musarrat2021fractional}. The authors proposed a nonlinear smooth sliding surface to alleviate the chattering phenomenon, where the provided simulation investigations validated the effectiveness of the developed approach. In another study \cite{falehi2020innovative}, the authors used a robust perturbation observer to design an FO-SMC in order to extract the maximum power and improve the fault ride-through capability. They employed the $sgn\left(\cdot\right)\rightarrow sat\left(\cdot\right)$ substitution in the switching control law to minimize the chattering. In addition, the multi-objective grasshopper optimization algorithm was implemented to optimize the controller parameters in \cite{falehi2020innovative}. Comparative investigations were carried out and as reported, the proposed controller demonstrated better fault-tolerant performance compared to perturbation observer-based SO-SMC and simple SO-SMC. An FO-SMC based on feedback linearization technique was proposed in \cite{li2020fractional} to mitigate sub-synchronous control interaction in DFIG-based wind farms under various operating conditions. The authors proposed a fractional sliding surface along with the $sat\left(\cdot\right)$ function to alleviate the chattering. The proposed control strategy was compared with conventional vector control, feedback linearization SMC, and high-order SMC, where, faster subsynchronous control interaction damping performance was reported.

\subsection{PMSG (FO-SMC)}
\label{sec:3.4.2}
\sloppy Authors in \cite{xiong2020output} developed two gravitational search algorithm (GSA) -optimized FO-SMC methods for PMSG to enhance its output power quality. One was proposed to control the rotor-side $dq$ axis currents in the machine side converter, and the other one was proposed to regulate the output voltage and the DC-link voltage of the GSC. They proposed a fractional-calculus-based sliding surface to enhance the convergence speed and reduce the chattering. However, as the reported comparative investigations with conventional SMC and PI controllers have shown, despite the better tracking precision and stronger robustness against parametric disturbances, some slight chattering still remained with the proposed control scheme. A nonlinear FO-SMC scheme to improve the power production efficiency of a 5-phase PMSG-based WECS was proposed in \cite{rhaili2021enhanced}. The provided comparative investigations with SMC revealed that the proposed nonlinear FO sliding surface with fractional integration and differentiation components effectively reduced the chattering phenomenon, yielding superior power production over SMC and PI approaches.

\subsection{Other Generators (FO-SMC)}
\label{sec:3.4.3}
Some studies have investigated other types of generators and motors, such as PMSM. As an example, as shown in Figure \ref{fig:4} \cite{rui2019fractional}, an FO-SMC and a minimum order load torque observer was developed for the high-accuracy speed regulation of PMSM in WT based on a speed regulating differential mechanism. In this Figure, $i_{q}=i_{qs,1}+i_{qs,2}$ represents the control law, $\hat{T}_{L}$ is the observed load torque, $\omega_{s}$ is the reference rotor speed of PMSM, and $k_{t}$ is the torque constant. Comparative experimental results under operating conditions of changing wind speeds was carried out and as reported, the proposed FO-SMC with load torque observer could effectively eliminate undesirable chattering effect and provide better control performance in terms of state error minimization.

FO-SMC \cite{mousavi2021maximum, talebi2018fractional} and fractional-order backstepping SMC \cite{talebi2019fractional} were proposed to improve the variable speed WT's efficiency at a below-rated wind speed. Although the existence of the $sgn\left(\cdot\right)$ function in the switching control law could impose some chattering, the proposed FO sliding surfaces could effectively decrease the chattering and enhance the controllers' performance. The proposed controllers were compared with their integer-order modes, and better external disturbance rejection and power extraction were reported. Authors in \cite{kerrouche2019fractional} presented an FO-SMC controller with dual sequence decomposition technique to enhance the grid-connected wind farm's dynamic performance with distribution static synchronous compensator (D-STATCOM) by eliminating the oscillations of the active and reactive powers exchanged between the wind generator and the grid. The authors proposed an FO sliding surface to mitigate the chattering. As reported, the proposed control strategy reduced the torque oscillations, and the fatigue on the turbine shaft and the gearbox was diminished.

\begin{figure}
\centering
\includegraphics[width=4.5 in]{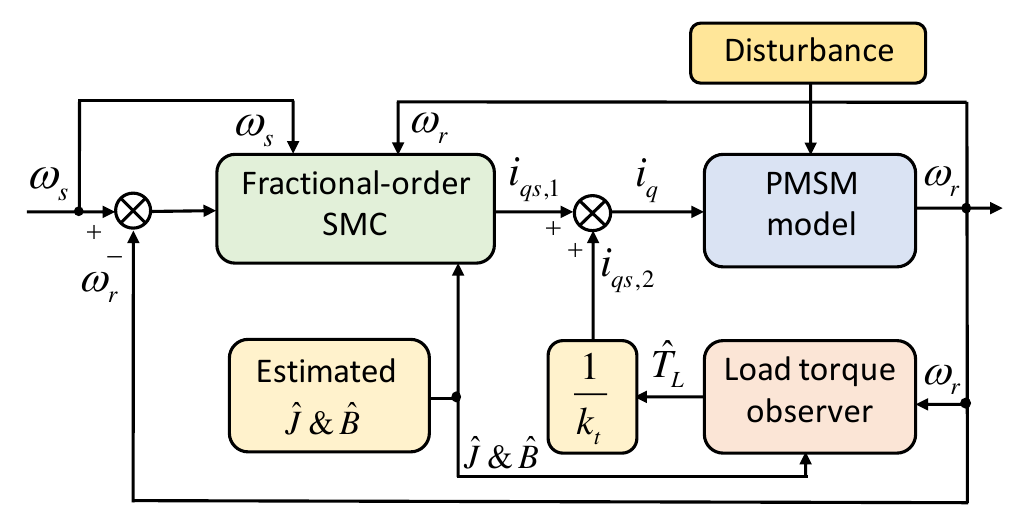}
\caption{The block diagram of the developed FO-SMC with load torque observer \cite{rui2019fractional}.}
\label{fig:4}
\end{figure}

\textcolor{black}{Table \ref{tab:4} summarizes the objectives, features, advantages, and disadvantages of the discussed studies on FO-SMC approaches for WECS control.}
\begin{table*}[h!]
\caption{\textcolor{black}{Summary of fractional-order SMC approaches for WECS control.}}
\centering
\label{tab:4}
\resizebox{\textwidth}{!}{
\begin{tabular}{l l l l l l}
\hline\hline \\[-3mm]
\begin{minipage}[t]{0.1\columnwidth} Work \end{minipage} & \begin{minipage}[t]{0.15\columnwidth} Objectives \end{minipage} & \begin{minipage}[t]{0.15\columnwidth} Operating region \end{minipage} & \begin{minipage}[t]{0.15\columnwidth} Generator \end{minipage} & \begin{minipage}[t]{0.6\columnwidth} Advantages \end{minipage}	& \begin{minipage}[t]{0.5\columnwidth}Disadvantages \end{minipage}\\ \hline
\begin{minipage}[t]{0.1\columnwidth} \cite{musarrat2021fractional}  \end{minipage} & \begin{minipage}[t]{0.15\columnwidth} DCVR, FTC  \end{minipage} &\begin{minipage}[t]{0.15\columnwidth} Full-load  \end{minipage} &\begin{minipage}[t]{0.12\columnwidth} DFIG  \end{minipage} & \begin{minipage}[t]{0.6\columnwidth} \begin{itemize}
\item Chattering reduced using FO sliding surface
\item Grid faults considered
\item Matched and mismatched disturbances estimated using an observer
\end{itemize} \end{minipage} & \begin{minipage}[t]{0.5\columnwidth} \begin{itemize}
\item No comparisons provided
\end{itemize} \end{minipage} \\ [12mm]
&&&&& \\
\rowcolor{lightgray}\begin{minipage}[t]{0.1\columnwidth} \cite{ullah2017adaptive}  \end{minipage} & \begin{minipage}[t]{0.15\columnwidth} RSC, GSC  \end{minipage} &\begin{minipage}[t]{0.15\columnwidth} Full-load  \end{minipage} &\begin{minipage}[t]{0.12\columnwidth} DFIG  \end{minipage} & \begin{minipage}[t]{0.6\columnwidth} \begin{itemize}
\item Chattering reduced using FO terminal sliding surface
\item Model uncertainties considered
\end{itemize} \end{minipage} & \begin{minipage}[t]{0.5\columnwidth} \begin{itemize}
\item Wind speed too smooth (not realistic)
\end{itemize} \end{minipage} \\ [12mm]
&&&&& \\
\begin{minipage}[t]{0.1\columnwidth} \cite{dash2018adaptive}  \end{minipage} & \begin{minipage}[t]{0.15\columnwidth} APC  \end{minipage} &\begin{minipage}[t]{0.15\columnwidth} Full-load  \end{minipage} &\begin{minipage}[t]{0.12\columnwidth} DFIG  \end{minipage} & \begin{minipage}[t]{0.6\columnwidth} \begin{itemize}
\item Chattering reduced using adaptive FO integral terminal sliding surface
\item Three-phase short-circuit fault considered as disturbance
\end{itemize} \end{minipage} & \begin{minipage}[t]{0.5\columnwidth} \begin{itemize}
\item Chattering reduction not completely investigated
\item Comparison only with PI
\end{itemize} \end{minipage} \\ [12mm]
&&&&& \\
\rowcolor{lightgray}\begin{minipage}[t]{0.1\columnwidth} \cite{mousavi2022active}  \end{minipage} & \begin{minipage}[t]{0.15\columnwidth} FTC, APC, RPC   \end{minipage} &\begin{minipage}[t]{0.15\columnwidth} Partial-load  \end{minipage} &\begin{minipage}[t]{0.12\columnwidth} DFIG  \end{minipage} & \begin{minipage}[t]{0.6\columnwidth} \begin{itemize}
\item Chattering reduced using adaptive FO integral terminal sliding surface (with full investigation)
\item Finite-time convergence ensured
\item Rotor current sensor faults and model uncertainties considered
\item Rotor dynamics estimated using an observer
\end{itemize} \end{minipage} & \begin{minipage}[t]{0.5\columnwidth} -- \end{minipage} \\ [12mm]
&&&&& \\
\begin{minipage}[t]{0.1\columnwidth} \cite{ebrahimkhani2016robust}  \end{minipage} & \begin{minipage}[t]{0.15\columnwidth} MPPT, APC  \end{minipage} &\begin{minipage}[t]{0.15\columnwidth} Full-load  \end{minipage} &\begin{minipage}[t]{0.12\columnwidth} DFIG  \end{minipage} & \begin{minipage}[t]{0.6\columnwidth} \begin{itemize}
\item Chattering reduced using adaptive FO sliding surface (with full investigation)
\item Finite-time convergence ensured
\item Rotor dynamics estimated using an observer
\end{itemize} \end{minipage} & \begin{minipage}[t]{0.5\columnwidth} \begin{itemize}
\item Comparison only with PI
\end{itemize} \end{minipage} \\ [12mm]
&&&&& \\
\rowcolor{lightgray}\begin{minipage}[t]{0.1\columnwidth} \cite{falehi2020innovative}  \end{minipage} & \begin{minipage}[t]{0.15\columnwidth} MPE, FTC  \end{minipage} &\begin{minipage}[t]{0.15\columnwidth} Full-load  \end{minipage} &\begin{minipage}[t]{0.12\columnwidth} DFIG  \end{minipage} & \begin{minipage}[t]{0.6\columnwidth} \begin{itemize}
\item Chattering reduced by $sgn\left(\cdot\right)\rightarrow sat\left(\cdot\right)$ substitution
\item Perturbation observer estimated model uncertainties
\item FO-SMC parameters tuned using Grasshopper algorithm
\end{itemize} \end{minipage} & \begin{minipage}[t]{0.5\columnwidth} \begin{itemize}
\item Chattering reduction was not achieved using FO surface, but by $sgn\left(\cdot\right)\rightarrow sat\left(\cdot\right)$ replacement
\end{itemize} \end{minipage} \\ [12mm] \hline
&&&&& \\
\begin{minipage}[t]{0.1\columnwidth} \cite{xiong2020output}  \end{minipage} & \begin{minipage}[t]{0.15\columnwidth} MPE, DCVR  \end{minipage} &\begin{minipage}[t]{0.15\columnwidth} Full-load  \end{minipage} &\begin{minipage}[t]{0.12\columnwidth} PMSG  \end{minipage} & \begin{minipage}[t]{0.6\columnwidth} \begin{itemize}
\item Chattering reduced using adaptive FO sliding surface
\item External disturbances considered
\end{itemize} \end{minipage} & \begin{minipage}[t]{0.5\columnwidth} \begin{itemize}
\item Chattering problem not properly tackled compared to SMC
\end{itemize} \end{minipage} \\ [12mm]
&&&&& \\
\rowcolor{lightgray}\begin{minipage}[t]{0.1\columnwidth} \cite{rhaili2021enhanced}  \end{minipage} & \begin{minipage}[t]{0.15\columnwidth} MPE  \end{minipage} &\begin{minipage}[t]{0.15\columnwidth} Full-load  \end{minipage} &\begin{minipage}[t]{0.12\columnwidth} PMSG  \end{minipage} & \begin{minipage}[t]{0.6\columnwidth} \begin{itemize}
\item Chattering reduced using FO sliding surface
\end{itemize} \end{minipage} & \begin{minipage}[t]{0.5\columnwidth} \begin{itemize}
\item External disturbances not considered
\item Chattering problem not properly tackled
\end{itemize} \end{minipage} \\ [12mm] \hline
&&&&& \\
\begin{minipage}[t]{0.1\columnwidth} \cite{mousavi2021maximum}  \end{minipage} & \begin{minipage}[t]{0.15\columnwidth} MPE, FTC  \end{minipage} &\begin{minipage}[t]{0.15\columnwidth} Partial-load  \end{minipage} &\begin{minipage}[t]{0.12\columnwidth} SSG  \end{minipage} & \begin{minipage}[t]{0.6\columnwidth} \begin{itemize}
\item Chattering mitigated using FO terminal sliding surface
\item Finite-time convergence ensured
\item Actuator faults considered
\end{itemize} \end{minipage} & \begin{minipage}[t]{0.5\columnwidth} -- \end{minipage} \\ [12mm]
&&&&& \\
\rowcolor{lightgray}\begin{minipage}[t]{0.1\columnwidth} \cite{rui2019fractional}  \end{minipage} & \begin{minipage}[t]{0.15\columnwidth} MPE  \end{minipage} &\begin{minipage}[t]{0.15\columnwidth} Partial \& Full-load  \end{minipage} &\begin{minipage}[t]{0.12\columnwidth} PMSM  \end{minipage} & \begin{minipage}[t]{0.6\columnwidth} \begin{itemize}
\item Chattering reduced using FO sliding surface
\item Grid load fluctuations considered
\item Load torque estimated using an observer
\end{itemize} \end{minipage} & \begin{minipage}[t]{0.5\columnwidth} -- \end{minipage} \\ [12mm]
&&&&& \\
\begin{minipage}[t]{0.1\columnwidth} \cite{talebi2018fractional}  \end{minipage} & \begin{minipage}[t]{0.15\columnwidth} MPE  \end{minipage} &\begin{minipage}[t]{0.15\columnwidth} Partial-load  \end{minipage} &\begin{minipage}[t]{0.12\columnwidth} SSG  \end{minipage} & \begin{minipage}[t]{0.6\columnwidth} \begin{itemize}
\item Chattering reduced using FO sliding surface
\item Finite-time convergence ensured
\item External disturbances and unmodelled dynamics considered
\end{itemize} \end{minipage} & \begin{minipage}[t]{0.5\columnwidth} \begin{itemize}
\item Chattering problem not properly tackled
\end{itemize} \end{minipage} \\ [12mm]
&&&&& \\
\rowcolor{lightgray}\begin{minipage}[t]{0.1\columnwidth} \cite{kerrouche2019fractional}  \end{minipage} & \begin{minipage}[t]{0.15\columnwidth} APC, RPC  \end{minipage} &\begin{minipage}[t]{0.15\columnwidth} Full-load  \end{minipage} &\begin{minipage}[t]{0.12\columnwidth} SSG  \end{minipage} & \begin{minipage}[t]{0.6\columnwidth} \begin{itemize}
\item Torque oscillations reduced and the fatigue on the turbine shaft and the gearbox was diminished
\end{itemize} \end{minipage} & \begin{minipage}[t]{0.5\columnwidth} \begin{itemize}
\item The chattering problem not investigated
\item Finite-time convergence not investigated
\item No comparisons provided
\end{itemize} \end{minipage} \\
[20mm]
\hline\hline
\end{tabular}}
\end{table*}

\section{Higher-order SMC for WECS}
\label{sec:3.5}
Despite the conventional SMCs that only consider the first-order derivative of the systems, higher-order SMCs (HO-SMCs) consider higher-order derivatives with respect to time. As a result, the sliding variable and its consecutive derivatives tend to zero in finite time with higher stabilization accuracy, the undesired effects of the chattering phenomenon are mitigated, and the control system is more robust against uncertainties and disturbances \cite{utkin2015discussion, mathiyalagan2020second, liu2014finite}. In the following subsections, HO-SMC control of WECS are categorised into three groups; (a) the SO-SMC, including the general SO-SMC and the super-twisting SMC (ST-SMC), which is a type of SO-SMC developed for systems with relative degree one to avoid the chattering in the main control loop \cite{gonzalez2011variable}, (b) the terminal SMC (TSMC) with remarkable merits including fast and finite-time convergence \cite{guo2020terminal}, well-proven chattering reduction property \cite{hou2019discrete,yi2019adaptive}, and strong robustness against system uncertainties and unmodelled dynamics \cite{wang2019discrete}, and (c) some other HO-SMC  methods. Figure \ref{fig:5} demonstrates the year-based distribution of studies devoted to SO-SMC and TSMC based control schemes for DFIG and PMSG based WECS during the past two decades.
\begin{figure}
\centering
\includegraphics[width=4.2 in]{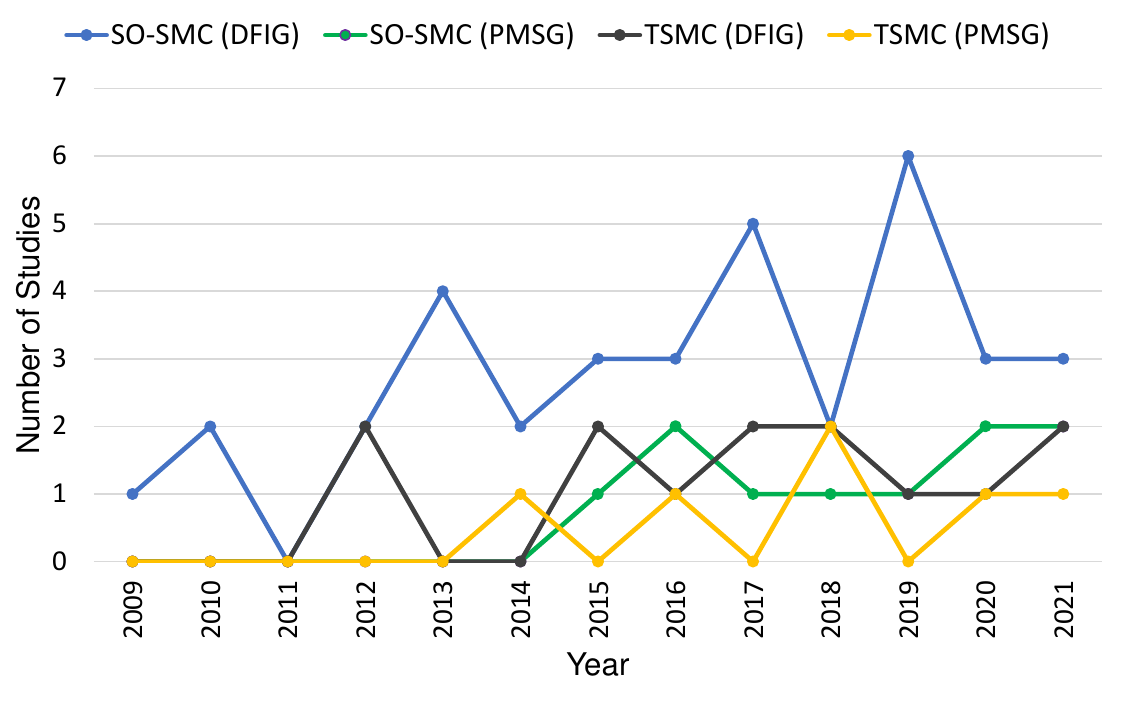}
\caption{The distribution of studies on SO-SMC and TSMC control of DFIG and PMSG based WECS over the past two decades.}
\label{fig:5}
\end{figure}

\textcolor{black}{A general representation of an adaptive ST-SMC for typical systems is illustrated in Figure \ref{fig:522}. It is worth mentioning that the sliding surface's design can be changed as per the designer's preference. Hence, Figure \ref{fig:522} shows an integer-order conventional sliding surface, while other types of sliding surfaces (such as fractional-order, fuzzy, etc.) can be embedded.}

\begin{figure}
\centering
\includegraphics[width=5 in]{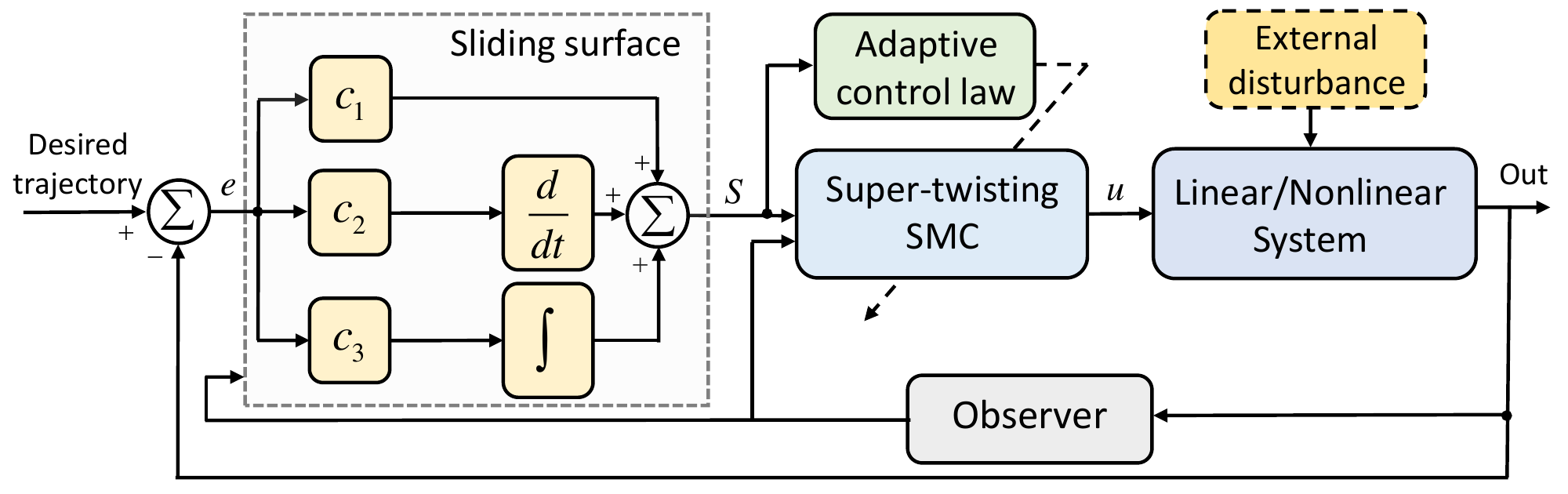}
\caption{\textcolor{black}{The general block diagram of an adaptive ST-SMC.}}
\label{fig:522}
\end{figure}

\subsection{Second-order SMC}
\label{sec:3.5.1}
In this section, the state-or-the-art literature of SO-SMC controllers used for different control problems of WECS is presented. The ST-SMC which is a type of SO-SMC \cite{evangelista2010wind, zhang2021individual} is also considered along with the general SO-SMC.

\subsubsection{DFIG (SO-SMC)}
\label{sec:3.5.1.1}
Authors in \cite{beltran2009high} designed a fault-tolerant SO-SMC for MPPT control of a DFIG-based 1.5 MW three-blade WT subjected to external disturbances and unmodelled dynamics, while later in \cite{benbouzid2014second}, the authors developed another SO-SMC controller to deal with the same problem in the presence of unbalanced voltage sags and grid frequency variations. Owing to utilizing higher-order sliding modes in both studies, the chattering were efficaciously mitigated. The comparative performance verifications in both works were provided, and superior fault ride-through performance of the proposed schemes were reported. In a similar study, an adaptive SO-SMC strategy was developed to maximize the captured energy and reduce the mechanical stress on the shaft of a WECS in the presence of model uncertainties, parameter variations, and external perturbations \cite{evangelista2012lyapunov}. The authors added a tuning parameter to the control law, allowing better behavior of the controlled system in terms of chattering elimination. Another adaptive SO-SMC control scheme was also investigated based on appropriate receding horizon adaptation time windows for DFIG-based WECS \cite{evangelista2016receding}. The authors added adaptive parameters to the control law to alleviate the chattering, where, as reported, the proposed method enhanced the adaptation strategy's reactivity against fast varying uncertainties with significantly reduced chattering.

Authors in \cite{martinez2017second, susperregui2013second} developed SO-SMC controllers for the grid synchronization and power control of 7MW DFIG-based WT in order to allow the WT to operate under distorted and unbalanced grid voltages. Both studies used the super-twisting (ST) algorithm to mitigate the chattering effect and achieve the desired dynamic performance. As the authors reported, the proposed control strategies demonstrated satisfactory robustness against parameter deviations and disturbances. Aiming at the active power maximization of a DFIG-based WT in region II, as well as applying power limitations in region III, SO-SMC strategies were presented in \cite{beltran2012second,evangelista2013active}. In order to guarantee the minimum chattering behavior of the control action, authors in \cite{ beltran2012second} presented an augmentation of the classic ST algorithm with the conventional SMC. In contrast, authors in \cite{ evangelista2013active} applied a modified ST algorithm with variable gains, which allows a wider range of operation with a smooth control action. According to the reported simulation and experimental results, both methods maintained desirable performances in terms of power maximization and minimum induced mechanical extra stress on the drivetrain.

A coordinated SO-SMC was developed in \cite{xiong2019coordinated} for grid synchronization and power optimization of a DFIG-based WECS. To maximize the captured wind energy, voltage regulation according to the grid requirements, and optimal tracking of the rotor speed, the authors utilized the ST algorithm. A reaching law was also proposed to accelerate the reaching speed, and the WT stator voltage was controlled to synchronize with the grid voltage without employing the current control loop. Figure \ref{fig:6} depicts the control structure investigated in the study, where $i_{r,dq}$ is the rotor current in $dq$ frame, $U_{g,dq}$ and $U_{r,dq}$ represent grid and rotor voltages in $dq$ frame, and $U_{s,abc}$ and $U_{s,dq}$ denote the stator voltages in $abc$ and $dq$ frames, respectively. As the authors reported, the proposed control scheme outperformed the conventional SMC in the presence of parameter perturbations.
\begin{figure}
\centering
\includegraphics[width=4.5 in]{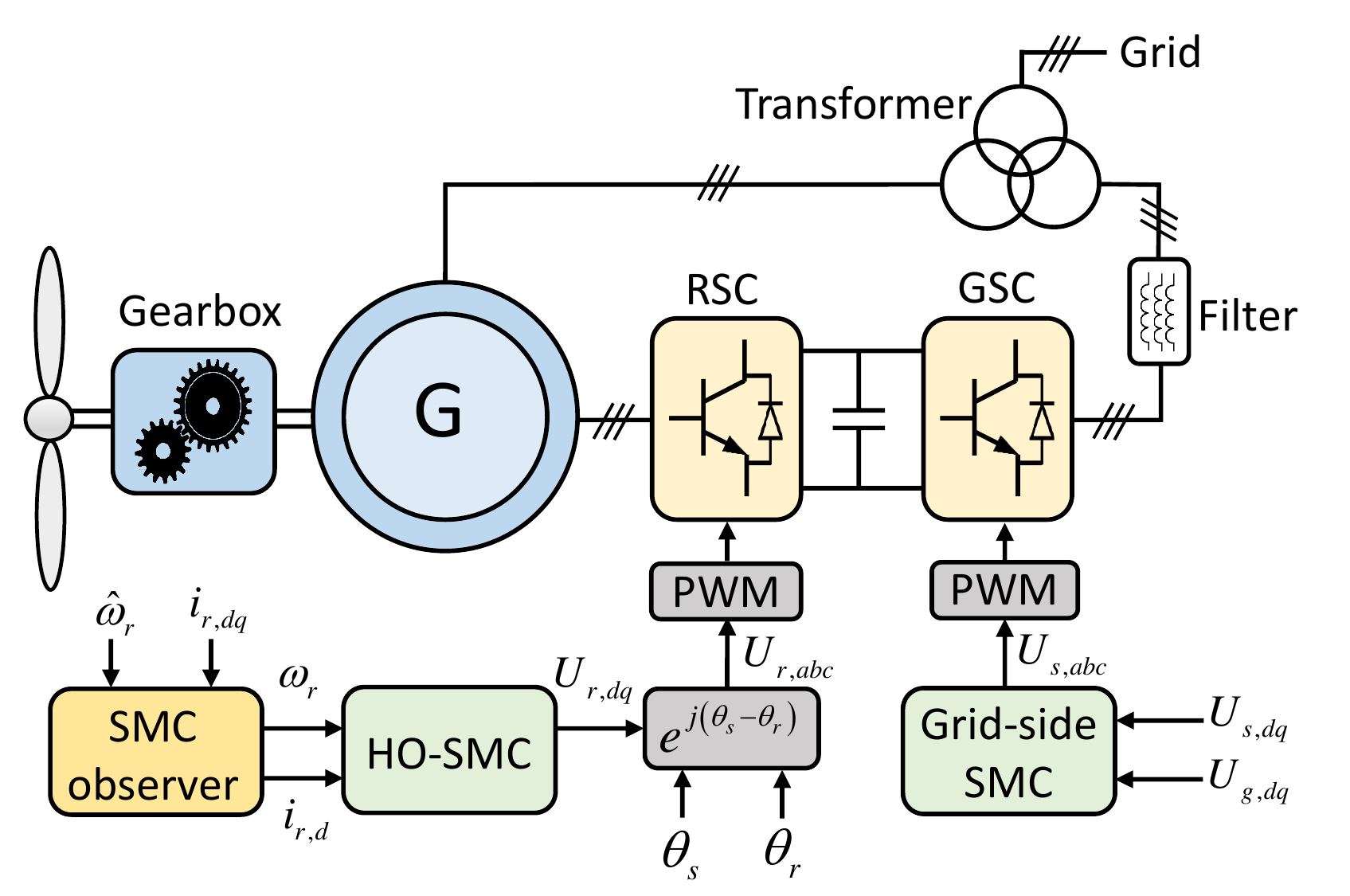}
\caption{The block diagram of the developed HO-SMC based control scheme with the SMC-based observer \cite{xiong2019coordinated}.}
\label{fig:6}
\end{figure}

In \cite{krim2018power}, a control strategy based on the augmentation of a power management supervisor for MPE and ST-SO-SMC for active and reactive load power control of standalone hybrid wind energy with battery energy storage system was developed. As the authors reported, the chattering problem was effectively handled using the ST algorithm. Furthermore, in comparison with conventional PI, the proposed control strategy demonstrated enhanced power quality of the installed load in the presence of sudden load uncertainties and external perturbations. In \cite{karabacak2019new}, by considering the energy stored in the turbine inertia, an inertial power-based perturb-and-observe third-order ST-SMC was proposed for MPPT control of WTs. Simulation and experimental comparisons were carried out, where as reported, the proposed scheme increased the MPPT efficiency for the step and ramp variations of the wind speed compared to the conventional SMC. In another study \cite{moussa2019super}, the authors implemented the strategy of oriented grid flux vector control and developed an ST-SMC controller for APC and RPC of DFIG-based WECS, while in contrast, authors in \cite{belabbas2019comparative} investigated the performance comparison of the backstepping SMC method and the ST-SMC for the same problem in hand. Both studies used the ST algorithm to reduce the chattering effect caused by the control switching; however, as the authors in \cite{belabbas2019comparative} reported, the backstepping control method demonstrated better robustness against parametric uncertainties in comparison with the developed ST- SMC.

Authors in \cite{evangelista2012multivariable} developed a multiple-input/multiple-output (MIMO) ST-SO-SMC controller for stator reactive power regulation and MPE of grid-connected DFIG-based WECS operating in region II. As reported, taking advantage of the ST algorithm, the proposed control scheme delivered a desirable MPE performance with minimum chattering. Later in \cite{liu2016second, benamor2019novel}, the authors dealt with the same problem by proposing ST-SO-SMC schemes and utilizing the quadratic form Lyapunov function method to guarantee the finite-time stabilization and determine the ST-SO-SMC parameters. Advantaging from the ST algorithm, reduced chattering phenomenon and higher sliding precision were reported. Both control schemes were compared with previously presented conventional methods, and their superiorities and robustness against uncertainties and external disturbances were validated. Authors in \cite{barambones2021real} proposed an SMO-based ST-SMC control scheme for MPE of DFIG-based WTs subjected to uncertainties. The ST-based torque observer was implemented to estimate the wind speed variations. They used the ST algorithm to alleviate the chattering phenomenon. Simulations and experimental investigations were conducted, where as reported, the proposed scheme maintained the MPE objective under system uncertainties and wind speed variations.

By employing the ST algorithm in the controller design to cope with the chattering phenomenon, authors in \cite{xiong2020high} proposed a SO-SMC-based direct power control strategy for DFIG-based WT operating under unbalanced grid voltage conditions. In similar studies, authors in \cite{morshed2018sliding,djilali2020first} proposed augmentations of ST sliding mode disturbance observer with SO-SMC. Comparisons with conventional methods were also provided, and the proposed strategies' better performance in terms of power quality improvement and maintaining constant power output under non-ideal grid conditions were reported. More studies can be found in \cite{tria2017integral, mazen2020modeling}.

\subsubsection{PMSG (SO-SMC)}
\label{sec:3.5.1.2}
Authors in \cite{meghni2017second} proposed an augmentation of FLC and SO-SMC to extract the maximum wind power in a hybrid system consisting of a PMSG-based variable speed WT and batteries as energy storage system. To attenuate the grid current harmonics and achieve a smooth regulation of grid active and reactive powers quantities, another SO-SMC was developed, and its performance in terms of grid power disturbances mitigation and maximum power extraction was reported in comparison with conventional SMC. In an experimental investigation \cite{merabet2016implementation}, a control scheme based on a SO-SMC and the disturbed single input-single output error model to control the generator and the grid-sides of a variable speed experimental WECS driven by the OPAL-RT real-time simulator was proposed. According to the authors, the proposed second-order sliding surface efficiently mitigated the chattering, and hence, the developed control scheme effectively controlled the generator speed and DC-link voltage in the presence of unknown disturbances, parametric variations, and uncertainties.

A two-stage cascade structured control strategy for a grid-connected 2 MW gearless PMSG-based WECS was proposed in \cite{valenciaga2015multiple}. The authors designed a second-order sliding surface to mitigate the chattering, resulting in a SO-SMC control scheme that simultaneously minimized the resistive losses into the generator and regulated the active and reactive powers delivered to the grid. The proposed scheme's robustness against unmodelled dynamics, external disturbances, and three-phase voltage dips was studied. As reported, the proposed control scheme guaranteed the finite-time convergence and alleviated the chattering problem. A comparative study of SO-SMC and conventional SMC through robust control of PMSG-based WECS in the presence of external disturbances was also presented in \cite{krim2018classical}. The system was connected to the grid through a resistor-inductor filter, and the chattering problem of SMC was reported to be attenuated in the SO-SMC by utilizing the ST algorithm. More studies can be found in \cite{krim2019second, zhang2021adaptive,li2021variable,hou2020composite,ma2021sensorless, ullah2020variable}.

\textcolor{black}{Table \ref{tab:5} summarizes the objectives, features, advantages, and disadvantages of the above-discussed SO-SMC designs for WECS control.}
\begin{table*}[h!]
\caption{\textcolor{black}{Summary of SO-SMC approaches for WECS control.}}
\centering
\label{tab:5}
\resizebox{\textwidth}{!}{
\begin{tabular}{l l l l l l}
\hline\hline \\[-3mm]
\begin{minipage}[t]{0.1\columnwidth} Work \end{minipage} & \begin{minipage}[t]{0.12\columnwidth} Objectives \end{minipage} & \begin{minipage}[t]{0.12\columnwidth} Operating region \end{minipage} & \begin{minipage}[t]{0.12\columnwidth} Generator \end{minipage} & \begin{minipage}[t]{0.6\columnwidth} Advantages \end{minipage}	& \begin{minipage}[t]{0.5\columnwidth}Disadvantages \end{minipage}\\ \hline
\begin{minipage}[t]{0.1\columnwidth} \cite{benbouzid2014second}  \end{minipage} & \begin{minipage}[t]{0.12\columnwidth} MPE  \end{minipage} &\begin{minipage}[t]{0.12\columnwidth} Full-load  \end{minipage} &\begin{minipage}[t]{0.12\columnwidth} DFIG  \end{minipage} & \begin{minipage}[t]{0.6\columnwidth} \begin{itemize}
\item Unbalanced voltage sags and grid frequency variations considered
\end{itemize} \end{minipage} & \begin{minipage}[t]{0.5\columnwidth} \begin{itemize}
\item Chattering reduction not investigated
\item No comparisons provided
\end{itemize} \end{minipage} \\ [12mm]
&&&&& \\
\rowcolor{lightgray}\begin{minipage}[t]{0.1\columnwidth} \cite{evangelista2012lyapunov}  \end{minipage} & \begin{minipage}[t]{0.12\columnwidth} MPE, MSM  \end{minipage} &\begin{minipage}[t]{0.12\columnwidth} Full-load  \end{minipage} &\begin{minipage}[t]{0.12\columnwidth} DFIG  \end{minipage} & \begin{minipage}[t]{0.6\columnwidth} \begin{itemize}
\item Chattering reduced using adaptive ST-SMC with variable gain
\item External disturbances and parameter uncertainties considered
\end{itemize} \end{minipage} & \begin{minipage}[t]{0.5\columnwidth} \begin{itemize}
\item Chattering reduction not completely investigated
\item No comparisons provided
\end{itemize} \end{minipage} \\ [12mm]
&&&&& \\
\begin{minipage}[t]{0.1\columnwidth} \cite{evangelista2016receding}  \end{minipage} & \begin{minipage}[t]{0.12\columnwidth} MPE  \end{minipage} &\begin{minipage}[t]{0.12\columnwidth} Partial-load  \end{minipage} &\begin{minipage}[t]{0.12\columnwidth} DFIG  \end{minipage} & \begin{minipage}[t]{0.6\columnwidth} \begin{itemize}
\item Chattering reduced using adaptive SO-SMC with variable gain
\item Fast varying disturbances considered
\end{itemize} \end{minipage} & \begin{minipage}[t]{0.5\columnwidth} \begin{itemize}
\item No comparisons provided
\end{itemize} \end{minipage} \\ [12mm]
&&&&& \\
\rowcolor{lightgray}\begin{minipage}[t]{0.1\columnwidth} \cite{martinez2017second, susperregui2013second}  \end{minipage} & \begin{minipage}[t]{0.12\columnwidth} APC, Grid synchronization  \end{minipage} &\begin{minipage}[t]{0.12\columnwidth} Full-load  \end{minipage} &\begin{minipage}[t]{0.12\columnwidth} DFIG  \end{minipage} & \begin{minipage}[t]{0.6\columnwidth} \begin{itemize}
\item Chattering reduced using adaptive ST-SMC
\item Unbalanced grid voltages considered
\item Experimental investigations provided
\end{itemize} \end{minipage} & \begin{minipage}[t]{0.5\columnwidth} \begin{itemize}
\item Chattering reduction not completely investigated
\item No comparisons provided
\end{itemize} \end{minipage} \\ [12mm]
&&&&& \\
\begin{minipage}[t]{0.1\columnwidth} \cite{beltran2012second}  \end{minipage} & \begin{minipage}[t]{0.12\columnwidth} MPE  \end{minipage} &\begin{minipage}[t]{0.12\columnwidth} Full-load  \end{minipage} &\begin{minipage}[t]{0.12\columnwidth} DFIG  \end{minipage} & \begin{minipage}[t]{0.6\columnwidth} \begin{itemize}
\item Chattering reduced using ST-SMC
\item Experimental investigations provided
\item External disturbances and unmodelled dynamics considered
\end{itemize} \end{minipage} & \begin{minipage}[t]{0.5\columnwidth} \begin{itemize}
\item Chattering reduction not completely investigated
\item No comparisons provided
\end{itemize} \end{minipage} \\ [12mm]
&&&&& \\
\rowcolor{lightgray}\begin{minipage}[t]{0.1\columnwidth} \cite{evangelista2013active}  \end{minipage} & \begin{minipage}[t]{0.12\columnwidth} MPE, MSM  \end{minipage} &\begin{minipage}[t]{0.12\columnwidth} Full-load  \end{minipage} &\begin{minipage}[t]{0.12\columnwidth} DFIG  \end{minipage} & \begin{minipage}[t]{0.6\columnwidth} \begin{itemize}
\item Chattering reduced using adaptive ST-SMC with variable gain
\item External disturbances considered
\end{itemize} \end{minipage} & \begin{minipage}[t]{0.5\columnwidth} \begin{itemize}
\item Chattering reduction not completely investigated
\item No comparisons provided
\end{itemize} \end{minipage} \\ [12mm]
&&&&& \\
\begin{minipage}[t]{0.1\columnwidth} \cite{xiong2019coordinated}  \end{minipage} & \begin{minipage}[t]{0.12\columnwidth} MPE, Grid synchronization  \end{minipage} &\begin{minipage}[t]{0.12\columnwidth} Partial-load  \end{minipage} &\begin{minipage}[t]{0.12\columnwidth} DFIG  \end{minipage} & \begin{minipage}[t]{0.6\columnwidth} \begin{itemize}
\item Chattering reduced using adaptive ST-SMC
\item Parameter perturbations estimated using an observer
\end{itemize} \end{minipage} & \begin{minipage}[t]{0.5\columnwidth} \begin{itemize}
\item Wind speed too smooth (not realistic)
\end{itemize} \end{minipage} \\ [12mm]
&&&&& \\
\rowcolor{lightgray}\begin{minipage}[t]{0.1\columnwidth} \cite{karabacak2019new}  \end{minipage} & \begin{minipage}[t]{0.12\columnwidth} MPE  \end{minipage} &\begin{minipage}[t]{0.12\columnwidth} Full-load  \end{minipage} &\begin{minipage}[t]{0.12\columnwidth} DFIG  \end{minipage} & \begin{minipage}[t]{0.6\columnwidth} \begin{itemize}
\item Chattering reduced using third-order ST-SMC
\item External perturbations estimated using an observer
\item Experimental investigations provided
\end{itemize} \end{minipage} & \begin{minipage}[t]{0.5\columnwidth} \begin{itemize}
\item Wind speed too smooth (not realistic)
\end{itemize} \end{minipage} \\ [12mm]
&&&&& \\
\begin{minipage}[t]{0.1\columnwidth} \cite{liu2016second,benamor2019novel}  \end{minipage} & \begin{minipage}[t]{0.12\columnwidth} MPE, RPC  \end{minipage} &\begin{minipage}[t]{0.12\columnwidth} Partial-load  \end{minipage} &\begin{minipage}[t]{0.12\columnwidth} DFIG  \end{minipage} & \begin{minipage}[t]{0.6\columnwidth} \begin{itemize}
\item Chattering reduced using adaptive ST-SMC
\item External disturbances considered
\end{itemize} \end{minipage} & \begin{minipage}[t]{0.5\columnwidth} \begin{itemize}
\item Wind speed too smooth (not realistic)
\item Chattering reduction not completely investigated
\end{itemize} \end{minipage} \\ [12mm]
&&&&& \\
\rowcolor{lightgray}\begin{minipage}[t]{0.1\columnwidth} \cite{barambones2021real}  \end{minipage} & \begin{minipage}[t]{0.12\columnwidth} MPE  \end{minipage} &\begin{minipage}[t]{0.12\columnwidth} Full-load  \end{minipage} &\begin{minipage}[t]{0.12\columnwidth} DFIG  \end{minipage} & \begin{minipage}[t]{0.6\columnwidth} \begin{itemize}
\item Chattering reduced using adaptive ST-SMC
\item SMO estimated the wind speed variations
\item Experimental investigations provided
\end{itemize} \end{minipage} & \begin{minipage}[t]{0.5\columnwidth} \begin{itemize}
\item Wind speed too smooth (not realistic)
\item Comparison only with PID (no advanced method)
\end{itemize} \end{minipage} \\ [12mm]
&&&&& \\
\begin{minipage}[t]{0.1\columnwidth} \cite{xiong2020high}  \end{minipage} & \begin{minipage}[t]{0.12\columnwidth} APC, RPC  \end{minipage} &\begin{minipage}[t]{0.12\columnwidth} Full-load  \end{minipage} &\begin{minipage}[t]{0.12\columnwidth} DFIG  \end{minipage} & \begin{minipage}[t]{0.6\columnwidth} \begin{itemize}
\item Chattering reduced using ST-SMC
\item Unbalanced grid voltages considered
\item Experimental investigations provided
\end{itemize} \end{minipage} & \begin{minipage}[t]{0.5\columnwidth} -- \end{minipage} \\ [12mm]
&&&&& \\ \hline
\rowcolor{lightgray}\begin{minipage}[t]{0.1\columnwidth} \cite{meghni2017second}  \end{minipage} & \begin{minipage}[t]{0.12\columnwidth} MPE, APC, RPC  \end{minipage} &\begin{minipage}[t]{0.12\columnwidth} Full-load  \end{minipage} &\begin{minipage}[t]{0.12\columnwidth} PMSG  \end{minipage} & \begin{minipage}[t]{0.6\columnwidth} \begin{itemize}
\item FLC augmented with ST-SMC to reduce chattering
\item Grid disturbances considered
\end{itemize} \end{minipage} & \begin{minipage}[t]{0.5\columnwidth} -- \end{minipage} \\ [12mm]
&&&&& \\
\begin{minipage}[t]{0.1\columnwidth} \cite{merabet2016implementation}  \end{minipage} & \begin{minipage}[t]{0.12\columnwidth} GSC, RSC, DCVR  \end{minipage} &\begin{minipage}[t]{0.12\columnwidth} Full-load  \end{minipage} &\begin{minipage}[t]{0.12\columnwidth} PMSG  \end{minipage} & \begin{minipage}[t]{0.6\columnwidth} \begin{itemize}
\item Chattering reduced using ST-SMC
\item External disturbances and parametric uncertainties considered
\item Experimental investigations provided
\end{itemize} \end{minipage} & \begin{minipage}[t]{0.5\columnwidth} \begin{itemize}
\item Chattering reduction not completely investigated
\item No comparisons provided
\end{itemize} \end{minipage} \\ [12mm]
&&&&& \\
\rowcolor{lightgray}\begin{minipage}[t]{0.1\columnwidth} \cite{valenciaga2015multiple}  \end{minipage} & \begin{minipage}[t]{0.12\columnwidth} APC, RPC, FTC  \end{minipage} &\begin{minipage}[t]{0.12\columnwidth} Full-load  \end{minipage} &\begin{minipage}[t]{0.12\columnwidth} PMSG  \end{minipage} & \begin{minipage}[t]{0.6\columnwidth} \begin{itemize}
\item Chattering reduced using SO-SMC
\item External disturbances and unmodelled dynamics considered
\end{itemize} \end{minipage} & \begin{minipage}[t]{0.5\columnwidth} \begin{itemize}
\item Chattering reduction not completely investigated
\item No comparisons provided
\end{itemize} \end{minipage} \\ [12mm] \hline
&&&&& \\
\begin{minipage}[t]{0.1\columnwidth} \cite{hou2020composite}  \end{minipage} & \begin{minipage}[t]{0.12\columnwidth} MPE  \end{minipage} &\begin{minipage}[t]{0.12\columnwidth} Partial-load  \end{minipage} &\begin{minipage}[t]{0.12\columnwidth} PMSM  \end{minipage} & \begin{minipage}[t]{0.6\columnwidth} \begin{itemize}
\item A disturbance observer estimated the lumped uncertainties
\item Experimental investigations provided
\end{itemize} \end{minipage} & \begin{minipage}[t]{0.5\columnwidth} \begin{itemize}
\item Wind speed too smooth (not realistic)
\item Chattering reduction not investigated
\end{itemize} \end{minipage} \\
[10mm]
\hline\hline
\end{tabular}}
\end{table*}

\subsection{Terminal SMC}
\label{sec:3.5.2}
Remarkable merits of TSMC controllers including fast and finite-time convergence \cite{guo2020terminal}, reduced chattering \cite{hou2019discrete, panchade2013survey,mousavi2022active}, and strong robustness against system uncertainties and unmodelled dynamics \cite{wang2019discrete,zheng2017full} have encouraged many scholars to investigate its performance in the control problem designs for WECS. Accordingly, this section provides an overview of the published literature in this area.

\subsubsection{DFIG (TSMC)}
\label{sec:3.5.2.1}
Authors in \cite{patnaik2015fast} dealt with the control problem of DFIG-based WECS subjected to diverse and challenging situations by developing a higher-order SMC. They utilized a fast adaptive TSMC with a nonlinear sliding surface to mitigate the chattering phenomenon. According to the authors, the proposed controller illustrated more robustness toward external disturbances and wind speed variations than the conventional PI and feedback linearization methods. In other studies, the authors developed incorporations of a TSMC and higher-order SMO \cite{benbouzid2012high}, fuzzy TSMC controller with a fuzzy adaptive observer \cite{wu2018maximal}, and FO-TSMC controller \cite{soomro2021wind} to control the active power of DFIG-based WT and mitigate the chattering phenomenon. As reported, in terms of steady-state tracking error and high control precision, the overall control schemes demonstrated superior performance over the conventional methods in the presence of external disturbances and unmodelled generator and turbine dynamics.

Authors in \cite{morshed2015comparison} proposed an integral TSMC approach to enhance the power quality of WTs in the presence of uncertainties, parameter variations, and severe voltage sag conditions. However, although the chattering was reduced in the developed control scheme compared to the conventional SMC, the existence of a $sgn\left(\cdot\right)$ function in the control law had led to an undesirable chattering problem. Accordingly, later, the same authors proposed an FLC-based auto-tuned integral TSMC \cite{morshed2019sliding} to simultaneously deal with the power quality enhancement of WTs and chattering elimination. As reported, the chattering was effectively mitigated. Compared to the conventional SMC, the superiorities of the proposed control scheme in terms of active power, currents, DC-link voltage, and electromagnetic torque quality were revealed. Active power enhancement and stability improvement of a DFIG-based wind farm was investigated by developing a third-order TSMC control scheme \cite{patnaik2020adaptive}. According to the authors, taking advantage of the terminal sliding surface, the chattering problem was mitigated; and hence, the proposed scheme maintained better synchronism between all generators in the multi-machine system for both fixed and chaotic wind profiles compared with the conventional PI.

Authors in \cite{patnaik2016adaptive} proposed an adaptive TSMC for APC and RPC of the rotor-side and the grid-side converters in a DFIG-based WT. The proposed strategy utilized the \textit{abc} frame of reference, and its gain parameters were defined dynamically, depending on the absolute value of the respective sliding surfaces. The authors proposed a nonlinear terminal sliding surface to reduce chattering. However, some slight chattering remained due to the existence of the discontinuous $sgn\left(\cdot\right)$ function in the switching law. Alternatively, they used the $sgn\left(\cdot\right)\rightarrow tanh\left(\cdot\right)$ replacement to overcome the issue and attenuate the chattering. The effectiveness of the proposed scheme for the WT control was investigated in a multi-machine environment in the presence of switching faults, reference tracking, and wind speed variations.

\subsubsection{PMSG (TSMC)}
\label{sec:3.5.2.2}
The three-phase GSC control problem of variable speed PMSG-based WECS was investigated in \cite{zheng2018integral}. The authors proposed two integral-type TSMC controllers with attenuated chattering to control the active and reactive powers exchanged between the converter and the grid. In order to reduce the chattering, they used integrators to soften the switching signals and generate continuous control signals. Experimental investigations were carried out and the performance of the proposed control strategy in terms of control energy wastage reduction and robustness against matched and unmatched parametric uncertainties was reported. In \cite{pradhan2018composite}, an augmentation of NTSMC with a soft-switching SMO was proposed to deal with the control problem of remotely located PMSG-based photovoltaic-wind hybrid systems. According to the experimental results reported by the authors, the chattering problem was effectively alleviated, and the incorporation of the NTSMC-based speed controller with the proposed observer-based disturbance rejection method guaranteed the robustness against model uncertainties and external disturbances.

\subsubsection{Other Generators (TSMC)}
\label{sec:3.5.2.3}
A fault-tolerant FO nonsingular TSMC strategy was proposed in \cite{mousavi2021maximum} for MPE of WTs subjected to actuator faults. The design was based on an FO nonsingular terminal sliding surface to guarantee the fast convergence speed of system states and simultaneously suppresses chattering. According to the authors, the proposed scheme demonstrated superior power production performance over the SO-SMC and conventional SMC approaches in fault-free and faulty situations. Authors in \cite{abolvafaei2019maximum} proposed a composite control strategy for the MPE of WT systems operating in region II; the incorporation of a second-order fast TSMC and SMC with PID-based sliding surface. According to the authors, the chattering problem and mechanical loads were reduced utilizing the PID-based SMC, whilst the TSMC extracted the maximum wind power and guaranteed the fast finite-time convergence. As reported, the proposed control scheme demonstrated superior performance compared to indirect speed control, the nonlinear static state feedback control, the nonlinear dynamic state feedback control, and the first-order SMC in the presence of parametric uncertainties, external disturbances, and unmodelled dynamics.

Maximum power capture \cite{rajendran2015adaptive} and pitch angle control \cite{aghaeinezhad2021individual} problems of variable speed WTs were investigated by developing adaptive nonsingular TSMC controllers. Authors in \cite{rajendran2015adaptive} used $tanh\left(\cdot\right)$ to reduce the chattering, while authors in \cite{aghaeinezhad2021individual} proposed a fractional-based sliding surface to deal with this phenomenon. As reported, the proposed schemes outperformed conventional methods in terms of robustness against input disturbances and parametric uncertainties.

\textcolor{black}{Table \ref{tab:6} summarizes the objectives, features, advantages, and disadvantages of the above-discussed TSMC approaches for WECS control.}
\begin{table*}[ht!]
\caption{\textcolor{black}{Summary of TSMC approaches for WECS control.}}
\centering
\label{tab:6}
\resizebox{\textwidth}{!}{
\begin{tabular}{l l l l l l}
\hline\hline \\[-3mm]
\begin{minipage}[t]{0.1\columnwidth} Work \end{minipage} & \begin{minipage}[t]{0.12\columnwidth} Objectives \end{minipage} & \begin{minipage}[t]{0.12\columnwidth} Operating region \end{minipage} & \begin{minipage}[t]{0.12\columnwidth} Generator \end{minipage} & \begin{minipage}[t]{0.6\columnwidth} Advantages \end{minipage}	& \begin{minipage}[t]{0.5\columnwidth}Disadvantages \end{minipage}\\ \hline
\begin{minipage}[t]{0.1\columnwidth} \cite{patnaik2015fast}  \end{minipage} & \begin{minipage}[t]{0.12\columnwidth} APC, RPC  \end{minipage} &\begin{minipage}[t]{0.12\columnwidth} Partial \& Full-load  \end{minipage} &\begin{minipage}[t]{0.12\columnwidth} DFIG  \end{minipage} & \begin{minipage}[t]{0.6\columnwidth} \begin{itemize}
\item Chattering reduced using fast adaptive TSMC
\item External disturbances and wind variations considered
\end{itemize} \end{minipage} & \begin{minipage}[t]{0.5\columnwidth} \begin{itemize}
\item Comparison only with PI (no advanced method)
\end{itemize} \end{minipage} \\ [12mm]
&&&&& \\
\rowcolor{lightgray}\begin{minipage}[t]{0.1\columnwidth} \cite{soomro2021wind}  \end{minipage} & \begin{minipage}[t]{0.12\columnwidth} MPE  \end{minipage} &\begin{minipage}[t]{0.12\columnwidth} Partial-load  \end{minipage} &\begin{minipage}[t]{0.12\columnwidth} DFIG  \end{minipage} & \begin{minipage}[t]{0.6\columnwidth} \begin{itemize}
\item Chattering reduced using ST-FO-TSMC
\item External disturbances and unmodelled dynamics considered
\end{itemize} \end{minipage} & \begin{minipage}[t]{0.5\columnwidth} \begin{itemize}
\item Comparison only with PI (no advanced method)
\item Not good chattering mitigation performance
\end{itemize} \end{minipage} \\ [12mm]
&&&&& \\
\begin{minipage}[t]{0.1\columnwidth} \cite{morshed2015comparison}  \end{minipage} & \begin{minipage}[t]{0.12\columnwidth} APC, FTC  \end{minipage} &\begin{minipage}[t]{0.12\columnwidth} Full-load  \end{minipage} &\begin{minipage}[t]{0.12\columnwidth} DFIG  \end{minipage} & \begin{minipage}[t]{0.6\columnwidth} \begin{itemize}
\item Chattering reduced using integral TSMC
\item Parameter variations and grid voltage sag considered
\end{itemize} \end{minipage} & \begin{minipage}[t]{0.5\columnwidth} \begin{itemize}
\item Not good chattering mitigation performance
\end{itemize} \end{minipage} \\ [12mm]
&&&&& \\
\rowcolor{lightgray}\begin{minipage}[t]{0.1\columnwidth} \cite{patnaik2020adaptive}  \end{minipage} & \begin{minipage}[t]{0.12\columnwidth} APC, FTC  \end{minipage} &\begin{minipage}[t]{0.12\columnwidth} Full-load  \end{minipage} &\begin{minipage}[t]{0.12\columnwidth} DFIG  \end{minipage} & \begin{minipage}[t]{0.6\columnwidth} \begin{itemize}
\item Chattering reduced using TSMC
\item Short circuit fault considered
\end{itemize} \end{minipage} & \begin{minipage}[t]{0.5\columnwidth} \begin{itemize}
\item Comparison only with PI (no advanced method)
\end{itemize} \end{minipage} \\ [12mm]
&&&&& \\
\begin{minipage}[t]{0.1\columnwidth} \cite{morshed2019sliding}  \end{minipage} & \begin{minipage}[t]{0.12\columnwidth} APC, FTC  \end{minipage} &\begin{minipage}[t]{0.12\columnwidth} Full-load  \end{minipage} &\begin{minipage}[t]{0.12\columnwidth} DFIG  \end{minipage} & \begin{minipage}[t]{0.6\columnwidth} \begin{itemize}
\item Chattering reduced using FLC-based integral TSMC
\item External disturbances and unmodelled dynamics estimated using SMO
\end{itemize} \end{minipage} & \begin{minipage}[t]{0.5\columnwidth} -- \end{minipage} \\ [12mm]
&&&&& \\ \hline
\rowcolor{lightgray}\begin{minipage}[t]{0.1\columnwidth} \cite{zheng2018integral}  \end{minipage} & \begin{minipage}[t]{0.12\columnwidth} APC, RPC  \end{minipage} &\begin{minipage}[t]{0.12\columnwidth} Full-load  \end{minipage} &\begin{minipage}[t]{0.12\columnwidth} PMSG  \end{minipage} & \begin{minipage}[t]{0.6\columnwidth} \begin{itemize}
\item Chattering reduced using integral TSMC
\item Matched and unmatched parameter uncertainties considered
\item Experimental investigations provided
\end{itemize} \end{minipage} & \begin{minipage}[t]{0.5\columnwidth} \begin{itemize}
\item No comparisons provided
\end{itemize} \end{minipage} \\ [12mm]
&&&&& \\
\begin{minipage}[t]{0.1\columnwidth} \cite{pradhan2018composite}  \end{minipage} & \begin{minipage}[t]{0.12\columnwidth} MPE, APC, FTC  \end{minipage} &\begin{minipage}[t]{0.12\columnwidth} Partial \& Full-load  \end{minipage} &\begin{minipage}[t]{0.12\columnwidth} PMSG  \end{minipage} & \begin{minipage}[t]{0.6\columnwidth} \begin{itemize}
\item Chattering reduced using TSMC
\item External disturbances estimated using SMO
\item Experimental investigations provided
\end{itemize} \end{minipage} & \begin{minipage}[t]{0.5\columnwidth} \begin{itemize}
\item No comparisons provided
\end{itemize} \end{minipage} \\ [12mm]
&&&&& \\ \hline
\rowcolor{lightgray}\begin{minipage}[t]{0.1\columnwidth} \cite{abolvafaei2019maximum}  \end{minipage} & \begin{minipage}[t]{0.12\columnwidth} MPE, MSM  \end{minipage} &\begin{minipage}[t]{0.12\columnwidth} Partial-load  \end{minipage} &\begin{minipage}[t]{0.12\columnwidth} SSG  \end{minipage} & \begin{minipage}[t]{0.6\columnwidth} \begin{itemize}
\item Chattering reduced using SO-TSMC
\item External disturbances and parametric uncertainties considered
\end{itemize} \end{minipage} & \begin{minipage}[t]{0.5\columnwidth} \begin{itemize}
\item Chattering reduction not completely investigated
\end{itemize} \end{minipage} \\ [12mm]
&&&&& \\
\begin{minipage}[t]{0.1\columnwidth} \cite{aghaeinezhad2021individual}  \end{minipage} & \begin{minipage}[t]{0.12\columnwidth} PAC  \end{minipage} &\begin{minipage}[t]{0.12\columnwidth} Full-load  \end{minipage} &\begin{minipage}[t]{0.12\columnwidth} SSG  \end{minipage} & \begin{minipage}[t]{0.6\columnwidth} \begin{itemize}
\item Chattering reduced using adaptive FO-TSMC
\item External disturbances and parametric uncertainties considered
\end{itemize} \end{minipage} & \begin{minipage}[t]{0.5\columnwidth} \begin{itemize}
\item Chattering reduction not completely investigated
\item No comparisons provided
\end{itemize} \end{minipage} \\
[20mm]
\hline\hline
\end{tabular}}
\end{table*}

\subsection{Other HO-SMC}
\label{sec:3.5.3}
Apart from WECS applications of SO-SMC, ST-SO-SMC, and TSMC, other HO-SMC approaches have also been investigated in the literature, which will be reviewed in this section.

\subsubsection{DFIG (other HO-SMC)}
\label{sec:3.5.3.1}
In \cite{golnary2018design, golnary2019dynamic}, the maximum wind energy capture of a 1.5MW WT was investigated. The authors proposed adaptive neuro-fuzzy inference system (ANFIS) -based wind speed estimators, and augmentations of the HO-SMC based observer and feedback linearization \cite{golnary2018design} and SO-SMC \cite{golnary2019dynamic} were developed to track the rotor speed. Both studies considered the derivation of generator torque as a virtual control function to avoid the chattering phenomenon. Comparative investigations were conducted and the superiorities of the proposed methods in terms of power extraction and load mitigation in comparison with the indirect speed control and quasi-continuous SMC methods were reported.

\subsubsection{Other Generators (other HO-SMC)}
\label{sec:3.5.3.2}
Authors in \cite{merida2014analysis} proposed non-homogeneous quasi-continuous HO-SMC to deal with simultaneous MPE and mechanical load reduction problem for variable speed WTs. They also incorporated a wind speed estimator with a sliding mode output feedback torque controller to ensure the achievement of optimal rotor speed. On the other hand, authors in \cite{pratap2018robust} dealt with the same problem by developing a quasi-sliding mode control approach. They introduced an auxiliary sliding variable with a positive tunable constant to attenuate the chattering. As reported, compared with conventional SMC, the proposed control scheme alleviated the chattering phenomenon while guaranteeing robustness against parametric uncertainties, external disturbances, and unmodelled dynamics. In another study \cite{zhang2009new}, the authors developed a quasi-continuous SMC to deal with the pitch control problem of WTs. However, although a superior performance of the developed control scheme than that of PI was reported, the chattering problem still existed with the controller. Hence, they mitigated the chattering effects by $sgn\left(\cdot\right)\rightarrow sat\left(\cdot\right)$ substitution in the control law. The incorporation of aerodynamic torque observer and HO-SMC was developed in \cite{beltran2008high} to deal with the control problem of variable speed WT working in regions II and III. Due to the proposed second-order SMO with the controller, no chattering in the rotor speed was reported. The proposed control strategy was reported effective in terms of power regulation and robustness against parametric uncertainties.

\section{Fuzzy SMC for WECS}
\label{sec:3.6}
The estimation of the system's parameter uncertainties increases the adjustability of the switching gain of SMC, leading to chattering attenuation and controller's robustness attainment \cite{xu2019event}. Since fuzzy inference systems (FISs) can approximate unknown continuous functions without any requirement for prior knowledge of parametric uncertainties and external disturbances bounds, they have been widely used together with SMCs in practical applications \cite{ma2018cooperative, sami2020sensorless, hwang2013adaptive}. Also, FIS have been utilized as the reaching control to replace the discontinuous switching control law and smooth the SMC's control signal; however, they cannot precisely drive the sliding surface to zero and thus cannot ensure an ideal sliding motion \cite{wu2018adaptive}. \textcolor{black}{A general representation of an adaptive Fuzzy-SMC (F-SMC) for typical systems is depicted in Figure \ref{fig:533}. The fuzzifier block converts the signal values into fuzzy linguistic values, and the fuzzy interface mechanism adaptively determines the control action based on the fuzzy rules. The control signal is then transferred to the system using the defuzzifier block. It is worth mentioning that the sliding surface's design is as per the designer's preference and the system's dynamics.}

\begin{figure}
\centering
\includegraphics[width=4.5 in]{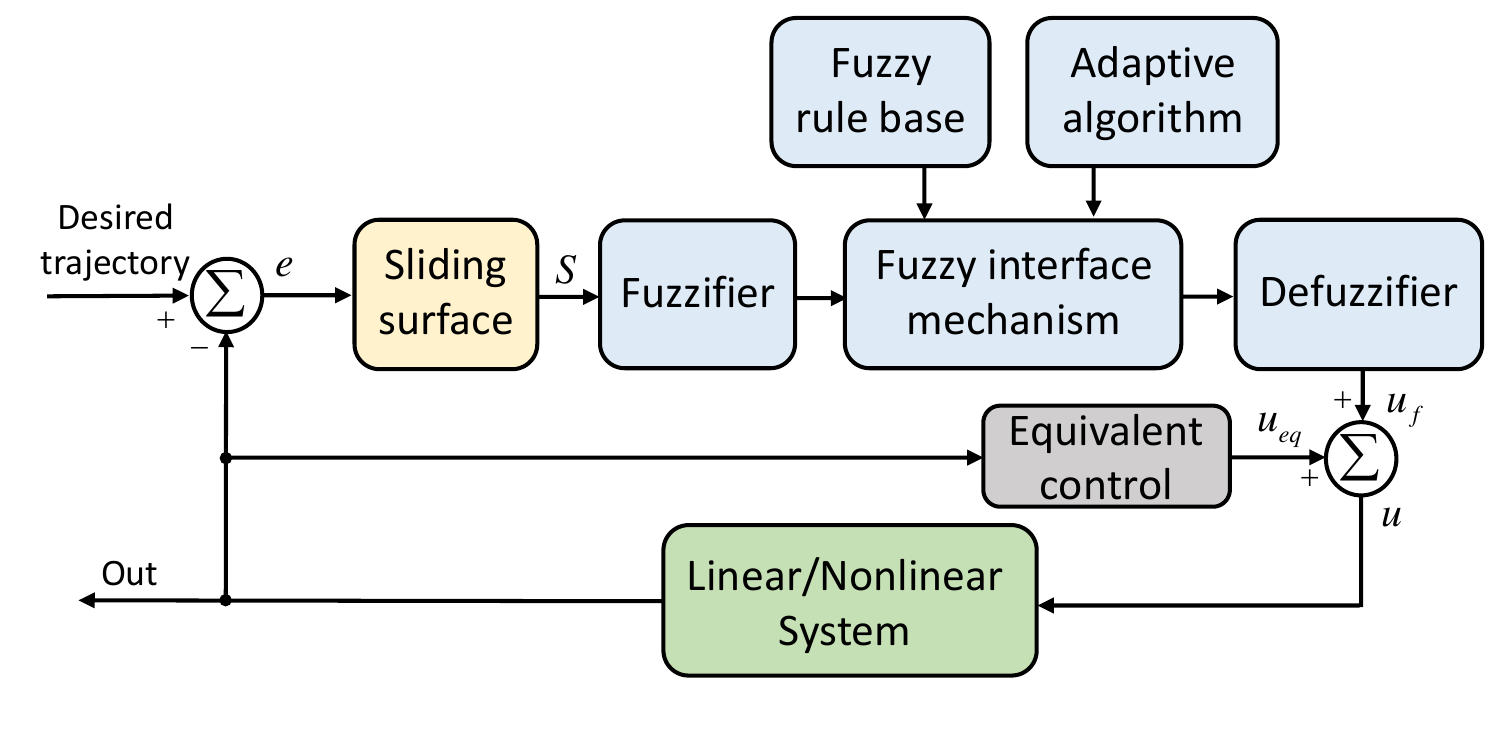}
\caption{\textcolor{black}{The general block diagram of an adaptive fuzzy SMC.}}
\label{fig:533}
\end{figure}

\subsection{DFIG (F-SMC)}
\label{sec:3.6.1}
An optimal F-SMC for DFIG-based WT was proposed \cite{bounar2019pso} to achieve the MPE with zero stator reactive power regulation. To this end, the FIS was deployed to avoid the undesirable chattering phenomenon in SMC, and a combination of particle swarm optimization (PSO) and GSA was proposed to tune the control parameters optimally. An adaptive F-SMC was designed for MPE of a variable speed adjustable pitch DFIG-based WT \cite{yao2007adaptive}. The authors incorporated an adaptive fuzzy-based weight with the SMC to reduce the chattering effects. Compared to the conventional SMC and PID controllers, the authors reported that their controller performs better with respect to parameter variations and load disturbances. Later, authors in \cite{rajendran2014variable, subramaniyam2021memory} developed adaptive fuzzy ISMC controllers to deal with the same problem as \cite{yao2007adaptive}. Similar to \cite{yao2007adaptive}, the authors used fuzzy gain adjustment procedures to adaptively tune the weighting parameters in the switching control law, and reduce the chattering. In addition, authors in \cite{rajendran2014variable} proposed a modified Newton-Raphson estimator for precise estimation of effective wind speed, and as a result, more robustness of the proposed strategy was reported in the presence of input torque disturbance, compared to conventional SMC and ISMC. In a similar study \cite{liu2014sliding}, an SMC with fuzzy switching gain adjustment was developed to tackle the MPE control problem. The authors used the $sgn\left(\cdot\right)\rightarrow tanh\left(\cdot\right)$ change to deal with the chattering problem. However, although the designed control system demonstrated exemplary performance in the presence of disturbances and uncertainties, they were not as satisfactory as \cite{rajendran2014variable}.

\subsection{PMSG (F-SMC)}
\label{sec:3.6.2}
In \cite{yin2015novel}, a fuzzy integral sliding mode current control strategy was proposed to extract the maximum wind power and high-order voltage harmonics mitigation for a direct-driven WECS. The fuzzy logic generator was employed to derive the peak power points based on the variations of measured DC-side voltage and current, while the integral SMC was to track the peak power points. The proposed MPPT control strategy is shown in Figure \ref{fig:7}, where the changes of DC-side current and voltage over a sampling period and the changes of optimum DC-side current are chosen as the input and output variables, respectively. Here, $\Delta{I_{dc}}$ and $\Delta{V_{dc}}$ represent the DC-side current and voltage variations, respectively, $V_{dc}\left[k\right]$, $I_{dc}\left[k\right]$, $V_{dc}\left[k-1\right]$, and $I_{dc}\left[k-1\right]$ denote the DC-side voltage and current at the sampling instants $k$ and $k-1$, respectively.
\begin{figure}
\centering
\includegraphics[width=4.2 in]{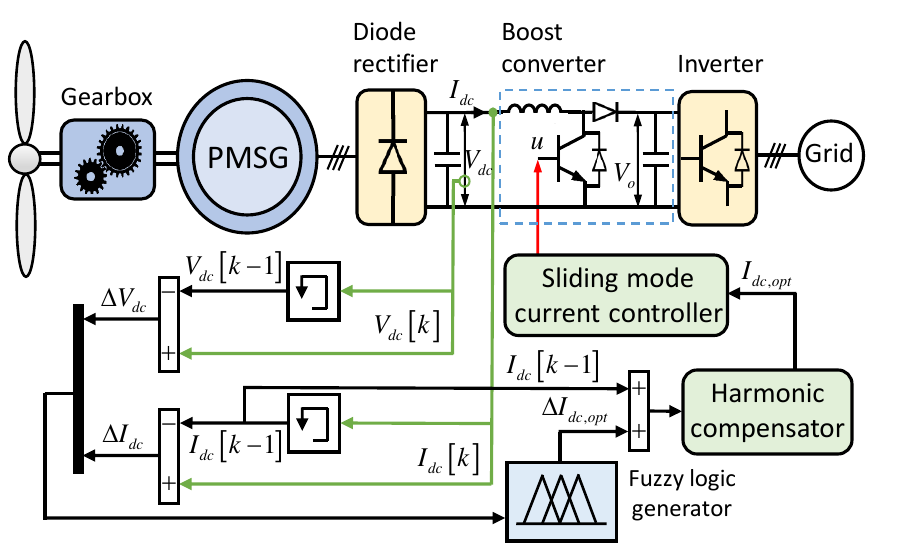}
\caption{\textcolor{black}{The block diagram of the developed F-SMC-based MPPT control of PMSG investigated in \cite{yin2015novel}.}}
\label{fig:7}
\end{figure}

Taking advantage of the Takagi--Sugeno (TS) fuzzy model, a disturbance observer-based integral F-SMC with diminished chattering was designed in \cite{hwang2019disturbance} for MPE of a PMSG-based WTs. The asymptotic stability of the control system with $H_{\infty }$ was derived in terms of the linear matrix inequality format, and the reachability of the sliding surface was derived based on the Lyapunov function. In a similar study \cite{do2016disturbance}, the authors developed an F-SMC for WECS control, combined with a nonlinear disturbance observer for estimation of aerodynamic torque and angular shaft speed reference. They embedded a fuzzy-based variable switching gain scheme with the controller to mitigate the chattering. In \cite{subramaniam2019passivity}, a passivity-based fuzzy ISMC for PMSG-based WECS was investigated. The authors developed a fuzzy-based sliding surface to alleviate the undesirable chattering effects. They used the double orthogonal complement and Frobenius theorem to present the fuzzy sliding surface's existence condition. Also, by constructing the proper Lyapunov function, the system's asymptotic stability was achieved.

\subsection{Other Generators (F-SMC)}
\label{sec:3.6.3}
An F-SMC control strategy was developed for MPPT control of a PMSM motor in \cite{benchabane2012improved}, while later, authors in \cite{lahlou2019sliding} designed an SMC combined with a type-2 fuzzy PID for MPPT of a variable speed WT. {They replaced the sliding surface with a type-2 fuzzy PID surface to reduce the chattering phenomenon.} Compared with the conventional SMC, better performance of the proposed control strategy in terms of better chattering reduction and robustness against large parametric uncertainties was reported. Authors in \cite{tahir2018new} proposed an F-SMC controller for wound-field synchronous generator-based variable-speed WECS in order to maximize the power extracted from the WT. The FIS was used for on-line adjustment of the switching gains, while a fast sigmoid function with a variable boundary layer was proposed to reduce the chattering phenomenon. An actuator faults diagnosis and FTC approach for WTs with a hydrostatic transmission was presented in \cite{schulte2015fault} by incorporating a TS fuzzy system and an SMO. According to the authors, the TS observer maintained the sliding motion on the surface with reduced chattering. As reported, in the presence of actuator faults, the proposed FTC yielded similar results to that of the fault-free case.

\textcolor{black}{Table \ref{tab:7} summarizes the objectives, features, advantages, and disadvantages of the discussed F-SMC designs for WECS control.}
\begin{table*}[ht!]
\caption{\textcolor{black}{Summary of F-SMC approaches for WECS control.}}
\centering
\label{tab:7}
\resizebox{\textwidth}{!}{
\begin{tabular}{l l l l l l}
\hline\hline \\[-3mm]
\begin{minipage}[t]{0.1\columnwidth} Work \end{minipage} & \begin{minipage}[t]{0.12\columnwidth} Objectives \end{minipage} & \begin{minipage}[t]{0.12\columnwidth} Operating region \end{minipage} & \begin{minipage}[t]{0.12\columnwidth} Generator \end{minipage} & \begin{minipage}[t]{0.6\columnwidth} Advantages \end{minipage}	& \begin{minipage}[t]{0.5\columnwidth}Disadvantages \end{minipage}\\ \hline
\begin{minipage}[t]{0.1\columnwidth} \cite{bounar2019pso}  \end{minipage} & \begin{minipage}[t]{0.12\columnwidth} MPE \end{minipage} &\begin{minipage}[t]{0.12\columnwidth} Partial-load  \end{minipage} &\begin{minipage}[t]{0.12\columnwidth} DFIG  \end{minipage} & \begin{minipage}[t]{0.6\columnwidth} \begin{itemize}
\item Chattering reduced using F-SMC
\item Controller parameters tuned by PSO and GSA
\item External disturbances and parametric uncertainties considered
\end{itemize} \end{minipage} & \begin{minipage}[t]{0.5\columnwidth} \begin{itemize}
\item Chattering reduction not completely investigated
\item No comparisons provided
\end{itemize} \end{minipage} \\ [12mm]
&&&&& \\
\rowcolor{lightgray}\begin{minipage}[t]{0.1\columnwidth} \cite{rajendran2014variable}  \end{minipage} & \begin{minipage}[t]{0.12\columnwidth} MPE \end{minipage} &\begin{minipage}[t]{0.12\columnwidth} Partial-load  \end{minipage} &\begin{minipage}[t]{0.12\columnwidth} DFIG  \end{minipage} & \begin{minipage}[t]{0.6\columnwidth} \begin{itemize}
\item Chattering reduced using adaptive integral F-SMC
\item Wind speed estimated using Newton-Raphson estimator
\item Input disturbances considered
\end{itemize} \end{minipage} & \begin{minipage}[t]{0.5\columnwidth} \begin{itemize}
\item Chattering reduction not completely investigated
\item No comparisons provided
\end{itemize} \end{minipage} \\ [12mm]
&&&&& \\
\begin{minipage}[t]{0.1\columnwidth} \cite{subramaniyam2021memory}  \end{minipage} & \begin{minipage}[t]{0.12\columnwidth} MPE \end{minipage} &\begin{minipage}[t]{0.12\columnwidth} Partial-load  \end{minipage} &\begin{minipage}[t]{0.12\columnwidth} DFIG  \end{minipage} & \begin{minipage}[t]{0.6\columnwidth} \begin{itemize}
\item Chattering reduced using adaptive integral F-SMC
\end{itemize} \end{minipage} & \begin{minipage}[t]{0.5\columnwidth} \begin{itemize}
\item Chattering reduction not completely investigated
\item No disturbances or uncertainties considered
\item No comparisons provided
\end{itemize} \end{minipage} \\ [12mm]
&&&&& \\
\rowcolor{lightgray}\begin{minipage}[t]{0.1\columnwidth} \cite{liu2014sliding}  \end{minipage} & \begin{minipage}[t]{0.12\columnwidth} MPE \end{minipage} &\begin{minipage}[t]{0.12\columnwidth} Partial-load  \end{minipage} &\begin{minipage}[t]{0.12\columnwidth} DFIG  \end{minipage} & \begin{minipage}[t]{0.6\columnwidth} \begin{itemize}
\item External disturbances and parametric uncertainties considered
\end{itemize} \end{minipage} & \begin{minipage}[t]{0.5\columnwidth} \begin{itemize}
\item Implementation of F-SMC failed to mitigate the chattering problem. Authors used the $sgn\left(\cdot\right)\rightarrow sat\left(\cdot\right)$ replacement to deal with the chattering
\end{itemize} \end{minipage} \\ [12mm] \hline
&&&&& \\
\begin{minipage}[t]{0.1\columnwidth} \cite{yin2015novel}  \end{minipage} & \begin{minipage}[t]{0.12\columnwidth} MPE, DCVR \end{minipage} &\begin{minipage}[t]{0.12\columnwidth} Partial-load  \end{minipage} &\begin{minipage}[t]{0.12\columnwidth} PMSG  \end{minipage} & \begin{minipage}[t]{0.6\columnwidth} \begin{itemize}
\item Chattering reduced using F-SMC
\item Voltage harmonics at the generator side eliminated
\end{itemize} \end{minipage} & \begin{minipage}[t]{0.5\columnwidth} \begin{itemize}
\item No disturbances or uncertainties considered
\item Comparison only with PI (no advanced method)
\end{itemize} \end{minipage} \\ [12mm]
&&&&& \\
\rowcolor{lightgray}\begin{minipage}[t]{0.1\columnwidth} \cite{hwang2019disturbance}  \end{minipage} & \begin{minipage}[t]{0.12\columnwidth} MPE \end{minipage} &\begin{minipage}[t]{0.12\columnwidth} Partial-load  \end{minipage} &\begin{minipage}[t]{0.12\columnwidth} PMSG  \end{minipage} & \begin{minipage}[t]{0.6\columnwidth} \begin{itemize}
\item MPE achieved using a F-SMC
\item External disturbances estimated using a fuzzy-based disturbance observer
\end{itemize} \end{minipage} & \begin{minipage}[t]{0.5\columnwidth} \begin{itemize}
\item No investigation on chattering mitigation provided
\end{itemize} \end{minipage} \\ [12mm]
&&&&& \\
\begin{minipage}[t]{0.1\columnwidth} \cite{do2016disturbance}  \end{minipage} & \begin{minipage}[t]{0.12\columnwidth} MPE \end{minipage} &\begin{minipage}[t]{0.12\columnwidth} Partial-load  \end{minipage} &\begin{minipage}[t]{0.12\columnwidth} PMSG  \end{minipage} & \begin{minipage}[t]{0.6\columnwidth} \begin{itemize}
\item Chattering reduced using F-SMC
\item Aerodynamic torque estimated using a nonlinear disturbance observer
\end{itemize} \end{minipage} & \begin{minipage}[t]{0.5\columnwidth} \begin{itemize}
\item No comparisons provided
\end{itemize} \end{minipage} \\ [12mm]
&&&&& \\
\rowcolor{lightgray}\begin{minipage}[t]{0.1\columnwidth} \cite{subramaniam2019passivity}  \end{minipage} & \begin{minipage}[t]{0.12\columnwidth} MPE \end{minipage} &\begin{minipage}[t]{0.12\columnwidth} Partial-load  \end{minipage} &\begin{minipage}[t]{0.12\columnwidth} PMSG  \end{minipage} & \begin{minipage}[t]{0.6\columnwidth} \begin{itemize}
\item Chattering reduced using integral F-SMC
\item External disturbances considered
\end{itemize} \end{minipage} & \begin{minipage}[t]{0.5\columnwidth} \begin{itemize}
\item No comparisons provided
\end{itemize} \end{minipage} \\ [12mm] \hline
&&&&& \\
\begin{minipage}[t]{0.1\columnwidth} \cite{schulte2015fault}  \end{minipage} & \begin{minipage}[t]{0.12\columnwidth} FTC \end{minipage} &\begin{minipage}[t]{0.12\columnwidth} Partial-load  \end{minipage} &\begin{minipage}[t]{0.12\columnwidth} SSG  \end{minipage} & \begin{minipage}[t]{0.6\columnwidth} \begin{itemize}
\item Chattering reduced using F-SMC
\item SMO estimated the actuator faults
\end{itemize} \end{minipage} & \begin{minipage}[t]{0.5\columnwidth} \begin{itemize}
\item Chattering not investigated
\item No comparisons provided
\end{itemize} \end{minipage} \\
[12mm]
\hline\hline
\end{tabular}}
\end{table*}

\section{Neural Network-based SMC for WECS}
\label{sec:3.7}
One of the main drawbacks of conventional SMC and HO-SMC controllers in practical applications is the requirement to know the parameter variations and external disturbances; information that is frequently not known in practice. Thus, designers often choose it as a conservatively large value, which results in considerably large control inputs, chattering occurrence, and also higher power consumption of the system. As an alternative, neural network (NN) -based control strategies have been used to approximate the uncertainties and disturbances along with their bounds \cite{ling150robust, qu2018neural, van2019adaptive}.

\subsection{DFIG (NN-SMC)}
\label{sec:3.7.1}
Authors in \cite{djilali2019real, djilali2017neural} presented real-time discrete sliding mode field-oriented controllers for the MPE problem of grid-connected DFIG-based WT systems in both balanced and unbalanced grid conditions. The proposed schemes were augmented with a recurrent high-order NN estimator for DC-link and DFIG mathematical models approximation. Simulation \cite{djilali2017neural} and experimental \cite{djilali2019real} investigations on a $1/4$ HP DFIG prototype in the presence of external disturbances and unknown dynamics were carried out. In comparison with the conventional PI and SMC methods, the merits of the proposed studies with regards to reference tracking, robustness to parameter variations, and sensitivity to speed changes were reported. Authors in \cite{djilali2018neural} developed a robust predictor neural SMC scheme for DFIG-based WT systems in the presence of parameter variations and external disturbances. In order to suppress the chattering, the authors used $tanh\left(\cdot\right)$ within the sliding surface. As reported, the higher-order NN identifier, which was trained on-line using an extended Kalman filter, was effectively augmented with the SMC strategy and successfully compensated the measurement delay in stator and rotor current.

\subsection{PMSG (NN-SMC)}
\label{sec:3.7.2}
In \cite{yin2020recurrent, boulouma2016rbf}, the MPE problem of PMSG-based WECS was investigated using a radial basis function neural network (RBFNN)-based SMC approach. The perturbation boundaries and full system states were assumed to be unknown, where, according to the authors, due to synthesizing a real-time dynamic learning law, the presented control scheme demonstrated superior performance in terms of chattering reduction and optimum generated power in comparison with SO-SMC. The block diagram of the control scheme investigated in \cite{yin2020recurrent} is illustrated in Figure \ref{fig:8}, where the RBFNN is used to identify the uncertain WT dynamics, and an online update algorithm is derived to update the weights of the RBFNN.
\begin{figure}
\centering
\includegraphics[width=3.2 in]{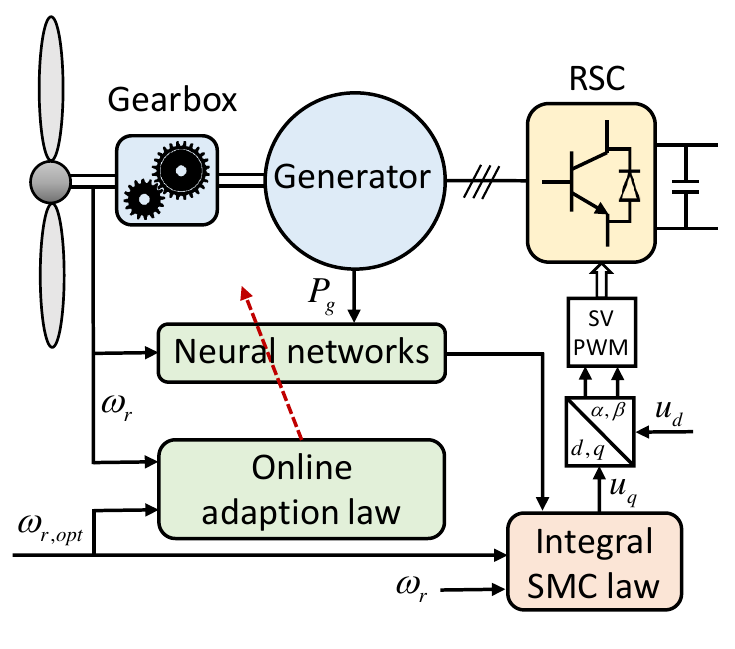}
\caption{Block diagram of the developed RBFNN-based SMC approach for maximum power capture in \cite{yin2020recurrent}.}
\label{fig:8}
\end{figure}

Generalized global SMC controllers incorporated with a feed-forward NN to estimate the nonlinear drift terms and input channels \cite{haq2020maximum} and support vector machine NN to mitigate the chattering \cite{xin2014chattering} were investigated to deal with the MPPT problem of PMSG-based WT systems. Authors in \cite{haq2020maximum} compared the proposed strategy with the feedback linearized control approach, where better performance was exhibited in the presence of wind speed and parametric variations.

\subsection{Other Generators (NN-SMC)}
\label{sec:3.7.3}
In \cite{hong2019enhanced}, a radial fuzzy wavelet NN was proposed and incorporated with SMC for a switched reluctance generator of variable speed WECS. The proposed strategy utilized the hill-climb searching, and as the authors reported, compared to the conventional PI and fuzzy control methods, it achieved faster convergence to MPPT. By assuming unknown system states and perturbations boundaries, incorporations of RBFNN with ISMC were investigated for MPE of SFG-based \cite{lamzouri2018robust} and self-excited induction generator (SEIG) \cite{kenne2017new} WT systems. In order to mitigate the chattering effects, authors in \cite{lamzouri2018robust} used the $sgn\left(\cdot\right)\rightarrow sat\left(\cdot\right)$ substitution, while authors in \cite{kenne2017new} used the RBFNN to adaptively tune the switching gain. Optimal NN-SMC control schemes were also investigated for variable speed WECS using adaptive PSO \cite{boufounas2014optimal} and genetic algorithm \cite{berrada2015optimal}. Although both NN-SMCs were reported with reduced chattering, authors in \cite{boufounas2014optimal} replaced advantaged from the $sgn\left(\cdot\right)\rightarrow sat\left(\cdot\right)$ replacement in the control law to achieve better results. According to the authors, utilization of evolutionary optimization algorithms have effectively enhanced the trajectory tracking performance of NN-SMC controllers.

\textcolor{black}{Table \ref{tab:8} summarizes the objectives, features, advantages, and disadvantages of the discussed studies on NN-SMC approaches for WECS control.}
\begin{table*}[t!]
\caption{\textcolor{black}{Summary of NN-SMC approaches for WECS control.}}
\centering
\label{tab:8}
\resizebox{\textwidth}{!}{
\begin{tabular}{l l l l l l}
\hline\hline \\[-3mm]
\begin{minipage}[t]{0.1\columnwidth} Work \end{minipage} & \begin{minipage}[t]{0.12\columnwidth} Objectives \end{minipage} & \begin{minipage}[t]{0.12\columnwidth} Operating region \end{minipage} & \begin{minipage}[t]{0.12\columnwidth} Generator \end{minipage} & \begin{minipage}[t]{0.6\columnwidth} Advantages \end{minipage}	& \begin{minipage}[t]{0.5\columnwidth}Disadvantages \end{minipage}\\ \hline
\begin{minipage}[t]{0.1\columnwidth} \cite{djilali2019real}  \end{minipage} & \begin{minipage}[t]{0.12\columnwidth} APC, RPC \end{minipage} &\begin{minipage}[t]{0.12\columnwidth} Full-load  \end{minipage} &\begin{minipage}[t]{0.12\columnwidth} DFIG  \end{minipage} & \begin{minipage}[t]{0.6\columnwidth} \begin{itemize}
\item Balanced and unbalanced grid conditions considered
\item Recurrent high-order NN estimator used for DC-link voltage approximation
\item Experimental investigations provided
\item External disturbances and unknown dynamics considered
\end{itemize} \end{minipage} & \begin{minipage}[t]{0.5\columnwidth} \begin{itemize}
\item Chattering problem not investigated
\end{itemize} \end{minipage} \\ [12mm]
&&&&& \\
\rowcolor{lightgray}\begin{minipage}[t]{0.1\columnwidth} \cite{djilali2018neural}  \end{minipage} & \begin{minipage}[t]{0.12\columnwidth} APC, RPC \end{minipage} &\begin{minipage}[t]{0.12\columnwidth} Full-load  \end{minipage} &\begin{minipage}[t]{0.12\columnwidth} DFIG  \end{minipage} & \begin{minipage}[t]{0.6\columnwidth} \begin{itemize}
\item Chattering reduced by the $sgn\left(\cdot\right)\rightarrow sat\left(\cdot\right)$ replacement
\item NN-SMC predictor compensated the current measurement delays
\item Parameter variations and external disturbances considered
\end{itemize} \end{minipage} & \begin{minipage}[t]{0.5\columnwidth} \begin{itemize}
\item No comparisons provided
\end{itemize} \end{minipage} \\ [12mm] \hline
&&&&& \\
\begin{minipage}[t]{0.1\columnwidth} \cite{yin2020recurrent}  \end{minipage} & \begin{minipage}[t]{0.12\columnwidth} MPE \end{minipage} &\begin{minipage}[t]{0.12\columnwidth} Partial-load  \end{minipage} &\begin{minipage}[t]{0.12\columnwidth} PMSG  \end{minipage} & \begin{minipage}[t]{0.6\columnwidth} \begin{itemize}
\item RBFNN compensated the external disturbances
\item RBFNN augmented with integral SMC
\end{itemize} \end{minipage} & \begin{minipage}[t]{0.5\columnwidth} \begin{itemize}
\item Comparison with only PI (no advance method)
\end{itemize} \end{minipage} \\ [12mm]
&&&&& \\
\rowcolor{lightgray}\begin{minipage}[t]{0.1\columnwidth} \cite{haq2020maximum}  \end{minipage} & \begin{minipage}[t]{0.12\columnwidth} MPE \end{minipage} &\begin{minipage}[t]{0.12\columnwidth} Partial-load  \end{minipage} &\begin{minipage}[t]{0.12\columnwidth} PMSG  \end{minipage} & \begin{minipage}[t]{0.6\columnwidth} \begin{itemize}
\item Chattering reduced using integral-type sliding surface
\item Feed-forward NN used to estimate the nonlinear drift terms
\end{itemize} \end{minipage} & \begin{minipage}[t]{0.5\columnwidth} \begin{itemize}
\item Not good chattering mitigation
\item No comparison with advanced methods
\end{itemize} \end{minipage} \\ [12mm] \hline
&&&&& \\
\begin{minipage}[t]{0.1\columnwidth} \cite{kenne2017new}  \end{minipage} & \begin{minipage}[t]{0.12\columnwidth} MPE \end{minipage} &\begin{minipage}[t]{0.12\columnwidth} Partial-load  \end{minipage} &\begin{minipage}[t]{0.12\columnwidth} SCIG  \end{minipage} & \begin{minipage}[t]{0.6\columnwidth} \begin{itemize}
\item Chattering reduced by augmentation of RBFNN with integral SMC
\item RBFNN compensated the external disturbances
\end{itemize} \end{minipage} & \begin{minipage}[t]{0.5\columnwidth} \begin{itemize}
\item Chattering problem not investigated
\end{itemize} \end{minipage} \\
[15mm]
\hline\hline
\end{tabular}}
\end{table*}

\section{Conclusions}
\label{sec:3.8}

This chapter provided a comprehensive review of the state-of-the-art studies on SMC-based control strategies for WECS control are provided. According to the literature, regardless of the well-known performance of conventional SMC in practical applications, it still exhibits some weaknesses in delivering the desired performance when it comes to WECS control, especially when the WT deals with various faults, parameter uncertainties, and external disturbances. In this regard, different modifications have been developed in the literature to mitigate the chattering problem, increase the convergence speed, achieve better tracking precision, and prevent unnecessarily large control signals from being produced in order to overcome the parametric uncertainties. Considering the reviewed studies, the advantages and drawbacks of SMC-based control strategies and their performance in dealing with various WECS control objectives have been investigated. Accordingly, it was revealed that advanced SMC-based controllers could be counted as highly reliable strategies capable of improving the power generation of WECS by mitigating the undesirable effects of faults, disturbances, and parameter variations.

\begin{table*}[th!]
\caption{Comparative study of the most-investigated modified SMCs for WECSs control.}
\centering
\label{tab:9}
\resizebox{\textwidth}{!}{
\begin{tabular}{l l l}
\hline\hline \\[-3mm]
\begin{minipage}[t]{0.4\columnwidth} Method \end{minipage} & \begin{minipage}[t]{0.8\columnwidth} Advantages (over approaches in the parenthesis)\end{minipage}	& \begin{minipage}[t]{0.7\columnwidth}Disadvantages \end{minipage}\\ \hline
\begin{minipage}[t]{0.4\columnwidth} Simple modified SMC (replacing $sgn\left(\cdot\right)$ with $sat\left(\cdot\right)$ or $tanh\left(\cdot\right)$, using exponential reaching law), Section \ref{sec:3.2}  \end{minipage} & \begin{minipage}[t]{0.8\columnwidth} \begin{itemize}
  \item Reduced chattering (SMC)
  \item Better control performance (SMC)
  \item Less steady-state error (SMC)
\end{itemize}  \end{minipage} & \begin{minipage}[t]{0.7\columnwidth} \begin{itemize}
  \item The chattering is reduced but not much
  \item Relatively large control signal
  \item Singularity problem in practical implementations
\end{itemize} \end{minipage} \\
&& \\ \hline
\rowcolor{lightgray}\begin{minipage}[t]{0.4\columnwidth} Adaptive SMC, Section \ref{sec:3.3} \end{minipage} & \begin{minipage}[t]{0.8\columnwidth} \begin{itemize}
  \item Reduced chattering (SMC)
  \item Better disturbance rejection (SMC)
  \item Less steady-state error (SMC)
\end{itemize} \end{minipage} & \begin{minipage}[t]{0.7\columnwidth} \begin{itemize}
  \item The chattering is reduced but not much
  \item Relatively large control signal (less than SMC)
\end{itemize} \end{minipage} \\ \hline
&& \\
\begin{minipage}[t]{0.4\columnwidth} Fractional SMC, Section \ref{sec:3.4}  \end{minipage} & \begin{minipage}[t]{0.8\columnwidth}\begin{itemize}
  \item Mitigated chattering (SMC, ASMC, SO-SMC, Terminal-SMC, Fuzzy-SMC, NN-SMC)
  \item Ensured finite-time convergence of the states
  \item More robustness against external disturbances, unmodelled dynamics, and parametric uncertainties (SMC, ASMC, SO-SMC)
  \item Faster convergence speed (SMC, ASMC, SO-SMC, Fuzzy-SMC, NN-SMC)
  \item More tunable parameters, leading to more precision
  \item Less steady-state error (SMC, ASMC, SO-SMC, Fuzzy-SMC, NN-SMC)
\end{itemize}  \end{minipage} & \begin{minipage}[t]{0.7\columnwidth} \begin{itemize}
\item Complex mathematical computation
\end{itemize} \end{minipage} \\
&& \\ \hline
\rowcolor{lightgray}\begin{minipage}[t]{0.4\columnwidth} Second-order SMC and Super-twisting SMC, Section \ref{sec:3.5.1} \end{minipage} & \begin{minipage}[t]{0.8\columnwidth} \begin{itemize}
  \item Reduced chattering (SMC, ASMC)
  \item Ensured finite-time convergence of the states
  \item More robustness against external disturbances and uncertainties (SMC, ASMC)
  \item Faster convergence speed (SMC, ASMC, Fuzzy-SMC, NN-SMC)
  \item Less steady-state error (SMC, ASMC)
\end{itemize} \end{minipage} & \begin{minipage}[t]{0.7\columnwidth} \begin{itemize}
  \item Complex mathematical computation (less than FO-SMC)
\end{itemize} \end{minipage} \\ \hline
&& \\
\begin{minipage}[t]{0.4\columnwidth} Terminal SMC and Nonsingular TSMC, Section \ref{sec:3.5.2} \end{minipage} & \begin{minipage}[t]{0.8\columnwidth} \begin{itemize}
  \item Mitigated chattering (SMC, ASMC, SO-SMC, Fuzzy-SMC, NN-SMC)
  \item Ensured finite-time convergence of the states
  \item More robustness against external disturbances, unmodelled dynamics, and parametric uncertainties (SMC, ASMC, SO-SMC)
  \item Faster convergence speed (SMC, ASMC, SO-SMC, Fractional-SMC, Fuzzy-SMC, NN-SMC)
  \item Less steady-state error (SMC, ASMC, SO-SMC, Fuzzy-SMC, NN-SMC)
  \item NTSMC can resolve the singularity problem (SMC, ASMC)
\end{itemize} \end{minipage} & \begin{minipage}[t]{0.7\columnwidth} \begin{itemize}
  \item Complex mathematical computation (less than FO-SMC)
\end{itemize} \end{minipage} \\
&& \\ \hline
\rowcolor{lightgray}\begin{minipage}[t]{0.4\columnwidth} Fuzzy SMC, Section \ref{sec:3.6}  \end{minipage} & \begin{minipage}[t]{0.8\columnwidth} \begin{itemize}
  \item Reduced chattering (SMC, ASMC, SO-SMC)
  \item More robustness against external disturbances, unmodelled dynamics, and parametric uncertainties (SMC, ASMC, SO-SMC)
  \item Adaptive control performance
  \item Faster convergence speed (SMC, ASMC)
  \item Less steady-state error (SMC, ASMC)
\end{itemize} \end{minipage} & \begin{minipage}[t]{0.7\columnwidth} \begin{itemize}
  \item Cannot ensure an ideal sliding motion
  \item The convergence speed is less than HO-SMC and FO-SMC approaches
  \item Requires more data
\end{itemize} \end{minipage} \\ \hline
&& \\
\begin{minipage}[t]{0.4\columnwidth} Neural Network SMC, Section \ref{sec:3.7} \end{minipage} & \begin{minipage}[t]{0.8\columnwidth} \begin{itemize}
  \item Reduced chattering (SMC, ASMC, SO-SMC)
  \item More robustness against external disturbances, unmodelled dynamics, and parametric uncertainties (SMC, ASMC, SO-SMC, HO-SMC)
  \item Adaptive control performance
  \item Faster convergence speed (SMC, ASMC)
  \item Less steady-state error (SMC, ASMC)
\end{itemize} \end{minipage} & \begin{minipage}[t]{0.7\columnwidth} \begin{itemize}
  \item The convergence speed is less than HO-SMC and FO-SMC approaches
  \item Requires more training data
\end{itemize}  \end{minipage} \\
[1ex]
\hline\hline
\end{tabular}}
\end{table*}

A comparative study of the advantages and disadvantages of the discussed modified SMCs for WECS control is illustrated in Table \ref{tab:9}. As the simplest common modification on SMC, some studies have substituted the $sgn\left(\cdot\right)$ function in the conventional SMC's sliding surface with $sat\left(\cdot\right)$ or $tanh\left(\cdot\right)$ to reduce the undesirable chattering effects. Few studies have employed reaching laws to achieve a faster convergence rate in finite time. Since the amplitude of chattering depends on the magnitude of control, decreasing the amplitude of the discontinuous control leads to chattering reduction. However, it can result in the undesirable system's transient response degradation. Hence, the reaching law approach provides a trade-off by decreasing the amplitude of the discontinuous control when the system states are near the sliding surface (to reduce the chattering) and increasing the amplitude when the system states are far from the sliding surface. On the other hand, fuzzy-based and neural networks-based SMCs have successfully overcome the conventional and simple modified SMCs with reduced chattering, effective disturbance rejection, and superior control performance. However, despite their chattering reduction and superior control performance compared to conventional SMCs, they cannot ensure an ideal sliding motion and require more data to deliver the desired performance, demanding further improvements in the control structure. In this respect, higher-order and fractional-order SMCs have successfully established themselves as desirable alternatives with outstanding performance and superiorities over other SMC-based approaches in terms of ensured fast finite-time convergence, mitigated chattering, robustness against external disturbances and model uncertainties, and higher control precision. However, despite the novelties presented in each method, more developments are yet to be done to achieve better control performances.

Though achieving the desired WECS control performance does not entirely depend on the developed controller structure. Proper tuning of the controller parameters and appropriate and realistic WECS modelling also plays a critical role in the performance validation of the developed control schemes. Furthermore, the controller parameters are often precomputed offline by trial-and-error, thus preventing them from consistently deliver the best performance, especially for WECS with varying operating conditions. Hence, optimal tuning procedures using optimization algorithms and adaptive soft computing-based methods such as fuzzy, neural networks, and learning approaches can result in more precise parameters, and lead to better control performances. Furthermore, barrier function (BF) -based SMC approaches have proven their remarkable performance in forcing the state trajectories to converge to a predefined neighborhood of zero in finite time without knowing the upper bound of disturbances \cite{mobayen2022barrier,laghrouche2021barrier}. The BF-based SMC approaches have been found to deliver effective approximations of the external disturbances, yielding a more stable closed-loop system. Accordingly, considering the WECS exposure to various disturbances, incorporating BFs with other modified SMCs can yield even better performances; an open research area that is still to be investigated through WECS control.

\chapter{Fault-Tolerant Optimal Pitch Control of Wind Turbines} 

\label{Chapter4} 

\section{Introduction}
\label{sec:4.1}

Renewable energy, especially WT systems, have gained considerable attention during the past decade due to the energy shortage and environmental issues \cite{ma2015optimal, li2017adaptive, madsen2020experimental}. Since WTs have contributed a considerable portion of the world's power production, demands on the development of reliable control approaches that guarantee the power generation and reduce the operational and maintenance costs have increased substantially. As mentioned in Chapter \ref{Chapter1}, in region III (of the four WT operational regions), the wind speed exceeds the rated value. In this region, pitch actuation is critical for limiting the power capture in high wind speed situations. Hence, pitch angle control strategies are used to control the pitch angle and keep the WT operating at its rated power. Numerous studies in the literature have dealt with the PAC problem in order to limit the aerodynamic power captured by the WT. For instance, in \cite{ren2016nonlinear}, the authors investigated the PAC based on a nonlinear PI controller together with a state and perturbation observer. In \cite{zhang2015load}, a PAC scheme consisting of the conventional PI and two resonant compensators was developed. Authors in \cite{van2015advanced} proposed an advanced PAC strategy based on FLC, while in \cite{venkaiah2020hydraulically,asgharnia2018performance} fuzzy PID and fractional-order fuzzy PID controllers were investigated to improve PAC performance.

The aforementioned studies consider the ideal situation, where the variations in the actuator dynamics and sensor faults are assumed to be negligible. However, in real operations, WTs are prone to different sets of sensor, actuator, and system faults, which degrade the WT stability and power production performance and impose maintenance costs. Accordingly, various fault-tolerant pitch control (FTPC) approaches have been investigated to compensate the fault effects in WT systems and achieve a robust system performance. To this end, in \cite{badihi2020fault}, a FTPC scheme is investigated based on an adaptive PI controller augmented with a fault detection strategy, \textcolor{black}{while authors in \cite{habibi2018adaptive} developed an adaptive PID. Although compared to other classical approaches, the controller demonstrated more acceptable performance in terms of handling non-linear dynamics, further improvements are required to mitigate the fault effects when unexpected actuator faults and wind speed fluctuations happen.} Authors in \cite{lan2018fault} incorporated a conventional PI along with a sliding mode observer to compensate for faults. According to the authors, the proposed control strategy is capable of recovering the nominal pitch actuation; however, only the case of low pressure actuator fault occurrence is considered in the study, where the controller's performance in more harsh situations is yet to be investigated. A Kalman filter was used to assess the blade pitch angle of the WT, together with a PI to deal with the FTPC problem in \cite{cho2018model}. 

During the past decade, the concept and applications of fractional calculus have attracted growing interests of scholars in various engineering fields \cite{mousavi2021robust,azarmi2015analytical,mousavi2020enhanced}. FO derivatives induce an infinite series, presenting a long memory of the past \cite{mousavi2018fractional}, whereas integer-order derivatives are local operators that imply a finite number of terms. Since wind energy and direction have chaotic behavior, the pitch actuation system needs to provide an immediate precise response, which in practice, leads to some slight errors. Thus, it is desirable to preserve all the past effective pitch angles, representing the memory of the pitch system characteristics. Therefore, WT systems are quite suitable processes to be used with FO controllers. Similar to PID controllers, fractional-order PID (FOPID) controllers have been extensively implemented in many applications \cite{angel2018fractional,ren2018optimal,naidu2020power}. FOPID controllers not only inherit the advantages of conventional PIDs such as simple structure and strong robustness but also expand the control range by adding more flexibility to the control system \cite{azarmi2015analytical, angel2018fractional}; however, the existence of two more tunable parameters has made the design problem more complicated. A variety of tuning rules and design methods have been investigated in the literature \cite{azarmi2015analytical,amoura2016closed}, while most of them suffer from the unavailability of the exact dynamic model in Laplace domain representation, especially for complex nonlinear systems \cite{angel2018fractional}. As an alternative solution, evolutionary algorithms (EAs) have been playing a crucial role in determining FOPID parameters \cite{mousavi2015memetic,lee2010fractional}. 
In this work, a fractional-calculus based extended memory of pitch angles is augmented with the controller to enhance its performance for adjusting the desired pitch angle of WT blades and improving the power generation of the WT in the presence of faults.

Metaheuristic optimization algorithms have been extensively employed to optimally tune the controllers' parameters for WT control systems \cite{xiong2020output,bounar2019pso,kumar2021power,benamor2019new}. 
\textcolor{black}{Firefly algorithm (FA), as one of the recently introduced EAs \cite{yang2009firefly}, has been effectively solved many optimization problems in recent years \cite{lv2018firefly,nayak2021hyper,fister2013comprehensive}.} Authors in \cite{shan2021distributed} developed a distributed parallel FA for parameter tuning of a variable pitch WT, where according to the authors, the proposed control scheme reduced the power fluctuation and improved the safety and reliability of WT. \textcolor{black}{FA has certain superiorities over some of the most used EAs. To name a few, a) FA is able to tune its scaling parameter and hence adapt to problem landscape,  b) FA can be counted as a generalization of PSO, differential evolution (DE), and simulated annealing (SA) \cite{yang2009firefly}}, which takes all the three algorithms' advantages, c) unlike PSO, FA does not use velocities, and thus, can avoid the drawbacks associated with the velocity initialization \cite{fister2013comprehensive}, d) since the fireflies aggregate more closely around each optimum, it has shown superior performance over genetic algorithm (GA) that jumps around randomly, and e) since local attraction is more substantial than long-distance attraction, FA can automatically subdivide its population into subgroups, which makes it a suitable method to efficiently \textcolor{black}{tackle nonlinear and multimodal problems \cite{yang2009firefly}}. However, the success of the search procedure in FA depends on a suitable trade-off between global search (exploration) and local search (exploitation) abilities, which corresponds to the attractiveness formulation and variation of light intensity. Both factors allow significant scope for the algorithm's improvements. Thus, investigations have been performed to enhance its performance taking advantage of other search methodologies in order to achieve even better performance \cite{wang2019novel,pazhoohesh2017optimal,altabeeb2019improved}. In this study, an enhanced FA is developed that explores the search space with a well-connected weighted parallel strategy that enriches the population diversity and increases the information exchange between the fireflies. The proposed strategy not only \textcolor{black}{expedites} the convergence speed of FA, but also reduces the possibility of getting trapped in local optima. A dynamic switching coefficient is also implemented that lets the algorithm perform more accurate exploitations. The switching coefficient makes the algorithm work more precisely and finds more reasonable solutions.

Even though various advanced control strategies have been developed for WTs such as sliding mode control, model predictive control, FLC-based control, etc. \cite{jain2018fault,xu2019event,mousavi2021maximum}, \textcolor{black}{PI/PID method is still the preferred approach} in real-world applications with some improvements \cite{lan2018fault,aissaoui2013fuzzy} due to its simplicity. In this regard, many researchers have utilized a simple PI controller through the pitch angle regulation process in region III \cite{habibi2018adaptive,lan2018fault,cho2018model}, and some studies have focused on applying rotor speed limitations \cite{tang2018active,luo2007strategies}. However, due to the existence of only two tuning parameters in PI controllers, the strategy of utilizing a simple PI does not guarantee the minimum steady-state error, especially when faults occur in the system  \cite{badihi2020fault,badihi2014wind}. Besides, although PI/PID control has attracted a wide range of attention in WT control systems, a significant limitation still remains; how to determine the controller's gains. Accordingly, despite the existence of various methods for tuning PID gains \cite{aissaoui2013fuzzy, bianchi2012gain, habibi2018adaptive}, there is no specific way to determine such gains for WT control, as they need be chosen by the designer which is neither a straightforward task nor optimal. This has motivated the attempt to construct an optimal controller for WT control. Hence, to effectively maintain constant power generation, an optimized FOPID controller augmented with extended memory of the pitch angles is developed to regulate the pitch angle and prevent the WT from over-speeding.

In this work, several performance evaluations of the proposed dynamic weighted parallel FA (DWPFA) in comparison to other conventional and modified EAs are investigated through solving well-defined 2017 IEEE congress on evolutionary computation (CEC2017) mathematical benchmark functions \cite{awad2017ensemble}. Non-parametric Friedman and Friedman Aligned statistical tests are also provided to statistically analyze the quality of the solution \cite{gui2019multi}. The proposed memory extension of pitch angles is incorporated in the FOPID controller (called EM-FOPID) to generate the desired WT pitch angle reference in the presence of sensor, actuator, and system faults, where the controller parameters are tuned using the proposed DWPFA algorithm. \textcolor{black}{This study contributes the literature as follows:}
\begin{enumerate}
\item Using the concept of fractional calculus, a fault-tolerant pitch control strategy with extended memory of pitch angles (EM-FOPID) is developed that improves the power generation of the WT, where the controller parameters are tuned using the proposed DWPFA.
\item A modified FA (DWPFA) is proposed that increases the convergence speed of the conventional FA, reduces the possibility of getting trapped in local optima, and increases the exploitation accuracy.
\item Comparative simulations are provided that reveal the remarkable performance of proposed optimal EM-FOPID with respect to optimal FOPID and conventional PI.
\end{enumerate}

The chapter is organized as follows. The problem statement, including the WT modelling, control objective, and fault scenarios, are described in Section \ref{sec:4.2}. The proposed DWPFA algorithm is explored in Section \ref{sec:4.3}. Section \ref{sec:4.4} presents the proposed EM-FOPID control strategy. The performance of DWPFA is evaluated, and the EM-FOPID design is verified in Section \ref{sec:4.5}. Finally, Section \ref{sec:4.6} concludes the chapter.

\section{Problem Statement}
\label{sec:4.2}
In this section, first, the model of a three-bladed variable speed WT is investigated, and then, the control objective of the study is described. The last subsection is devoted to introducing the different sensor, actuator, and system fault scenarios with various levels of severity, which are considered to occur to the WT.

\subsection{Wind Turbine Modelling}
\label{sec:4.2.1}
In this work, a 4.8 MW three-bladed variable speed horizontal axis wind turbine (HAWT) is considered \cite{odgaard2013fault}. The system consists of three main units: \textcolor{black}{the generator-converter, drivetrain, and the blade and pitch model}. The aerodynamic and pitch system models are combined to form the blade and pitch model, where the former denotes the \textcolor{black}{transformation} of wind power to rotational energy, and the latter rotates the blades around their longitude axis. The drivetrain provides the required rotational speed of the generator. To convert the mechanical wind energy to electrical energy, the coupled converter-generator system is utilized. Owing to the controlled pitching of the blades changing the aerodynamic efficiency of the WT, aerodynamic wind energy is transformed into effective mechanical energy. Thus, the captured power significantly depends on the available wind energy and the geometry of the blade aerofoils and their pitch, which thus affect the responding capability of the machine to wind fluctuations. Considering $\omega _{r}=\eta _{g}\omega _{g} $, where $\eta _{g} $ represents the generator's efficiency, and $\omega _{g} $ and $\omega _{r} $ define the generator and rotor rotational speed, respectively, the optimal rotor speed $\omega _{r,opt} $ can be achieved as $\omega _{r,opt} ={\lambda _{opt} \upsilon _{w} }/{R}$.

The pitch actuator model consists of a hydraulic and a mechanical machinery, and \textcolor{black}{can be expressed as the following second-order system }\cite{odgaard2013fault}:
\begin{equation}
\frac{\beta \left(s\right)}{\beta _{r} \left(s\right)} =\frac{\omega _{n}^{2} }{s^{2} +2\xi \omega _{n} s+\omega _{n}^{2} },
\label{GrindEQ__4-6_}
\end{equation}
\textcolor{black}{where $\beta _{r} $ denotes the command signal for the pitch angle being produced by the WT controller, $\beta $ stands for the actual pitch angle produced by the actuator, $\xi $ denotes the damping factor, and $\omega _{n} $ represents the natural frequency in [\si{rad/s}].}
\begin{rem}
\label{rem:4-2}
\textcolor{black}{In a pitch actuator, the pitch actuator constraints play a critical role. In this regard, the pitch rate constraints are considered between $-8^{\circ/s} $ and $8^{\circ/s} $ so that they are operationally possible and at which the rotor will not stall, while the operational range of the pitch angle is considered as $-3^{\circ} \le \beta \le 90^{\circ} $ \cite{azizi2019fault}.}
\end{rem}

The mechanical part of WT, namely the drivetrain, is a rather complex system which consists of a gearbox, low-speed shaft, and high-speed shaft that converts the low-speed torque of the rotor-side shaft to a high-speed torque of the generator-side shaft. Since the WT is coupled to the generator through a gearbox, the generator torque $T_{G} $ can regulate the rotor speed. The rotor inertia $J_{R} $ [\si{kg.m^{2}}] is driven at speed $\omega _{r} $ by the aerodynamic torque $T_{a} $ [\si{N.m}], and the generator inertia $J_{G} $ [\si{kg.m^{2}}] is driven by high-speed torque of the generator-side shaft at speed $\omega _{g} $ and is braked by the generator torque $T_{G} $. \textcolor{black}{This study considers the following dynamic model of a two-mass drivetrain model \cite{tang2018active}:}
\begin{equation}
\left[\begin{array}{c} {\dot{\omega }_{r} } \\ {\dot{\omega }_{g} } \\ {\dot{\theta }_{\Delta } } \end{array}\right]=\left[\begin{array}{ccc} {\vartheta _{11} } & {\vartheta _{12} } & {\vartheta _{13} } \\ {\vartheta _{21} } & {\vartheta _{22} } & {\vartheta _{23} } \\ {\vartheta _{31} } & {\vartheta _{32} } & {\vartheta _{33} } \end{array}\right]\left[\begin{array}{c} {\omega _{r} } \\ {\omega _{g} } \\ {\theta _{\Delta } } \end{array}\right]+\left[\begin{array}{c} {\psi _{1} } \\ {0} \\ {0} \end{array}\right]T_{a} +\left[\begin{array}{c} {0} \\ {\psi _{2} } \\ {0} \end{array}\right]T_{G},
\label{GrindEQ__4-7_}
\end{equation}
where $\vartheta _{11} =-{D_{LS} +B_{r} }/{J_{R} }$, $\vartheta _{12} ={D_{LS} }/{J_{R} N_{GB} }$, $\vartheta _{13} =-{K_{LS} }/{J_{R} }$, $\vartheta _{21} ={\eta _{dt} D_{LS} }/{J_{G} N_{GB} }$, $\vartheta _{22} =-{\eta _{dt} D_{LS} }/{J_{G} N_{GB}^{2} } -{B_{G} }/{J_{G} }$, $\vartheta _{23} ={\eta _{dt} K_{LS} }/{J_{G} N_{GB} }$, $\vartheta _{31} =1$, $\vartheta _{32} =-{1}/{N_{GB}} $, $\vartheta _{33} =0$, $\psi _{1} ={1}/{J_{R}}$, $\psi _{2} =-{1}/{J_{G} }$. \textcolor{black}{$K_{LS} $ [\si{Nm/rad}] and $D_{LS}$ [\si{Nms/rad}] denote the low-speed shaft stiffness and damping coefficient, respectively. $B_{G}$ stands for the viscous friction of the high-speed shaft in [\si{Nms/rad}], $N_{GB} $ denotes the gearbox ratio, and $\eta _{dt} $ and $\theta _{\Delta } $ represent the drivetrain efficiency and torsion angle, respectively.}

The generator is responsible for converting the shaft kinetic energy into electrical energy. Since compared to the WT dynamics, the dynamics of the electrical system are noticeably faster, the following first-order model with fast dynamics can be used to model the generator and converter dynamics \cite{odgaard2013fault}.
\begin{equation}
\frac{T_{G} \left(s\right)}{T_{G,ref} \left(s\right)} =\frac{\alpha _{gc} }{s+\alpha _{gc} },
\label{GrindEQ__4-9_}
\end{equation}
where $T_{G,ref} $ stands for the torque reference to the generator, and $\alpha _{gc}={1}/{\tau_{gc}}$ denotes the generator and converter unit coefficient, with $\tau_{gc}$ representing the time constant. The generated power by the generator can be achieved by $P_{g} =\eta _{g} T_{G} \omega _{g}$.

\subsection{Control Objectives}
\label{sec:4.2.2}
The baseline control system in WTs consists of two individual control sections to regulate the generator torque and the pitch angle of the blades, where the overall performance of the WT directly depends on the performance of both controllers.

The wind energy is not always constant, and it holds different flow profiles (laminar, turbulent, etc.), and gusts, which can result in deceleration of the rotor speed to a critical speed, which brings instability and damage to the WT. Besides, as the power generated depends on the generator torque, it is evident from \eqref{GrindEQ__4-9_} that any change (due to faults or uncertainties) in the generator would directly affect the rated power. Thus, generator faults can impose severe problems in tracking the maximum power point and rated power in regions II and III, respectively. In this regard, in region II, the power reference $P_{ref} $ tracking is switched to maximum power point tracking (MPPT) to stabilize the WT while maximizing the power capture \cite{tang2018active}. The switch works with respect to the wind speed, as in Figure \ref{fig:3-4}, the dotted arrow from the Wind Profile block to the switch block shows its dependency on wind speed. Accordingly, reference torque to the converter can be represented as $T_{G,ref} =K_{opt} \omega _{r}^{2}$, where $K_{opt} ={1}/{2} (\rho \pi R^{5} {\mathcal C_{P,\max }}/{\lambda _{opt}^{3} })$.
Figure \ref{fig:3-4} depicts the block diagram of the proposed WT control scheme, which comprises three main blocks: (a) the WT model which is prone to actuator, system, and sensor faults, (b) the proposed DWPFA-based pitch control block, and (c) the generator torque assignment block, which includes the power reference tracking (${P_{ref}/\omega _{r} }$) and MPPT ($K_{opt} \omega _{r}^{2} $).
\begin{figure}
\centering
\includegraphics[width=4.9 in]{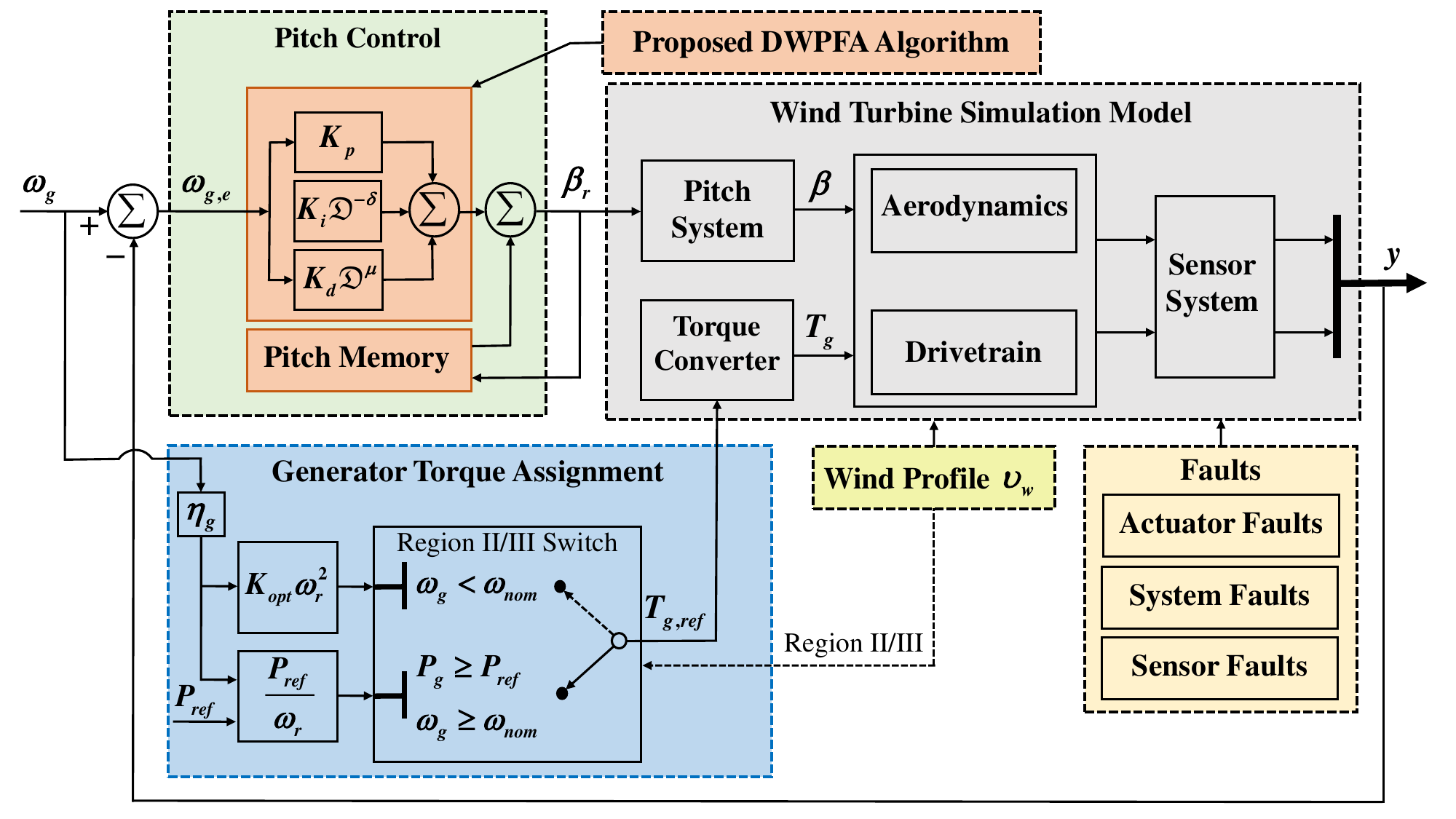}
\caption{Proposed control scheme diagram for a wind turbine system.}
\label{fig:3-4}
\end{figure}

\subsection{Fault Scenarios}
\label{sec:4.2.3}
Faults occurring in a WT can affect the system characteristics or lead to inoperable conditions. Wind turbine faults may be classified in terms of those that are highly serious, where the WT needs to be shut-down in order to prevent irreparable damage, and those faults that can be accommodated by suitable controllers, leading the WT to stay operational with some possible performance detriment. The faults modelled in this work include sensor (F1-F4), actuator (F5, F6), and system (F7) faults with various levels of severity, as summarized in Table \ref{tab:4-1}; where each one can cause performance degradation or slight damage to the WT. \textcolor{black}{The development time for F1-F5 is considered medium (operational malfunctioning that can happen during a few hours), while it is slow and very slow for F6 and F7 (malfunctions that can happen during months and/or years of operation), respectively.}

\begin{table}[!t]
\caption{Sensor, actuator, and system fault scenarios considered.}
\centering
\label{tab:4-1}
\scalebox{0.8}{
\begin{tabular}{l l l l l}
\hline\hline \\[-3mm]
\multicolumn{1}{c}{Fault} & \multicolumn{1}{c}{Faulted (Occurrence time (s)) } & \multicolumn{1}{c}{Severity}  \\[1ex] \hline
F1 & $\beta _{1,m1} =5^\circ$ (2000-2100), $\beta _{1,m1} =5^\circ$ (2700-2900)& low \\
F2  & $\beta _{2,m2} =1.2^\circ$ (2400-2500) & low  \\
F3  & $\beta _{3,m2} =5^\circ$ (2600-2700) & low  \\
F4   & $\omega _{r,m2} =1.1$ (3805-4400), $\omega _{g,m1} =0.9$ (3805-4400) & low \\
F5  & Hydraulic pressure drop (2900-3000) & high \\
F6  & Air content increment in the oil (3500-3600) & medium \\
F7  & Friction changes in the drivetrain (4100-4300) & medium \\ [1ex]
\hline\hline
\end{tabular}}
\end{table}

Sensor faults mainly originate from mechanical or electrical faults in the sensors, due to drift, noise, and external factors such as lightning, heavy rain, moisture, storms, and corrosion; and also misalignment of one or more blades at the installation step or blade imbalancement during operation \cite{habibi2019reliability}. Considering the fact that the pitch position measurements act as a reference for the internal pitch system controller, sensor faults can negatively affect the pitch positions if the control system fails to handle them properly, which leads to performance degradation of WT \cite{cho2018model}. Additionally, since the generator and rotor speed measurements are carried out utilizing encoders, and due to possible malfunctions of the electrical components of the encoders, they can be faulty as well. The faults can be in the form of a fixed value that prevents the encoder from being updated with new values, or a changed gain factor on the measurements which causes the encoder to read more marks on the rotating part than are actually present \cite{habibi2019reliability}.

At a basic level, in a WT system, \textcolor{black}{faults can occur on the converter and the pitch actuator system}. Faults in the pitch actuator cause changes in the dynamics due to three factors; a hydraulic leakage, a drop pressure in the pump wear, or a high air content in the hydraulic oil, where the latter may happen in various levels due to compressible nature of the air \cite{ruiz2018wind}. The source of converter faults is either in changed dynamics of the converter arising from an internal fault in the converter's electronic components, or an offset in the converter torque estimation, which is more severe. The converter controller can deal with the faults in the electronic components, and since the torque balance in the WT power train is changed by torque offset, it is possible to detect and accommodate them \cite{odgaard2013fault}. System faults result in changes in the dynamic of parts of the system, which mainly happen in the drivetrain. Although compared to the system dynamics, the drivetrain friction coefficient changes more slowly with respect to time, it could be detected by observing the changes in the frequency spectrum of the vibration measurements. In this work, this fault is considered as a small change of the friction coefficient.

\subsection{Fault Injection}
\label{sec:4.2.4}
This section briefly explains the type of changes happening in the fault-free model as each fault occurs. According to Table \ref{tab:4-1}, faults F1, F2, and F3 correspond to fixed values on $\beta _{1,m1}$, $\beta _{2,m2}$, and $\beta _{3,m2}$, respectively. Each fault occurs in a certain time interval, as shown in Table \ref{tab:4-1}. The fault F4 happens in the time interval of 3805-4400 seconds by changing the gain factors on $\omega _{r,m2}$ and $\omega _{g,m1}$ to 1.1 and 0.9, respectively. As two of the main pitch actuator faults, the hydraulic pressure drop (F5) and high air content in the oil (F6) are considered in this work. The effects of these faults are reflected in the damping ratio and natural frequency of the pitch system, where each one influences the system dynamics differently. \textcolor{black}{A drop in the hydraulic pressure changes $\omega _{n}$ and $\xi$ from their nominal values $\omega _{n,0}$ and $\xi_0$ to their low pressure values $\omega _{n,f}$ and $\xi_f$, which influences the pitch system dynamics}. Under this gradual low-pressure fault, $\omega _{n}^{2}$ and $\xi\omega _{n}$ in (\ref{GrindEQ__4-6_}) can be modelled as follows \cite{habibi2018adaptive}:

\begin{subequations}
\label{GrindEQ__4-FF1_}
\begin{align}
\omega _{n}^{2}&=\omega _{n,0}^{2}+f. \left(\omega _{n,f}^{2}-\omega _{n,0}^{2}\right), \\
\xi\omega _{n}&=\xi_0\omega_{n,0}+f.\left(\xi_f\omega_{n,f}-\xi_0\omega_{n,0}\right),
\end{align}
\end{subequations}
\textcolor{black}{where $f\in\left[0,1\right]$ represents the fault indicator at which $f=0$ and $f=1$ correspond to the normal pressure and low pressure up to 50\% pressure drop, respectively}.

Changing the fault indicator corresponds to changes in the natural frequency and the damping ratio, where Table \ref{tab:4-FF2} presents the effects of their changes on the hydraulic pressure drop (F5) and the air content in the oil (F6). For a better demonstration of hydraulic pressure drop, the fault indicator is changed gradually, which corresponds to different values for $\omega_n$ and $\xi$, where in each step, the percentage of change in hydraulic pressure drop is given. Besides, $\omega_n=5.73$ and $\xi=0.45$ correspond to the maximum percentage of 15\% change for the air content in the oil occurring during the time period of 3500-3600 seconds. Another considered fault is the friction changes in the drivetrain (F7) which will be investigated in Section \ref{sec:4.5} with different levels of severity with 5\%, 10\%, 50\%, and 100\% increase in the coefficient.

\begin{table}[!ht]
\caption{Different faults effect on the pitch system dynamics.}
\centering
\label{tab:4-FF2}
\scalebox{0.8}{
\begin{tabular}{l c c c c}
\hline\hline \\[-3mm]
\multicolumn{1}{c}{Faults} & \multicolumn{1}{c}{Fault indicator} & \multicolumn{1}{c}{$\omega_n$ (change \%)} & \multicolumn{1}{c}{$\xi$ (change \%)} & \multicolumn{1}{c}{Fault (\%)} \\[1ex] \hline
Fault free & 0.0 & 11.1100 & 0.600 & 0.0 \% \\ \hline
Hydraulic pressure drop (HPD) & 0.1 & 10.5952 (-4.63 \%) & 0.5953 (-0.78 \%) & 5 \% \\
 & 0.2 & 10.0541 (-9.50 \%) & 0.5916 (-1.40 \%) & 10 \% \\
 & 0.3 & 9.4822 (-14.65 \%) & 0.5895 (-1.75 \%) & 15 \% \\
 & 0.4 & 8.8734 (-20.13 \%) & 0.5895 (-1.75 \%) & 20 \% \\
 & 0.5 & 8.2197 (-26.01 \%) & 0.5927 (-1.21 \%) & 25 \% \\
 & 0.6 & 7.5094 (-32.40 \%) & 0.6010 (+0.16 \%) & 30 \% \\
 & 0.7 & 6.7244 (-39.47 \%) & 0.6178 (+2.96 \%) & 35 \% \\
 & 0.8 & 5.8347 (-47.48 \%) & 0.6505 (+8.41 \%) & 40 \% \\
 & 0.9 & 4.7823 (-56.95 \%) & 0.7187 (+19.78 \%) & 45 \% \\
 & 1.0 & 3.4200 (-69.21 \%) & 0.9000 (+50.00 \%) & 50 \% \\ \hline
High air content in the oil (HAC) & -- & 5.7300 (-48.42 \%) & 0.4500 (-25.00 \%) & 15 \%  \\ [1ex]
\hline\hline
\end{tabular}}
\end{table}

Figure \ref{fig:FF3-4} depicts the step responses of the pitch system to different fault situations. Accordingly, as the hydraulic pressure drops, it slows the pitch actuator dynamics, resulting in the degradation of pitching performance.

\begin{figure}
\centering
\includegraphics[width=3.5 in]{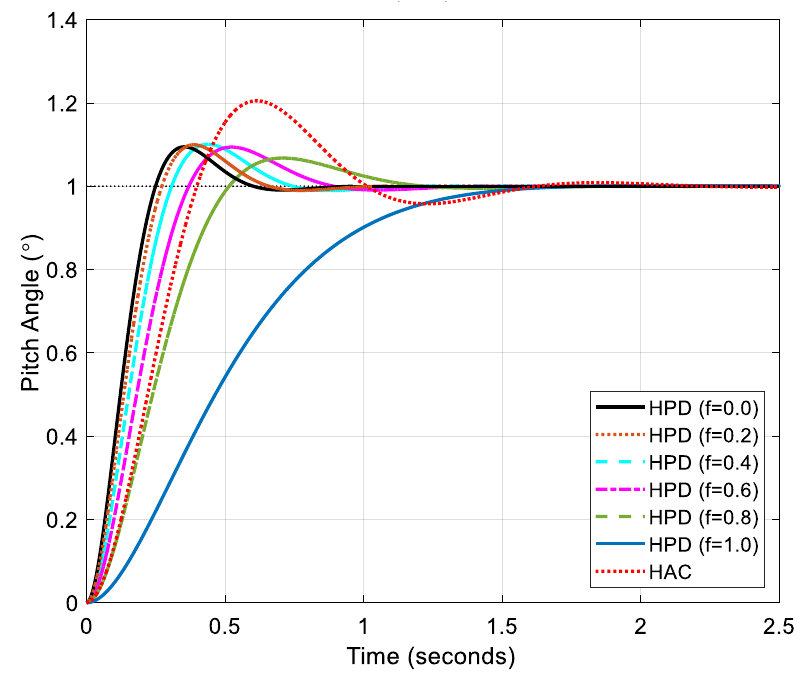}
\caption{Step response of the pitch system under various fault conditions.}
\label{fig:FF3-4}
\end{figure}

\begin{rem}
\label{rem:4-281}
The effects of pitch actuator faults on the pitch system are reflected in $\omega_n$ and $\xi$ from their nominal values to faulty values by changing the fault indicator $f$, as expressed in \eqref{GrindEQ__4-FF1_}. Accordingly, fault-free and faulty situations are being considered in the design process of the developed controller {\eqref{GrindEQ__4-25_}}.
\end{rem}

\section{Proposed Dynamic Weighted Parallel Firefly Algorithm}
\label{sec:4.3}
In this section, the conventional FA is first introduced, and then, the proposed dynamic weighted parallel FA will be investigated in detail.
\subsection{Basic Principles of FA}
\label{sec:4.3.1}
FA is an optimization algorithm that mimics the social behavior of fireflies and their flashing light patterns \cite{yang2009firefly}. The swarm of fireflies is randomly located in the search space, where each one represents a possible solution to the problem. Fireflies with better solutions acquire more light intensity, while other swarm members update their positions by moving toward brighter and more attractive fireflies. For simplicity of development, the FA utilizes the three rules: (a) the fireflies are unisex; thus, they only get drawn to brighter ones, (b) attractiveness corresponds to brightness. The less bright ones always move toward the brighter fireflies, and if no brighter one is left, it moves randomly, and (c) the analytical form of the problem affects the brightness of a firefly, where, brightness is proportional to the value of the objective function.

In the firefly algorithm, attractiveness is proportional to the light intensity seen by adjoining fireflies. Accordingly, since decreasing the distance from the source leads to increment of light intensity, attractiveness increases as the distance between any two fireflies decreases. \textcolor{black}{Two critical issues to be considered in FA are the attractiveness formulization and the light intensity variation. The light intensity $\mathfrak I=\mathfrak I_{0} e^{-\gamma_f \mathfrak r^{2} }$ alters with the distance $\mathfrak r$}, where $\gamma_f $ and $\mathfrak I_{0} $ represent the light absorption coefficient and the initial light intensity, respectively. The firefly's attractiveness $\chi_f$ is defined as $\chi_{f,\mathfrak r} =\chi_{f,0} e^{-\gamma_f \mathfrak r^{2} }$, where $\chi _{f,0} $ denotes the attractiveness at $\mathfrak r=0$. \textcolor{black}{The Euclidian distance $\mathfrak r_{ij} =\left\|\mathfrak X_{i} -\mathfrak X_{j} \right\| _{2} $ or $\ell _{2} $ norm can express the distance $r_{ij} $ between any two fireflies $i$ and $j$ at $\mathfrak X_{i} $ and $\mathfrak X_{j} $}.

\begin{rem}
\label{rem:4-233}
Considering the attractiveness $\chi_{f,\mathfrak r} =\chi_{f,0} e^{-\gamma_f \mathfrak r^{2} }$, it can be seen that there are two limiting cases with firefly algorithm related to small and large values of $\gamma_f$ (i.e. $\gamma_f\rightarrow0$ and $\gamma_f\rightarrow \infty$). \textcolor{black}{When $\gamma_f$ tends to zero, the brightness and attractiveness become constant; in other words, a firefly is visible to all other fireflies.} In contrast, when $\gamma_f$ is very large, the attractiveness considerably decreases, and the fireflies are short-sighted or equivalently fly in a dense foggy environment. \textcolor{black}{Large values of $\gamma_f$ imply an almost randomly movement of fireflies, that refers to a random search procedure.} As a result, the FA usually performs between these two cases, where the attractiveness coefficient plays a critical role in fireflies' movements.
\end{rem}

According to Remark \ref{rem:4-233}, firefly $i$ moves toward a more attractive (brighter) one as follows \cite{yang2009firefly},
\begin{equation}
\mathfrak X_{i} \left(t+1\right)=\mathfrak X_{i} \left(t\right)+\chi_{f,\mathfrak r} \left(\mathfrak X_{j}\left(t\right) -\mathfrak X_{i}\left(t\right) \right)+\eta_f \cdot (\psi_f -0.5),
\label{GrindEQ__4-15_}
\end{equation}
\textcolor{black}{where randomization is performed with $\eta_f $ and $\psi_f $ being random numbers within the interval [0,1].}

It is worth noting that, the case $\gamma_f\rightarrow0$ corresponds to a special case of PSO with $\chi_{f,0}\approx2$ \cite{yang2009firefly}. However, although according to \cite{yang2009firefly}, $\chi_{f,0}=1$ is considered in the standard FA for most cases, other ranges of $\chi_{f,0}$ are reported and used in the literature, where $0<\chi_{f,0}<2$ has found to deliver the best performance \cite{mousavi2018fractional,lv2018firefly,nayak2021hyper,fister2013comprehensive,wang2019novel,shan2021distributed,trachanatzi2020firefly}. Hence, initializing $\chi_{f,0}$ within the interval $(0,2)$ (corresponding to $\left|1-\chi _{f,\mathfrak r} \right|<1$) is a reasonable choice to achieve better performances from the firefly algorithm. However, other values of $\chi_{f,0}\ge 2$ can also be set for the algorithm.

\subsection{Proposed Dynamic Weighted Parallel FA}
\label{sec:4.3.2}
Although the conventional FA has its advantages, it also has some shortcomings, such as premature convergence leading to being trapped in local minima and lack of a suitable trade-off between exploitation and exploration abilities \cite{fister2013comprehensive}. Besides, in FA, the brightest member always moves randomly in the search area which tends to decrease its intensity, especially at high dimensions. In this regard, many studies have incorporated external global/local search procedures (algorithms) into the FA, in order to enhance its search abilities and performance \cite{niknam2012reserve,pazhoohesh2017optimal,altabeeb2019improved}.

\textcolor{black}{This study proposes a modified version of FA} that explores the search space with a well-connected weighted parallel strategy, which effectively accelerates the convergence speed of the conventional FA while reducing the possibility of becoming trapped in local optima. In addition, taking advantage of a dynamic switching coefficient, as the damping coefficient decreases, the switching coefficient increases to let the algorithm to perform more accurate exploitations.

\sloppy A population of randomly generated fireflies is firstly initialized, where each individual stands for a possible solution. The damping coefficient $R_f=2\left(1-n_{it} /n_{it,\max } \right)$ is incorporated into \eqref{GrindEQ__4-15_} to improve the movement pattern of individuals in the exploration process as they move towards the brightest one, as follows:
\begin{equation}
\mathfrak X_{i} \left(t+1\right)=\mathfrak X_{i} \left(t\right)+R_f\chi _{f,\mathfrak r} \left(\mathfrak X_{j}\left(t\right) -\mathfrak X_{i}\left(t\right) \right)+\eta_f \cdot (\psi_f -0.5),
\label{GrindEQ__4-16_}
\end{equation}
where $n_{it} $ and $n_{it,\max } $ denote the number of the current iteration and max iteration, respectively, and $R_f$ linearly decreases from two to zero over the number of iterations.

The objective value of each individual is evaluated to determine the best solution. In the exploitation process, the population can be divided into $n$ number of semi-independent subgroups with \textcolor{black}{equal} number of members, where within each subgroup, the individuals are updated in parallel aiming at finding better solutions. The weighting coefficient associated with each subgroup can be calculated by $w_{sg}=\frac{2.5\times n_{sg}}{n_{sg}+1}$, where $n_{sg}$ denotes the number of predefined groups.

\begin{rem}
\label{rem:4-344}
It should be noted that for each individual, the weighting coefficient associated with its current group should be twice the weighting coefficient associated with other groups. In addition, the total number of members in the populations ($n_{p}$) should be divisible by the defined number of subgroups ($n_{sg}$); thereby, all subgroups would have the same number of fireflies through the individual assignment procedure.
\end{rem}

\begin{rem}
\label{rem:4-345}
In this work, the subgroups members assignment procedure presented by \eqref{GrindEQ__4-17_} and Figure \ref{fig:4-4}, as well as the position update \eqref{GrindEQ__4-19_} are presented for the case $n_{sg}=4$. The procedure for other values of $n_{sg}$ can be carried out analogously.
\end{rem}

The population is divided into four semi-independent subgroups $\{A,B,C,D\}$, where the best four individuals of the exploration process are assigned as subgroup leaders (see Figure \ref{fig:4-4}), and other members join the subgroups as follows:

\begin{subequations}
\label{GrindEQ__4-17_}
\begin{align}
A &=\left\{A_{1} ,A_{2} ,A_{3} ,...,A_{n} \right\}=\left\{1,5,9,...,A_{n} \right\}, \\
B &=\left\{B_{1} ,B_{2} ,B_{3} ,...,B_{n} \right\}=\left\{2,6,10,...,B_{n} \right\}, \\
C &=\left\{C_{1} ,C_{2} ,C_{3} ,...,C_{n} \right\}=\left\{3,7,11,...,C_{n} \right\}, \\
D &=\left\{D_{1} ,D_{2} ,D_{3} ,...,D_{n} \right\}=\left\{4,8,12,...,D_{n} \right\},
\end{align}
\end{subequations}
where $n=n_{p} /n_{sg}=4$.

The weighting coefficients $w_{A} $, $w_{B} $, $w_{C} $, and $w_{D} $ are defined for subgroups in order to establish a connection between the whole population and subgroups leaders to exchange their location information. The weighting coefficients are chosen in such a way that $w_{total}=w_{A}+w_{B}+w_{C}+w_{D}=10$. Accordingly, in every step of the exploitation process, the individuals update their positions, taking into account the position of the leaders with different weights.
\begin{figure}
\centering
\includegraphics[width=4 in]{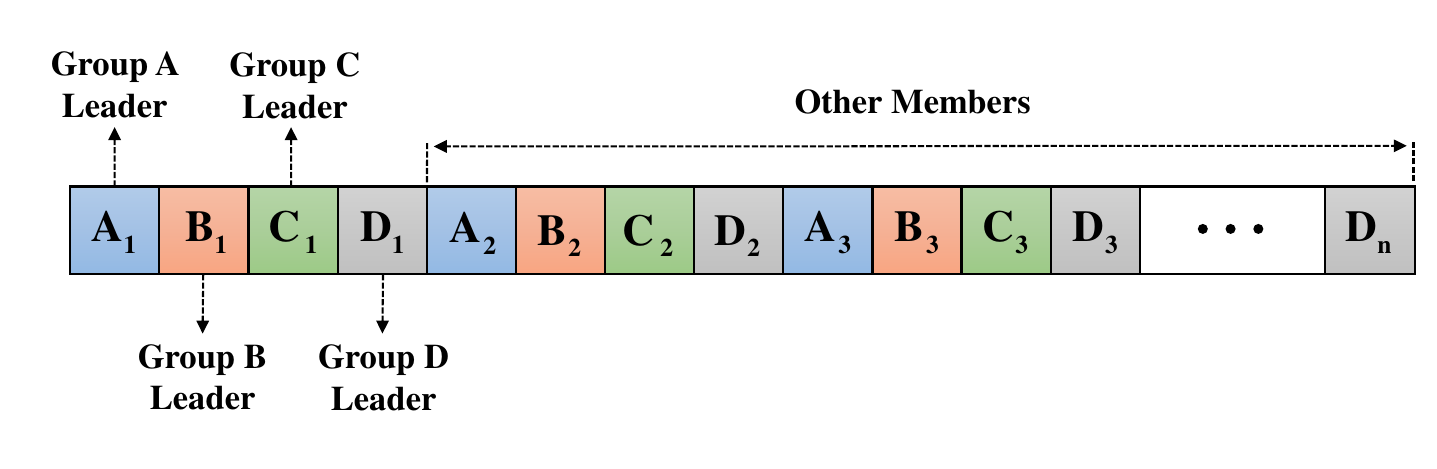}
\caption{The subgroups members assignment procedure.}
\label{fig:4-4}
\end{figure}
The switching coefficient $\Theta =\si{round}\left(\kappa_f /e^{R_f} \right)$ represents the number of exploitation steps and dynamically makes a trade-off between the exploitation and exploration functions of the algorithm, where $\kappa_f $ is the switching coefficient sensitivity parameter, an arbitrary positive number that determines the depth of the exploitation process. Members of each subgroup are updated according to \eqref{GrindEQ__4-19_}, taking into consideration their subgroup leader and other leaders as follows:
\begin{align}
\label{GrindEQ__4-19_}
\mathfrak X\left(t+1\right)&=\mathfrak X+R_f\chi _{f,\mathfrak r} \left(\mathfrak X_{GL} -\mathfrak X\right)\\ \nonumber
&+\left(\begin{array}{l} {w_{A} r_{1} \left(\mathfrak X_{A} -\mathfrak X\right)+w_{B} r_{2} \left(\mathfrak X_{B} -\mathfrak X\right)} \\ {+w_{C} r_{3} \left(\mathfrak X_{C} -\mathfrak X\right)+w_{D} r_{4} \left(\mathfrak X_{D} -\mathfrak X\right)} \end{array}\right)/w_{total}+\eta_f \cdot \left(\psi_f -0.5\right),
\end{align}
\textcolor{black}{where $\mathfrak X_{GL} $ represents the subgroup leader's position} and $r_{i} \in \left(0,1\right),i=1,2,3,4$ is a random number that directs the movement.

\begin{rem}
\label{rem:4-346}
Considering \eqref{GrindEQ__4-19_}, suppose that the position of a firefly in group $B$ is to be updated. Since the firefly belongs to group $B$, $\mathfrak X_{GL} \equiv \mathfrak X_{B}$ denotes the position of its subgroup leader. The damping coefficient $R_f$ in the second term of \eqref{GrindEQ__4-19_} linearly decreases from two to zero over the number of iterations; hence, $R_f$ increases the exploration performance at the beginning of the algorithm, and as the algorithm proceeds and the fireflies approach the subgroup leader, it increases the exploitation performance around the leader. This helps the individuals perform a semi-local search behavior to better search space around the optimum solution. The third term in \eqref{GrindEQ__4-19_} benefits from the position of other subgroups' leaders and their associated weighting coefficients ($w_{sg}$), giving more weight to each firefly's current leader as mentioned in Remark \ref{rem:4-344}. This helps the fireflies consider other leaders' positions as they move towards their own group's leader, resulting in increased diversity of movements.
\end{rem}

After each exploitation step, the individuals are evaluated, and the best solution becomes the leader of the subgroup. By completing the exploitation process, the four subgroups are merged into one big group so that agents can share location information amongst the search space. According to \eqref{GrindEQ__4-19_}, it is observed that utilizing the subgroup method increases the information exchange between individuals, and effectively \textcolor{black}{increases the algorithm's convergence speed}. Besides, despite the conventional FA, the brightest member's movement is not entirely random, and it performs a more sophisticated search through the exploitation process. Algorithm 1 presents the pseudo-code of the DWPFA, where NFE represents the number of function evaluations.
\begin{algorithm}[t!]
\caption{Pseudo-code of DWPFA.}
  \footnotesize
\begin{algorithmic}[1]
    \State Objective function $f({\bf \mathfrak X}),\quad {\bf \mathfrak X}=(\mathfrak X_{1} ,\mathfrak X_{2} ,...,\mathfrak X_{d} )$
    \State NFE=0;
    \State Initialize the population randomly;
    \State Define the light absorption coefficient $\gamma_f $;
    \While{not stopping criterion}
        \For{$i=1:n_{p}$}
            \For{$j=1:n_{p}$}
                \If{$\mathfrak I_{j} >\mathfrak I_{i} $}
                    \State Move individual $i$ towards $j$ in all dimensions;
                \EndIf
                \State \textbf{end if}
                \State Update $\mathfrak X_{i} $ using \eqref{GrindEQ__4-15_};
                \State Update light intensity $\mathfrak I_{i} $;
                \State NFE=0;
            \EndFor
            \State \textbf{end for}
        \EndFor
        \State \textbf{end for}
        \State Rank the individuals and find the current best;
        \State Define the subgroups and assign the leaders and members;
        \For {$k=1:\Theta $}
            \For {$i=2:n_{p} $}
                \If{$\mathfrak I_{GL} >\mathfrak I_{i} $}
                    \State Move individual $i$ towards the subgroup leader in all dimensions;
                \EndIf
                \State \textbf{end if}
                \State Update $\mathfrak X_{i} $ using \eqref{GrindEQ__4-19_};
                \State Update light intensity $\mathfrak I_{i} $;
                \State NFE=0;
                \State Rank the individuals and define the subgroup best as leader;
            \EndFor
            \State \textbf{end for}
        \EndFor
        \State \textbf{end for}
        \State Merge the subgroups into one group;
        \State Rank the individuals and find the current best;
    \EndWhile
    \State \textbf{end while}
\end{algorithmic}
\end{algorithm}

\subsection{On the computational complexity of the proposed DWPFA}
\label{sec:4.3.3}
The computational complexity of an EA is an indicator of its execution time and is controlled by its structure. Let $O\left(F\right)$ denote the computational complexity of the fitness evaluation function $F\left(\cdot\right)$. The conventional FA has a computational complexity of $O\left(It_{max} \times n^2_{p} \times F\right)$ \cite{wang2017firefly}, where $It_{max}$ represents the maximum number of iterations and $n_{p}$ denotes the population size. However, compared to the relatively small $n_{p}$, the study of the number of attractions and movements during $It_{max}$ iterations could be more important. It is noteworthy that although larger $n_{p}$ can result in significant benefits in terms of the algorithm's performance, its negative consequence is a substantial increase in calculation time.

In the conventional FA, each firefly is compared with all other members of the population, and at each comparison step, one of the agents is moved. Hence, it can be concluded that each agent is moved with an average of $\left(n_{p}-1\right)/2$ times per iteration \cite{wang2017firefly}. Consequently, at each iteration of the conventional FA, $n_{p} \times \left(n_{p}-1\right)/2$ attractions are performed. It should be noted that, although the attraction enables the agents to find new optimal solutions, if a high number of attractions does not come along with better exploitation performance (leading to higher convergence speed of the algorithm), \textcolor{black}{excessive attractions can induce oscillations during the search process}, and simultaneously impose a high computational burden with less optimal solutions. Accordingly, an effective EA should provide a trade-off between the accuracy and computational cost. \textcolor{black}{In this regard}, efforts have been made to enhance the FAs performance considering the abovementioned objectives (\textit{i.e.} better convergence and keeping the computational complexity as low as possible). Authors in \cite{lv2018firefly} proposed a modified FA with Gaussian disturbance and local search (GDLSFA). \textcolor{black}{As reported, the solutions' accuracy and convergence speed has increased}; however, compared to the conventional FA, a higher number of attractions $3/2\times n_{p} \times \left(n_{p}-1\right) $ is achieved. Authors in \cite{nayak2021hyper} developed a memetic FA (MFA) to enhance the solutions' accuracy of FA; however, the optimization objective is achieved with an increased number of attractions as $n_{p} \times \left(n_{p}-1\right)$ due to the embedded exploitation process.

In the proposed DWPFA algorithm, the exploration process has the same number of attractions as the conventional FA has. During the exploitation process, each firefly moves with an average of $\Theta$ times per iteration. Consequently, at each iteration, $\Theta \times n_{sg} \times \left(n_{p}-n_{sg}\right)$ attractions are performed during the exploitation process. That is to say, the total number of each agent’s movement per iteration is $\left(n_{p}-1\right)/2+\Theta$, with an attraction number of $ n_{p} \times \left(n_{p}-1\right)/2+\left(\Theta \times n_{sg} \times \left(n_{p}-n_{sg}\right)\right)$. As a result, the total number of attractions per iteration is within the range of $\left[n_{p} \times \left(n_{p}-1\right)/2+\left(n_{p}-n_{sg}\right), n_{p} \times \left(n_{p}-1\right)/2+\left(\kappa_f \times n_{sg} \times \left(n_{p}-n_{sg}\right)\right)\right]$, which is slightly more than that of FA with the same computational complexity. The experiments in support of this work have shown $\kappa_f=4$ is a good selection for DWPFA. Figure \ref{fig:334-4} illustrates a comparative study on the number of attractions associated with FA, GDLSFA, MFA, and the proposed DWPFA with different numbers of populations. According to the foregoing analysis, it can be observed that the number of attractions of DWPFA under the full attraction model is much lower than the abovementioned studies, showing much less imposed computational burden. \textcolor{black}{In addition, lower number of attractions demonstrate faster performance with less computational complexity. In this respect, as shown in Figure \ref{fig:334-4}, the proposed DWPFA at its low ranges of attraction illustrates a similar complexity to that of conventional FA, and still outperforms other approaches at its high ranges, which shows its low computational complexity while delivering superior performance.} In addition, considering the DWPFA's superior exploitation performance and solutions accuracy compared to FA, the slight increase in its computational burden is negligible.

\begin{figure}
\centering
\includegraphics[width=4.8 in]{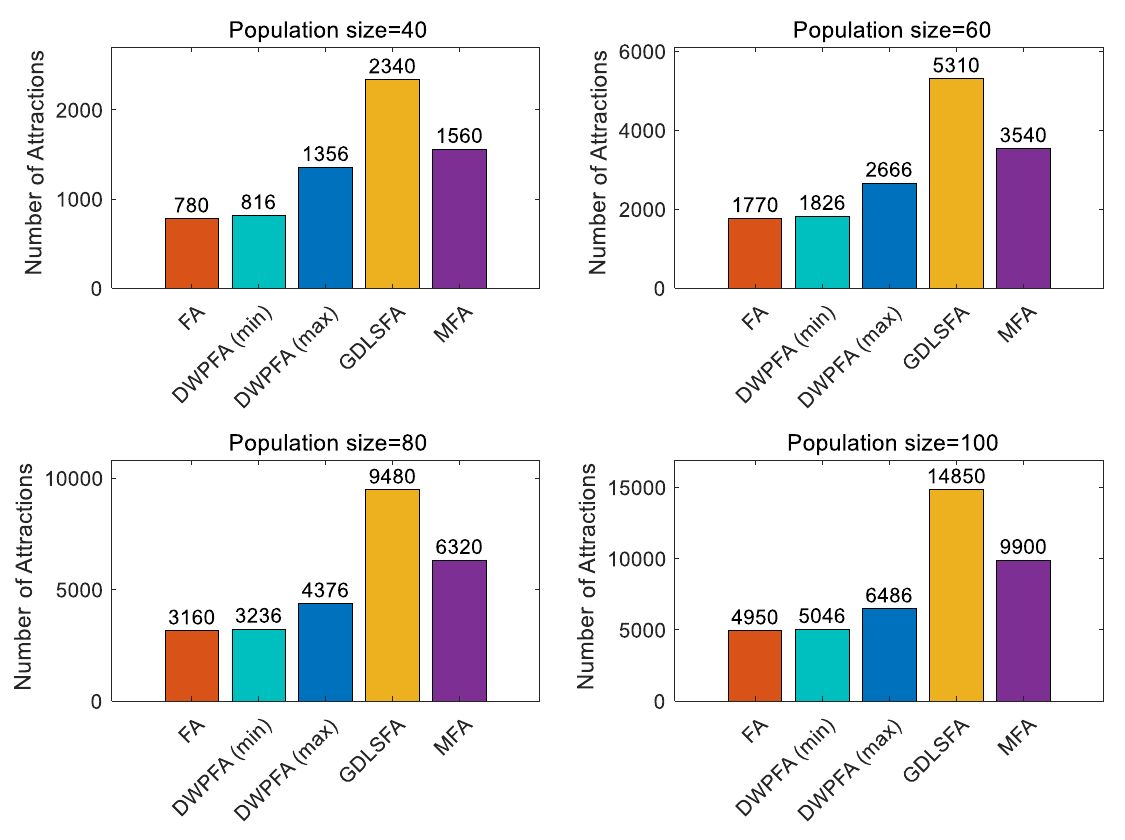}
\caption{\textcolor{black}{Comparative illustration on changes in the number of attractions associated with FA, GDLSFA \cite{lv2018firefly}, MFA \cite{nayak2021hyper}, and DWPFA.}}
\label{fig:334-4}
\end{figure}

\section{Proposed Extended Memory Pitch Control Strategy}
\label{sec:4.4}
PID controllers have been extensively applied in industrial applications owing to their design and implementation simplicity, low computational complexity, and robustness in the presence of external disturbances. FOPID controllers have demonstrated more flexibility to controller design, and more robustness in comparison with conventional PIDs \cite{mousavi2015memetic,viola2017design}. FOPIDs involve two additional degrees of freedom to the conventional PID; namely, the non-integer integral $\delta $ and derivative $\mu $ orders, leading to a more promising performance with five adjustable parameters \cite{mousavi2015memetic}. Since the wind energy level and direction changes continuously, and the pitch actuation system cannot provide immediate precise responses, considering the slight time delay between these changes can play an effective role in enhancing the control performance. Accordingly, it may be desirable to keep track of past effective pitch angles, as they serve as memory of the pitch system characteristics. Thus, extending the memory of pitch angles results in acquiring more data and a more comprehensive perspective of the system behavior. In this regard, in this work, the incorporation of fractional-calculus-based extended memory of pitch angles with an optimal FOPID controller is proposed in order to generate the desired pitch angle reference for the WT in region III.

Fractional calculus generalizes the integration and differentiation of a function to non-integer order, represented by ${\rm {\mathfrak D}}^{\gamma } $ operator, where $\gamma $ denotes the fractional order. Fractional-order derivatives and integrals can be derived through various definitions \cite{sabatier2007advances}. However, the Grunwald--Letnikov approximation is the most prominent definition in fractional-order calculus \cite{mousavi2018fractional,angel2018fractional,ren2018optimal,naidu2020power}. The G-L fractional derivative of a function $x\left(t\right)$ in discrete-time is expressed as follows.
\begin{equation}
{\rm {\mathfrak D}}^{\gamma } \left[x\left(t\right)\right]=\frac{1}{T^{\gamma } } \sum _{k=0}^{\Upsilon}\left(-1\right)^{k} \frac{\Gamma \left(\gamma +1\right)x\left(t-kT\right)}{\Gamma \left(k+1\right)\Gamma \left(\gamma -k+1\right)},
\label{GrindEQ__4-20_}
\end{equation}
\textcolor{black}{where $T$ and $\Upsilon$ denote the sampling period in [s] and the truncation order, respectively.} $\Gamma \left(\cdot\right)$ is Euler's gamma function, where $\Gamma \left(z\right)=\int_{0}^{\infty} \mathfrak n^{z-1} e^{-\mathfrak n} \si{d \mathfrak n},\, \, \si{Re}\left(z\right)>0 $.

The main control actions in region III are carried out by the pitch system, through designing a controller to minimize the error $e\left(t\right)=\omega _{nom} -\omega _{g} \left(t\right)$. The proposed EM-FOPID is implemented in the time domain as follows,
\begin{equation}
\beta _{r} \left(t\right)=\beta _{r} \left(t-1\right)+K_{p} e\left(t\right)+K_{i} {\rm {\mathfrak D}}^{-\delta } e\left(t\right)+K_{d} {\rm {\mathfrak D}}^{\mu } e\left(t\right).
\label{GrindEQ__4-21_}
\end{equation}
\begin{rem}
\label{rem:4-3}
According to fractional calculus concepts, despite the integer-order derivative that represents a finite series, the fractional-order derivative involves an infinite number of terms \cite{sabatier2007advances}. This characteristic leads to acquiring a memory of all past pitch angles and can be controlled by the fractional order $0\le \gamma \le 1$. \textcolor{black}{It is noteworthy that} the controller is implemented in discrete time to benefit from the memory preservation characteristics of the discrete-time G-L fractional derivative. Hence, all considered time-dependent variables discrete-time.
\end{rem}

In this perspective, \eqref{GrindEQ__4-21_} can be rearranged as
\begin{equation}
\beta _{r} \left(t\right)-\beta _{r} \left(t-1\right)=K_{p} e\left(t\right)+K_{i} {\rm {\mathfrak D}}^{-\delta } e\left(t\right)+K_{d} {\rm {\mathfrak D}}^{\mu } e\left(t\right).
\label{GrindEQ__4-22_}
\end{equation}
\begin{rem}
\label{rem:4-11}
Assuming $T=1$, the left side of \eqref{GrindEQ__4-22_} represents the G-L fractional derivative with order $\gamma =1$, which leads to:
\begin{equation}
{\rm {\mathfrak D}}^{\gamma } \left[\beta _{r} \left(t\right)\right]=K_{p} e\left(t\right)+K_{i} {\rm {\mathfrak D}}^{-\delta } e\left(t\right)+K_{d} {\rm {\mathfrak D}}^{\mu } e\left(t\right).
\label{GrindEQ__4-23_}
\end{equation}
\end{rem}

Thus, considering the first $\Upsilon =4$ terms of differential derivative \eqref{GrindEQ__4-20_}, \eqref{GrindEQ__4-23_} can be rewritten as follows:
\begin{align}
\label{GrindEQ__4-25_}
\beta _{r} \left(t\right) & =\gamma \beta _{r} \left(t-1\right) \\ \nonumber
&+\frac{1}{2{\rm !}} \gamma \left(1-\gamma \right)\beta _{r} \left(t-2\right)+\frac{1}{3{\rm !}} \gamma \left(1-\gamma \right)\left(2-\gamma \right)\beta _{r} \left(t-3\right) \\ \nonumber
& +\frac{1}{4{\rm !}} \gamma \left(1-\gamma \right)\left(2-\gamma \right)\left(3-\gamma \right)\beta _{r} \left(t-4\right) \\ \nonumber
&+K_{p} e\left(t\right)+K_{i} {\rm {\mathfrak D}}^{-\delta } e\left(t\right)+K_{d} {\rm {\mathfrak D}}^{\mu } e\left(t\right).
\end{align}

\begin{rem}
\label{rem:4-114}
Experimental tests have shown that larger values of $\Upsilon>4$ lead to similar results \cite{mousavi2015memetic}.
\end{rem}

In this work, the proposed DWPFA is utilized to achieve the optimum parameters $\left\{K_{p} ,K_{i} ,\delta ,K_{d} ,\mu \right\}$ for the proposed EM-FOPID controller. To simultaneously improve the transient and steady-state error, the integral of time multiplied squared error (ITSE) index is \textcolor{black}{incorporated} in the objective function. Also, to avoid large control signals and simultaneously reduce its deviations, which may lead to actuator saturation, the integral of squared control signal (ISCO) is embedded alongside the ITSE. The objective function evaluated to determine the controller parameters \textcolor{black}{is defined} as $J=\int _{0}^{\infty }\left[te^{2} \left(t\right)+u^{2} \left(t\right)\right] \si{dt}+V_c$ subjected to the pitch actuator constraints stated in Remark \ref{rem:4-2}, where $u\left(t\right)$ denotes the controller's output, and $V_c$ is the constraints violation coefficient such that $V_c$ is a very small value when the constraints are respected during the optimal design procedure, and $V_c=\infty$ in case of any violations. Figure \ref{fig:44-4} depicts the procedure of tuning the proposed EM-FOPID using DWPFA algorithm.

\begin{rem}
\label{rem:4-456}
In this study, the pitch actuator \eqref{GrindEQ__4-6_} is modelled as a second-order transfer function with constraints to ensure the feasible operational range of the WT \cite{odgaard2013fault}. A common drawback associated with some studies, such as \cite{venkaiah2020hydraulically}, is that the pitch actuator constraints are not considered in the controller's design procedure, which may cause the wind-up phenomenon and consequently degrade the WT performance if the control input reaches the saturation limits. In this regard, in many studies magnitude and rate limiters are implemented to deal with the constraints \cite{odgaard2013fault,lan2018fault,shan2021distributed}. In this work, constraints are explicitly checked at each iteration, with high penalties added to the objective function in case of any violations, to ensure adherence to constraints in future iterations.
\end{rem}

\begin{figure}[t!]
\centering
\includegraphics[width=3 in]{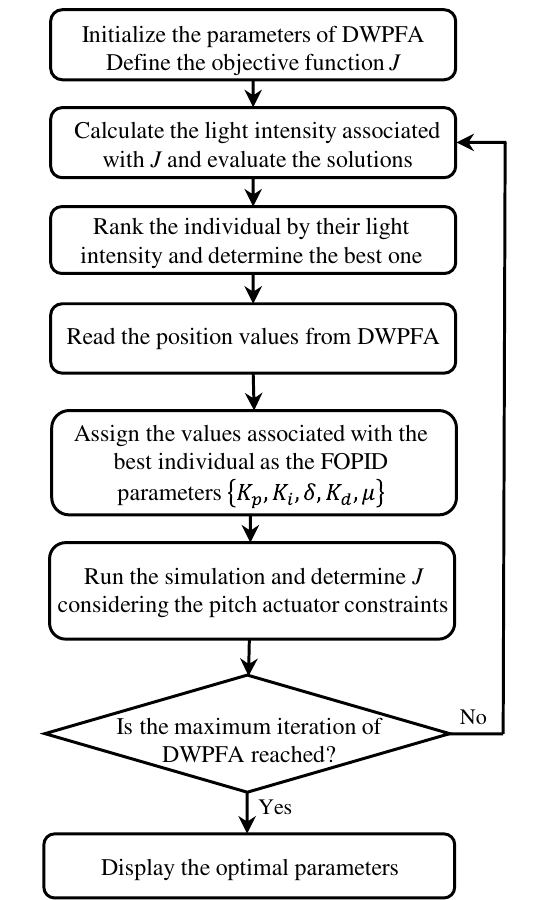}
\caption{The flowchart of the proposed EM-FOPID tuning using DWPFA.}
\label{fig:44-4}
\end{figure}

\section{Simulation Results and Discussions}
\label{sec:4.5}
The performance of the proposed DWPFA algorithm is first evaluated in comparison with other EAs through solving the CEC2017 mathematical benchmark functions. Then, the proposed DWPFA-optimized EM-FOPID controller is utilized to adjust the pitch angle of WT blades, where its performance is compared to PI and DWPFA-optimized conventional FOPID approaches under sensor, actuator, and system faults. Simulations are carried out using MATLAB R2020a (9.8.0.1417392) 64-bit, on a ASUS laptop with 64-bit win10 operating system, processor: Intel{\tiny{\textregistered}} core${}^{\si{TM}}$ i7-8550U CPU 2.50 GHz, installed memory: 8.00 GB, and VGA: GeForce NVidia 620M-4GB.

\subsection{Performance Evaluation of DWPFA}
\label{sec:4.5.1}
In this section, well-defined CEC2017 special session mathematical benchmark functions \cite{awad2017ensemble} are used as objective functions, to \textcolor{black}{assess} the performance of the proposed DWPFA \textcolor{black}{compared to} other EAs. In this regard, 30 test functions are used and are categorized as follows: unimodal functions $\left(f_{1} -f_{3} \right)$, simple multimodal functions $\left(f_{4} -f_{10} \right)$, hybrid functions $\left(f_{11} -f_{20} \right)$, and composition functions $\left(f_{21} -f_{30} \right)$. To testify the performance of the proposed DWPFA algorithm, it is compared with PSO, fractional PSO-based memetic algorithm (FPSOMA) \cite{mousavi2015memetic}, grey wolf optimizer (GWO), enhanced GWO (EGWO) \cite{luo2019enhanced}, enhanced bacterial foraging optimization (CCGBFO) \cite{chen2020enhanced}, FA, and fractional order FA (FOFA) \cite{mousavi2018fractional}. In the proposed algorithm, $\gamma_f=1$ and $\chi_{f,0}=1.2$ are considered. \textcolor{black}{In the performed experiments}, the benchmark functions' dimension is set to $D=50$, and all EAs have a population size of 40. Each algorithm is run 200 times independently for each test instance, and the allowed number of maximum function evaluation (NFE) is set to $10000\times D$. Since achieving zero error on CEC2017 functions is a demanding task for the algorithms, the constant $\varepsilon $ is defined as an acceptable threshold value of a satisfactory solution near the optimal solution for each function. In this work, $\varepsilon =50$ is set for $\left(f_{1} -f_{20} \right)$ and $\varepsilon =500$ is set for $\left(f_{21} -f_{30} \right)$. In Table \ref{tab:4-2}, the first column illustrates the sequence of 30 CEC2017 benchmark functions, and the next columns show the mean results achieved, while the minimum value obtained for each function is emphasized in bold. In addition, a comparison is carried out by reporting the experimental results of CEC2017 test suit as a logarithmic radar graph (spider plot) in Figure \ref{fig:5-4}.

\begin{table*}
\caption{\textcolor{black}{Obtained results for CEC2017 test functions with $D=50$.}}
\centering
\label{tab:4-2}
\resizebox{\textwidth}{!}{
\begin{tabular}{l c c c c c c c c c}
\hline\hline \\[-3mm]
\multicolumn{1}{c}{Fun.} & \multicolumn{1}{c}{PSO} & \multicolumn{1}{c}{FPSOMA} & \multicolumn{1}{c}{GWO} & \multicolumn{1}{c}{EGWO}  & \multicolumn{1}{c}{BFO} & \multicolumn{1}{c}{CCGBFO} & \multicolumn{1}{c}{FA} & \multicolumn{1}{c}{FOFA} & \multicolumn{1}{c}{DWPFA} \\[1.2ex] \hline
$f_{1} $ & $1.20\times 10^4$ & $1.06\times 10^2$ & $2.95\times 10^3$ & $2.17\times 10^2$ & $3.94\times 10^2$ & $3.41\times 10^2$ & $1.88\times 10^2$ & $1.47\times 10^1$ & \textbf{$1.05\times 10^1$} \\
$f_{2} $ & $5.08\times 10^7$ & $1.22\times 10^5$ & $2.61\times 10^7$ & $4.70\times 10^5$ & $4.90\times 10^5$ & $3.22\times 10^5$ & $2.93\times 10^6$ & $1.49\times 10^5$ & \textbf{$1.29\times 10^4$} \\
$f_{3} $ & $3.75\times 10^5$ & $2.25\times 10^0$ & $3.64\times 10^4$ & $2.13\times 10^2$ & $2.23\times 10^4$ & $2.14\times 10^2$ & $4.69\times 10^2$ & $2.50\times 10^{-3}$ & \textbf{$2.19\times 10^{-3}$} \\
$f_{4} $ & $2.30\times 10^4$ & $1.69\times 10^1$ & $2.37\times 10^3$ & $3.73\times 10^1$ & $2.67\times 10^2$ & $2.12\times 10^2$ & $4.63\times 10^2$ & $1.99\times 10^1$ & \textbf{$1.10\times 10^1$} \\
$f_{5} $ & $1.68\times 10^3$ & $4.49\times 10^1$ & $1.94\times 10^2$ & $1.67\times 10^1$ & $3.89\times 10^2$ & $5.53\times 10^1$ & $1.82\times 10^2$ & \textbf{$1.21\times 10^1$} & $1.86\times 10^1$ \\
$f_{6} $ & $2.19\times 10^3$ & $1.89\times 10^{-2}$ & $3.65\times 10^1$ & $2.40\times 10^{-2}$ & $2.60\times 10^1$ & $1.13\times 10^{-1}$ & $3.60\times 10^1$ & \textbf{$1.53\times 10^{-3}$} & $1.58\times 10^{-2}$ \\
$f_{7} $ & $2.14\times 10^3$ & $1.84\times 10^1$ & $2.78\times 10^2$ & $2.44\times 10^2$ & $2.30\times 10^2$ & $1.75\times 10^2$ & $3.54\times 10^2$ & $1.73\times 10^1$ & \textbf{$1.06\times 10^1$} \\
$f_{8} $ & $1.50\times 10^3$ & $1.71\times 10^1$ & $5.75\times 10^2$ & $2.47\times 10^1$ & $3.26\times 10^2$ & $2.13\times 10^2$ & $1.50\times 10^2$ & $1.43\times 10^1$ & \textbf{$1.38\times 10^1$} \\
$f_{9} $ & $1.74\times 10^3$ & $1.61\times 10^{-1}$ & $5.59\times 10^2$ & $2.63\times 10^1$ & $3.80\times 10^2$ & $1.34\times 10^2$ & $2.54\times 10^1$ & $1.03\times 10^{-1}$ & \textbf{$1.92\times 10^{-2}$} \\
$f_{10} $ & $1.44\times 10^4$ & $1.28\times 10^3$ & $3.30\times 10^3$ & $1.35\times 10^3$ & $3.16\times 10^3$ & $2.57\times 10^3$ & $3.44\times 10^3$ & $2.51\times 10^3$ & \textbf{$1.11\times 10^3$} \\
$f_{11} $ & $1.43\times 10^4$ & $5.57\times 10^1$ & $4.75\times 10^1$ & $1.37\times 10^1$ & $5.19\times 10^2$ & $2.57\times 10^1$ & $4.10\times 10^1$ & $1.69\times 10^1$ & \textbf{$1.26\times 10^1$} \\
$f_{12} $ & $5.87\times 10^4$ & $4.32\times 10^3$ & $3.10\times 10^3$ & $4.89\times 10^3$ & $2.51\times 10^4$ & $5.21\times 10^3$ & $4.93\times 10^3$ & $3.02\times 10^3$ & \textbf{$1.08\times 10^3$} \\
$f_{13} $ & $1.31\times 10^3$ & $1.45\times 10^1$ & $1.89\times 10^2$ & $4.77\times 10^1$ & $5.75\times 10^1$ & $3.65\times 10^1$ & $5.38\times 10^2$ & $4.62\times 10^1$ & \textbf{$1.05\times 10^1$} \\
$f_{14} $ & $1.63\times 10^4$ & $3.77\times 10^1$ & $1.51\times 10^2$ & $2.96\times 10^1$ & $6.65\times 10^1$ & $2.92\times 10^1$ & $5.83\times 10^1$ & \textbf{$1.89\times 10^1$} & $2.17\times 10^1$ \\
$f_{15} $ & $6.52\times 10^3$ & $1.84\times 10^1$ & $2.63\times 10^2$ & $6.87\times 10^1$ & $2.95\times 10^2$ & $4.51\times 10^1$ & $3.47\times 10^2$ & $4.13\times 10^1$ & \textbf{$1.09\times 10^1$} \\
$f_{16} $ & $6.62\times 10^3$ & $2.19\times 10^2$ & $2.85\times 10^3$ & $1.75\times 10^2$ & $3.39\times 10^2$ & $2.72\times 10^2$ & $4.41\times 10^2$ & \textbf{$1.13\times 10^2$} & $1.83\times 10^2$ \\
$f_{17} $ & $2.24\times 10^8$ & $4.28\times 10^1$ & $3.43\times 10^3$ & $4.58\times 10^2$ & $2.86\times 10^3$ & $5.03\times 10^2$ & $2.30\times 10^3$ & $5.45\times 10^2$ & \textbf{$2.85\times 10^1$} \\
$f_{18} $ & $9.42\times 10^3$ & $1.47\times 10^1$ & $3.49\times 10^2$ & $5.39\times 10^1$ & $3.12\times 10^2$ & $3.30\times 10^1$ & $6.12\times 10^2$ & $2.39\times 10^1$ & \textbf{$1.18\times 10^1$} \\
$f_{19} $ & $2.68\times 10^3$ & $4.52\times 10^1$ & $5.10\times 10^2$ & $4.29\times 10^1$ & $2.44\times 10^3$ & $6.09\times 10^1$ & $6.01\times 10^2$ & $3.36\times 10^1$ & \textbf{$2.23\times 10^1$} \\
$f_{20} $ & $6.94\times 10^5$ & $1.97\times 10^2$ & $2.94\times 10^3$ & $3.14\times 10^2$ & $6.57\times 10^2$ & $4.40\times 10^2$ & $3.27\times 10^2$ & $1.87\times 10^2$ & \textbf{$1.34\times 10^2$} \\
$f_{21} $ & $4.36\times 10^7$ & $2.36\times 10^2$ & $2.11\times 10^3$ & $2.01\times 10^2$ & $3.51\times 10^4$ & $2.01\times 10^2$ & $4.54\times 10^3$ & $3.37\times 10^2$ & \textbf{$1.29\times 10^2$} \\
$f_{22} $ & $6.61\times 10^3$ & $4.99\times 10^2$ & $3.18\times 10^3$ & $2.53\times 10^2$ & $5.60\times 10^3$ & $4.75\times 10^2$ & $1.07\times 10^3$ & $3.95\times 10^2$ & \textbf{$1.12\times 10^2$} \\
$f_{23} $ & $4.24\times 10^3$ & \textbf{$1.08\times 10^2$} & $1.99\times 10^2$ & $1.24\times 10^2$ & $1.58\times 10^2$ & $1.34\times 10^2$ & $2.70\times 10^2$ & $1.14\times 10^2$ & $1.32\times 10^2$ \\
$f_{24} $ & $2.34\times 10^3$ & $1.55\times 10^2$ & $3.23\times 10^2$ & $2.34\times 10^2$ & $1.47\times 10^3$ & $1.67\times 10^2$ & $1.03\times 10^3$ & $3.71\times 10^2$ & \textbf{$1.09\times 10^2$} \\
$f_{25} $ & $5.86\times 10^3$ & $1.97\times 10^2$ & $2.56\times 10^2$ & $1.95\times 10^2$ & $3.59\times 10^2$ & $2.95\times 10^2$ & $1.87\times 10^2$ & $1.36\times 10^2$ & \textbf{$1.19\times 10^2$} \\
$f_{26} $ & $2.71\times 10^3$ & $1.72\times 10^2$ & $5.12\times 10^2$ & $3.39\times 10^2$ & $2.11\times 10^3$ & $4.99\times 10^2$ & $1.44\times 10^3$ & $6.35\times 10^2$ & \textbf{$1.47\times 10^2$} \\
$f_{27} $ & $4.75\times 10^3$ & $1.54\times 10^2$ & $5.42\times 10^2$ & $1.84\times 10^2$ & $2.95\times 10^2$ & $2.04\times 10^2$ & $2.22\times 10^2$ & $1.36\times 10^2$ & \textbf{$1.33\times 10^2$} \\
$f_{28} $ & $9.28\times 10^2$ & $1.37\times 10^2$ & $1.78\times 10^2$ & $1.45\times 10^2$ & $2.85\times 10^2$ & $1.52\times 10^2$ & $4.67\times 10^2$ & $1.24\times 10^2$ & \textbf{$1.02\times 10^2$} \\
$f_{29} $ & $5.96\times 10^7$ & $1.95\times 10^2$ & $3.89\times 10^3$ & $3.67\times 10^2$ & $3.81\times 10^3$ & $4.33\times 10^2$ & $4.47\times 10^2$ & $1.70\times 10^2$ & \textbf{$1.27\times 10^2$} \\
$f_{30} $ & $1.96\times 10^6$ & $1.79\times 10^3$ & $2.08\times 10^4$ & $1.52\times 10^3$ & $2.97\times 10^4$ & $3.43\times 10^3$ & $5.61\times 10^4$ & \textbf{$1.16\times 10^3$} & $1.46\times 10^3$ \\
[1ex]
\hline\hline
\end{tabular}}
\end{table*}

\begin{figure}
\centering
\includegraphics[width=4.6 in]{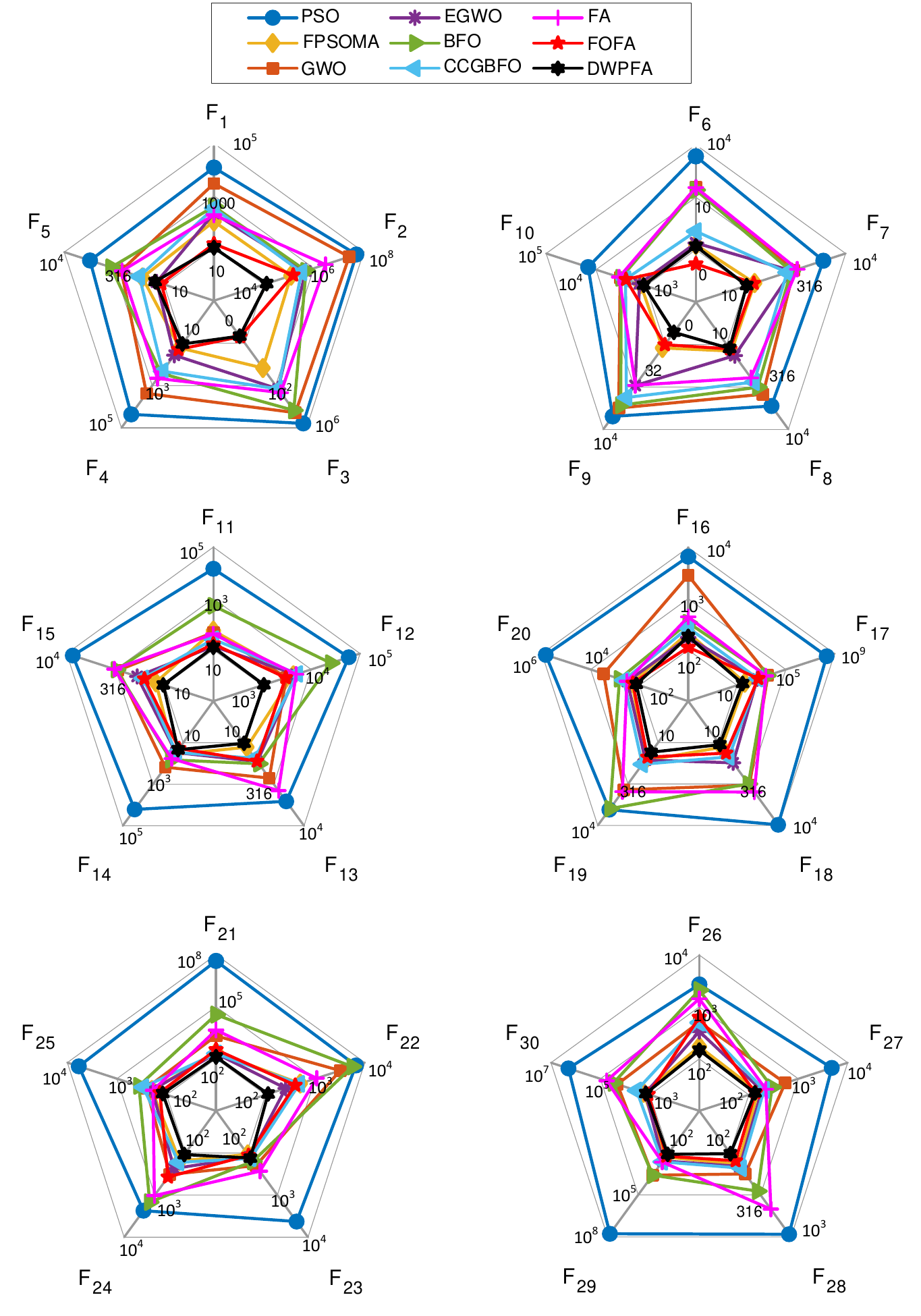}
\caption{Spider plot of results achieved on CEC2017 benchmark functions.}
\label{fig:5-4}
\end{figure}

According to Table \ref{tab:4-2}, it can be seen that the optimization performance obtained by the proposed DWPFA are markedly less than other algorithms in solving real-parameter optimization problems in the CEC2017 suite. Results illustrated in Table \ref{tab:4-2} and Figure \ref{fig:5-4}  demonstrate the superiority of DWPFA, yielding outstanding performance with 24 best solutions achieved out of 30 test problems, followed by FOFA and FPSOMA with 5 and 1 best solutions, respectively.

To intuitively show the performance of DWPFA, Figure \ref{fig:6-4} is plotted to illustrate the box-and-whisker diagrams of solutions obtained on different selected problems of each category, for all 200 runs with $D=50$. The vertical and horizontal axes indicate the optimal solution and the nine algorithms, respectively. From Figure \ref{fig:6-4} it is observed that despite the high complexity of functions, DWPFA provides promising results maintaining fewer values and shorter distribution of solutions comparing to other algorithms under evaluation, indicating excellent and steady performances of it. This implies that DWPFA is more effective for optimizing functions, and thus its superiority is apparent.
\begin{figure}[t!]
\centering
\includegraphics[width=4.5 in]{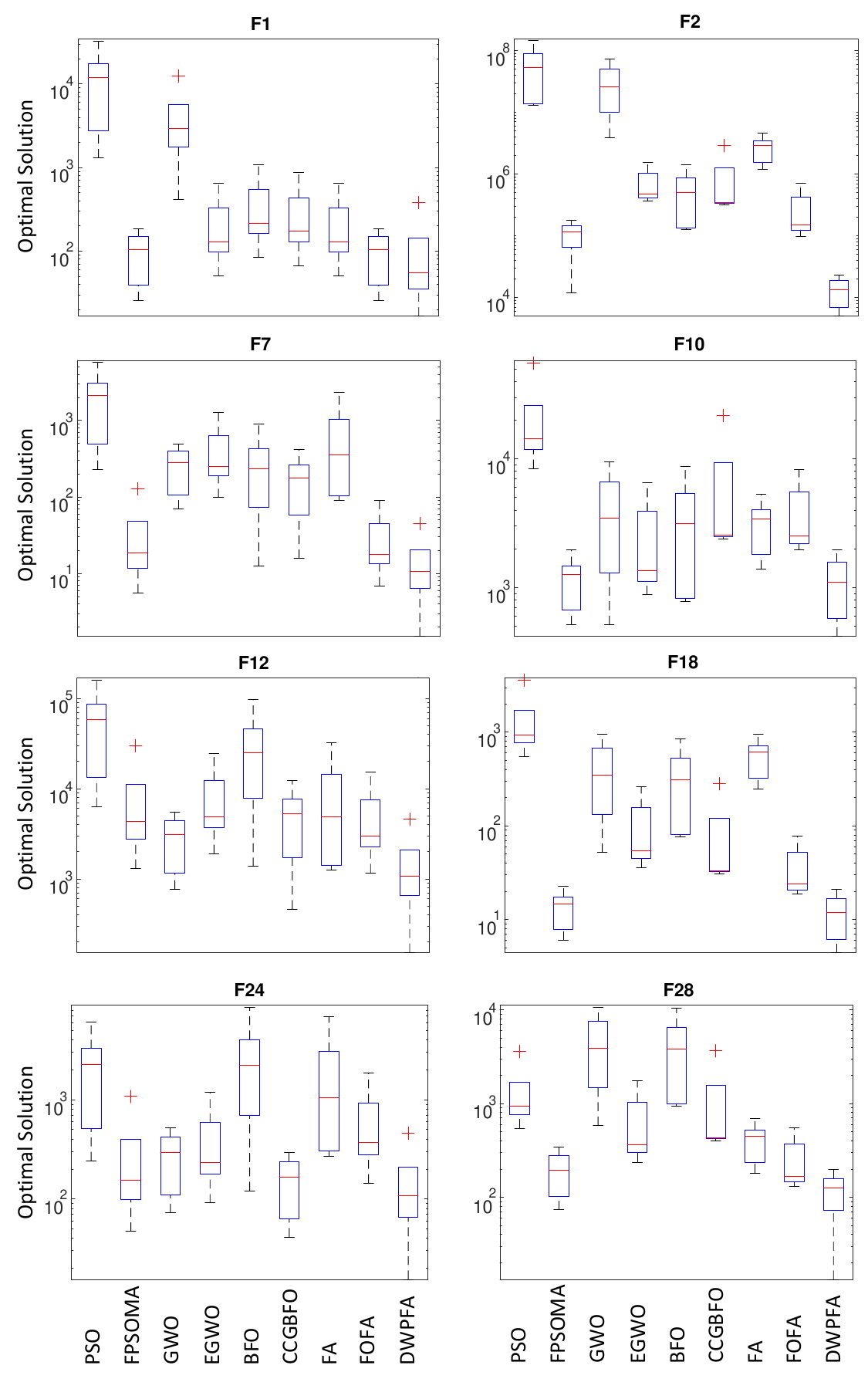}
\caption{The box-and-whisker comparative performance diagrams on the selected functions, $D=50$.}
\label{fig:6-4}
\end{figure}
In order to determine whether to accept or reject the null hypothesis, non-parametric tests can be utilized. Non-parametric tests determine whether the data sets to be compared have the same variance \cite{carrasco2020recent}. Accordingly, to statistically compare and analyze the quality of the solution, two non-parametric statistical hypothesis tests were used to compare the results, namely the Friedman test and the Friedman Aligned Ranks test \cite{carrasco2020recent}. The null hypothesis for the Friedman test represents the equality of medians between the populations, while the ranks assigned to the resulting differences (aligned observations) are called Friedman aligned ranks. Table \ref{tab:4-3} demonstrates the Friedman and Friedman Aligned test results sorted by the performance order \say{Rank}. The results indicate that DWPFA obtains the best rank, followed by FOFA and FPSOMA. The ranking of all algorithms in both tests is the same, except for EGWO, where its rank improved from fifth to fourth, taking CCGBFO's place.
\begin{table}[t!]
\caption{The Friedman and Friedman Aligned test results over CEC2017 test functions, $D=50$.}
\centering
\label{tab:4-3}
\scalebox{0.8}{
\begin{tabular}{l c c c l c c}
\hline\hline \\[-3mm]
\multicolumn{3}{c}{Friedman} &  & \multicolumn{3}{c}{Friedman Aligned} \\ \hline
\multicolumn{1}{c}{Algorithm} & \multicolumn{1}{c}{Score} & \multicolumn{1}{c}{Rank} & \multicolumn{1}{c}{ } & \multicolumn{1}{c}{Algorithm} & \multicolumn{1}{c}{Score} & \multicolumn{1}{c}{Rank}  \\[1ex] \hline
DWPFA & 2.1925 & 1 &  & DWPFA & 33.263 & 1 \\
FOFA & 2.3852 & 2 &  & FOFA & 34.174 & 2 \\
FPSOMA & 2.5682 & 3 &  & FPSOMA & 34.368 & 3 \\
CCGBFO & 2.7268 & 4 &  & EGWO & 34.404 & 4 \\
EGWO & 3.1481 & 5 &  & CCGBFO & 34.671 & 5 \\
FA & 3.1692 & 6 &  & FA & 35.816 & 6 \\
BFO & 3.3744 & 7 &  & BFO & 36.221 & 7 \\
GWO & 3.5760 & 8 &  & GWO & 36.414 & 8 \\
PSO & 3.8767 & 9 &  & PSO & 38.162 & 9 \\
[1ex]
\hline\hline
\end{tabular}}
\end{table}

\subsection{Numerical Example}
\label{sec:4.5.33}
\textcolor{black}{This section investigates the closed-loop performance of the proposed EM-FOPID approach compared to the results obtained by other relevant methods in terms of ITSE performance criterion}. Consider the following system adopted from literature \cite{mousavi2015memetic},
\begin{equation}
m_1{\rm {\mathfrak D}}^{n_1}{x}\left(t\right)+ m_2{\rm {\mathfrak D}}^{n_2}{x}\left(t\right)+ m_3x\left(t\right)=u\left(t\right),
\label{GrindEQ__4-S1_}
\end{equation}
where $m_1=0.8$, $m_2=0.9$, $m_3=1$, $n_1=2.2$, and $n_2=0.5$.

The control variable $u(t)$ can be considered as
\begin{equation}
u\left(t\right)=u\left(t-1\right)+K_{p} e\left(t\right)+K_{i} {\rm {\mathfrak D}}^{-\delta } e\left(t\right)+K_{d} {\rm {\mathfrak D}}^{\mu } e\left(t\right).
\label{GrindEQ__4-S2_}
\end{equation}

Considering Remark \ref{rem:4-11} and employing the extended memory characteristics, \eqref{GrindEQ__4-S2_} can be rewritten as,
\begin{align}
\label{GrindEQ__4-S3_}
u\left(t\right)& =\gamma u \left(t-1\right) \\ \nonumber
&+\frac{1}{2{\rm !}} \gamma \left(1-\gamma \right)u \left(t-2\right)+\frac{1}{3{\rm !}} \gamma \left(1-\gamma \right)\left(2-\gamma \right)u \left(t-3\right) \\ \nonumber
& +\frac{1}{4{\rm !}} \gamma \left(1-\gamma \right)\left(2-\gamma \right)\left(3-\gamma \right)u \left(t-4\right) \\ \nonumber
&+K_{p} e\left(t\right)+K_{i} {\rm {\mathfrak D}}^{-\delta } e\left(t\right)+K_{d} {\rm {\mathfrak D}}^{\mu } e\left(t\right).
\end{align}

Table \ref{tab:4-444} summarizes the step responses obtained by the controllers under study. Accordingly, it is observed that, compared to PID controllers, the FOPIDs demonstrate superior performance in terms of maximum overshoot and rise time. In addition, the proposed DWPFA algorithm delivers better tuning performance compared to other algorithms. Comparison study involving DWPFA-optimized EM-FOPID and other methods validates the effectiveness of embedding the memory characteristics in the controller. Accordingly, the DWPFA-optimized EM-FOPID outperforms all other methods and demonstrates more preferable performance with a maximum overshoot of 0.75\% and rise time of 0.002 seconds, followed by the DWPFA-optimized FOPID with 1.25\% overshoot and 0.003 seconds. In addition, although the FPSOMA-based FOPID \cite{mousavi2015memetic} provides an acceptable performance, large controller gains $K_p$, $K_i$, and $K_d$ are required, whereas, with similar fractional integral and derivative orders $\delta $ and $\mu $, the impact of incorporating the memory effects to the controller has led to smaller controller gains with an even better control performance.
\begin{table}[t!]
\caption{Performance comparison of the controllers.}
\centering
\label{tab:4-444}
\resizebox{\textwidth}{!}{
\begin{tabular}{l c c c c c c c c}
\hline\hline \\[-3mm]
\multicolumn{1}{l}{Controller} & \multicolumn{1}{c}{$K_{p} $} & \multicolumn{1}{c}{$K_{i} $ } & \multicolumn{1}{c}{$\delta $ } & \multicolumn{1}{c}{$K_{d} $} & \multicolumn{1}{c}{$\mu $} & \multicolumn{1}{c}{$\gamma $} & \multicolumn{1}{c}{Overshoot (\%)} & \multicolumn{1}{c}{Rise Time [\si{s}]}  \\[1ex] \hline
PID Ziegler–Nichols \cite{mousavi2015memetic} & 16.6281 & 9.4422 & - & 7.2230 & - & - & 25.92 & 0.223 \\
PID-DWPFA & 5.3100 & 10.8620 & - & 18.1150 & - & - & 12.28 & 0.137 \\
FOPID-FPSOMA \cite{mousavi2015memetic} & 393.9550 & 353.9850 & 0.12 & 117.8490 & 1.2240 & - & 1.73 & 0.003 \\
FOPID-DE \cite{biswas2009design} & 21.2200 & 1.3700 & 0.92 & 12.0500 & 0.93 & - & 7.69 & 0.023 \\
FOPID-DWPFA & 36.2200 & 27.0500 & 0.24 & 112.2000 & 1.13 & - & 1.25 & 0.003 \\
EM-FOPID-DWPFA & 24.8400 & 19.8540 & 0.16 & 84.6800 & 1.34 & 0.7 & 0.75 & 0.002 \\
[1ex]
\hline\hline
\end{tabular}}
\end{table}

\subsection{Optimal Pitch Angle Control of WT}
\label{sec:4.5.2}
In this section, the proposed DWPFA-optimized EM-FOPID is applied to generate the desired pitch angle reference of WT blades in region III, and its performance is evaluated with respect to the conventional PI and DWPFA-optimized conventional FOPID approaches under fault-free and faulty conditions. In this study, it is assumed that all required system signals are available for the controller. In practice, it is often necessary to estimate the wind speed and other signals. The WT model parameters are listed in Table \ref{tab:4-4}, and the WT is subjected to sensor, actuator, and system faults described in Table \ref{tab:4-1}. The proposed DWPFA algorithm is used to tune the controller parameters (illustrated in Table \ref{tab:4-5}), and comparative simulations are conducted to validate the efficiency of the proposed optimal EM-FOPID.
\begin{table}[t!]
\caption{WT model parameters.}
\centering
\label{tab:4-4}
\scalebox{0.8}{
\begin{tabular}{l l l l l l l}
\hline\hline \\[-3mm]
\multicolumn{1}{l}{Parameter} & \multicolumn{1}{l}{Value} & \multicolumn{1}{l}{Unit} & \multicolumn{1}{c}{ } & \multicolumn{1}{l}{Parameter} & \multicolumn{1}{l}{Value} & \multicolumn{1}{l}{Unit}  \\[1ex] \hline
$R$ & 57.5 & \si{m} &  & $D_{LS} $ & 775.49  & \si{Nms/rad} \\
$\rho $ & 1.225  & \si{kg/m^{3}} &  & $B_{G} $ & 46.6  & \si{Nms/rad} \\
$\omega _{n} $ & 11.11  & \si{rad/s} &  & $N_{GB} $ & 95 & - \\
$\xi $ & 0.6 & - &  & \si{\eta _{dt}} & 0.97 & - \\
$J_{R} $ & 55E+06  & \si{kg.m^{2}} &  & $\alpha _{gc} $ & 50  & \si{rad/s} \\
$J_{G} $ & 390  & \si{kg.m^{2}} &  & $\omega _{nom} $ & 162  & \si{rad/s} \\
$K_{LS} $ & 2.7E+09  & \si{Nm/rad} &  & $K_{opt} $ & 1.2171 & - \\
[1ex]
\hline\hline
\end{tabular}}
\end{table}

\begin{table}[t!]
\caption{Controllers parameters for the WT system.}
\centering
\label{tab:4-5}
\scalebox{0.8}{
\begin{tabular}{l c c c c c c}
\hline\hline \\[-3mm]
\multicolumn{1}{l}{Controller} & \multicolumn{1}{c}{$K_{p} $} & \multicolumn{1}{c}{$K_{i} $ } & \multicolumn{1}{c}{$\delta $ } & \multicolumn{1}{c}{$K_{d} $} & \multicolumn{1}{c}{$\mu $} & \multicolumn{1}{c}{$\gamma $}  \\[1ex] \hline
Conventional PI \cite{odgaard2013fault} & 4 & 1 & - & - & - & - \\
DWPFA-FOPID & 8.2 & 3 & 0.7 & 2 & 0.5 & - \\
DWPFA-EM-FOPID & 5.5 & 2.1 & 0.8 & 1.8 & 0.45 & 0.7 \\
[1ex]
\hline\hline
\end{tabular}}
\end{table}

\textcolor{black}{A critical issue in designing a controller is to ensure closed-loop stability; this applies whether the system is linear or non-linear. However, although it can generally be achieved for linear systems and some classes of non-linear systems, analytical investigation of the closed-loop stability for FOPID controllers for a 4.8-MW WT is not a straightforward task} due to the unavailability of the exact dynamic model in Laplace domain representation and the nonlinearities associated with the WT system. \textcolor{black}{On the contrary}, PID/FOPID controllers generally have the ability to destabilize a system if they are poorly designed. However, since the objective function uses ITSE, the optimization algorithm is expected to ensure stability and avoid any parameters that destabilize the system. Hence, if a set of control parameters cause instability at any wind speeds, the cost function would be a large value, and unstable modes would not be found to be optimal and will be avoided over successive generations. Consequently, the closed-loop system can be guaranteed to remain stable during the optimization process.

\textcolor{black}{The design parameters are considered within the search ranges $-500<\left(K_p,K_i,K_d\right)<500$, and $0<\left(\delta,\mu\right)<2$.} It is often essential that the controller's parameters are not chosen close to the marginal stability regions, which may lead to performance degradation in this particular case. Hence, to guarantee the system's stability, the controller's gains $K_p$, $K_i$, and $K_d$ must be non-negative, and the fractional orders $\delta$ and $\mu$ should be chosen such that they maintain a trade-off between a) getting as far as possible from stability margins and b) providing desirable integration (leading to higher precision) and derivation (leading to more stability) performances. The marginal stability regions are depicted in Figure \ref{fig:222-4}(a) via Venn diagram. It is worth mentioning that, although the whole area $0<\left(\delta,\mu\right)<2$ (as shown in Figure \ref{fig:222-4}(b)) maintain the stability and acceptable performance, the golden zone $0.3<\left(\delta,\mu\right)<0.95$ has found to deliver the best performance, where, as shown in Figure \ref{fig:222-4}(b) the optimal integral and derivative orders $\delta=0.8$ and $\mu=0.45$ are found within this zone.
\begin{figure}
\centering
\includegraphics[width=4.5 in]{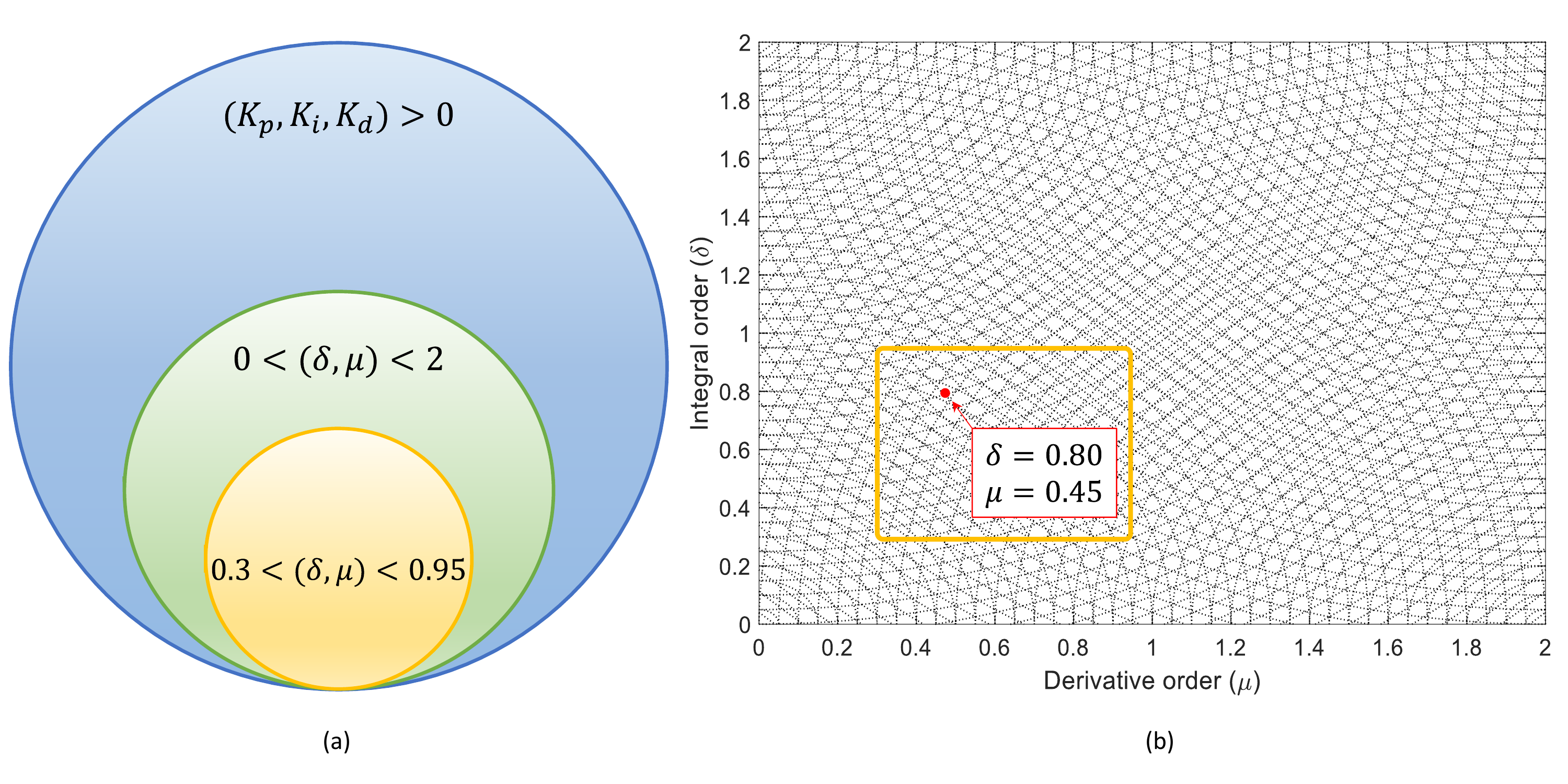}
\caption{The stability regions for $K_p$, $K_i$, $K_d$, $\delta$, and $\mu$.}
\label{fig:222-4}
\end{figure}

The wind profile covering a wind speed range of 5-20 \si{m/s} along with the occurrence time intervals of each fault scenario is depicted in Figure \ref{fig:7-4}. It consists of slow wind variations $\upsilon _{m} \left(t\right)$, stochastic wind behavior $\upsilon _{s} \left(t\right)$, the wind shear effects $\upsilon _{ws} \left(t\right)$, and the tower shadow effects $\upsilon _{ts} \left(t\right)$ \cite{dolan2006simulation} expressed as follows:
\begin{equation}
\upsilon _{w} \left(t\right)=\upsilon _{m} \left(t\right)+\upsilon _{s} \left(t\right)+\upsilon _{ws} \left(t\right)+\upsilon _{ts} \left(t\right).
\label{GrindEQ__4-27_}
\end{equation}
\begin{figure}
\centering
\includegraphics[width=4 in]{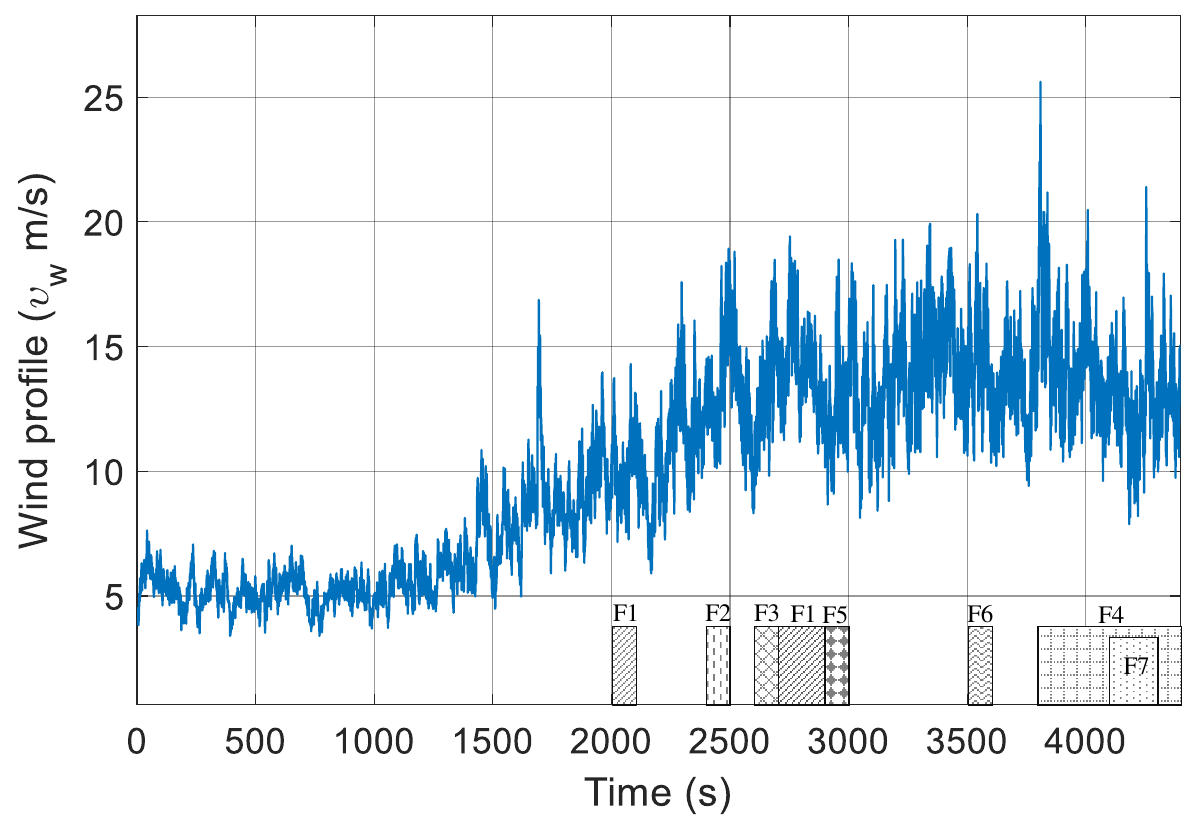}
\caption{The wind speed profile and the occurrence time intervals of each fault scenarios.}
\label{fig:7-4}
\end{figure}
According to Table \ref{tab:4-1}, four different sensor fault scenarios (F1-F4) with low levels of severity occur between the time intervals of 2000-4400 seconds. Also, two actuator faults (F5, F6) and one system fault (F7) with medium and high levels of severity occur between the time intervals of 2900-3600 seconds and 4100-4300 seconds, respectively. Simulations \textcolor{black}{are performed} using MATLAB/Simulink environment for the WT model presented in Section \ref{sec:4.2}.

\begin{rem}
\label{rem:4-4}
When actuator faults (hydraulic pressure drop or increment of air content in the oil) occur, the pitch angle changes accordingly, which degrades the reference tracking of the generator. This tracking ability degradation can lead to large fluctuations in the generator speed. It is worth mentioning that these faults occur relatively slow; thus, the pitch angle can track the reference. However, they need to be compensated in order to prevent the deterioration of tracking performance.
\end{rem}

Figure \ref{fig:8-4} depicts the power generated by the WT under the control of all three controllers under consideration. It can be seen that the profiles cover the full range of operation, demonstrating the suitability of this profile for comparison under the various fault scenarios. It is noteworthy to mention that, due to the stochastic wind behavior and its deviations, whenever the wind speed decreases largely in region III, the generated power also decreases largely. However, faults also impose the effects based on their severity and result in the decrement of generated power. According to Figure \ref{fig:8-4}, it can be seen that all three control approaches can compensate the effects of the fault, delivering different levels of performance. \textcolor{black}{The target of the power generation is to maintain the maximum power at all times (\textit{i.e.} $4.8E+6$[\si{W}]), especially when various faults happen. Here, the conventional PI controller's parameters are chosen as in \cite{odgaard2013fault}. From Figure \ref{fig:8-4}, it is evident that the effects of faults (F1-F3) with low severity are satisfactorily accommodated using all three controllers. However, when the highly severe fault (F5) happens due to hydraulic pressure drop, the conventional PI and optimal FOPID (OFOPID) controllers demonstrated several drops in the power generation and hence, could not deliver a satisfactorily fault accommodation performance as the proposed EM-OFOPID did. In the event of the air content increment in the oil with a medium level of severity (F6), the same performance is achieved with the conventional PI; however, in this case, the optimal FOPID has performed as well as the proposed EM-OFOPID, outperforming the conventional PI. It can also be seen that while the low severity fault (F4) is occurring in the time interval of 3805-4400 seconds, at some points, the conventional PI fails to accommodate the effects of the fault and lower power prodiction performance is demonstrated.} Besides, a significant decrease in the wind speed happens \textcolor{black}{within the interval} of 4180-4260 seconds, which associates with the previous fault (F4) and another fault (F7) \textcolor{black}{during} 4100-4300 seconds. Accordingly, the generated power is decreased; however, from the zoomed-in inset, it is evident that the proposed EM-OFOPID demonstrates superior performance compared to other controllers in terms of fault accommodation and power generation. \textcolor{black}{Figure \ref{fig:777-4} compares the performance of PI and EM-OFOPID controllers in terms of power generation during the time interval of fault F7 occurrence in the system. The fault corresponds to the slow friction changes in the drivetrain with different levels of severity with a 5\%, 10\%, 50\%, and 100\% increase in the coefficient during the time interval of 4100-4300 seconds, to demonstrate an insight of friction change in practice.} Accordingly, it can be observed that increasing the severity of F7 degrades the performance of PI, while EM-OFOPID can effectively compensate for the fault and demonstrate superior performance. In addition, the comparative $\Vert{P_g-P_{g,opt}}\Vert_2$ in all controllers validate the superior performance of EM-OFOPID with 1.4591e+08 over OFOPID with 1.4923e+08 and PI with 1.6814e+08. It is readily observed that taking advantage of more design parameters and DWPFA to tune them; the FOPID schemes increase the power generation compared to the conventional PI scheme. The results also reveal that utilizing a memory of pitch angles with the optimal FOPID effectively enhances the control performance, while the superiority of DWPFA-optimized EM-FOPID is apparent in comparison with the DWPFA-optimized conventional FOPID. To further testify the performance of DWPFA \textcolor{black}{compared to other EAs}, the comparative $\Vert{P_g-P_{g,opt}}\Vert_2$ of EM-OFOPID and OFOPID controllers tuned by optimization algorithms is demonstrated and sorted by the performance order \say{Rank} in Table \ref{tab:4-33}. The results manifest the superior performance of DWPFA over other algorithms, obtaining the best rank.

\begin{table}[t!]
\caption{The comparative $\Vert{P_g-P_{g,opt}}\Vert_2$ of EM-OFOPID and OFOPID controllers tuned by optimization algorithms.}
\centering
\label{tab:4-33}
\scalebox{0.8}{
\begin{tabular}{l c c c l c c}
\hline\hline \\[-3mm]
\multicolumn{3}{c}{OFOPID} &  & \multicolumn{3}{c}{EM-OFOPID} \\ \hline
\multicolumn{1}{c}{Algorithm} & \multicolumn{1}{c}{$\Vert{P_g-P_{g,opt}}\Vert_2$} & \multicolumn{1}{c}{Rank} & \multicolumn{1}{c}{ } & \multicolumn{1}{c}{Algorithm} & \multicolumn{1}{c}{$\Vert{P_g-P_{g,opt}}\Vert_2$} & \multicolumn{1}{c}{Rank}  \\[1ex] \hline
DWPFA & 1.4923e+08 & 1 &  & DWPFA & 1.4591e+08 & 1 \\
FPSOMA & 1.4940e+08 & 2 &  & FOFA & 1.4610e+08 & 2 \\
FOFA & 1.4961e+08 & 3 &  & FPSOMA & 1.4667e+08 & 3 \\
CCGBFO & 1.4996e+08 & 4 &  & CCGBFO & 1.4704e+08 & 4 \\
EGWO & 1.5058e+08 & 5 &  & EGWO & 1.4735e+08 & 5 \\
FA & 1.5094e+08 & 6 &  & FA & 1.4758e+08 & 6 \\
GWO & 1.5143e+08 & 7 &  & BFO & 1.4777e+08 & 7 \\
BFO & 1.5196e+08 & 8 &  & GWO & 1.4792e+08 & 8 \\
PSO & 1.5224e+08 & 9 &  & PSO & 1.4811e+08 & 9 \\
[1ex]
\hline\hline
\end{tabular}}
\end{table}

\begin{figure}[t!]
\centering
\includegraphics[width=4.5 in]{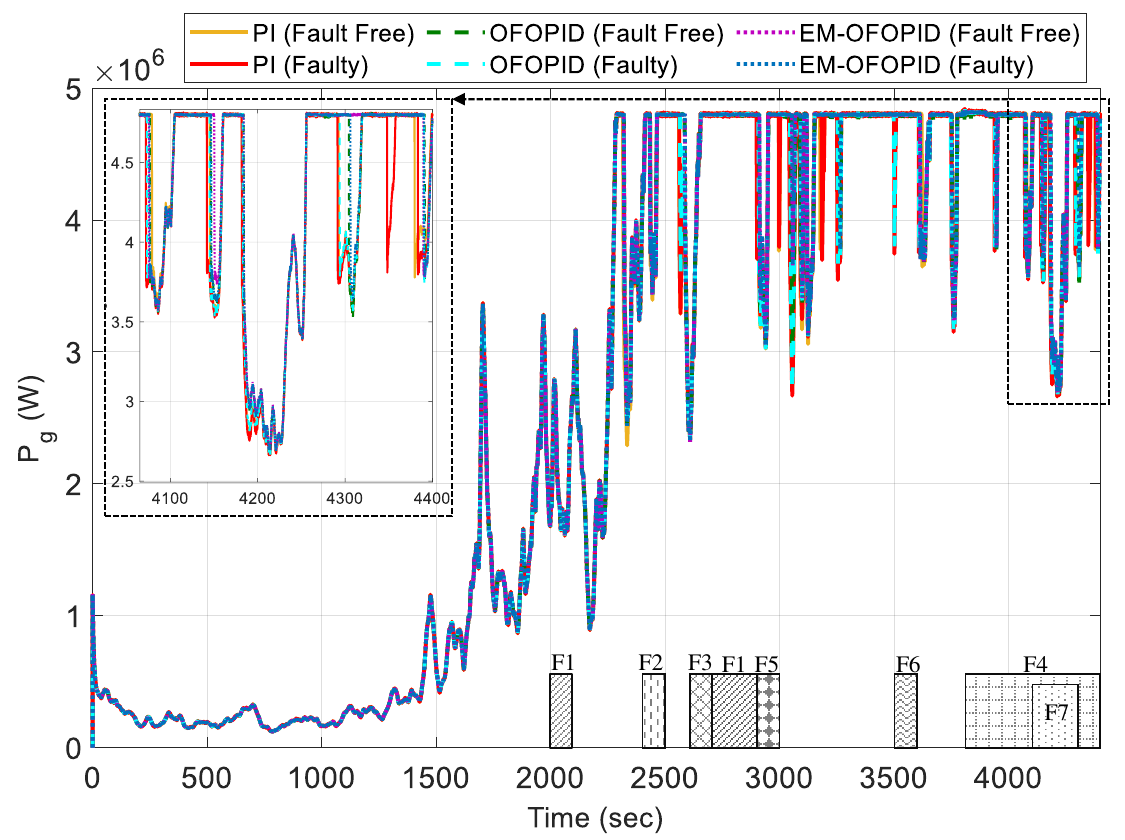}
\caption[Generator power under fault-free and faulty conditions; a comparison between PI, OFOPID, and EM-OFOPID methods.]{Generator power under fault-free and faulty conditions; a comparison between PI, OFOPID, and EM-OFOPID methods. The insets exhibit the dashed-line-highlighted regions.}
\label{fig:8-4}
\end{figure}
\begin{figure}[t!]
\centering
\includegraphics[width=4.5 in]{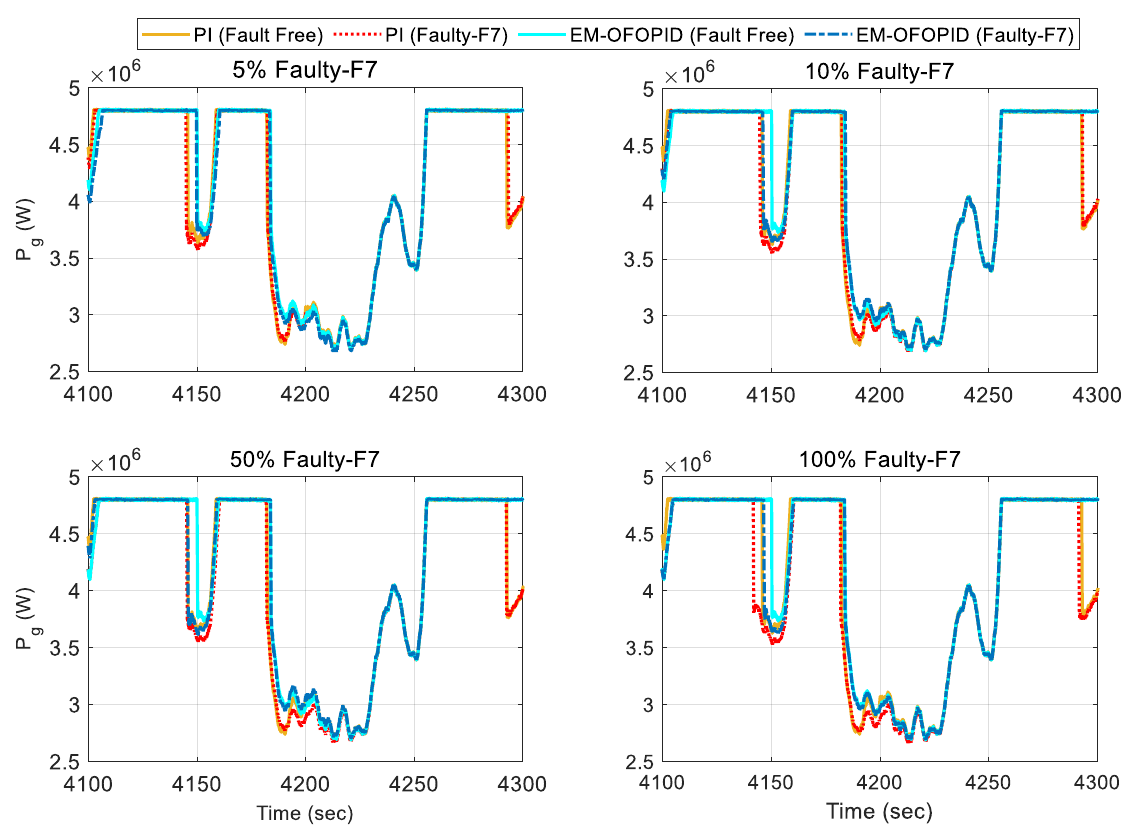}
\caption{Generator power under occurrence of faults F7 during the time interval of 4100-4300 s; a comparison between PI and EM-OFOPID methods.}
\label{fig:777-4}
\end{figure}
Figures \ref{fig:9-4} and \ref{fig:10-4} show the rotor and generator speed, respectively. As it is observed, despite the occurrence of different faults, the conventional PI and optimal FOPID approaches deliver sort of acceptable performance. In this regard, as the zoomed-in insets show in Figures. \ref{fig:9-4} and \ref{fig:10-4}, compared to the fault-free case, the performance of conventional PI degrades as the faults occur. The optimal FOPID also shows similar behavior; however, its performance is significantly better than the conventional PI. Considering Figures. \ref{fig:9-4} and \ref{fig:10-4} it is obvious that although each fault imposes its effects based on its severity level, the EM-OFOPID evidently demonstrates improved performance, that is to say, the proposed fault-tolerant EM-OFOPID controller with extended memory of pitch angles can work well even at the situation of simultaneous sensor, actuator, and system faults. \textcolor{black}{Figure \ref{fig:11-4} shows the scatter plot of generated power, where it is observed that, in comparison to other methods, the proposed control scheme delivers less fluctuations in the generated power and tends to be more consistent in the power generated at a given wind speed.} Figures \ref{fig:12-4}, \ref{fig:13-4}, and \ref{fig:14-4} respectively depict the measured pitch angle of blade 1 from sensor 1, blade 2 from sensor 2, and blade 3 from sensor 2 in the presence of different fault scenarios. As it is observed, the operational constraints on pitch angle ($-3^{\circ} \le \beta \le 90^{\circ} $) are respected.
\begin{figure}[t!]
\centering
\includegraphics[width=4.5 in]{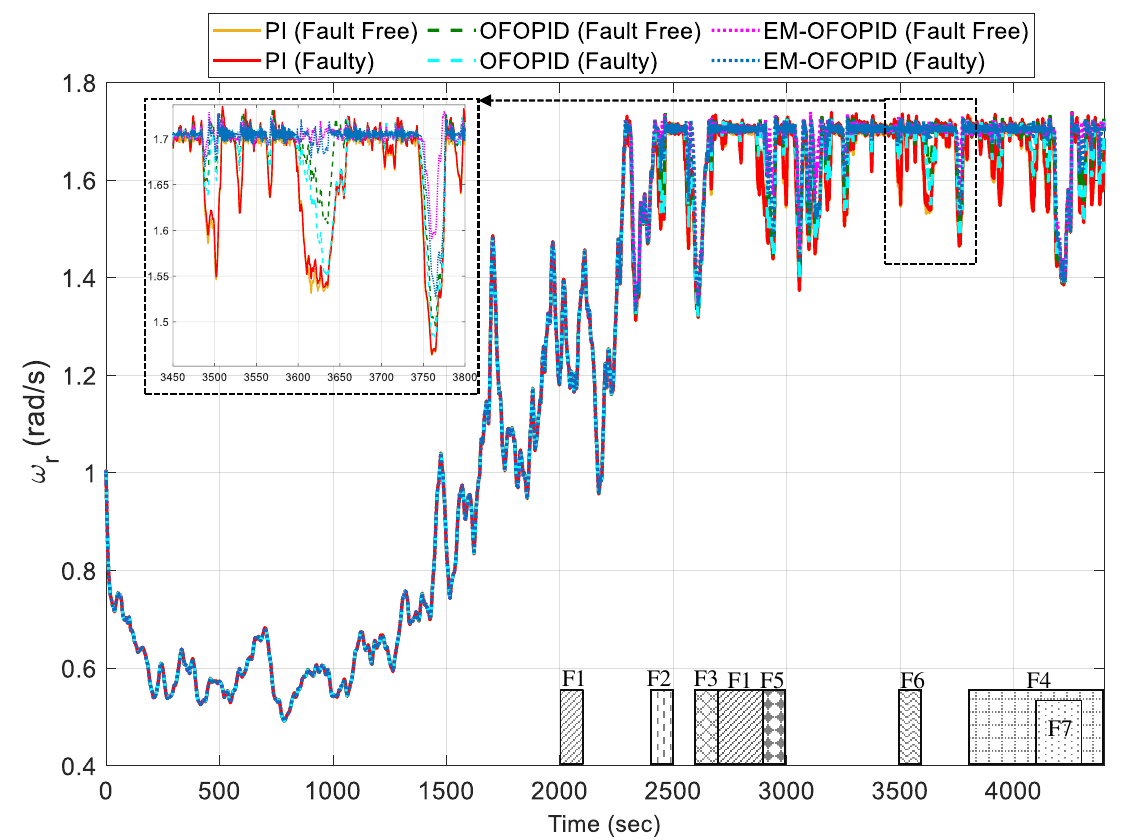}
\caption[Rotor speed under fault-free and faulty conditions; a comparison between PI, OFOPID, and EM-OFOPID methods.]{Rotor speed under fault-free and faulty conditions; a comparison between PI, OFOPID, and EM-OFOPID methods. The insets exhibit the dashed-line-highlighted regions.}
\label{fig:9-4}
\end{figure}
\begin{figure}[t!]
\centering
\includegraphics[width=4.5 in]{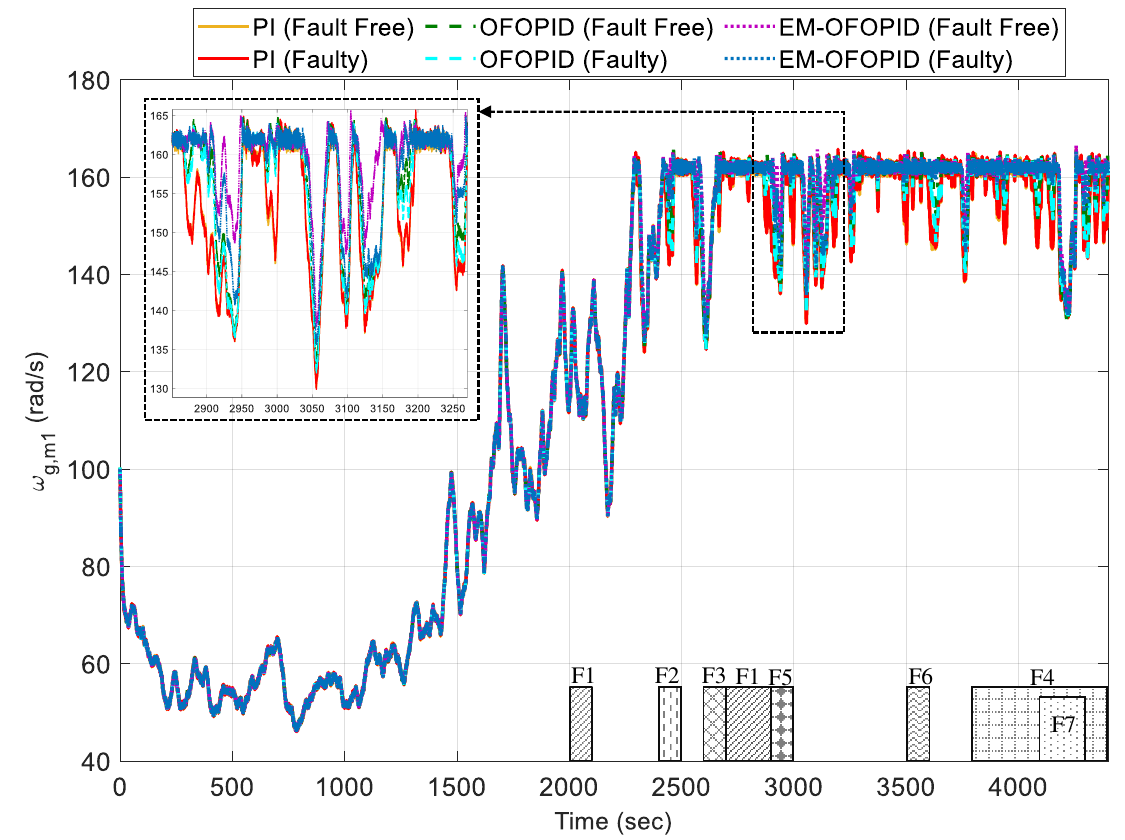}
\caption[Generator speed under fault-free and faulty conditions; a comparison between PI, OFOPID, and EM-OFOPID methods.]{Generator speed under fault-free and faulty conditions; a comparison between PI, OFOPID, and EM-OFOPID methods. The insets exhibit the dashed-line-highlighted regions.}
\label{fig:10-4}
\end{figure}
\begin{figure}[t!]
\centering
\includegraphics[width=5.5 in]{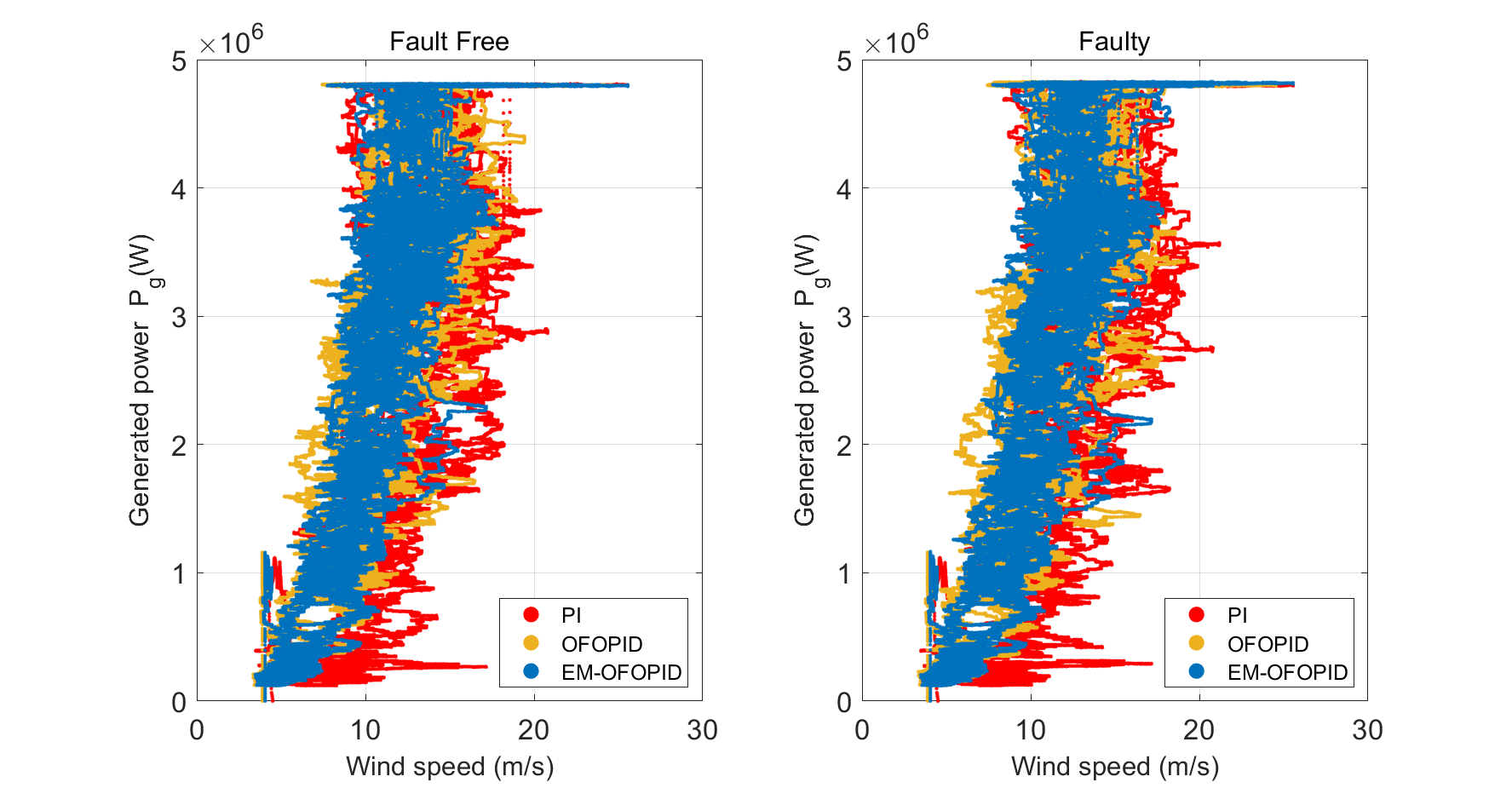}
\caption{Scatter plot of generated power; a comparison between PI, OFOPID, and EM-OFOPID methods.}
\label{fig:11-4}
\end{figure}
\begin{figure}[t!]
\centering
\includegraphics[width=4.5 in]{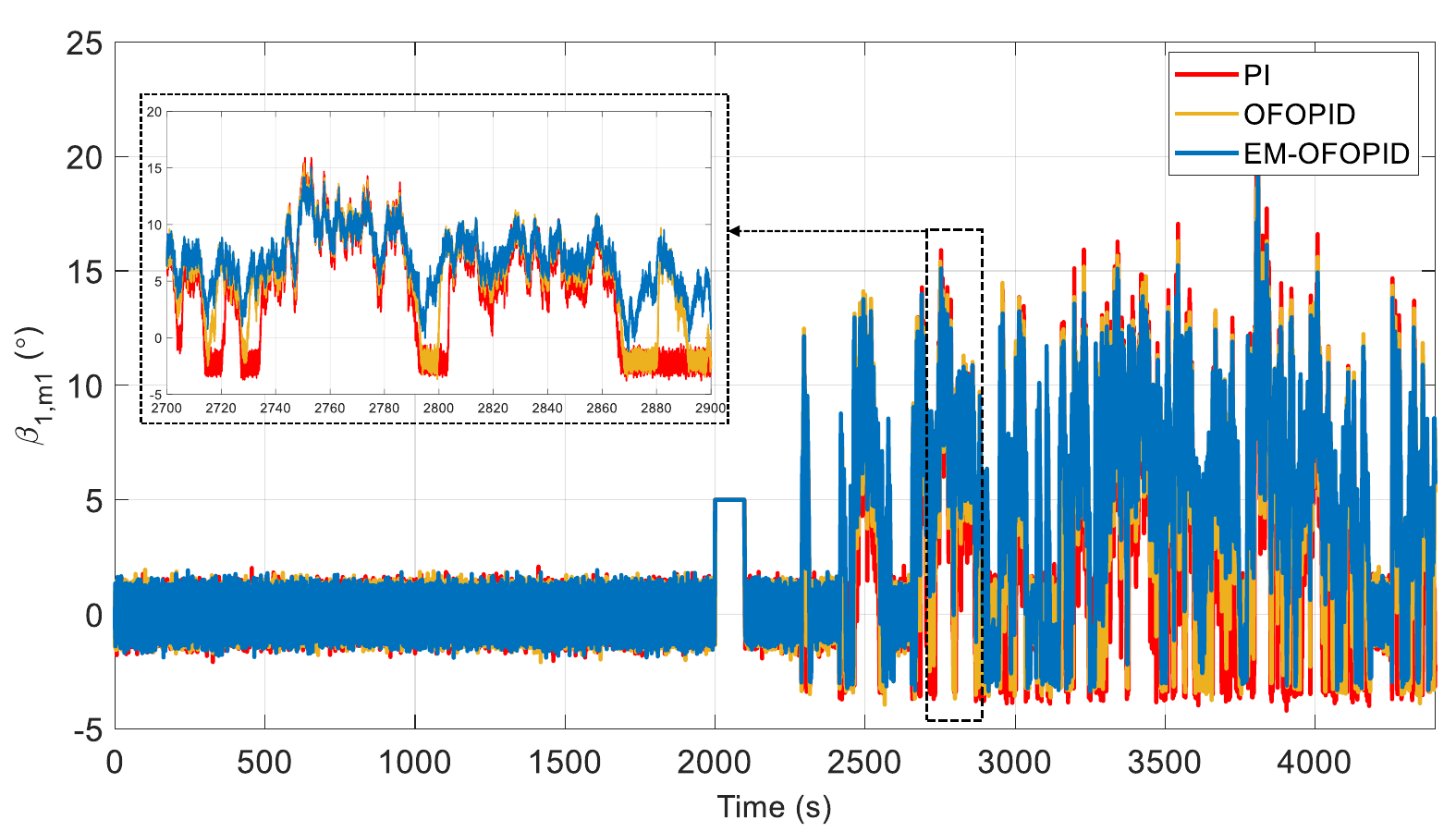}
\caption[Pitch angle of blade 1 under faulty condition; a comparison between PI, OFOPID, and EM-OFOPID methods.]{Pitch angle of blade 1 under faulty condition; a comparison between PI, OFOPID, and EM-OFOPID methods. The insets exhibit the dashed-line-highlighted regions.}
\label{fig:12-4}
\end{figure}
\begin{figure}[t!]
\centering
\includegraphics[width=4.5 in]{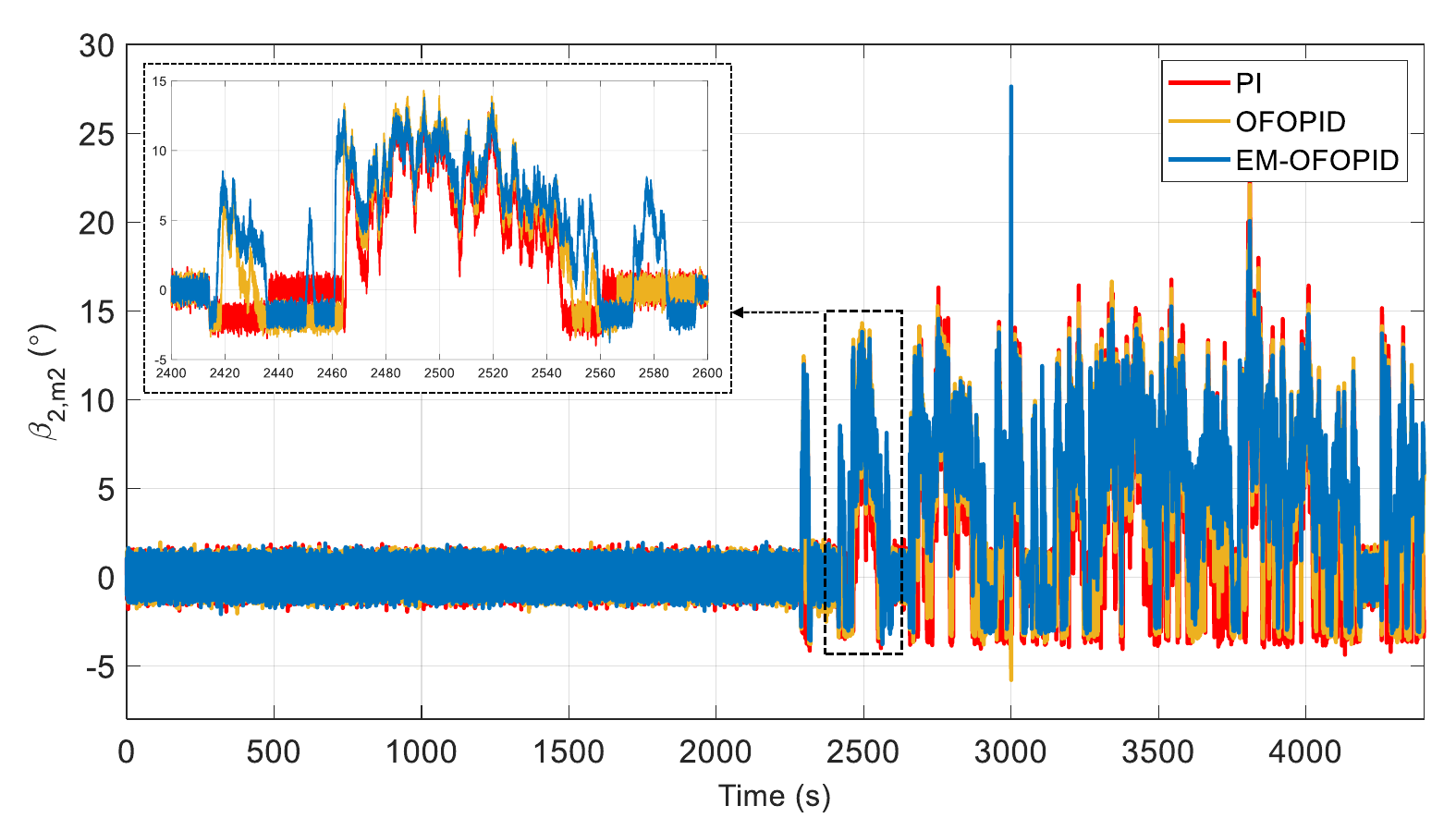}
\caption[Pitch angle of blade 2 under faulty condition; a comparison between PI, OFOPID, and EM-OFOPID methods.]{Pitch angle of blade 2 under faulty condition; a comparison between PI, OFOPID, and EM-OFOPID methods. The insets exhibit the dashed-line-highlighted regions.}
\label{fig:13-4}
\end{figure}
\begin{figure}[t!]
\centering
\includegraphics[width=4.5 in]{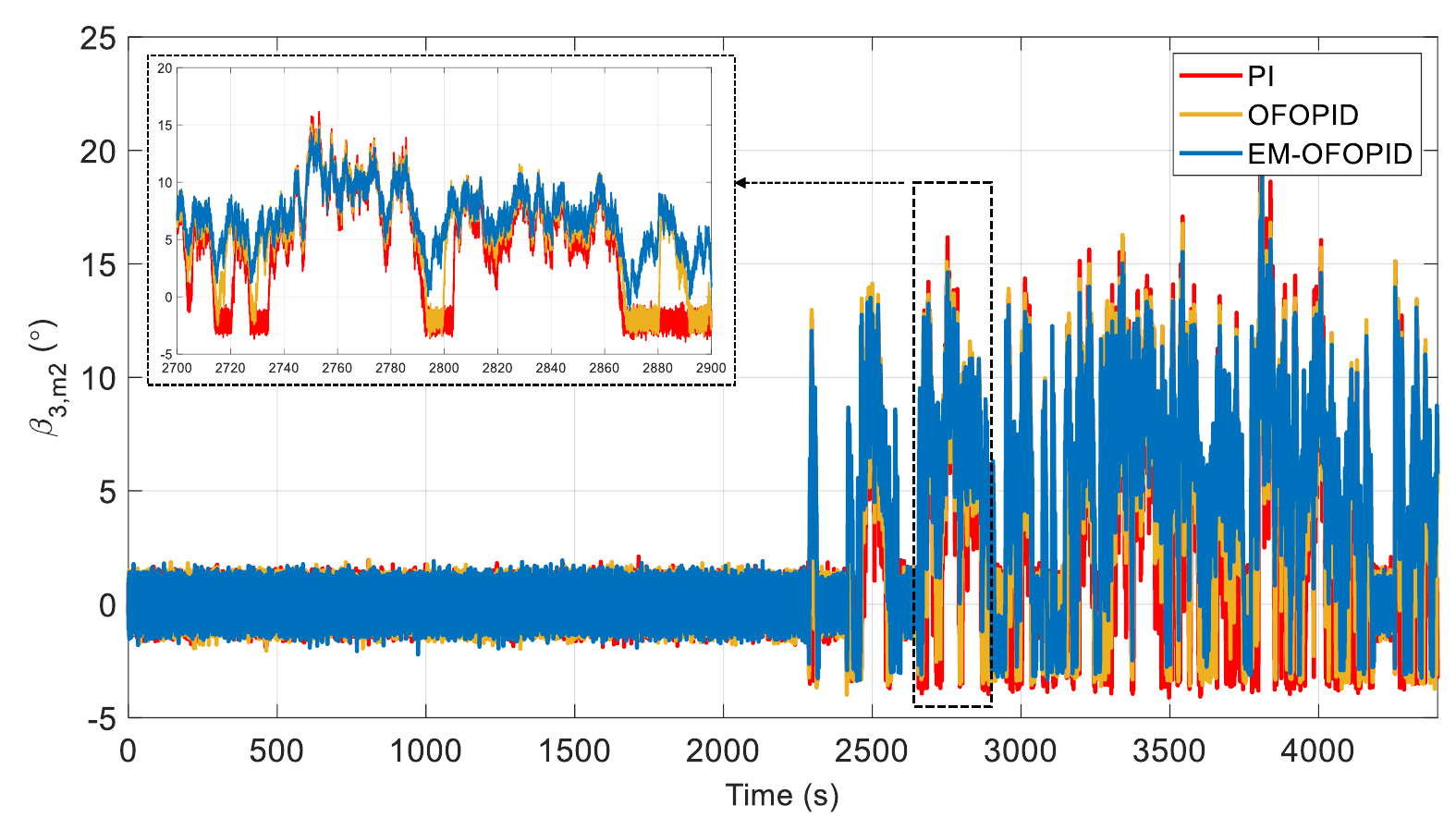}
\caption[Pitch angle of blade 3 under faulty condition; a comparison between PI, OFOPID, and EM-OFOPID methods.]{Pitch angle of blade 3 under faulty condition; a comparison between PI, OFOPID, and EM-OFOPID methods. The insets exhibit the dashed-line-highlighted regions.}
\label{fig:14-4}
\end{figure}

To sum up, according to the simulation results illustrated, the investigated conventional PI, DWPFA optimized FOPID, and DWPFA optimized EM-FOPID controllers, \textcolor{black}{efficiently} tolerate the effects of sensor, actuator, and system faults. However, as investigated, the EM-OFOPID demonstrated the best performance in mitigating the effects of fault scenarios and improving the power generation of the WT.

\begin{rem}
\label{rem:4-55}
Although non-PID approaches can often produce superior behaviour, there is a strong industrial preference for PID controllers due to some main reasons such as: a) simplicity of design and implementation, so PID controller do not require overly complex mathematical models for the design process, b) they can be re-tuned in the field if necessary, by operators who can make small changes to improve performance without having to go back to re-do complicated analysis, and c) considering the industry's current infrastructures and the hardship and costly efforts that have to be done to install the required hardware for other complex methods, the benefits of the proposed method in terms of no requirements for costly or complex hardware installments will be salient.
\end{rem}

\section{Conclusions}
\label{sec:4.6}
This chapter proposed a new fault-tolerant pitch control scheme to adjust the pitch angle of WT blades subjected to sensor, actuator, and system faults. The proposed FTPC scheme comprises a fractional-calculus-based extended memory of pitch angles augmented with FOPID controller to maintain improvement in power generation performance of WT. Furthermore, a novel dynamic weighted parallel firefly algorithm (DWPFA) has been proposed, and its performance was evaluated through well-defined CEC2017 benchmark functions in comparison with other EAs. Non-parametric Friedman and Friedman Aligned statistical tests were utilized to analyze the quality of solutions. Comparative simulation results revealed the superiority of DWPFA over other EAs. The performance of the proposed fault-tolerant EM-FOPID has been investigated in comparison to conventional PI and optimal FOPID approaches, on a 4.8-MW WT model in region III, where the controller parameters were tuned using DWPFA. \textcolor{black}{Simulation results demonstrated the efficaciousness of the proposed FTPC strategy} under fault-free and faulty conditions. \textcolor{black}{Accordingly}, the proposed DWPFA optimized EM-FOPID not only demonstrated the best performance in mitigating the effects of fault scenarios, but also improved the power generation of the WT.

A limitation for realization of the proposed EM-FOPID control scheme for the blade pitch control system is a slight increase in the computational complexity due to the memory requirements based on the fractional-order operators and the higher number of parameters that must be tuned, compared to the conventional controllers. Since approximations must be considered to implement such controllers, fractional-order operators' implementations are relatively complex and costly compared to their integer-order counterparts. However, for the specific controller presented here, it only takes 1.8\% and 6.4\% more time as compared to that of FOPID and PI controllers, respectively, to perform the blade pitch control. Besides, although the conventional PI and the optimal FOPID approaches provide lower computational complexities with respect to the proposed EM-FOPID, prioritizing more power generation and better fault-tolerant performances will make the proposed method a more preferred candidate providing a viable solution that can be implemented with ease without the needs of costly or complex hardware installments.



\chapter{Fault-Tolerant Maximum Power Extraction from Wind Turbines} 

\label{Chapter5} 

\section{Introduction}
\label{sec:5.1}
This chapter presents a nonlinear control approach to maximize power extraction of WECSs operating below their rated wind speeds. Due to nonlinearities associated with the dynamics of WECSs, the stochastic nature of wind, and the inevitable presence of faults in practice, developing reliable fault-tolerant control strategies to guarantee maximum power production of WECSs has always been considered important. A fault-tolerant fractional-order nonsingular terminal sliding mode control (FNTSMC) strategy to maximize the captured power of WTs subjected to actuator faults is developed. A nonsingular terminal sliding surface is proposed to ensure fast finite-time convergence, whereas the incorporation of fractional calculus in the controller enhances the convergence speed of system states and simultaneously suppresses chattering, resulting in extracted power maximization by precisely tracking the optimum rotor speed. Closed-loop stability is analysed and validated through the Lyapunov stability criterion.

Various linear and nonlinear control strategies have tackled this tracking problem to achieve the maximum power extraction objective such as adaptive neural network-based control \cite{ganjefar2014improving}, backstepping-based cascade control \cite{wang2018non},  optimal control \cite{mokhtari2018high}, optimal nonlinear model predictive control  \cite{song2021maximum}, and neuro-adaptive sliding mode control \cite{haq2020maximum}. Recently, authors in \cite{abolvafaei2021two} developed two fractional-order fast terminal sliding mode controllers to reduce the mechanical stress on the drivetrain and maximize the captured power of variable-speed WTs. A fractional-calculus-based model of the WT was presented, where the proposed controllers have successfully performed the maximum power extraction task. However, although the aforementioned studies have successfully dealt with the power maximization problem of WT systems despite the wind speed variations; a critical issue in WT control systems remains neglected; the existence of actuator faults that degrade the overall system's stability and power production performance. Hence, developing a fault-tolerant robust control scheme capable of accommodating the faults' effects can be favourable to ensure the desired power production performance. In this regard, an active FTC scheme for a DFIG-based WT with actuator fault and disturbance was developed in \cite{li2018active}. The control structure comprised a Takagi--Sugeno fuzzy observer to estimate the faults and disturbances and an FTC scheme to reduce their effects. As reported, the developed active FTC scheme has successfully reduced the peak current in the transient process. In another study \cite{habibi2018adaptive}, an adaptive modified PID controller was proposed to maximize the captured power of WT systems. According to the authors, the developed approach demonstrated more acceptable performance compared to other classical methods in dealing with unexpected actuator faults and wind speed fluctuations. However, the controller requires further improvements in order to mitigate the fault effects. Authors in \cite{schulte2015fault} proposed an actuator fault diagnosis and FTC approach by incorporating a Takagi--Sugeno fuzzy system and a sliding mode observer for WTs with a hydrostatic transmission. The simulation results were reported to demonstrate similar performance of the fault-free and faulty cases, revealing the desirable performance of the FTC approach. In another study \cite{azizi2019fault}, an active FTC scheme was developed for rotor speed regulation and maximum power extraction of a WT in the presence of actuator faults and uncertainties. In this regard, the authors designed a full-order compensator for fault and disturbance attenuation and an adaptive output feedback SMC with an integral surface to perform the FTC. The proposed strategy was reported to demonstrate better fault-tolerant capability and more robust behavior with fewer fluctuations and less fatigue on the rotor speed and output power than conventional PID and disturbance accommodation controllers.

Stemming from the desirable merits of SMC approaches such as fast dynamic response, good transient performance, stability, and robustness to matched parameter uncertainties, it has established itself as one of the most effective strategies to deal with different linear and nonlinear control problems \cite{mofid2021adaptive,benbouhenni2021synergetic,riaz2021review,mobayen2021adaptive}. In this regard, due to the highly nonlinear behavior of WECSs, power control and performance enhancement of these systems have been the topic of many SMC-based control strategies during the past decade \cite{eddine2016comprehensive, liu2018dfig, nayeh2020multivariable}. Authors in \cite{eddine2016comprehensive} developed an improved SMC controller with reduced chattering for power maximization of a grid-connected DFIG-based WECS under bounded uncertainties and disturbances. In another study \cite{liu2018dfig}, an exponential reaching law was proposed to reduce the chattering phenomenon and enhance the WT active and reactive power control performance in an SMC controller. In a similar study \cite{nayeh2020multivariable}, an improved SMC was developed to deal with the active and reactive power control problems of a DFIG-based WT subjected to various uncertainties. As reported, comparative investigations of the developed SMC approach and the $H_{\infty }$ robust control method demonstrated superior performance in terms of tracking error, overshoot, and settling time. The conventional SMC is relatively straightforward to design and implement. However, despite its satisfactory performance in practical applications, it has some defects such as the chattering problem, failing to establish a finite-time convergence of the systems states to the equilibrium point, and producing unnecessarily large control signals \cite{kelkoul2020stability}. Accordingly, to overcome these shortcomings and enhance their performance, various modifications have been developed in the literature, such as adaptive SMC \cite{jing2019adaptive}, higher-order SMC \cite{laghrouche2021barrier, mousavi2021robustp, zheng2018integral}, soft computing-based SMC \cite{xu2019event,hashtarkhani2019neural}, and fractional calculus-based SMC \cite{mousavi2021robust, xie2021coupled}. Higher-order SMCs, such as TSMC approaches, have successfully dealt with the finite-time convergence and large control signal problems associated with conventional SMCs \cite{ali2020lcc}. However, regardless of their provided improvements, they still need further chattering mitigation and convergence speed improvements. On the other hand, due to the distinctive memory features of fractional-order derivatives \cite{mousavi2018fractional}, the augmentation of fractional-order calculus with linear and nonlinear controllers has led to enhanced performance in \textcolor{black}{many control applications \cite{fahad2020advanced, mousavi2015memetic}.} In this context, the synthesis of fractional calculus with SMC controllers has shown to be an effective amendment to the controllers' performance by mitigating the chattering phenomenon and delivering faster convergence speed \cite{nicola2021fractional, sami2020super, mousavi2021robust}.

In accordance with the above-discussed literature and considering the desirable performance of SMC approaches in WT control problems, this chapter proposes a fractional-order nonsingular terminal sliding mode controller to maximize the power extraction of WECSs operating in the partial-load region. Its main contributions are as follows:
\begin{itemize}
\item	\textcolor{black}{A design that integrates the fractional calculus into NTSMC to effectively enhance the finite-time convergence speed and simultaneously alleviate the chattering phenomena. Therefore, the optimum rotor speed tracking is achieved with little error, resulting in more power extracted from the wind;}
\item	Validation and performance assessment of the fault-tolerant capability of proposed design using partial loss on the generator torque;
\item	\textcolor{black}{Comparative performance analysis of the developed control strategy with conventional SMC \cite{merida2014analysis} and second-order fast terminal SMC \cite{abolvafaei2019maximum}. Accordingly, taking advantage of the proposed control law, a desirable optimum rotor speed tracking performance with fewer fluctuations and faster transient response is achieved.}
\end{itemize}

The remainder of the chapter are organized as follows. Section \ref{sec:5.2} presents the modelling of the two-mass WT along with the problem statement and fault description. Section \ref{sec:5.3} is dedicated to the proposed controller's design process and presents the stability analysis based on the Lyapunov stability theorem. Section \ref{sec:5.4} presents the comparative simulation results, and finally, some conclusions are given in Section \ref{sec:5.5}.

\section{Problem Formulation}
\label{sec:5.2}
In this section the two-mass WT model under study is first presented. The power capture maximization problem in the partial-load region alongside the considered actuator fault are then introduced.

\subsection{WT Mechanical Model}
\label{sec:5.2.1}
The drivetrain provides the generator's required rotational speed by converting high torque on the low-speed shaft to the low torque on the high-speed shaft to be transferred to the generator unit. The mechanical model of the two-mass WT represented by Figure \ref{fig:1-5} can be described as follows \cite{azizi2019fault},
\begin{subequations}
\label{ENERGIES__4_}
\begin{align}
J_{R}\dot{\omega}_{r}&=T_a-T_{LS}-D_R\omega_r, \\
J_{G}\dot{\omega}_{g}&=T_{HS}-T_{G}-D_G\omega_g, \\
\dot{\theta}_{\Delta}&=\omega_r-\frac{\omega_g}{N_{GB}},
\end{align}
\end{subequations}
where $T_{LS}=k_{ls}\left(\theta_{HS}-\theta_{LS}\right)+D_{LS}\left(\omega_r-\omega_{LS}\right)$ and $T_{HS}= T_{LS}/N_{GB}$ represent the low and high speed shaft torque, respectively. $\theta_{\Delta}=\theta_{HS}-\theta_{LS}$ denotes the torsion angle of drivetrain, and $\omega_{LS}$ is the low shaft speed, and $\theta_{R}$ and $\theta_{LS}$ represent the rotation angle of the rotor and generator shafts, respectively. $D_R$, $D_G$, and $D_{LS}$ are the rotor external damping, the generator external damping, and low-speed shaft damping, respectively. The gearbox ratio is expressed as:
\begin{equation}
\label{ENERGIES__5_}
N_{GB}=\frac{\omega_g}{\omega_{LS}}=\frac{T_{LS}}{T_{HS}}.
\end{equation}

Using \eqref{ENERGIES__4_} and \eqref{ENERGIES__5_}, one can obtain
\begin{equation}
J_{t}\dot{\omega}_{r} =T_{a}-D_{t}\omega_{r}-T_{G},
\label{ENERGIES__6_}
\end{equation}
where $J_t=J_R+N^{2}_{GB}J_G $, and $D_t=D_R+N^{2}_{GB}D_G $ denote the induced total inertia, generator torque on the rotor side, and the induced total external damping on the rotor side, respectively.
\begin{figure}
\centering
\includegraphics[width=3 in]{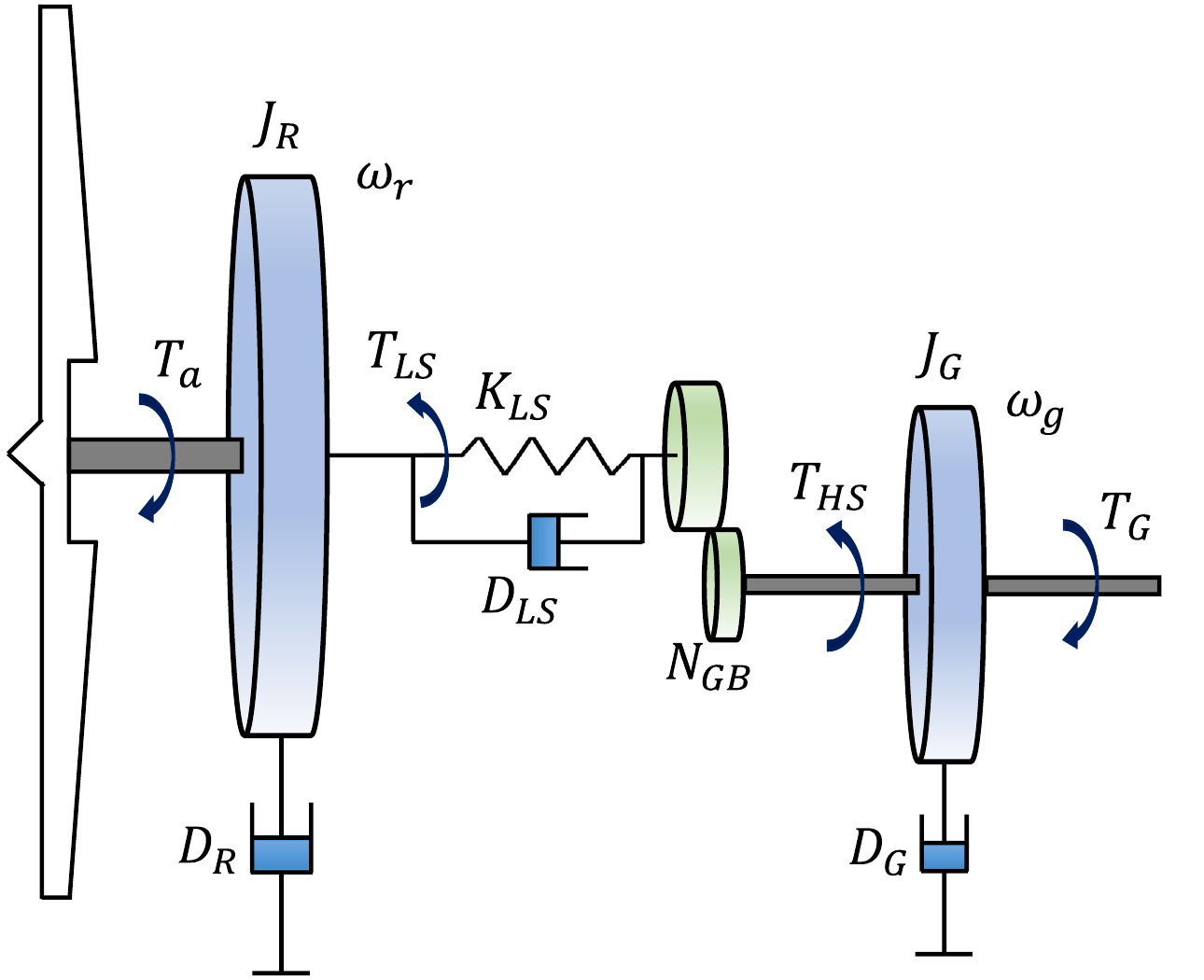}
\caption{\textcolor{black}{Schematic of the two-mass model. (Note: the blades are not included as a separate mass, and are shown as an illustration).}}
\label{fig:1-5}
\end{figure}

The dynamics of the generator are modelled as a first-order system to track the requested generator torque $T_{G,ref}$ as $\dot{T_{G}}=-T_G/\tau_T+T_{e,ref}/\tau_T$ \cite{johnson2006control}, where $\tau_T$ denotes the time constant. Accordingly, the electrical power produced in the generator can be expressed as $P_e=\eta_g\omega_g T_G $.

\subsection{Problem Statement}
\label{sec:5.2.2}
The control objective is to extract the maximum power from the wind energy in the partial-load region. To this end, the power coefficient $C_{P} $ needs to be obtained based on the optimum pitch angle $\beta_{opt}\left(t\right)$ and the optimum tip speed ratio $\lambda_{opt}\left(t\right)$, \textit{i.e.} $C_{P,\max }\triangleq C_{P} \left(\lambda _{opt}\left(t\right),\beta _{opt}\left(t\right) \right)$, where $C_{P,\max } $ denotes the maximum power coefficient. As a common procedure, when the wind speed $\upsilon_{w}\left(t\right)$ exceeds the rated wind speed, $\upsilon _{w,rated}\left(t\right)$, some pitch angle control strategies such as adaptive PI controller \cite{badihi2020fault}, fuzzy-PI controller \cite{badihi2014wind}, and gain-scheduling FOPID \cite{asgharnia2020load} are implemented to adjust $\beta\left(t\right)$ and ensure $\omega_r\left(t\right)$ tracks the rated rotor speed $\omega_{r,rated}\left(t\right)$.

In this work, the focus is on the case that $\upsilon_{w}\left(t\right)$ is lower than $\upsilon _{w,rated}\left(t\right)$ (\textit{i.e.} region II) and $\beta\left(t\right)$ is fixed at $\beta=0$. Accordingly, rewriting the aerodynamic power \eqref{Equation_1_} we have
\begin{equation}
P_{a}^*\left(t\right) =\frac{1}{2} \rho \pi R^{2} \upsilon_{w}^{3}\left(t\right) C_{P,opt} \Big(\lambda_{opt}\left(t\right) ,\beta\left(t\right) \Big),
\label{ENERGIES__7_}
\end{equation}
where $C_{P,\max}\triangleq C_{P} \left(\lambda_{opt}\left(t\right),0\right)$. Hence, the control objective is to define a control law that maximizes the power extraction by maintaining the maximum rotor efficiency during operation, by adjusting the rotor speed $\omega_r\left(t\right)$ to follow the optimum rotor speed $\omega_{r,opt}\left(t\right)$ and as a result, ensure $\lambda\left(t\right)=\lambda_{opt}\left(t\right)$ for $t\ge 0$. Accordingly, the reference rotor speed can be derived as follows,
\begin{equation}
\omega _{r,opt}\left(t\right) ={\lambda_{opt}\left(t\right) \upsilon_{w}\left(t\right) }/{R}.
\label{ENERGIES__8_}
\end{equation}

Therefore, the maximum power extraction is achieved when $e\left(t\right)=\omega_{r,opt}\left(t\right)-\omega_{r}\left(t\right)$ converges to zero as $t$ goes to infinity.

\subsection{Actuator Faults}
\label{sec:5.2.3}
Faults in a WT system can be classified into two main categories in terms of severity. The first category consists of highly extreme faults such as actuator/pitch breakdown, which requires immediate shutdown or grid disconnections to ensure the system's safety. In contrast, the second category comprises the non-extreme faults such as sensor or actuator degradation, where fault-tolerant strategies are usually adopted to preserve the system's operation with minimum performance degradation. Since this study focuses on the power maximization and tracking control at below-rated wind speeds (\textit{i.e.} $ \upsilon_{w}\left(t\right)< \upsilon_{w,rated}\left(t\right)$), the non-extreme actuator fault scenario will be considered. To this end, the actual control input $u_f\left(t\right)$ and the designed control input $u\left(t\right)$ are expressed as follows \cite{li2017adaptive}:
\begin{equation}
u_{f}\left(t\right)=\zeta_f\left(t\right) u\left(t\right),
\label{ENERGIES__88_}
\end{equation}
where the actuator efficiency factor (or the \textit{health indicator} \cite{li2017adaptive}) $\zeta_f\left(t\right)$ is a time-varying scalar function within interval $\zeta_f \in (0,1]$, where \say{0} and \say{1} correspond to total power loss and healthy actuation, respectively.

The actuator fault considered in this study is a partial loss of the generator output torque. The WT encounters an actuator failure and loses partial actuation power after $t=600 s$. The actuator efficiency coefficient is chosen as follows and depicted in Figure \ref{fig:2-5}.

\begin{equation}
\zeta_f\left(t\right) =
\begin{cases}
1 & \text{,  if $t<600$}\\
\zeta_f\left(t-1\right)*0.995 & \text{,  if $600\leq t<670$}\\
\frac{4}{5}+\frac{1}{5}\exp\left(\frac{-\left(t-670\right)}{20}\right)-\frac{1}{20}\mathsf{chirp}\left(\frac{\pi\left(t-670\right)}{1000}\right) & \text{,  if $t\geq670$}
\end{cases}
\label{ENERGIES__9_}
\end{equation}

\textcolor{black}{The actuator failures' complex fluctuation characteristics that are found in practice are simulated using the \say{$\mathsf {chirp}$} function. The \say{$\mathsf {chirp}$} sweep signal is used to verify the controller's performance in the case of faults with varying frequencies as occurs in real engineering applications.}

\begin{figure}
\centering
\includegraphics[width=4.5 in]{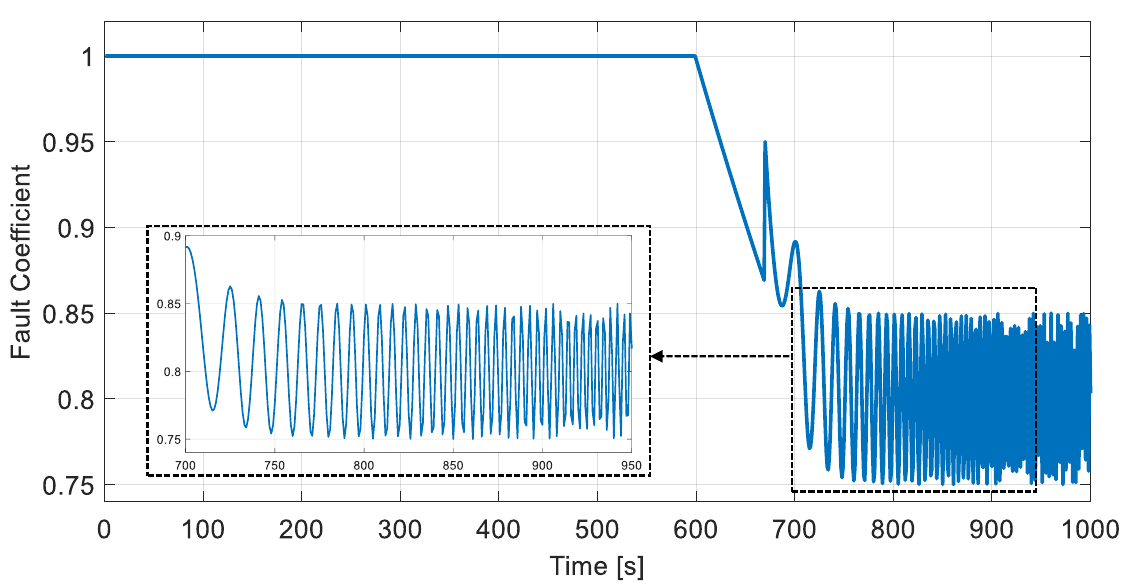}
\caption[Illustration of the actuator health indicator.]{Illustration of the actuator health indicator. $\zeta_f\left(t\right)=1$ denotes a healthy actuation and $\zeta_f\left(t\right)<1$ represents a partial loss of actuation power.}
\label{fig:2-5}
\end{figure}

\section{Proposed FNTSMC Controller Design}
\label{sec:5.3}

In this section, the maximum power extraction of WT is accomplished by rotor speed regulation so that the rotor speed tracking error $e=\omega_{r,opt}-\omega_{r}$ is minimized. For simplification of writings, the $t$ index is neglected in the equations. A fractional-order nonsingular terminal sliding (FNTS) surface is proposed as follows,
\begin{equation}
S=\phi e+{\rm {\mathfrak D}}^{\gamma-1 } e+\frac{1}{\alpha_1 }e^{\frac{l}{h}}+\frac{1}{\alpha_2 }\dot{e}^{\frac{p}{q}},
\label{ENERGIES__C1_}
\end{equation}
where $\phi>0$, $\alpha_1>0$, $\alpha_2>0$, $l>0$, $h>0$, $p>0$, $p>0$, are odd positive integers satisfying $1<\frac{p}{q}<2$, $1<\frac{l}{h}<2$, and $\frac{p}{q}<\frac{l}{h}$. ${\rm {\mathfrak D}}^{\gamma } \left(\cdot\right)$ represents the RL fractional derivative of order $1<\gamma <2$.

\begin{rem}
\label{rem:5-2}
When the system state is far from the equilibrium point, the term ${\rm {\mathfrak D}}^{\gamma-1} e$ in \eqref{ENERGIES__C1_} plays the main role by dominating the term $\frac{1}{\alpha_1 }e^{l/h} +\frac{1}{\alpha_2 }\dot{e}^{p/q}$, and guarantees a high convergence rate. Subsequently, as the system state approaches the equilibrium point, this time, the term $\frac{1}{\alpha_1 }e^{l/h} +\frac{1}{\alpha_2 }\dot{e}^{p/q}$ plays the main role and ensures a finite-time convergence.
\end{rem}

Differentiating \eqref{ENERGIES__C1_} with respect to time yields,
\begin{equation}
\label{ENERGIES__C2_}
\dot{S}=\phi \dot{e}+{\rm {\mathfrak D}}^{\gamma} e+\frac{l}{\alpha_1 h}e^{\frac{l}{h}-1}\dot{e}+\frac{p}{\alpha_2 q}\dot{e}^{\frac{p}{q}-1}\ddot{e}.
\end{equation}

Moreover, since $1<\frac{p}{q}<2$, $1<\frac{l}{h}<2$, and $\frac{p}{q}<\frac{l}{h}$, the singularity problem during the convergence of the terminal SMC is avoided. Considering \eqref{ENERGIES__6_}, \eqref{ENERGIES__C2_} can be rewritten as
\begin{equation}
\label{ENERGIES__C3_}
\dot{S}={\rm {\mathfrak D}}^{\gamma}e+\frac{p}{\alpha_2 q}\dot{e}^{\frac{p}{q}-2}\Bigg(\mathfrak{B}\dot{e}^{3-\frac{p}{q}}+\dot{e}\left(\ddot{\omega}_{r,opt}-\frac{\dot{T}_a-D_t\dot{\omega}_r}{J_t}+\frac{N_{GB}}{J_t}\dot{T}_{G}\right)\Bigg),
\end{equation}
where $\mathfrak{B}=\left(\alpha_1 p h\right)^{-1}\left(\alpha_1\alpha_2\phi q h+\alpha_2 q l e^{\frac{l}{h}-1}\right)$.

Setting $\dot{S}=0$, the following control law can be derived,
\begin{equation}
\label{ENERGIES__C4_}
\dot{T}_{G}=\left(\frac{N_{GB}}{J_{t}}\right)^{-1}\left(-\frac{\alpha_2 q}{p}\dot{e}^{1-\frac{p}{q}}{\rm {\mathfrak D}}^{\gamma} e-\mathfrak{B}\dot{e}^{2-\frac{p}{q}}-\ddot{\omega}_{r,opt}+\frac{\dot{T}_a-D_t\dot{\omega}_r}{J_t}\right).
\end{equation}

\begin{rem}
\label{rem:5-3}
It is worth mentioning that taking advantage of the developed control law \eqref{ENERGIES__C4_}, the system state remains on the sliding surface \eqref{ENERGIES__C1_} and satisfies the condition $\dot{s}=0$. Hence, the finite-time convergence of the tracking error to zero is guaranteed. However, in order to force the state toward the sliding surface in finite time and satisfy the sliding condition, a fractional-order switching law is suggested as follows:
\begin{equation}
\label{ENERGIES__C5_}
\dot{T}_{G,sw}=-\left(\frac{N_{GB}}{J_{t}}\right)^{-1}\frac{\alpha_2 q }{p} \dot{e}^{2-\frac{p}{q} } \left({\rm {\mathfrak D}}^{\gamma -1} \mathcal F_{s} \mathsf{sgn}\left(S\right)+\psi S\right),
\end{equation}
where $\psi>0$, ${\rm {\mathfrak D}}^{\gamma } \mathcal F_{s} =\kappa \left|S\right|,\; \kappa >0$, and $\mathcal F_{s} $ is an arbitrary positive auxiliary function.
\end{rem}

Accordingly, combining \eqref{ENERGIES__C4_} and \eqref{ENERGIES__C5_}, the equivalent control law can be expressed as:

\begin{multline}
\label{ENERGIES__C6_}
\dot{T}_{G,eq}=\dot{T}_{G}+\dot{T}_{G,sw}=\left(\frac{N_{GB}}{J_{t}}\right)^{-1}\Bigg(-\frac{\alpha_2 q}{p}\dot{e}^{1-\frac{p}{q}}{\rm {\mathfrak D}}^{\gamma} e-\mathfrak{B}\dot{e}^{2-\frac{p}{q}} -\ddot{\omega}_{r,opt} \\
+\frac{\dot{T}_a-D_t\dot{\omega}_r}{J_t}+\frac{\alpha_2 q }{p} \dot{e}^{2-\frac{p}{q} } \left({\rm {\mathfrak D}}^{\gamma -1} \mathcal F_{s} \mathsf{sgn}\left(S\right)+\psi S\right)\Bigg).
\end{multline}

\textcolor{black}{The block diagram of the proposed control scheme is illustrated in Figure \ref{fig:2526-5}.}
\begin{figure}
\centering
\includegraphics[width=4.5 in]{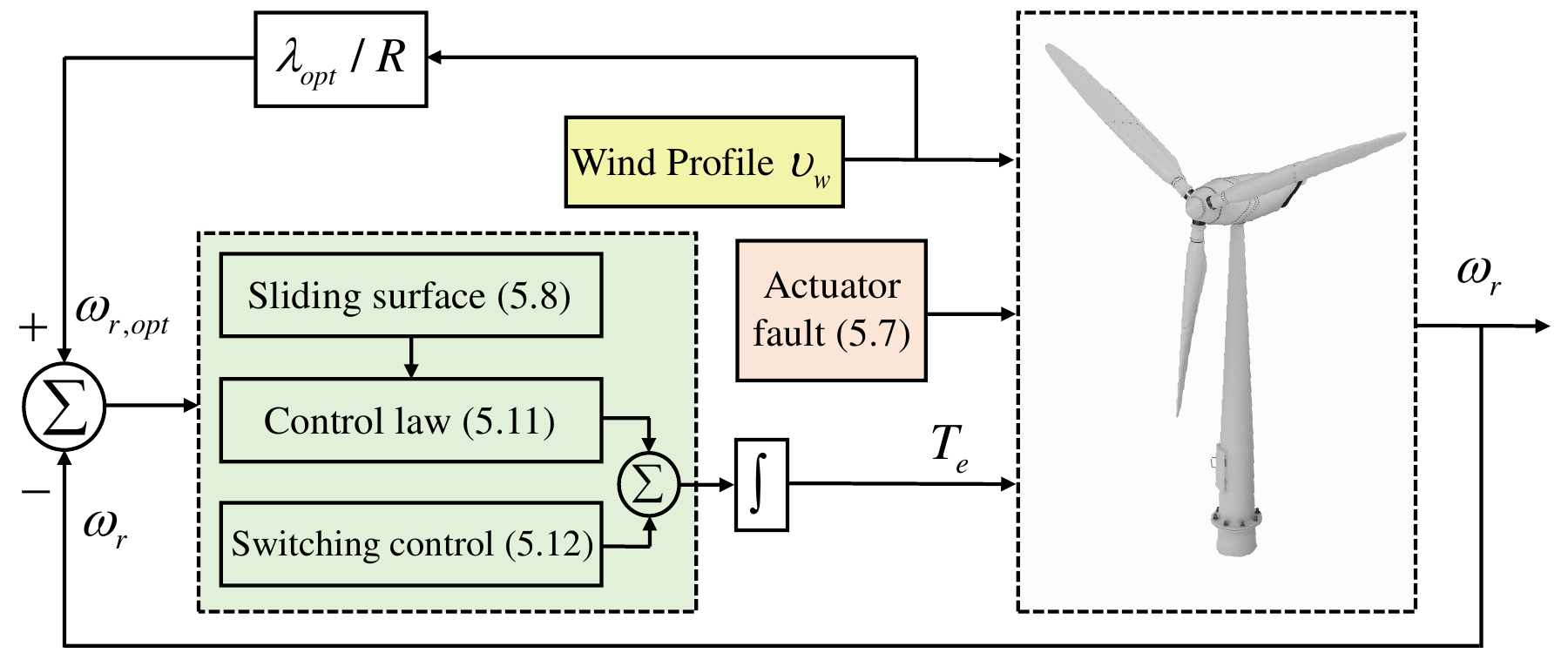}
\caption{\textcolor{black}{Block diagram of the proposed control scheme.}}
\label{fig:2526-5}
\end{figure}
\begin{thm}
By employing the FNTS surface \eqref{ENERGIES__C1_} and the switching control law \eqref{ENERGIES__C5_}, the tracking error dynamics reach the sliding surface in finite time and then converge to zero asymptotically.
\label{thm:1}
\end{thm}

\begin{proof}
Consider the following Lyapunov function candidate:
\begin{equation}
\label{ENERGIES__C7_}
V=\frac{1}{2} S^{2} +\frac{1}{2\kappa } \left({\rm {\mathfrak D}}^{\gamma -1} \mathcal F_{s} \right)^{2}.
\end{equation}

Derivation of $V$ with respect to time yields
\begin{equation}
\label{ENERGIES__C8_}
\dot{V}=S\dot{S}+\frac{1}{\kappa } \left({\rm {\mathfrak D}}^{\gamma -1} \mathcal F_{s} \right){\rm {\mathfrak D}}^{\gamma } \mathcal F_{s} =S\dot{S}+{\rm {\mathfrak D}}^{\gamma -1} \mathcal F_{s} \left|S\right|.
\end{equation}

Taking the equivalent control law in the form of \eqref{ENERGIES__C6_} and substituting \eqref{ENERGIES__C3_} into \eqref{ENERGIES__C8_}, we obtain
\begin{equation}
\label{ENERGIES__C9_}
\dot{V}=S\Bigg[{\rm {\mathfrak D}}^{\gamma}e+\frac{p}{\alpha_2 q}\dot{e}^{\frac{p}{q}-2}\Bigg(\mathfrak{B}\dot{e}^{3-\frac{p}{q}}+\dot{e}\left(\ddot{\omega}_{r,opt} -\frac{\dot{T}_a-D_t\dot{\omega}_r}{J_t}+\frac{N_{GB}}{J_t}\dot{T}_{e,eq}\right)\Bigg)\Bigg]+{\rm {\mathfrak D}}^{\gamma -1} \mathcal F_{s} \left|S\right| \le 0.
\end{equation}

Considering \eqref{ENERGIES__C9_}, it can be verified that
\begin{align}
\label{ENERGIES__C10_}
\dot{V}&=S\left(-{\rm {\mathfrak D}}^{\gamma -1} \mathcal F_{s} \mathsf{sgn}\left(S\right)-\psi S\right)+{\rm {\mathfrak D}}^{\gamma -1} \mathcal F_{s} \left|S\right| \\ \nonumber
 &{\leq {\rm -{\mathfrak D}}^{\gamma -1} \mathcal F_{s} \left|S\right|-\psi S^{2} +{\rm {\mathfrak D}}^{\gamma -1} \mathcal F_{s} \left|S\right|} \\ \nonumber
 &{\leq -\psi S^{2} <0}.
\end{align}

Thus, asymptotical convergence of the system states to the FNTS surface $S\left(t\right)=0$ is achieved according to the Lyapunov stability criterion. From the Lyapunov function {\eqref{ENERGIES__C7_}} we have
\begin{equation}
\label{ENERGIES__C111_}
S^{2} =-2V+\frac{1}{\kappa }\left(\mathfrak D^{\gamma -1} \mathcal F_{s} \right)^{2}.
\end{equation}

Hence, the finite-time convergence can be investigated by rewriting {\eqref{ENERGIES__C10_}} as follows,
\begin{equation}
\label{ENERGIES__C11_}
\dot{V}=\frac{dV}{dt} \le -\psi S^{2} =2\psi V-\frac{\psi }{\kappa } \left(\mathfrak D^{\gamma -1} \mathcal F_{s} \right)^{2}.
\end{equation}

From {\eqref{ENERGIES__C11_}} the following inequality can be obtained,
\begin{equation}
\label{ENERGIES__C12_}
dt\le \frac{-dV}{2\psi V-\frac{\psi }{\kappa } \left(\mathfrak D^{\gamma -1} \mathcal F_{s} \right)^{2} }.
\end{equation}

Let $t_{e,r} $ represents the reaching time at which the regulation error reaches the sliding surface ($e\left(0\right)\ne 0\to e=0$, \textit{i.e.} $V_{(t_{e,r} )}=0$). Then, integrating both sides of {\eqref{ENERGIES__C12_}} yields
\begin{equation}
\label{ENERGIES__C13_}
\int _{0}^{t_{e,r} }dt \le \int _{V_{\left(0\right)}}^{V_{\left(t_{e,r} \right)}}\frac{-dV}{2\psi V-\frac{\psi }{\kappa } \left(\mathfrak D^{\gamma -1} \mathcal F_{s} \right)^{2} } =\left[\frac{-1}{2\psi } \ln \left(2\psi V-\frac{\psi }{\kappa } \left(\mathfrak D^{\gamma -1} \mathcal F_{s} \right)^{2} \right)\right] _{V_{\left(0\right)}}^{V_{\left(t_{e,r} \right)}},
\end{equation}
which further yields
\begin{equation}
\label{ENERGIES__C14_}
t_{e,r} \le \frac{1}{2\psi } \ln \left(\frac{\frac{\psi }{\kappa } \left(\mathfrak D^{\gamma -1} \mathcal F_{s} \right)^{2}-2\psi V_{\left(0\right)} }{\frac{\psi }{\kappa } \left(\mathfrak D^{\gamma -1} \mathcal F_{s} \right)^{2} } \right).
\end{equation}

According to the Lyapunov stability criterion and the provided analysis \eqref{ENERGIES__C111_}-\eqref{ENERGIES__C14_}, the finite-time convergence of the NFTS surface {\eqref{ENERGIES__C1_}} to zero is achieved at $t_{e,r} $. Moreover, $s=0$ results in the finite-time convergence of tracking error to zero.

This completes the Proof.

\end{proof}

\section{Simulation Results}
\label{sec:5.4}
This section studies the power extraction performance of the proposed FNTSMC algorithm. Accordingly, comparative investigations are provided to testify the proposed scheme’s performance in comparison with conventional SMC \cite{merida2014analysis} and second-order fast terminal SMC (SOFTSMC) \cite{abolvafaei2019maximum} approaches. The numerical simulations are carried out on a two-mass WT whose characteristics are given in Table \ref{tab:5-1}. The parameters correspond to the two-bladed variable-speed variable pitch controls advanced research turbine (CART) with hub height of 36 \si{m} \cite{stol2004geometry}. The wind profile consists of 1000 s, within the speed range of 3.4-6.7 \si{m/s} with an average speed of 5.2 \si{m/s} in region II is shown in Figure \ref{fig:3-5}.

\begin{table}[!t]
\centering
\caption{Two-mass wind turbine model parameters.}
\label{tab:5-1}
\scalebox{0.8}{
\begin{tabular}{llllllll}
\hline\hline \\[-3mm]
\textbf{Parameter}	& \textbf{Value}	& \textbf{Unit}  &  &  & \textbf{Parameter}	& \textbf{Value}	& \textbf{Unit} \\
\hline
$R$ & 21.65 & \si{m} &  &  &  $\rho$ & 1.308 & \si{kg/m^3} \\
$J_R$ & 3.25E+05 & \si{kg.m^2} &  &  &  $J_G$ & 34.4 & \si{kg.m^2} \\
$D_R$ & 27.36 & \si{N.m.s/rad} &  &  &  $D_G$ & 0.2 & \si{N.m.s/rad} \\
$D_{LS}$ & 2.691E+05 & \si{N.m.s/rad} &  &  &  $k_{ls}$ & 9500 & \si{N.m/rad} \\
$N_{GB}$ & 43.165 & \si{-} &  &  &  $P_{e,nom}$ & 600E+03 & \si{W} \\
\hline\hline
\end{tabular}}
\end{table}

\begin{figure}
\centering
\includegraphics[width=4.5 in]{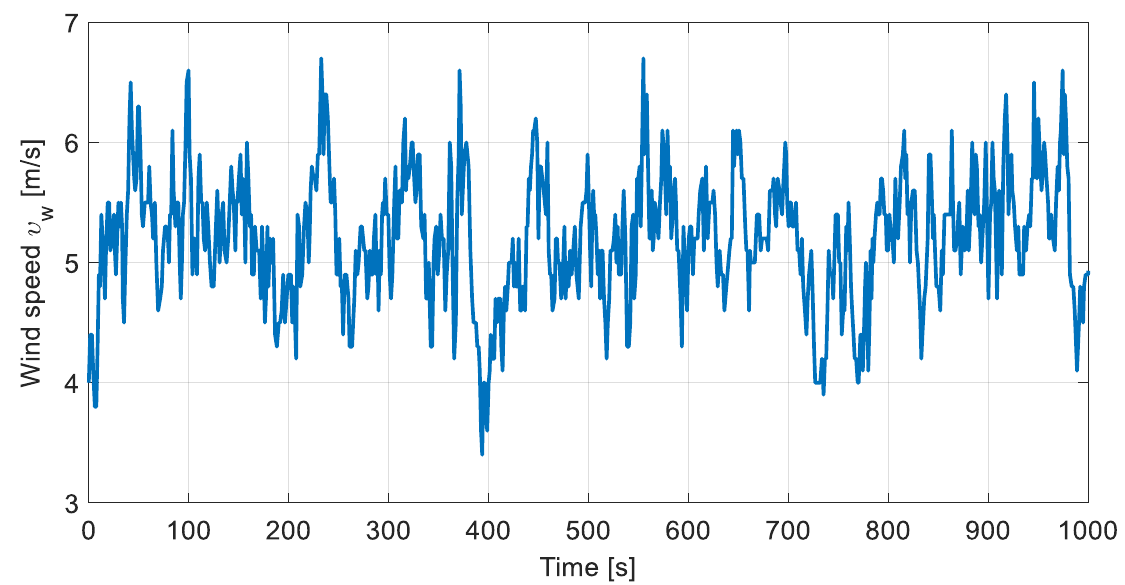}
\caption{Wind profile in region II.}
\label{fig:3-5}
\end{figure}

As stated in Section \ref{sec:5.2.2}, the control objective is to capture the maximum power by tracking the rotor speed with and without actuator fault. A partial loss on the generator's torque is considered as stated in Section \ref{sec:5.2.3} to demonstrate and verify the fault-tolerant performance of the controllers. Accordingly, during the first 600 seconds, the wind profile is applied to the system without the existence of any actuator faults. At $t=600$ s, the actuator failure starts as a partial linear loss of actuation power and ultimately changes its behavior to nonlinear mode as $t=670$ s.

Figure \ref{fig:4-5} shows the comparative performance illustration of SMC, SOFTSMC, and FNTSMC approaches in tracking the optimum rotor speed. \textcolor{black}{As the zoomed-in insets show, when the fault happens, the SMC controller is unable to supply a convenient control torque for efficient tracking of the optimal rotor speed. In addition, although the SOFTSMC delivers much better performance than that of the conventional SMC, its tracking performance is still inferior to the proposed FNTSMC's. As one can observe, the proposed FNTSMC controller outperforms other methods and presents a desirable tracking performance with fewer fluctuations and faster transient response.} The rotor speed tracking error and generator speed are presented in Figures \ref{fig:5-5} and \ref{fig:6-5}. From Figure \ref{fig:5-5} it can be seen that the tracking error associated with the proposed FNTSMC fluctuates in a small region, being 6 and 4 times smaller than similar regions for SMC and SOFTSMC, respectively. From Figures \ref{fig:4-5}-\ref{fig:6-5} one can observe that, during the initial seconds of fault occurrence (\textit{i.e.} $t=600-670$ s), compared to the SMC and SOFTSMC methods, the FNTSMC demonstrates excellent tracking performance. However, regardless of a bit of fluctuation in the FNTSMC’s performance at $t=670$ s, it gets back on track quickly, converges to its previous (no fault) small error region, and successfully tolerates the actuator fault. Furthermore, although SMC and SOFTSMC approaches have delivered acceptable fault-tolerant performances, their performance has been degraded after actuator fault occurrence, showing their weakness in the fault-tolerant task.

\begin{figure}
\centering
\includegraphics[width=4.5 in]{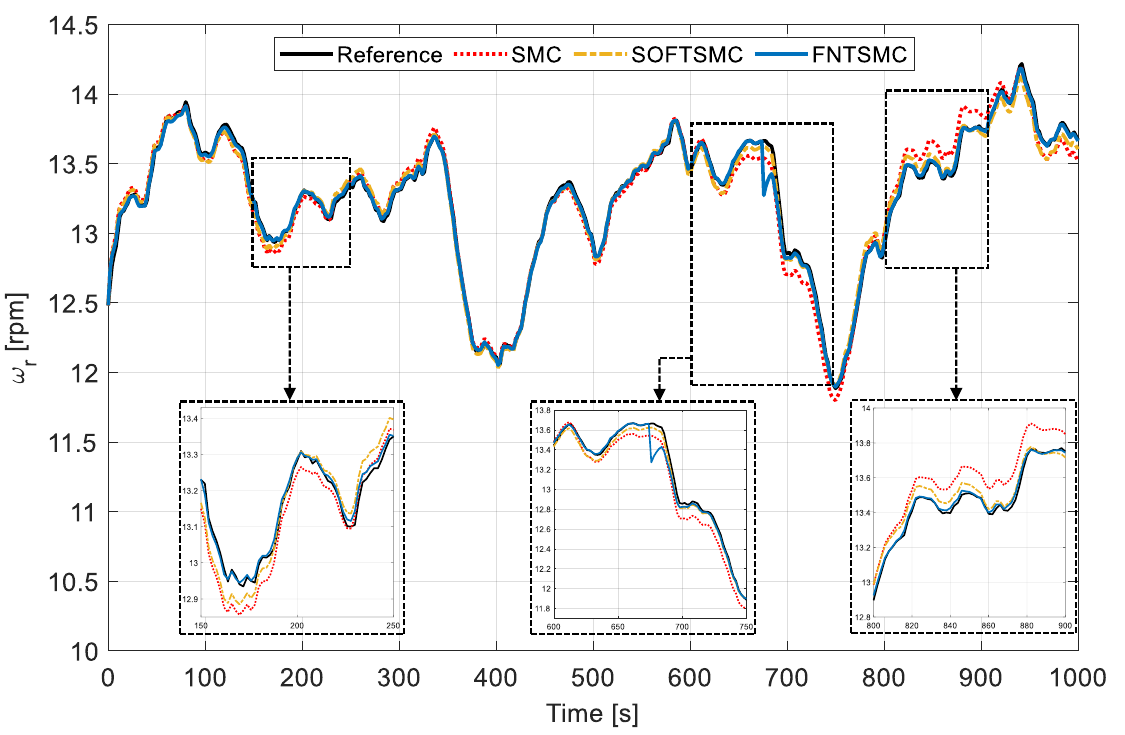}
\caption[Rotor speed tracking; a comparison between SMC, SOFTSMC, and FNTSMC approaches.]{Rotor speed tracking; a comparison between SMC, SOFTSMC, and FNTSMC approaches. The insets show the detail of the regions highlighted by dashed black lines.}
\label{fig:4-5}
\end{figure}

\begin{figure}
\centering
\includegraphics[width=4.5 in]{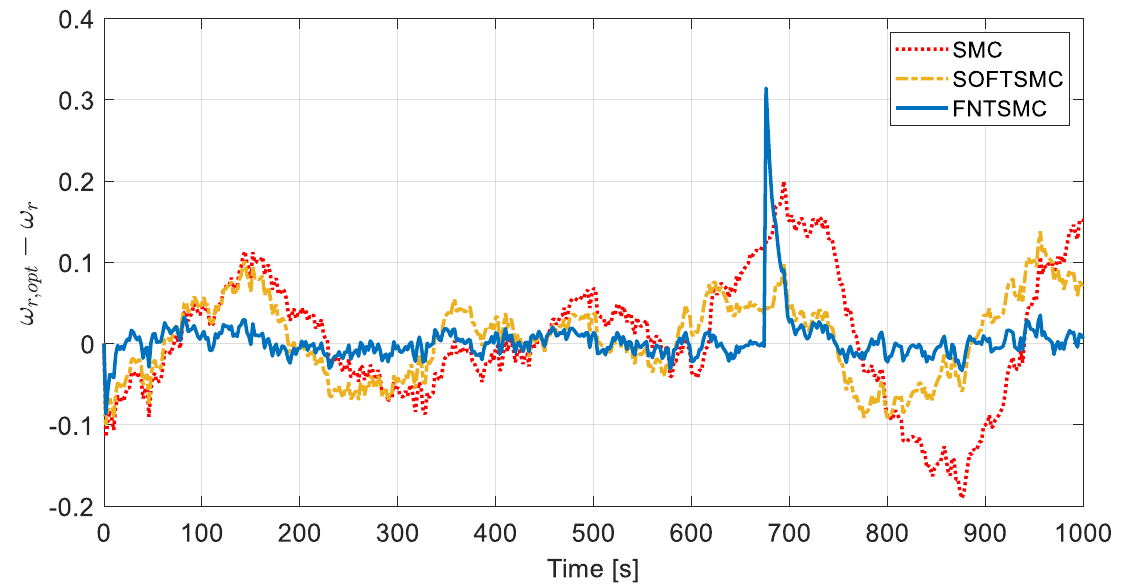}
\caption{\textcolor{black}{Rotor speed tracking error; a comparison between SMC, SOFTSMC, and FNTSMC approaches.}}
\label{fig:5-5}
\end{figure}

\begin{figure}
\centering
\includegraphics[width=4.5 in]{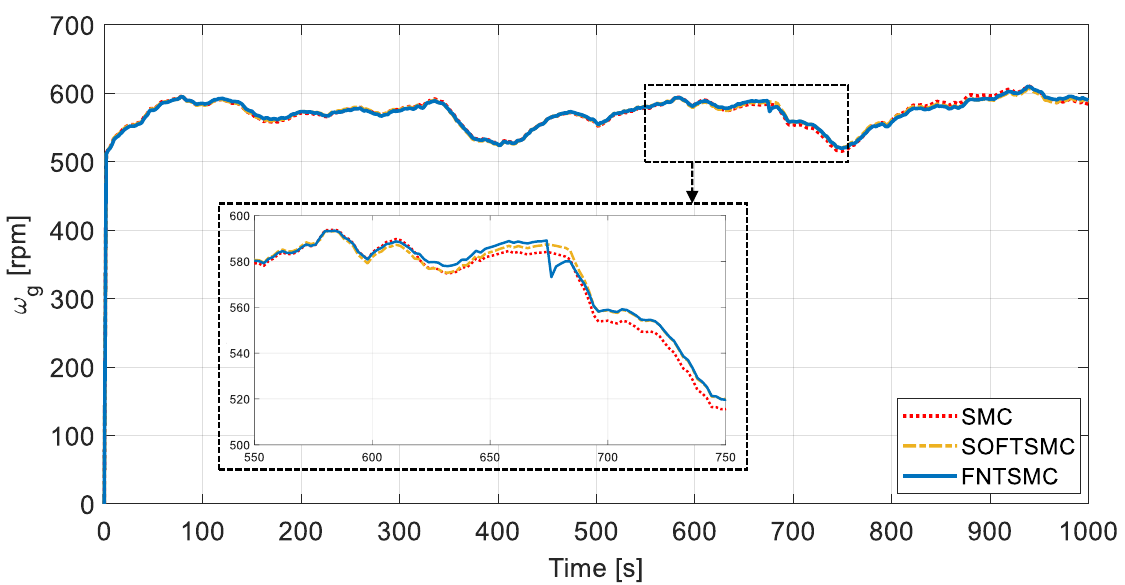}
\caption[Generator speed; a comparison between SMC, SOFTSMC, and FNTSMC approaches.]{Generator speed; a comparison between SMC, SOFTSMC, and FNTSMC approaches. The inset shows the detail of the region highlighted by dashed black lines.}
\label{fig:6-5}
\end{figure}

Figures \ref{fig:7-5}-\ref{fig:9-5} respectively illustrate the generator torque, low-speed shaft torque, and the electric power comparisons between the control approaches under study. \textcolor{black}{As mentioned in Section \ref{sec:5.2.3}, an actuator failure and partial loss of the actuation power happens after $t=600 s$.} \textcolor{black}{According to \ref{fig:8-5} it can be seen that, during the fault occurrence, the proposed FNTSMC controller excites less the drivetrain while providing a better power capture.} Considering Figures \ref{fig:7-5}-\ref{fig:9-5}, it can be observed that due to embedding the fractional-order component alongside the nonsingular terminal SMC design, the proposed FNTSMC algorithm tracks the rotor speed accurately and provides a better power capture in comparison with SMC and SOFTSMC approaches. Accordingly, the weakness of SMC and SOFTSMC approaches in supplying a suitable control torque to efficiently track the rotor speed $\omega_{r,opt}$ is apparent.

\begin{figure}
\centering
\includegraphics[width=4.5 in]{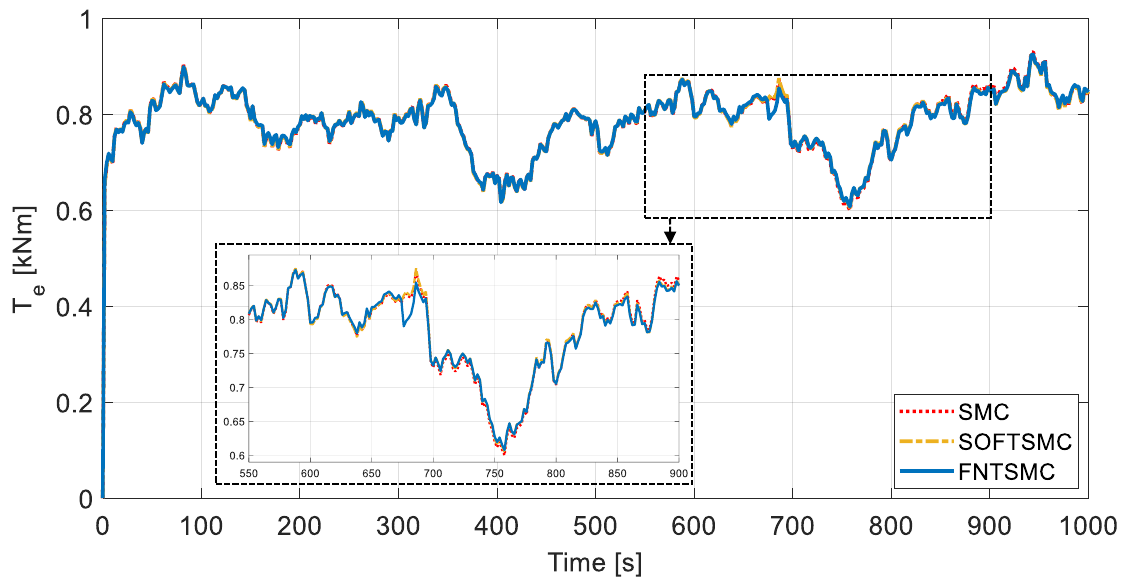}
\caption[Generator torque; a comparison between SMC, SOFTSMC, and FNTSMC approaches.]{Generator torque; a comparison between SMC, SOFTSMC, and FNTSMC approaches. The inset shows the detail of the region highlighted by dashed black lines.}
\label{fig:7-5}
\end{figure}

\begin{figure}
\centering
\includegraphics[width=4.5 in]{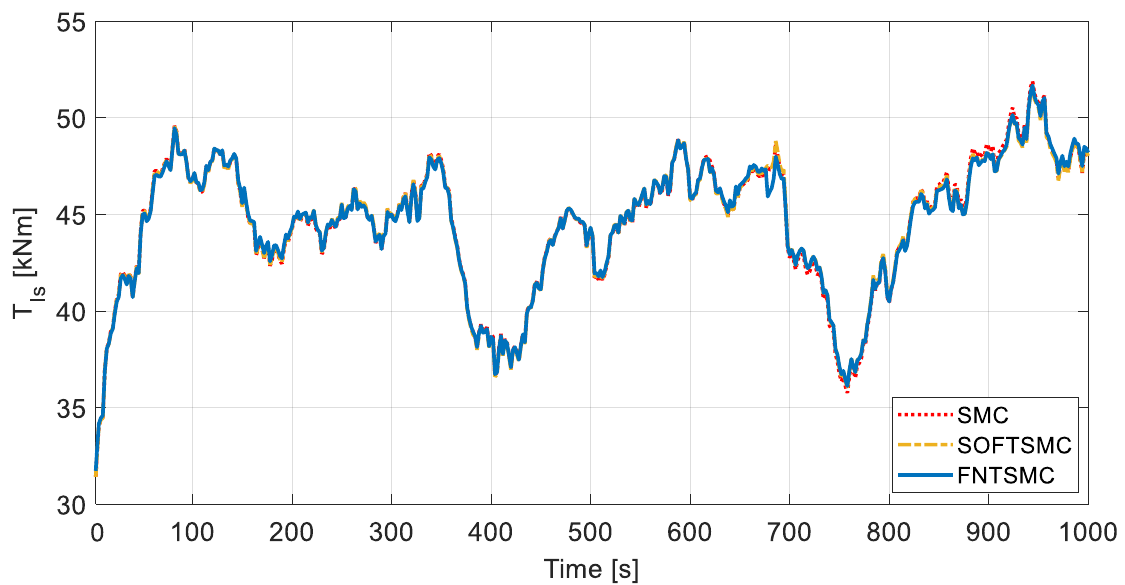}
\caption{Low-speed shaft torque; a comparison between SMC, SOFTSMC, and FNTSMC approaches.}
\label{fig:8-5}
\end{figure}

\begin{figure}
\centering
\includegraphics[width=4.5 in]{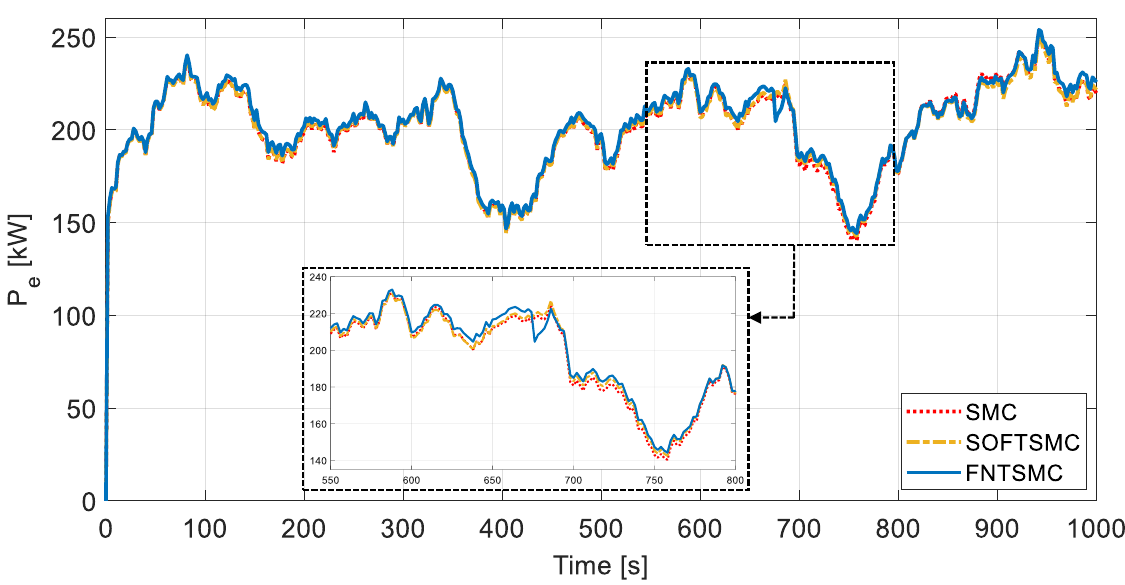}
\caption[Electric power; a comparison between SMC, SOFTSMC, and FNTSMC approaches.]{Electric power; a comparison between SMC, SOFTSMC, and FNTSMC approaches. The inset shows the detail of the region highlighted by dashed black lines.}
\label{fig:9-5}
\end{figure}

\textcolor{black}{To further evaluate the controllers’ performance, a comparative study in terms of the aerodynamic $\eta_{aero}$ and electrical $\eta_{elec}$ efficiencies is carried out, where $\eta_{aero}$ and $\eta_{elec}$ can be calculated as follows:}
\begin{equation}
\label{ENERGIES__S1_}
\textcolor{black}{\eta_{aero}\left(\%\right)=\frac{\int_{t_i}^{t_f} P_{a}\left(t\right) dt}{\int_{t_i}^{t_f} P_{a}^*\left(t\right) dt}\times 100,\; \; \; \; \; \; \eta_{elec}\left(\%\right)=\frac{\int_{t_i}^{t_f} P_{e}\left(t\right) dt}{\int_{t_i}^{t_f} P_{a}^*\left(t\right) dt}\times 100}
\end{equation}
\textcolor{black}{where $t_i$ and $t_f$ denote the initial and final times, respectively.}
\textcolor{black}{Accordingly, aerodynamic and electrical efficiencies of $\eta_{aero,SMC}=72.72 \%$, $\eta_{aero,SOFTSMC}=88.18 \%$, $\eta_{aero,FNTSMC}=91.34 \%$, $\eta_{elec,SMC}=69.64 \%$, $\eta_{elec,SOFTSMC}=94.82 \%$, and $\eta_{elec,FNTSMC}=97.03 \%$ are obtained for the control approaches under study. The achieved foregoing results indicate the superiority of the FNTSMC approach over SMC and SOFTSMC in terms of power capture, and aerodynamic and electric efficiency.}

\section{Conclusions}
\label{sec:5.5}
This chapter investigated the maximum power extraction problem of wind energy conversion systems operating below their rated wind speeds in the presence of actuator faults. Accordingly, a fractional-order nonsingular terminal sliding mode controller (FNTSMC) with enhanced finite-time convergence speed of system states and alleviated chattering was proposed to track the optimum rotor speed and maximize the power production. The closed-loop stability of the system and finite-time convergence of tracking error to the equilibrium point were guaranteed using the Lyapunov stability theory. An actuator fault in the form of a partial loss on the generator's torque was considered to evaluate the fault-tolerant performance and efficaciousness of the proposed method. The performance of the proposed FNTSMC scheme was investigated in comparison with conventional SMC and second-order fast SMC approaches on a two-mass WT system. \textcolor{black}{Simulation results and analysis (Figures \ref{fig:4-5} and \ref{fig:5-5}) demonstrated the notable rotor speed tracking performance of the proposed FNTSMC with fewer fluctuations and faster transient response.} As a result, its satisfactory power extraction performance with less excitation of the drivetrain was validated compared to conventional SMC and SOFTSMC approaches, both in fault-free and faulty situations. In addition, to further investigate the effectiveness of FNTSMC in terms of speed tracking and power extraction, a comparative study on the aerodynamic and electrical efficiencies of the control approaches was provided. Accordingly, taking advantage of the proposed FNTS surface, a superior power extraction performance was revealed compared to SMC and SOFTSMC.



\chapter{Active Fault-tolerant Control of DFIG-based Wind Turbines} 

\label{Chapter6} 

\section{Introduction}
\label{sec:6.1}
This chapter presents active fault-tolerant nonlinear control strategies for the rotor-side converter (RSC) control of DFIG driven WECS subjected to model uncertainties and rotor current sensor faults. Two fractional-order nonsingular terminal sliding mode controllers are proposed for rotor current regulation and speed trajectory tracking. Furthermore, the control scheme is incorporated with two state observers to estimate the rotor current dynamics during sensors' faults. Benefitting from the proposed sliding surfaces, fast finite-time convergence of system states is guaranteed, and the chattering is effectively suppressed. Closed-loop stability analysis is investigated in the sense of the Lyapunov stability criterion.

The DFIG comprises two back-to-back connected power rotor-side and grid-side converters (RSC and GSC), where the active and reactive powers are controlled via the RSC. In the context of power converters' control, various control approaches have been developed for DFIG-based WECSs \cite{patel2021nonlinear,kong2022nonlinear}. Although the aforementioned approaches have demonstrated desirable control performances, SMC-based techniques have been found more attractive and suitable due to fast response, insensitiveness to parametric uncertainties, and implementation ease \cite{hussain2019efficient}. Particularly, SMC-based methods can theoretically determine the final tracking precision according to the developed sliding surface and the reaching law, even if the controlled system suffers from uncertainties and disturbances \cite{mousavi2021robust,mousavi2021robustt}. Aiming at power point tracking, authors in \cite{merabet2018power} proposed an SMC control scheme for power converters' control of a DFIG-based WECS subjected to parametric uncertainties. They developed a rotor current control scheme for the RSC and a cascade control loop for the dc-link voltage regulation. In another study \cite{li2019sliding}, a nonlinear control paradigm based on SMC and feedback linearization technique was investigated for RSC control of a DFIG-based WT.

Despite the well-shown performance of the conventional SMC approaches in controlling the WECSs, to assure the finite-time reachability of the control system to the stable state as well as chattering mitigation, many researchers have developed higher-order SMC methods \cite{mousavi2021maximum,zaihidee2019application,morshed2019sliding}. On the other hand, embedding fractional calculus into the control structure can lead to increased freedom for parameter tuning, the flexibility of the controller design, and the system's accuracy \cite{mousavi2021fault,sebtahmadi2016current,mousavi2015memetic}. This has encouraged many researchers to use FO-SMC methods to tackle the SMC deficiencies. Authors in \cite{ullah2017adaptive} designed a fractional-order adaptive TSMC for both the RSC and GSC of the DFIG-based WT. Authors in \cite{patnaik2015fast} dealt with the RSC control problem of DFIG-based WECS subjected to diverse and challenging situations by developing a fast adaptive TSMC. Authors in \cite{morshed2015comparison} proposed an integral TSMC approach to enhance the power quality of WTs in the presence of uncertainties, parameter variations, and severe voltage sag conditions. However, although the chattering was reduced in the developed control scheme compared to the conventional SMC, the existence of a sign function in the control law had led to an undesirable chattering problem. Accordingly, the same authors later proposed an FLC-based auto-tuned integral TSMC \cite{morshed2019sliding} to simultaneously deal with the power quality enhancement of WTs and chattering elimination.

The measurement of the stator and rotor currents is vital for the current control of WECSs. Accordingly, the reliability and performance of the system would be degraded if the sensor(s) fail to provide the required measurements. Accordingly, numerous model-based and signal-based fault detection and isolation approaches have been investigated in the literature \cite{hussain2019efficient,abubakar2020induction,yang2019pcsmc,mamani2007algebraic}. Authors in \cite{mousavi2022active} proposed an observer-based fault-tolerant FTSMC scheme for the RSC control of a DFIG-driven WECS subjected to model uncertainties and rotor current sensor faults. In another study \cite{yang2019pcsmc}, a sliding mode perturbation and state observer was developed to aggregate the parameter uncertainties of WECS. This chapter proposes active fault-tolerant FNTSMC schemes for the RSC of DFIG-based WT. Accordingly, a current controller is developed to regulate the rotor current, and a speed controller carries out the speed trajectory tracking task. Furthermore, first, estimation and reconstruction of the rotor current during sensor faults is accomplished by an algebraic state observer. Then, a robust SMO is developed to estimate and reconstruct the rotor current during sensor faults. Moreover, by the virtue of inevitable false fault detections due to unavoidable performance degradations in the current sensors over time, a tolerance limit is defined for actual faults occurrence. Hence, the developed SMO can estimate and reconstruct the rotor current during sensors' faults with a high level of reliability. Comparative performance assessments are provided and reveal the promising performance of the proposed control scheme with respect to the well-performed fractional SMC (FSMC) \cite{li2020mitigating}.

The chapter is organized as follows. The DFIG-based WT model is presented in Section \ref{sec:6.2}. The proposed active fault-tolerant controller with a state estimator is described in \ref{sec:6.3}, and the simulation results are provided. Section \ref{sec:6.4} presents the developed SMO-based proposed active fault-tolerant controller together with the associated simulations and comparisons. Finally, Section \ref{sec:6.5} concludes the chapter.

\section{Modelling of DFIG-based WECS}
\label{sec:6.2}

\textcolor{black}{The dynamic model of DFIG-based WT (whose schematic has been show in Figure \ref{fig:2_222}) in a synchronously rotating $dq$ reference frame can be expressed as follows \cite{musarrat2021improved}:}

\begin{subequations}
\label{Equation_6-7}
\begin{align}
V_{ds}&=R_s I_{ds}+\frac{d}{dt}\psi_{ds}-\omega_s\psi_{qs}, \\
V_{qs}&=R_s I_{qs}+\frac{d}{dt}\psi_{qs}-\omega_s\psi_{ds}, \\
V_{dr}&=R_r I_{ds}+\frac{d}{dt}\psi_{dr}-\left(\omega_s-\omega_r\right)\psi_{qr}, \\
V_{dr}&=R_r I_{qs}+\frac{d}{dt}\psi_{qr}+\left(\omega_s-\omega_r\right)\psi_{dr},
\end{align}
\end{subequations}

with the flux equations of stator and rotor expressed as
\begin{subequations}
\label{Equation_6-8}
\begin{align}
\psi_{ds} &=L_{s} I_{ds} +L_{m} I_{dr}, \\
\psi_{qs} &=L_{s} I_{qs} +L_{m} I_{qr}, \\
\psi_{dr} &=L_{r} I_{dr} +L_{m} I_{ds}, \\
\psi_{qr} &=L_{r} I_{qr} +L_{m} I_{qs},
\end{align}
\end{subequations}
where the subscripts $r$ and $s$ stand for rotor and stator, respectively. $\left\{V_{ds} ,V_{qs} \right\}$ and $\left\{V_{dr} ,V_{qr} \right\}$ represent the $d$-axis and $q$-axis voltage components in [V], and $\left\{I_{ds} ,I_{qs} \right\}$ and $\left\{I_{dr} ,I_{qr} \right\}$ represent the current components in [A], respectively. $R_{s} $ and $R_{r} $ are the stator and rotor resistances in [$\Omega$], $L_{s} $, $L_{r} $, and $L_{m} $ denote the stator, rotor, and magnetizing inductances in [H], respectively. $\psi _{ds}$ and $\psi _{qs}$ represent the stator flux components, and $\psi _{dr}$ and $\psi _{qr}$ are the rotor flux components in [Wb]. $\omega_s$ and $\omega _{r} =N_P\Omega_{r}$ [\si{rad/s}] are the stator and rotor angular speeds, respectively, with $N_P$ being the number of pole pairs.

The rotating parts' mechanical dynamics can be expressed as $J\dot{\Omega}_{r}=T_{em}-T_r-f_r\Omega_{r}$, where $f_r$ is the friction coefficient, $J$ represents the rotating components' moment of inertia, and the electromagnetic torque $T_{em} $ can be expressed with stator fluxes and currents as
\begin{equation}
T_{em} =N_P \frac{L_m}{L_s} \left(\psi _{qs} I_{dr} -\psi _{ds} I_{qr} \right).
\label{Equation_6-10}
\end{equation}

Choosing a reference frame aligned to the d-axis of the stator flux $\psi_s$ leads to $\psi_{ds}=\psi_s$ and $\psi_{qs}=0$, yields
\begin{equation}
T_{em} =N_P \frac{L_m V_s}{\omega_s L_s} I_{qr}.
\label{Equation_6-11}
\end{equation}

Neglecting the per phase stator resistance we have $V_{qs}=V_s=\omega_s \psi_s$ and $V_{ds}=0$. Consequently, the rotor voltages and the stator active and reactive powers can be written as follows
\begin{subequations}
\label{Equation_6-12}
\begin{align}
V_{dr}&=R_r I_{dr}+\sigma L_r\frac{d}{dt}I_{dr}-\sigma L_r s\omega_s I_{qr}, \\
V_{qr}&=R_r I_{qr}+\sigma L_r\frac{d}{dt}I_{qr}-\sigma s\omega_s I_{dr}+s\frac{L_m V_s}{L_s},
\end{align}
\end{subequations}
\begin{subequations}
\label{Equation_6-13}
\begin{align}
P_s&=-\frac{L_m V_s}{L_s}I_{qr}, \\
Q_s&=\frac{V^{2}_{s}}{\omega_s L_s}-\frac{L_m V_s}{L_s}I_{dr},
\end{align}
\end{subequations}
where $\sigma=1-L_m/\left(L_r L_s\right)$ and the generator slip is defined as $s=\left(\omega_s -\omega_r\right)/{\omega_s}$.

The nonlinear system can be expressed as follows
\begin{equation}
\label{Equation_6-14}
\dot{x}=F\left(x\right)+Hu=f\left(x\right)+hu+d,
\end{equation}
where $x=[I_{dr} \, \, I_{qr}]^T$ denotes the state vector, $u=[V_{dr} \, \, V_{qr}]^T$ is the control signal, and $f\left(x\right)$ and $h$ represent the best approximation of $F\left(x\right)$ and $H$, respectively. The unknown lumped uncertainty is denoted by $d=\Delta f+\Delta hu=\left[d_1\, \, d_2\right]^T$, where $\Delta f$ and $\Delta h$ are the model and input uncertainties, respectively.

\begin{subequations}
\label{Equation_6-16}
\begin{align}
\label{}
f\left(x\right)&=
\begin{bmatrix}
-\frac{ R_r}{\sigma  L_r} I_{dr}+s\omega_s  I_{qr} \\
-\frac{ R_r}{\sigma  L_r} I_{qr}-s\omega_s  I_{dr}+s\frac{ L_m  V_s}{\sigma  L_r  L_s} \\
\end{bmatrix} \\
h&=
\begin{bmatrix}
\frac{1}{\sigma  L_r} & 0 \\
0 & \frac{1}{\sigma  L_r},
\end{bmatrix}
\end{align}
\end{subequations}

\section{Fault-Tolerant Controller Design}
\label{sec:6.3}
An active fault-tolerant FTSMC scheme is proposed in this section for the RSC of DFIG-based WT. Accordingly, a current controller is proposed to regulate the rotor current and the electromagnetic torque, while the speed trajectory tracking is accomplished with a speed controller. Furthermore, since the $dq$ components of the measured rotor currents are essential for reliable current control performance, a state estimator adopted from \cite{mamani2007algebraic} is used to estimate the rotor currents in the $dq$ reference frame. It is noteworthy that the $t$ index is neglected in the equations for simplification of writing.

\subsection{Current Controller Design}
\label{sec:6.3.1}
Defining the reference trajectory as $x_{ref}=\left[I_{dr-ref}\,\,\,I_{qr-ref}\right]^T$, and considering \eqref{Equation_6-11} we have
\begin{equation}
\label{Equation_6-17}
I_{qr-ref}=\frac{\omega_s L_s}{N_P L_m V_s}T_{em,ref}
\end{equation}
where $T_{em,ref}$ will be obtained from the speed controller later in section \ref{sec:6.3.2}. The active power reference corresponds to $I_{qr-ref}$, and the reactive power reference can be expressed as
\begin{equation}
\label{Equation_6-18}
Q_{s-ref}=\frac{V^2_s}{\omega_s L_s}-\frac{L_m V_s}{L_s}I_{dr}
\end{equation}
\begin{rem}
To achieve an improved power factor, $Q_{s-ref}=0$ is considered.
\end{rem}

Accordingly, $I_{dr-ref}$ can be obtained considering \eqref{Equation_6-18} as $I_{dr-ref}={V_s}/\left(\omega_s L_m\right)$. Defining the tracking error $\dot{e}=\dot{x}-\dot{x}_{ref}=f\left(x\right)+hu+d-\dot{x}_{ref}$, the following FTSMC sliding surface $S=\left[S_1\,\,\,S_2\right]^T$ is proposed
\begin{equation}
\label{Equation_6-19}
S=
\begin{bmatrix}
S_1 \\
S_2
\end{bmatrix}
=
\begin{bmatrix}
\kappa_1e_1+\kappa_2 \mathfrak D^{\gamma-1}e_1+\frac{1}{\beta_1} \mathfrak D^{\gamma}\left|e_1\right|^{\frac{p_1}{q_1}} \\
\kappa_3e_2+\kappa_4 \mathfrak D^{\gamma-1}e_2+\frac{1}{\beta_2} \mathfrak D^{\gamma}\left|e_2\right|^{\frac{p_2}{q_2}}
\end{bmatrix}
\end{equation}
where $\kappa_i>0$, $i=1,2,3,4$, $\beta >0$. $q_i$ and $p_i$, $i=1,2$, are odd positive integers satisfying $1<\left(p_i/q_i\right)<2$, and ${\rm {\mathfrak D}}^{\gamma} \left(\cdot\right)$ represents the Riemann-Liouville fractional derivative of order $0<\gamma <1$.

Differentiating \eqref{Equation_6-19} with respect to time yields
\begin{equation}
\label{Equation_6-20}
\dot{S}=
\begin{bmatrix}
\dot{S}_1 \\
\dot{S}_2
\end{bmatrix}
=
\begin{bmatrix}
\kappa_1 \dot{e}_1 +\kappa_2 \mathfrak D^{\gamma}e_1+\frac{p_1}{\beta_1 q_1}\mathfrak D^{\gamma}\left|e_1\right|^{\frac{p_1}{q_1}-1}\\
\kappa_3 \dot{e}_2 +\kappa_4 \mathfrak D^{\gamma}e_2+\frac{p_2}{\beta_2 q_2}\mathfrak D^{\gamma}\left|e_2\right|^{\frac{p_2}{q_2}-1}
\end{bmatrix}
\end{equation}
where
\begin{equation}
\label{Equation_6-21}
\dot{e}=
\begin{bmatrix}
\dot{e}_1 \\
\dot{e}_2
\end{bmatrix}
=
\begin{bmatrix}
f\left(x\right)+hu_1+d_1-\dot{I}_{dr-ref}\\
f\left(x\right)+hu_2+d_2-\dot{I}_{qr-ref}
\end{bmatrix}
\end{equation}

Substituting \eqref{Equation_6-21} into \eqref{Equation_6-20} yields

\begin{equation}
\label{Equation_6-22}
\dot{S}=
\begin{bmatrix}
\kappa_1 \left(f\left(x\right)+hu_1+d_1-\dot{I}_{dr-ref}\right) \dot{e}_1 +\kappa_2 \mathfrak D^{\gamma}e_1+\frac{p_1}{\beta_1 q_1}\mathfrak D^{\gamma}\left|e_1\right|^{\frac{p_1}{q_1}-1}\\
\kappa_3 \left(f\left(x\right)+hu_2+d_2-\dot{I}_{qr-ref}\right)\dot{e}_2 +\kappa_4 \mathfrak D^{\gamma}e_2+\frac{p_2}{\beta_2 q_2}\mathfrak D^{\gamma}\left|e_2\right|^{\frac{p_2}{q_2}-1}
\end{bmatrix}
\end{equation}

The control law is then synthesized as
\begin{align}
\label{Equation_6-23}
u&=\begin{bmatrix}
u_1 \\
u_2
\end{bmatrix} \\ \nonumber
&=\begin{bmatrix}
h^{-1}\kappa^{-1}_1\left(-f\left(x\right)-d_1+\dot{I}_{dr-ref}-\kappa_2 \mathfrak D^{\gamma}e_1-\frac{p_1}{\beta_1 q_1}\mathfrak D^{\gamma}\left|e_1\right|^{\frac{p_1}{q_1}-1}\right)\\
h^{-1}\kappa^{-1}_3\left(-f\left(x\right)-d_2+\dot{I}_{qr-ref}-\kappa_4 \mathfrak D^{\gamma}e_2-\frac{p_2}{\beta_2 q_2}\mathfrak D^{\gamma}\left|e_2\right|^{\frac{p_2}{q_2}-1}\right)
\end{bmatrix}
\end{align}

The equivalent control law can be expressed as $u_{eq}=u+u_{sw}$, where $u_{sw}$ compensates the lumped uncertainties' effects, and can be expressed as
\begin{equation}
\label{Equation_6-25}
u_{sw}=
\begin{bmatrix}
u_{sw,1} \\
u_{sw,2}
\end{bmatrix}
=
\begin{bmatrix}
-\kappa^{-1}_1\left(\eta_1 \mathsf{sgn}\left(S_1\right)+\eta_2 S_1\right)\\
-\kappa^{-1}_3\left(\eta_3 \mathsf{sgn}\left(S_2\right)+\eta_4 S_2\right)
 \end{bmatrix}
\end{equation}
where $\eta_i>0$, $i=1,2,3,4$.

\begin{thm}
\label{thm:1}
Applying the sliding surface \eqref{Equation_6-19} and the proposed control $u_{eq}$, the finite-time convergence of the tracking error dynamics to the sliding surface along with their asymptotical convergence to zero is achieved.
\end{thm}

\begin{proof}
Consider the Lyapunov function candidate $V=S^2/2$. Differentiating $V$ with respect to time and substituting \eqref{Equation_6-19} and \eqref{Equation_6-22} yields
\begin{align}
\label{Equation_6-27}
\dot{V}=
\begin{bmatrix}
\dot{V}_1 \\
\dot{V}_2
\end{bmatrix}
&=
\begin{bmatrix}
S_1 \dot{S}_1\\
S_2 \dot{S}_2
\end{bmatrix}
=
\begin{bmatrix}
-S_1 \left(\eta_1 \mathsf{sgn}\left(S_1\right)+\eta_2 S_1\right) \\
-S_2 \left(\eta_3 \mathsf{sgn}\left(S_2\right)+\eta_4 S_2\right)
\end{bmatrix} \\ \nonumber
&\leq
\begin{bmatrix}
-\eta_1 \left|S_1\right| -\eta_2 S^2_1 \\
-\eta_3 \left|S_2\right| -\eta_4 S^2_2
\end{bmatrix}
\end{align}

From \eqref{Equation_6-27} one can observe that $\dot{V}\leq 0$. Accordingly, the system states' asymptotical convergence to $S=0$ is achieved. Considering $S^{2} =2V$, the finite-time convergence can be studied as follows. Rewriting \eqref{Equation_6-27} yields,
\begin{equation}
\label{Equation_6-28}
\dot{V}=
\begin{bmatrix}
dV_1/dt_1 \\
dV_2/dt_2
\end{bmatrix} =
\begin{bmatrix}
-\eta_1 \left|\sqrt{2V_1}\right|-2\eta_2 V_1\\
-\eta_3 \left|\sqrt{2V_2}\right|-2\eta_4 V_2
\end{bmatrix}
\end{equation}

From \eqref{Equation_6-28} and defining $m=1/2$ one can obtain that
\begin{align}
\label{Equation_6-29}
\begin{bmatrix}
dV_1/dt_1 \\
dV_2/dt_2
\end{bmatrix}
&\leq
\begin{bmatrix}
-\frac{dV_1}{\sqrt{2}\eta_1\sqrt{V_1}+2\eta_2 V_1}\\
-\frac{dV_2}{\sqrt{2}\eta_3\sqrt{V_2}+2\eta_4 V_2}
\end{bmatrix}
\leq
\begin{bmatrix}
-\frac{dV_1}{2^m\eta_1 V^m_1+2\eta_2 V_1}\\
-\frac{dV_2}{2^m\eta_3 V^m_2+2\eta_4 V_2}
\end{bmatrix} \\ \nonumber
&\leq
\begin{bmatrix}
-\frac{V^{-m}_1dV_1}{2^m\eta_1+2\eta_2 V^{m-1}_1} \\
-\frac{V^{-m}_2dV_2}{2^m\eta_3+2\eta_4 V^{m-1}_2}
\end{bmatrix}
\leq
\begin{bmatrix}
-\frac{dV^{1-m}_1}{\left(1-m\right)\left(2^m\eta_1+2\eta_2 V^{m-1}_1\right)} \\
-\frac{dV^{1-m}_2}{\left(1-m\right)\left(2^m\eta_3+2\eta_4 V^{m-1}_2\right)}
\end{bmatrix}
\end{align}

Solving \eqref{Equation_6-29} leads to the convergence time as
\begin{align}
\label{Equation_6-30}
\begin{bmatrix}
t_{s,1}-t_{r,1} \\
t_{s,2}-t_{r,2}
\end{bmatrix}
&\leq
\begin{bmatrix}
-\frac{1}{2\eta_2\left(1-m\right)}\int_{V_{1}(t_{r,1})}^{V_1(t_{s,1})}\frac{dV^{1-m}_1}{\frac{2^m \eta_1}{2\eta_2}+V^{1-m}_1} \\
-\frac{1}{2\eta_4\left(1-m\right)}\int_{V_{2}(t_{r,2})}^{V_2(t_{s,2})}\frac{dV^{1-m}_2}{\frac{2^m \eta_3}{2\eta_4}+V^{1-m}_2}
\end{bmatrix} \\ \nonumber
&\leq
\begin{bmatrix}
\frac{1}{2\eta_2\left(2-m\right)}\ln\left(V^{1-m}_1+\frac{2^m \eta_1}{2\eta_2}\right)^{V_1(t_{s,1})}_{V_{1}(t_{r,1})} \\
\frac{1}{2\eta_4\left(2-m\right)}\ln\left(V^{1-m}_2+\frac{2^m \eta_3}{2\eta_4}\right)^{V_2(t_{s,2})}_{V_{2}(t_{r,2})}
\end{bmatrix}
\end{align}
where $V_i\left(t_{s,i}\right)$ and $V_i\left(t_{r,i}\right)$, $i=1,2$, denote the settling and reaching time, respectively, and $V_0\left(t_{s,i}\right)=0$.

Accordingly, when the system acts on the sliding surface \eqref{Equation_6-19}, the states converge to zero at the following finite time

\begin{equation}
\label{Equation_6-31}
\begin{bmatrix}
t_{s,1} \\
t_{s,2}
\end{bmatrix}
\leq
\begin{bmatrix}
\frac{1}{2\eta_2\left(2-m\right)}\ln\left(V^{1-m}_1(t_{r,1})+\frac{2^m \eta_1}{2\eta_2}\right)+(t_{r,1}) \\
\frac{1}{2\eta_4\left(2-m\right)}\ln\left(V^{1-m}_2(t_{r,2})+\frac{2^m \eta_3}{2\eta_4}\right)+(t_{r,2})
\end{bmatrix}
\end{equation}
\end{proof}

\subsection{Speed Controller Design}
\label{sec:6.3.2}
The mechanical system dynamics $\dot{\Omega}_r$ can be rewritten as
\begin{equation}
\label{Equation_6-32}
\dot{\Omega}_r=\frac{T_{em}}{J}+d_3
\end{equation}
where $T_{em}$ denotes the control input, and $d_3=-\frac{T_r}{J}-\frac{f_r \Omega_r}{J}$ is the lumped uncertainty. Defining the speed tracking error as $e_3=\Omega_r-\Omega_{r-ref}$, differentiating it with respect to time $\dot{e}_3=\dot{\Omega}_r-\dot{\Omega}_{r-ref}$, and substituting it into \eqref{Equation_6-32} yields
\begin{equation}
\label{Equation_6-34}
\dot{e}_3=\frac{T_{em}}{J}+d_3-\dot{\Omega}_{r-ref}
\end{equation}

The sliding surface $S_3$ can be defined as
\begin{equation}
\label{Equation_6-35}
S_3=\kappa_5e_3+\kappa_6 \mathfrak D^{\gamma-1}e_3+\frac{1}{\beta_3} \mathfrak D^{\gamma}\left|e_3\right|^{\frac{p_3}{q_3}}
\end{equation}

Differentiating \eqref{Equation_6-35} with respect to time yields
\begin{equation}
\label{Equation_6-36}
\dot{S}_3=\kappa_5 \left(\frac{T_{em}}{J}+d_3-\dot{\Omega}_{r-ref}\right)+\kappa_6 \mathfrak D^{\gamma}e_3+\frac{p_3}{\beta_3 q_3}\mathfrak D^{\gamma} \left|e_3\right|^{\frac{p_3}{q_3}-1}
\end{equation}

Setting $\dot{S}_3=0$ derives the speed control law as follows
\begin{equation}
\label{Equation_6-37}
T_{em}=J\kappa^{-1}_5 \left(-d_3+\dot{\Omega}_{r-ref}-\kappa_6 \mathfrak D^{\gamma}e_3-\frac{p_3}{\beta_3 q_3}\mathfrak D^{\gamma}\left|e_3\right|^{\frac{p_3}{q_3}-1}\right)
\end{equation}

The equivalent control law can be expressed as $T_{em,eq}=T_{em}+T_{em,sw}$, where $T_{em,sw}=-\kappa^{-1}_5\left(\eta_5 \mathsf{sgn}\left(S_3\right)+\eta_6 S_3\right)$. In order to investigate the stability analysis, consider the Lyapunov function candidate as $V_3=S^2_3/2$. Differentiating $V_3$ with respect to time and substituting \eqref{Equation_6-36} and $T_{em,eq}$ yields
\begin{equation}
\label{Equation_6-38}
\dot{V}_3=S_3 \dot{S}_3=-S_3 \left(\eta_5 \mathsf{sgn}\left(S_3\right)+\eta_6 S_3\right)
\leq -\eta_5 \left|S_3\right| -\eta_6 S^2_3
\end{equation}

According to \eqref{Equation_6-38}, $\dot{V}_3\leq 0$; hence, the system states' asymptotical convergence to the sliding surface is achieved. A similar states convergence investigation as in section \ref{sec:6.3.1} results in the following finite time.
\begin{equation}
\label{Equation_6-39}
t_{s,3} \leq \frac{1}{2\eta_6\left(2-m\right)}\ln\left(V^{1-m}_3(t_{r,3})+\frac{2^m \eta_5}{2\eta_6}\right)+(t_{r,3})
\end{equation}

\section{Fault-Tolerant Control using Algebraic State Estimator for Current Estimation}
\label{sec:6.4}
In this section, first, an algebraic state estimator (ASE) inspired from \cite{mamani2007algebraic} is addressed to estimate the rotor currents in the $dq$ reference frame. Then, the feasibility of the proposed control scheme is validated in comparison with the well-performing FSMC \cite{li2020mitigating} through simulations.

\subsection{Algebraic State Estimator}
\label{sec:6.4.1}
Assuming constant parametric uncertainties as constant disturbance, \eqref{Equation_6-12} can be rewritten as follows
\begin{equation}
\label{Equation_6-40}
\dot{I}=\mathcal A_n I+\mathcal B_n V+\mathcal D_n V_s+\xi^*
\end{equation}
where $\xi^*=\Delta \mathcal AI+\Delta \mathcal BV + \Delta \mathcal DV_s+d$ represents the slow varying lumped disturbance, $I=[I_{dr}\,\,\, I_{qr}]^T$, $V=[V_{dr}\,\,\, V_{qr}]^T$, $\mathcal A_n=\begin{bmatrix}
-\frac{R_r}{\sigma L_r} & s\omega_s \\
\frac{s\omega_s}{L_r} & -\frac{R_r}{\sigma L_r}
\end{bmatrix}$, $\mathcal B_n=\left[\frac{1}{\sigma L_r}\,\,\,\frac{1}{\sigma L_r}\right]^T$, and $\mathcal D_n=\left[0\,\,\,-\frac{sL_m}{\sigma L_r L_s}\right]^T$. The Laplace transformation of \eqref{Equation_6-40} is
\begin{equation}
\label{Equation_6-41}
\mathfrak{p}I-I(0)=\mathcal A_n I+\mathcal B_n V +\mathcal D_n V_s+\frac{\xi^*}{\mathfrak{p}}
\end{equation}

Multiplying \eqref{Equation_6-41} by "$\mathfrak{p}$" and taking its second derivation with respect to "$\mathfrak{p}$", yields
\begin{equation}
\label{Equation_6-42}
\frac{d^2}{d\mathfrak{p}^2}\left[\mathfrak{p}^2 I\right]=\mathcal A_n \frac{d^2}{d\mathfrak{p}^2}\left[\mathfrak{p}I\right]+\mathcal B_n \frac{d^2}{d\mathfrak{p}^2}\left[\mathfrak{p}V\right]+\mathcal D_n\frac{d^2}{d\mathfrak{p}^2}\left[\mathfrak{p}V_s\right]
\end{equation}

Since $\frac{d^2}{d\mathfrak{p}^2}\left[\xi^*\right]\cong0$ and $\frac{d^2}{d\mathfrak{p}^2}\left[-\mathfrak{p}I(0)\right]=0$, it can be observed that \eqref{Equation_6-42} is independent of lumped disturbances and initial conditions. Expanding \eqref{Equation_6-42}, it can be expressed in the time domain as
\begin{align}
\label{Equation_6-44}
t^2I&-4\smallint tI\,dt+2\iint I\,dt^2 \\ \nonumber
&=\mathcal A_n\Big(\smallint t^2I\,dt-2\iint tI\,dt^2\Big) \\ \nonumber
&+\mathcal B_n\Big(\smallint t^2V\,dt-2\iint tV\,dt^2\Big) \\ \nonumber
&+\mathcal D_n\Big(\smallint t^2V_s\,dt-2\iint tV_s\,dt^2\Big)
\end{align}

Accordingly, the estimated current vector can be synthesized as
\begin{align}
\label{Equation_6-45}
\hat{I}&=\frac{4}{t^2}\smallint tI\,dt-\frac{2}{t^2}\iint I\,dt^2 \\ \nonumber
&+\frac{1}{t^2}\Big[\mathcal A_n\Big(\smallint t^2I\,dt-2\iint tI\,dt^2\Big)\Big] \\ \nonumber
&+\frac{1}{t^2}\Big[\mathcal B_n\Big(\smallint t^2V\,dt-2\iint tV\,dt^2\Big)\Big] \\ \nonumber
&+\frac{1}{t^2}\Big[\mathcal D_n\Big(\smallint t^2V_s\,dt-2\iint tV_s\,dt^2\Big)\Big]
\end{align}

The difference between the measured and the estimated rotor current in the $dq$ reference frame can be expressed by $\mathcal R=|I-\hat{I}|$, where $\hat{I}=[\hat{I}_{dr}\,\,\,\hat{I}_{qr}]^T$. Ideally, $\mathcal R=0$ and $\mathcal R>0$ should denote the no-fault and faulty situations, respectively. However, false fault detections are inevitable since some slight performance degradations might happen in the current sensors over time, which cannot be counted as actual faults. Accordingly, a tolerance limit $\mathcal T$ is defined for actual faults occurrence, where $\mathcal R \leq \mathcal T$ and $\mathcal R > \mathcal T$ denote the no-fault and faulty situations, respectively. Hence, an observer is developed for rotor current estimation and reconstruction during sensors' faults. The stator flux linkage and voltage can be expressed in the stator reference frame as
\begin{subequations}
\label{Equation_6-46}
\begin{align}
\psi_{s\alpha}&=L_s i_{s\alpha} + L_{m}i_{r\alpha}, \\
\psi_{s\beta}&=L_s i_{s\beta} + L_{m}i_{r\beta},
\end{align}
\end{subequations}
\begin{subequations}
\label{Equation_6-47}
\begin{align}
v_{s\alpha}&=R_s i_{s\alpha}+\dot{\psi}_{s\alpha}, \\
v_{s\beta}&=R_s i_{s\beta}+\dot{\psi}_{s\beta},
\end{align}
\end{subequations}

Substituting \eqref{Equation_6-46} into \eqref{Equation_6-47} yields
\begin{subequations}
\label{Equation_6-48}
\begin{align}
\dot{i}_{r\alpha}&=\frac{v_{s\alpha}-R_s i_{s\alpha}-L_s \dot{i}_{s\alpha}}{L_m}, \\
\dot{i}_{r\beta}&=\frac{v_{s\beta}-R_s i_{s\beta}-L_s \dot{i}_{s\beta}}{L_m},
\end{align}
\end{subequations}

According to \eqref{Equation_6-48}, the approximated rotor current in $\alpha-\beta$ frame can be expressed as
\begin{subequations}
\label{Equation_6-49}
\begin{align}
\hat{i}_{r\alpha}&=\frac{\int v_{s\alpha}-R_s \int i_{s\alpha}-L_s i_{s\alpha}}{L_m}, \\
\hat{i}_{r\beta}&=\frac{\int v_{s\beta}-R_s \int i_{s\beta}-L_s i_{s\beta}}{L_m},
\end{align}
\end{subequations}

Considering \eqref{Equation_6-49}, the dependency of rotor current approximations on the voltage measurements and stator currents is apparent. So, if $\mathcal R>\mathcal T$, the rotor current approximation \eqref{Equation_6-49} will be used to derive the current control law $u_{eq}$.

\subsection{Simulation Results}
\label{sec:6.4.2}
This section provides a comparative performance investigation of the proposed active fault-tolerant FTSMC scheme with that of FSMC \cite{li2020mitigating} for the rotor-side converter of DFIG-based 1.5 MW WT, whose characteristics are illustrated in Table \ref{tab:6-1}. The WT is considered to be subjected to parametric uncertainties and sensor faults. Accordingly, it is assumed that the sensor fault associated with $I_{dr}$ happens during 25-35 s and two faults associated with $I_{qr}$ happen during 15-20 s and 35-45 s. Furthermore, the lumped uncertainty $d=50\% f+50\% hu$ is considered to affect the system. The wind profile with an average speed of 5.3 \si{m/s} within the speed range of 3.7-6.7 \si{m/s} is illustrated in Figure \ref{fig:6-1}. Figure \ref{fig:6-2} depicts the rotor speed tracking comparative illustration. Accordingly, one can observe that both control approaches demonstrate a similar performance; however, the proposed FTSMC outperforms FSMC with more precise trajectory tracking performance along with much less tracking error.

\begin{table}[!t]
\centering
\caption{1.5MW DFIG Based WT Parameters.}
\label{tab:6-1}
\scalebox{0.8}{
\begin{tabular}{llllllll}
\hline\hline \\[-3mm]
\textbf{Parameter}	& \textbf{Symbol}	& \textbf{Value} & \textbf{Unit}  \\
\hline
Stator resistance & $R_s$ & 0.023 & \si{pu} \\
Rotor resistance & $R_r$ & 0.016  & \si{pu} \\
Stator inductance & $L_s$ & 0.18  & \si{pu} \\
Rotor inductance & $L_r$ & 0.16 & \si{pu} \\
Magnetizing inductance & $L_m$ & 2.9 & \si{pu} \\
Pole pairs & $N_P$  & 3 & \si{-} \\
Blade radius & $R$  & 35 & \si{m} \\
Rotor friction coefficient & $f_r$ & 0.00015 & \si{N.m.s/rad} \\
Rotor inertia & $J_r$ & 765.6 & \si{kg.m^2} \\
Gearbox ratio & $N_{GB}$ & 62.5  & \si{-} \\
Air density & $\rho$ & 1.225  & \si{kg/m^3} \\
Optimum tip-speed ratio & $\lambda_{opt}$  & 6.325 & \si{-} \\
\hline\hline
\end{tabular}}
\end{table}

\begin{figure}
\centering
\includegraphics[width=4 in]{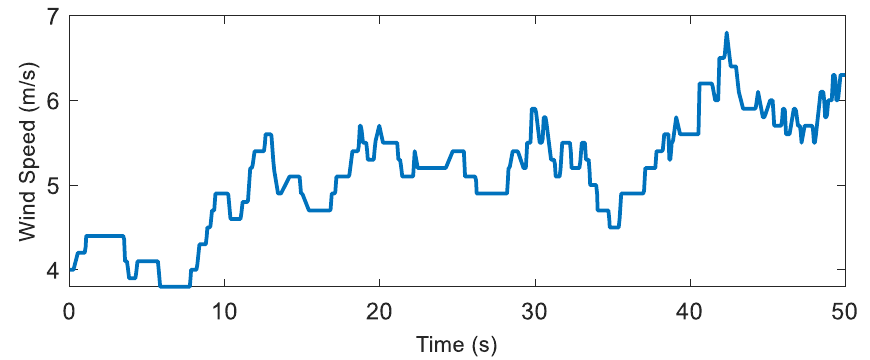}
\caption{Wind speed profile (50 s).}
\label{fig:6-1}
\end{figure}

The state observer's rotor current dynamics estimation and reconstruction performance in case of faulty sensors and lumped uncertainties is shown in Figure \ref{fig:6-3}. It can be seen that the estimation is fulfilled in a very short time with minimal error, demonstrating its remarkable performance. Figure \ref{fig:6-4} demonstrates the comparative grid active and reactive power tracking performance. As evident, both control approaches successfully carry out the grid active power tracking; however, in the case of reactive power regulation, the proposed FTSMC clearly surpass the FSMC with a precise power tracking performance and a small transient response. The sliding surfaces associated with FSMC and the proposed FTSMC are shown in Figure \ref{fig:6-5}. Accordingly, one can observe that although both controllers deliver bounded sliding surfaces, the proposed FTSMC outperforms the FSMC with more desirable chattering mitigation. The obtained foregoing results indicate the effective rotor currents reconstruction performance of the proposed fault-tolerant control technique and its superiority over FSMC in terms of speed and power tracking.

\begin{figure}
\centering
\includegraphics[width=4.5 in]{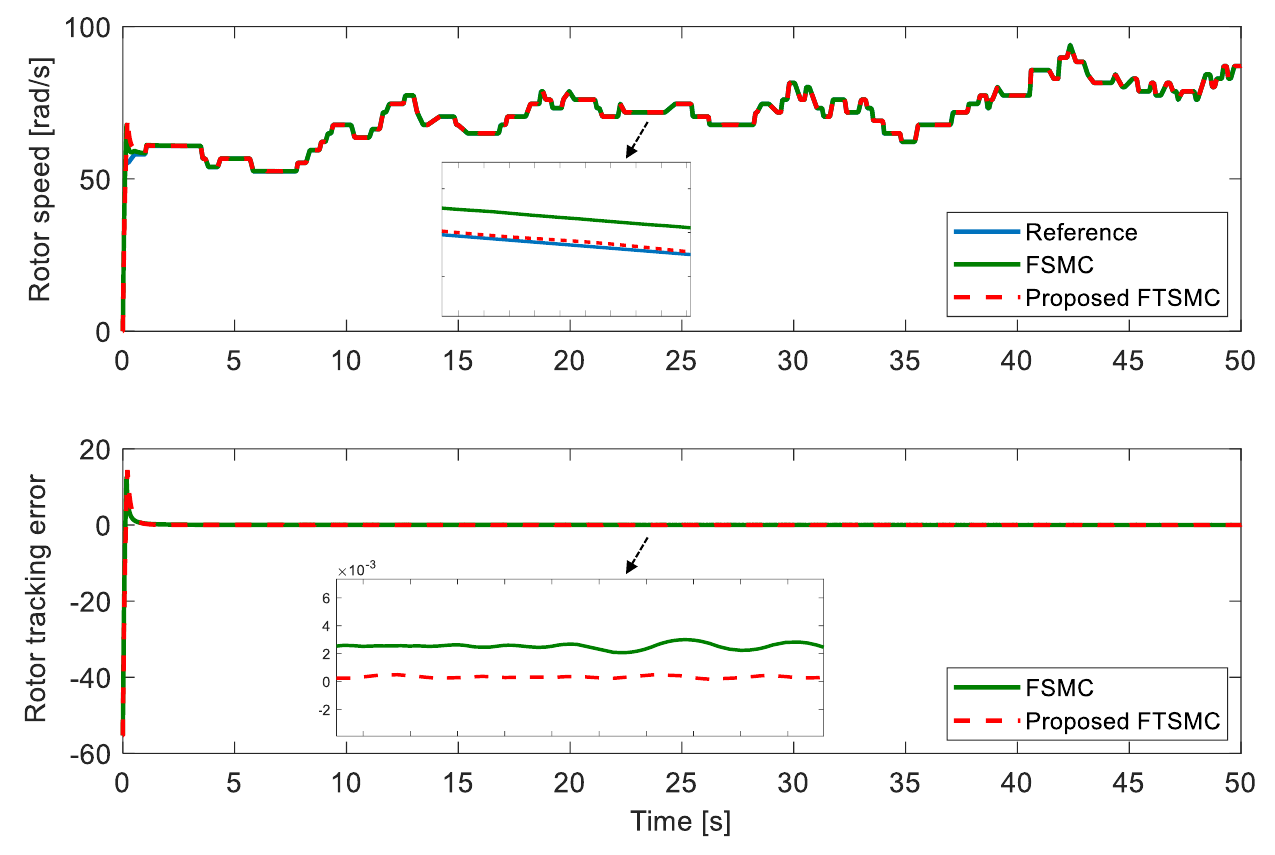}
\caption{Comparative rotor speed tracking performance.}
\label{fig:6-2}
\end{figure}

\begin{figure}
\centering
\includegraphics[width=4.5 in]{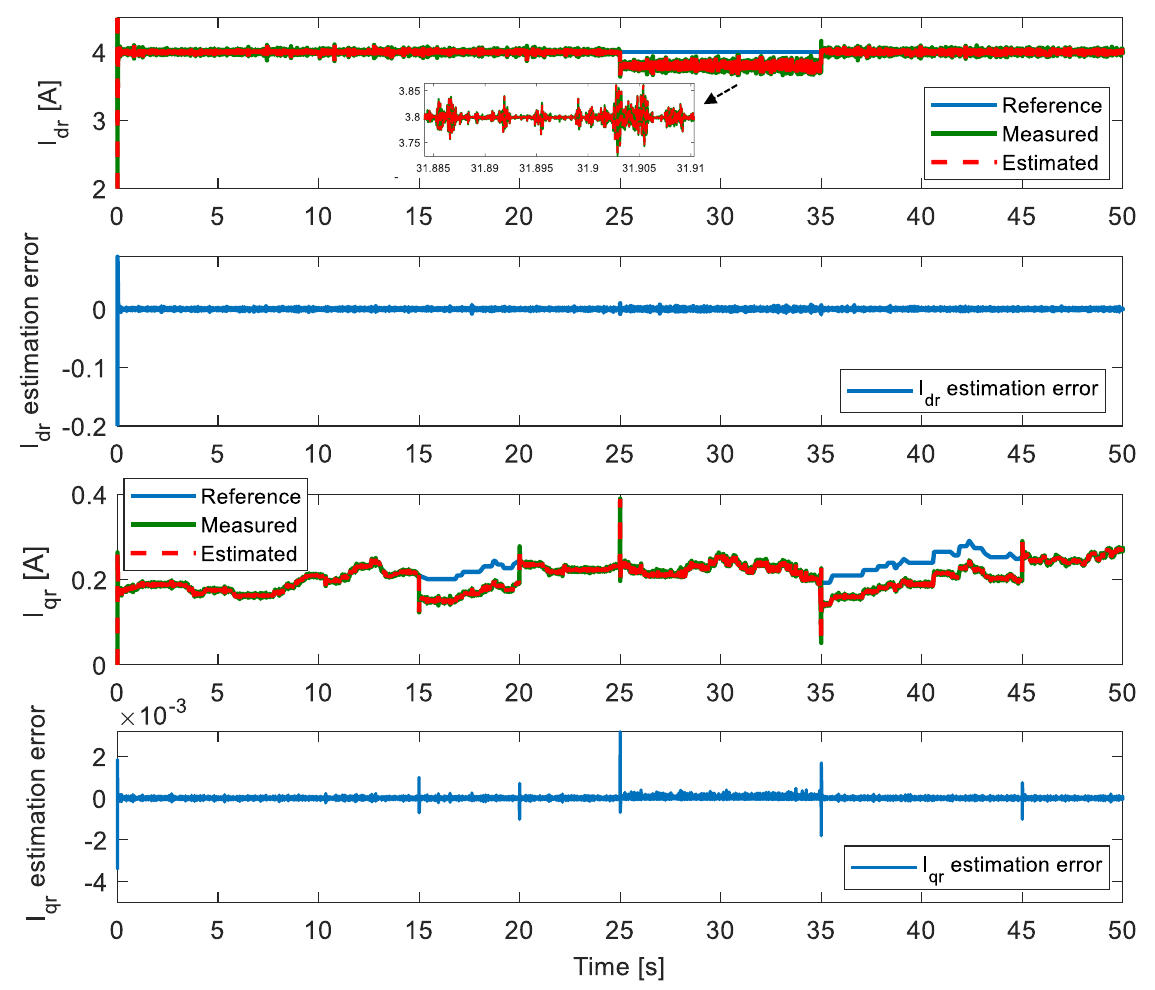}
\caption{Rotor current dynamics estimation.}
\label{fig:6-3}
\end{figure}

\begin{figure}
\centering
\includegraphics[width=4.5 in]{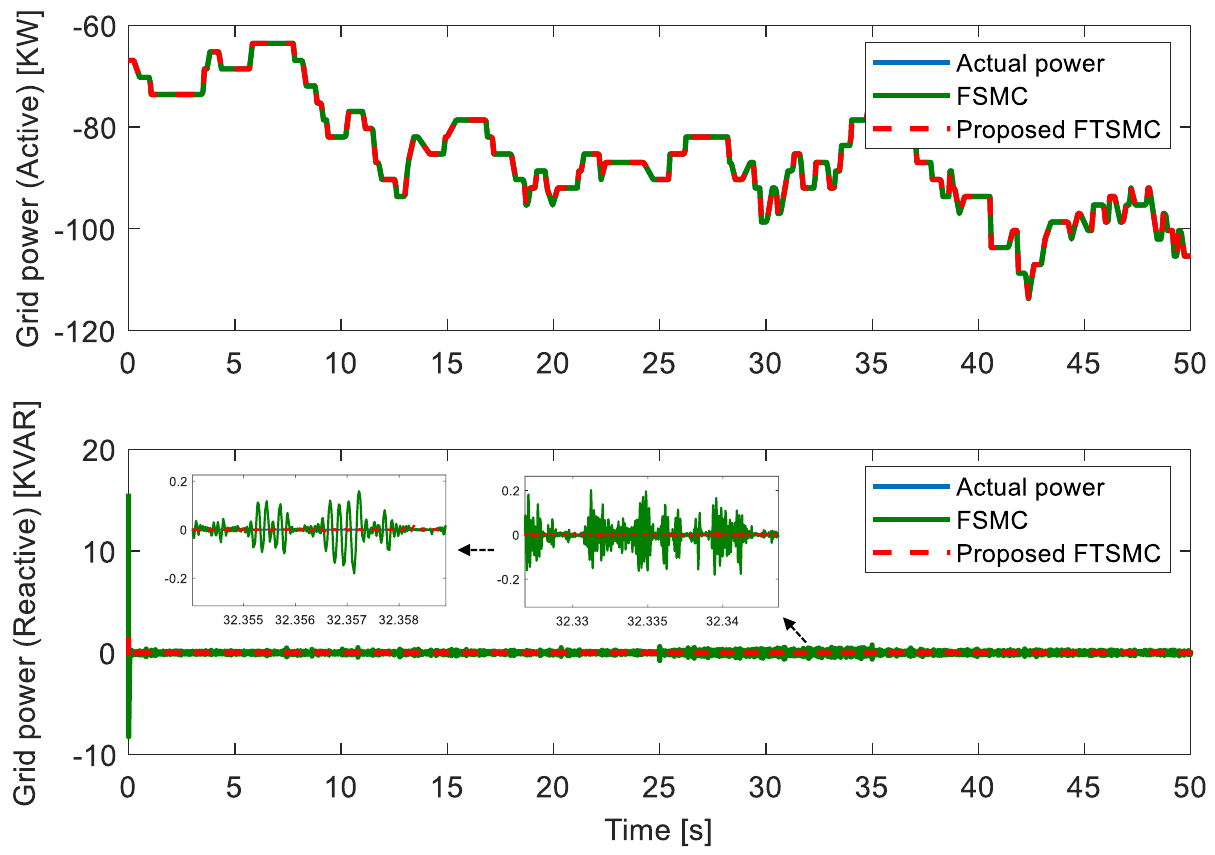}
\caption{Comparative active and reactive power tracking.}
\label{fig:6-4}
\end{figure}

\begin{figure}[ht!]
\centering
\includegraphics[width=4.5 in]{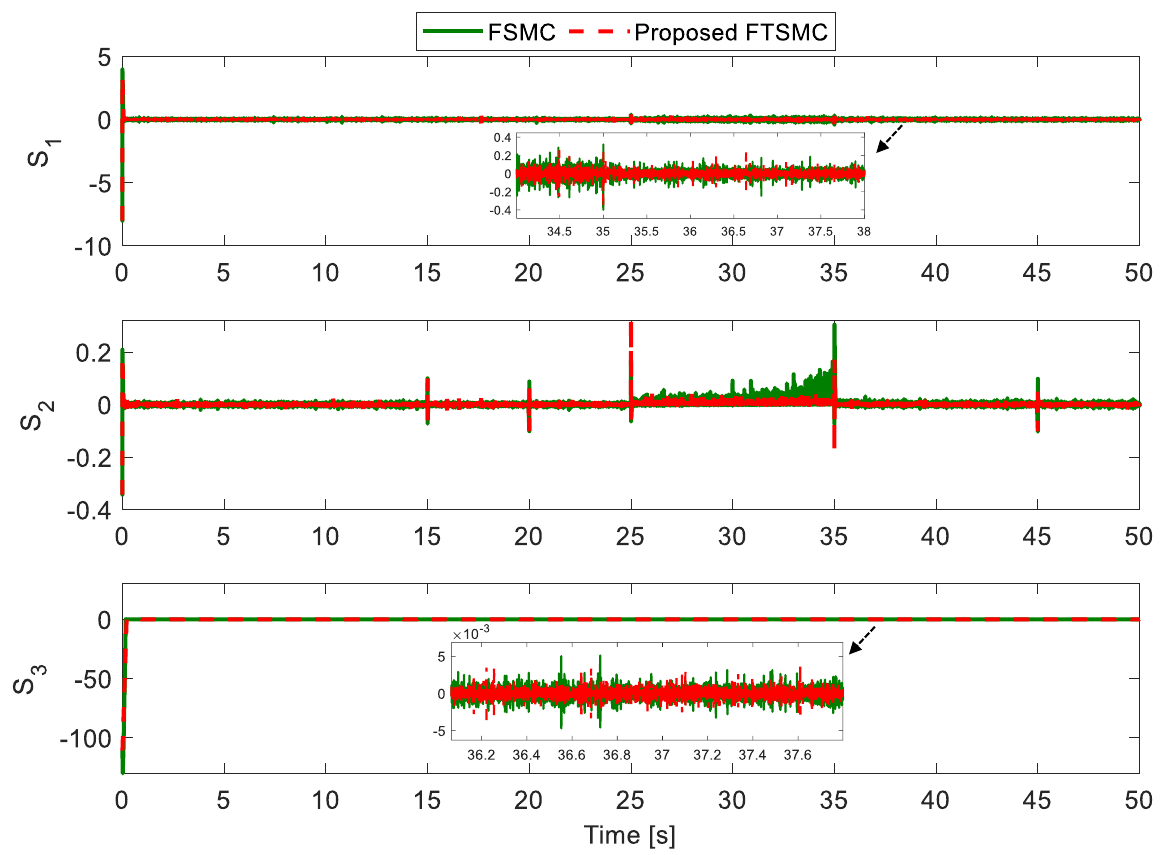}
\caption{Sliding surfaces.}
\label{fig:6-5}
\end{figure}

\section{Fault-Tolerant Control using Sliding Mode Observer for Current Estimation}
\label{sec:6.5}
In this section, first, a sliding mode observer is developed to estimate the rotor currents in the $dq$ reference frame. Then, the feasibility of the proposed control scheme is validated in comparison with the well-performed FSMC \cite{li2020mitigating} and the proposed ASE-based FTSMC approach through simulations.

\subsection{Sliding Mode Observer Design}
\label{sec:6.5.1}
In this section, a robust sliding mode observer is developed to estimate and reconstruct the rotor current in the $d-q$ reference frame during sensor faults.

\begin{assump}
\label{assump:6-1}
The system is subjected to slowly varying lumped disturbance $\xi^*=\Delta \mathcal A\mathcal I+\Delta \mathcal B \mathcal V + \Delta \mathcal D \mathcal V_s+d$, that satisfies the Lipschitz condition $\Vert \xi^*-\hat{\xi}^* \Vert \le J\Vert \mathcal I-\hat{\mathcal I} \Vert$.
\end{assump}
\begin{assump}
\label{assump:6-2}
It is assumed that a positive constant $\chi$ exists such that satisfies $\Vert \xi^*\Vert \le \chi$.
\end{assump}

Considering Assumption \ref{assump:6-1}, \eqref{Equation_6-16} can be rewritten as follows
\begin{equation}
\label{Equation_SMO1}
\dot{\mathcal I}=\mathcal A_n \mathcal I+\mathcal B_n \mathcal V+\mathcal D_n \mathcal V_s+\xi^*
\end{equation}
where $\mathcal I=[\mathcal I_{dr}\,\,\, \mathcal I_{qr}]^T$, $\mathcal V=[\mathcal V_{dr}\,\,\, \mathcal V_{qr}]^T$, and
\begin{align*}
\mathcal A_n&=\begin{bmatrix}
-\frac{\mathcal R_r}{\sigma \mathcal L_r} & s\omega_s \\
\frac{s\omega_s}{\mathcal L_r} & -\frac{\mathcal R_r}{\sigma \mathcal L_r}
\end{bmatrix}, \mathcal B_n=\left[\frac{1}{\sigma \mathcal L_r}\,\,\,\frac{1}{\sigma \mathcal L_r}\right]^T, \\ \nonumber
\mathcal D_n&=\left[0\,\,\,-\frac{s\mathcal L_m}{\sigma \mathcal L_r \mathcal L_s}\right]^T.
\end{align*}

The SMO is designed as
\begin{equation}
\label{Equation_SMO2}
\dot{\hat{\mathcal I}}=\mathcal A_n \hat{\mathcal I}+\mathcal B_n \mathcal V+\mathcal D_n \mathcal V_s+\mathcal N h\left(\varsigma\right)+\hat{\xi}^*,
\end{equation}
where $\varsigma=\mathcal I-\hat{\mathcal I}$ is the observer's sliding surface, $\hat{\mathcal I}=[\hat{\mathcal I}_{dr}\,\,\,\hat{\mathcal I}_{qr}]^T$, $h\left(\varsigma\right)=2/\left(1+e^{-a \varsigma}\right)-1$, $a>0$ is an arbitrary scalar, and $\mathcal N$ is a constant matrix with appropriate dimension.

Subtracting \eqref{Equation_SMO1} from \eqref{Equation_SMO2}, the error dynamic equation of the SMO can be obtained as follows
\begin{equation}
\label{Equation_SMO3}
\dot{\tilde{\mathcal I}}=\mathcal A_n \tilde{\mathcal I}-\mathcal N h\left(\varsigma\right)+\tilde{\xi}^*,
\end{equation}
where $\tilde{\mathcal I}=e=\mathcal I-\hat{\mathcal I}$ denotes the estimation error, and $\tilde{\xi}^*=\xi^*-\hat{\xi}^*$.

\begin{lem}
\label{lem:66-1}
\cite{gentle2007matrix} Matrices $M$ and $Q$ with proper dimensions exist, satisfying
\begin{equation}
\label{Equation_SMO4}
M^TQ+Q^TM\le \zeta M^TM+\zeta^{-1}Q^TQ.
\end{equation}
\end{lem}
\begin{lem}
\label{lem:66-2}
\cite{gentle2007matrix} Given the positive definite matrix $J$ and $\mathcal A_n$ being Hurwitz, the positive definite matrix $\mathcal P$ exists, satisfying
\begin{equation}
\label{Equation_SMO5}
\mathcal A^T_n \mathcal P+\mathcal P \mathcal A_n+\mu J^T J +\mu^{-1}\mathcal P^2 <0.
\end{equation}
\end{lem}

\begin{thm}
\label{thm:66-1}
Given an appropriate matrix $\mathcal{N}$ satisfying $\Vert \mathcal{N}\Vert = \Vert \chi \Vert$ yields locally bounded stability of the error system \eqref{Equation_SMO3}.
\end{thm}
\begin{proof}
Consider the Lyapunov function candidate as
\begin{equation}
\label{Equation_SMO5_1}
V_4=e^T \mathcal P e.
\end{equation}

The time-dependant derivation of \eqref{Equation_SMO5_1} yields
\begin{equation}
\label{Equation_SMO6}
\dot{V}_4=e^T\left(\mathcal{A}^T_n \mathcal{P}\right)e-2e^T \mathcal{N}h\left(e\right) +\tilde{\xi}^{*T} \mathcal{P}e+e^T\mathcal{P}\tilde{\xi}^*.
\end{equation}

From Lemma \ref{lem:66-2} one obtains
\begin{align}
\label{Equation_SMO7}
\tilde{\xi}^{*T} \mathcal{P}e+e^T\mathcal{P}\tilde{\xi}^* &\le \mu \tilde{\xi}^{*T}\tilde{\xi}^*+\mu^{-1}e^T\mathcal{P}^2 e \\ \nonumber
&=\mu\tilde{\xi}^{*2}+\mu^{-1}e^T\mathcal{P}^2e \le \mu J^2 \Vert \tilde{\xi}^* \Vert +\mu^{-1}e^T \mathcal{P}^2 e \\ \nonumber
&=e^T\left(\mu J^T J+\mu^{-1} \mathcal{P}^2\right)e.
\end{align}

Defining $\Psi=e^T\left(\mathcal{A}^T_n \mathcal{P}\right)e+\tilde{\xi}^{*T} \mathcal{P}e+e^T\mathcal{P}\tilde{\xi}^*$ one can obtain
\begin{equation}
\label{Equation_SMO8}
\Psi \le e^T\left(\mathcal{A}^T_n \mathcal{P}\right)e+e^T\left(\mu J^T J+\mu^{-1} \mathcal{P}^2\right)e <0,
\end{equation}
which results in $\dot{V}_4=\Psi-2e^T \mathcal{N}h\left(e\right) <0$. Accordingly, the locally bounded stability of the error system \eqref{Equation_SMO3} is achieved, and the developed SMO can perform the rotor currents estimation.
\end{proof}

Similar to Section \ref{sec:6.4}, let us define the difference between the measured and the estimated rotor current in the $d-q$ reference frame by $\mathcal R=|\tilde{I}|$. Flawlessly, $\mathcal R=0$ and $\mathcal R>0$ should denote the no-fault and faulty situations, respectively. However, false fault detections are inevitable due to unavoidable performance degradations in the current sensors over time, which can violate the above assumption in practice. To avoid this issue, a tolerance limit $\mathcal T$ can be defined for actual faults occurrence, where $\mathcal R \leq \mathcal T$ and $\mathcal R > \mathcal T$ illustrate the no-fault and faulty situations, respectively. As a result, the developed SMO can estimate and reconstruct the rotor current during sensors' faults. The stator flux linkage and voltage can be expressed in the stator reference frame as

\begin{equation}
\label{Equation_SMO9}
\begin{cases}
\varphi_{s\alpha}=\mathcal L_s i_{s\alpha} + \mathcal L_{m}i_{r\alpha}, \\
\varphi_{s\beta}=\mathcal L_s i_{s\beta} + \mathcal L_{m}i_{r\beta},
\end{cases}
\end{equation}
\begin{equation}
\label{Equation_SMO10}
\begin{cases}
v_{s\alpha}=\mathcal R_s i_{s\alpha}+\dot{\varphi}_{s\alpha}, \\ \nonumber
v_{s\beta}=\mathcal R_s i_{s\beta}+\dot{\varphi}_{s\beta}.
\end{cases}
\end{equation}

Considering the time derivative of \eqref{Equation_SMO9} in \eqref{Equation_SMO10} yields,
\begin{subequations}
\label{Equation_SMO11}
\begin{align}
\dot{i}_{r\alpha}&=\frac{1}{\mathcal L_m}v_{s\alpha}-\frac{\mathcal R_s}{\mathcal L_m}i_{s\alpha}-\frac{\mathcal L_s}{\mathcal L_m} \dot{i}_{s\alpha}, \\
\dot{i}_{r\beta}&=\frac{1}{\mathcal L_m}v_{s\beta}-\frac{\mathcal R_s}{\mathcal L_m}i_{s\beta}-\frac{\mathcal L_s}{\mathcal L_m} \dot{i}_{s\beta}.
\end{align}
\end{subequations}

Considering \eqref{Equation_SMO11}, the rotor current can be approximated in $\alpha-\beta$ frame as
\begin{subequations}
\label{Equation_SMO12}
\begin{align}
\hat{i}_{r\alpha}&=\frac{1}{\mathcal L_m}\int v_{s\alpha}-\frac{\mathcal R_s} {\mathcal L_m}\int i_{s\alpha}-\frac{\mathcal L_s}{\mathcal L_m} i_{s\alpha}, \\
\hat{i}_{r\beta}&=\frac{1}{\mathcal L_m}\int v_{s\beta}-\frac{\mathcal R_s}{\mathcal L_m}\int i_{s\beta}-\frac{\mathcal L_s}{\mathcal L_m} i_{s\beta}.
\end{align}
\end{subequations}

As apparent from \eqref{Equation_SMO12}, the approximations of rotor current are dependant on the stator currents and voltage measurements. As a result, if $\mathfrak R>\mathcal Q$, the current control law $u_{eq}$ is derived from the rotor current approximation \eqref{Equation_SMO12}.

\subsection{Simulation Results}
\label{sec:6.5.2}
The effectiveness of the proposed SMO-based active fault-tolerant FTSMC scheme for the rotor current regulation and speed trajectory tracking of a 1.5 MW DFIG-based WT is tested with respect to FSMC \cite{li2020mitigating} and algebraic state-observer-based FTSMC \cite{mousavi2022active} approaches through simulations in the MATLAB/Simulink platform. Multiple fault occurrences are assumed to happen during the time intervals of 35-50 s and 120-140 s for $I_{dr}$, and 15-35 s, 65-75 s, and 150-170 s for $I_{qr}$ while evaluating the controllers. Furthermore, similar to Section \ref{sec:6.4} the lumped uncertainty $d=50\% f+50\% hu$ is considered to affect the system. Figure \ref{fig:66-2} shows the wind profile with an average speed of 5.2 \si{m/s} within the speed range of 3.8-5.8 \si{m/s}. Figure \ref{fig:66-3} shows a comparative illustration of rotor speed tracking. It can be seen from Figure \ref{fig:66-3} that all three approaches demonstrate a similar tracking performance. However, the zoomed-in insets reveal the superior performance of the proposed SMO-AFTSMC approach, with more precise trajectory tracking performance. Furthermore, much less tracking error of the proposed scheme compared with other methods can be observed, indicating its superior performance.

\begin{figure}
\centering
\includegraphics[width=4.2 in]{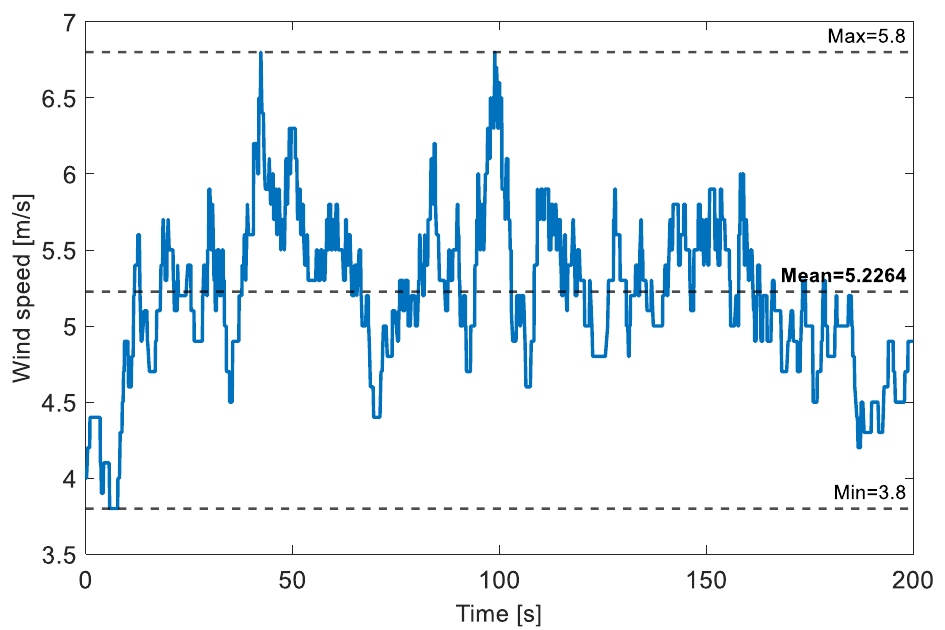}
\caption{Wind speed profile (200 s).}
\label{fig:66-2}
\end{figure}

\begin{figure}
\centering
\includegraphics[width=4.5 in]{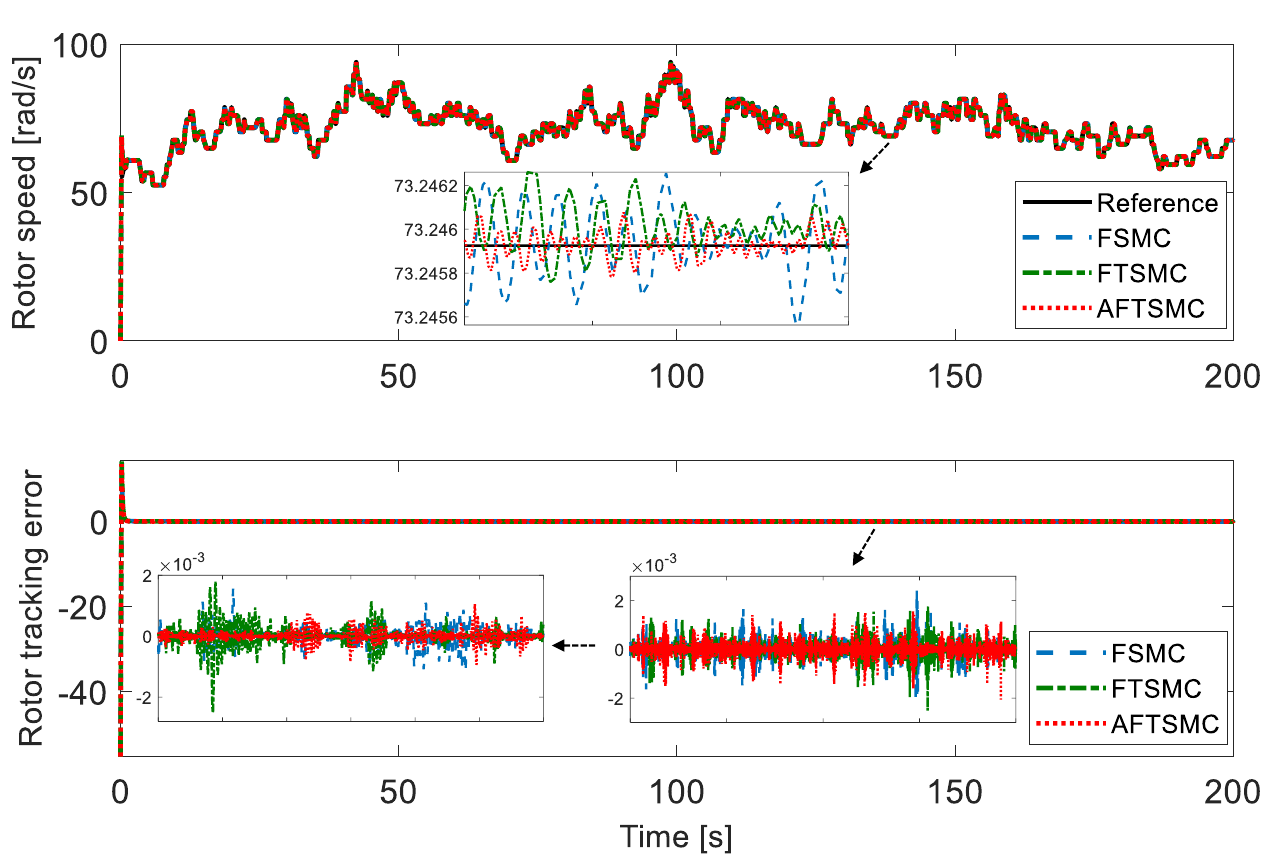}
\caption{Comparative rotor speed tracking performance.}
\label{fig:66-3}
\end{figure}

Figures \ref{fig:66-4} and \ref{fig:66-5} depict the rotor current dynamics estimation and reconstruction performance of the developed SMO in the presence of faulty sensors and lumped uncertainties. As observed from Figs. \ref{fig:66-4} and \ref{fig:66-5}, the estimation and reconstruction of $I_{dr}$ and $I_{qr}$ is desirably fulfilled in a very short time with minimal error, demonstrating its remarkable performance. In addition, the provided comparisons with the algebraic state observer (ASE) \cite{mousavi2022active} demonstrate the superiority of the developed SMO with less estimation error.

\begin{figure}
\centering
\includegraphics[width=4.5 in]{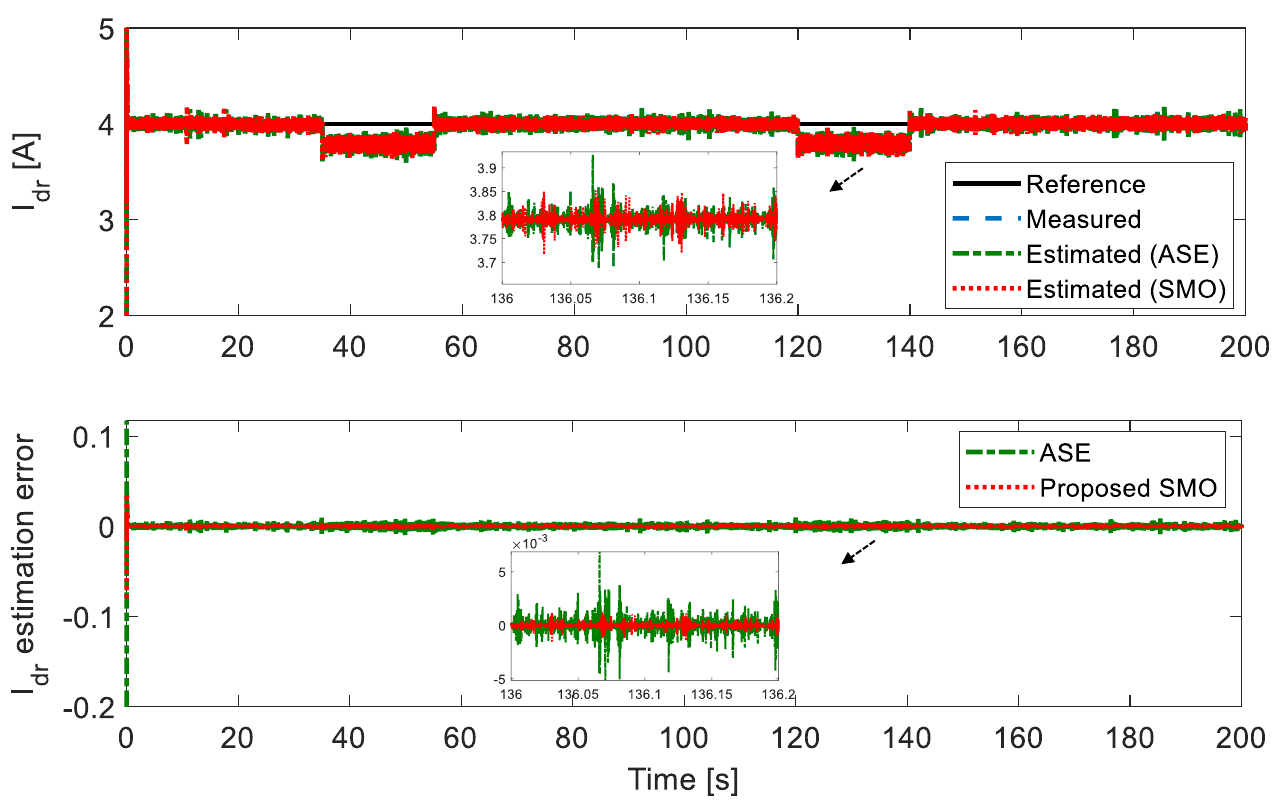}
\caption{Rotor current dynamics estimation ($I_{dr}$).}
\label{fig:66-4}
\end{figure}

\begin{figure}[ht!]
\centering
\includegraphics[width=4.5 in]{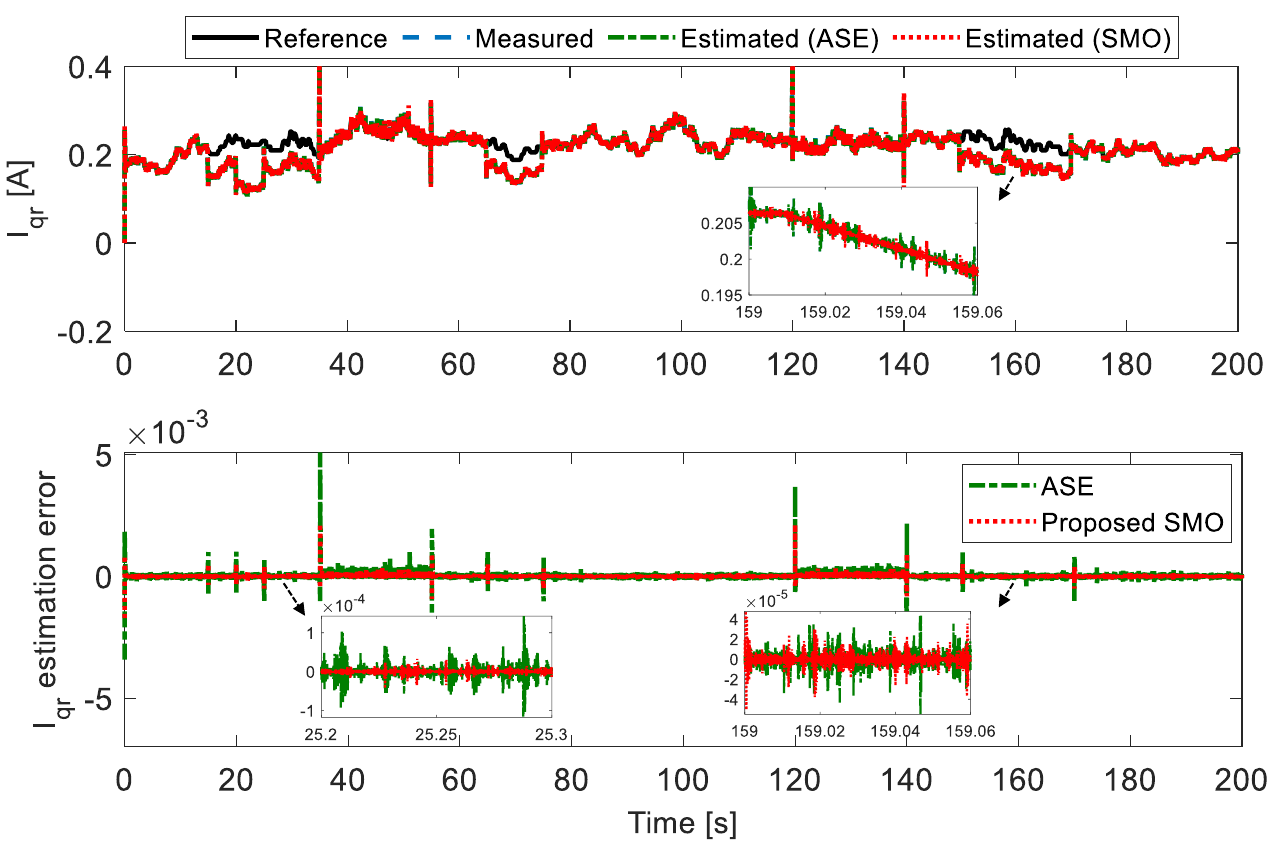}
\caption{Rotor current dynamics estimation ($I_{qr}$).}
\label{fig:66-5}
\end{figure}

The grid active and reactive power tracking performances are shown in Figure \ref{fig:66-6}. Accordingly, one can observe that all the control schemes have successfully carried out the grid active power tracking. However, as clearly seen from the zoomed-in insets, with a precise power tracking performance and a small transient response, the proposed SMO-AFTSMC has delivered a more improved power factor with superior reactive power regulation, surpassing the ASE-FTSMC and FSMC approaches. Figure \ref{fig:66-7} depicts the comparative sliding surfaces of the current and speed FSMC and FTSMC controllers. In this sense, although both controllers deliver bounded sliding surfaces, the proposed FTSMC provides better chattering mitigation compared to the FSMC. From the obtained foregoing results, it is evident that the developed SMO was found to deliver superior performance with respect to the ASE observer. In addition, the proposed fault-tolerant control scheme offers superior speed and power tracking performance compared to the FSMC approach.

\begin{figure}[ht!]
\centering
\includegraphics[width=4.5 in]{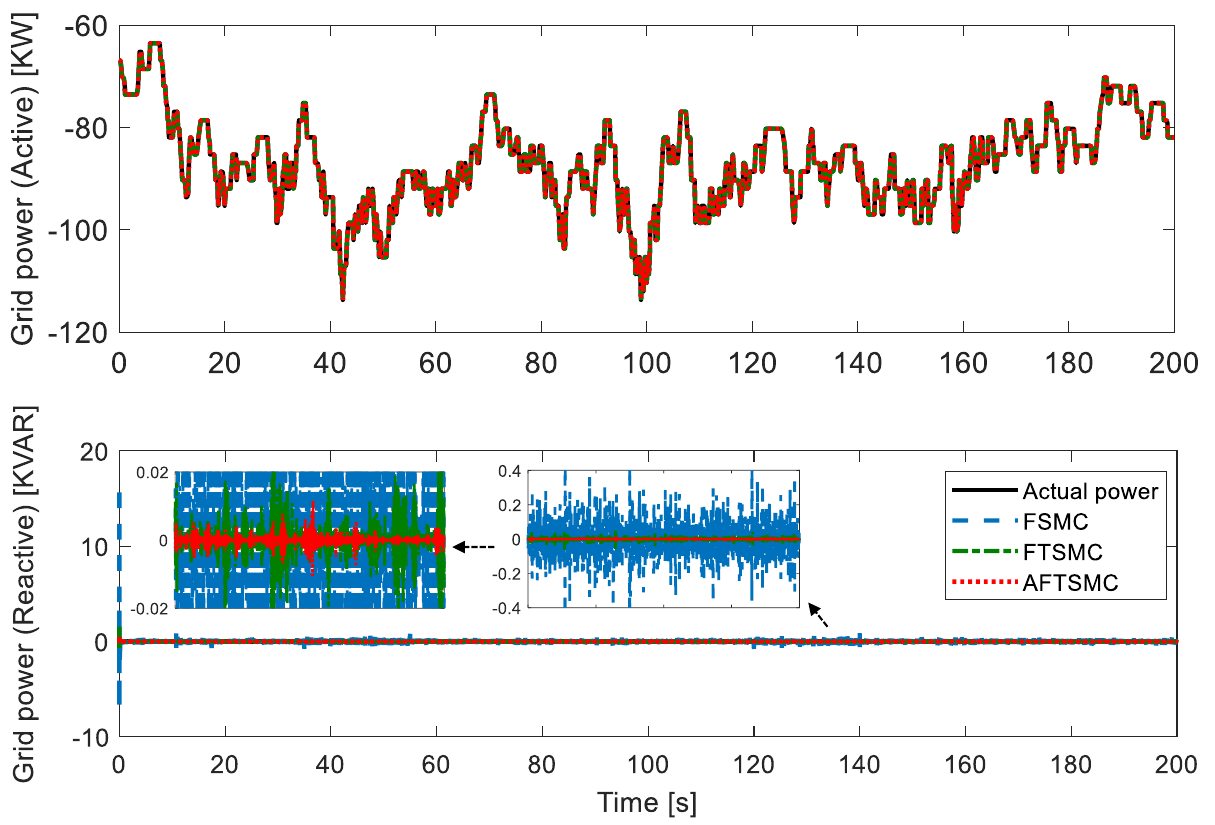}
\caption{Comparative active and reactive power tracking.}
\label{fig:66-6}
\end{figure}

\begin{figure}[ht!]
\centering
\includegraphics[width=4.5 in]{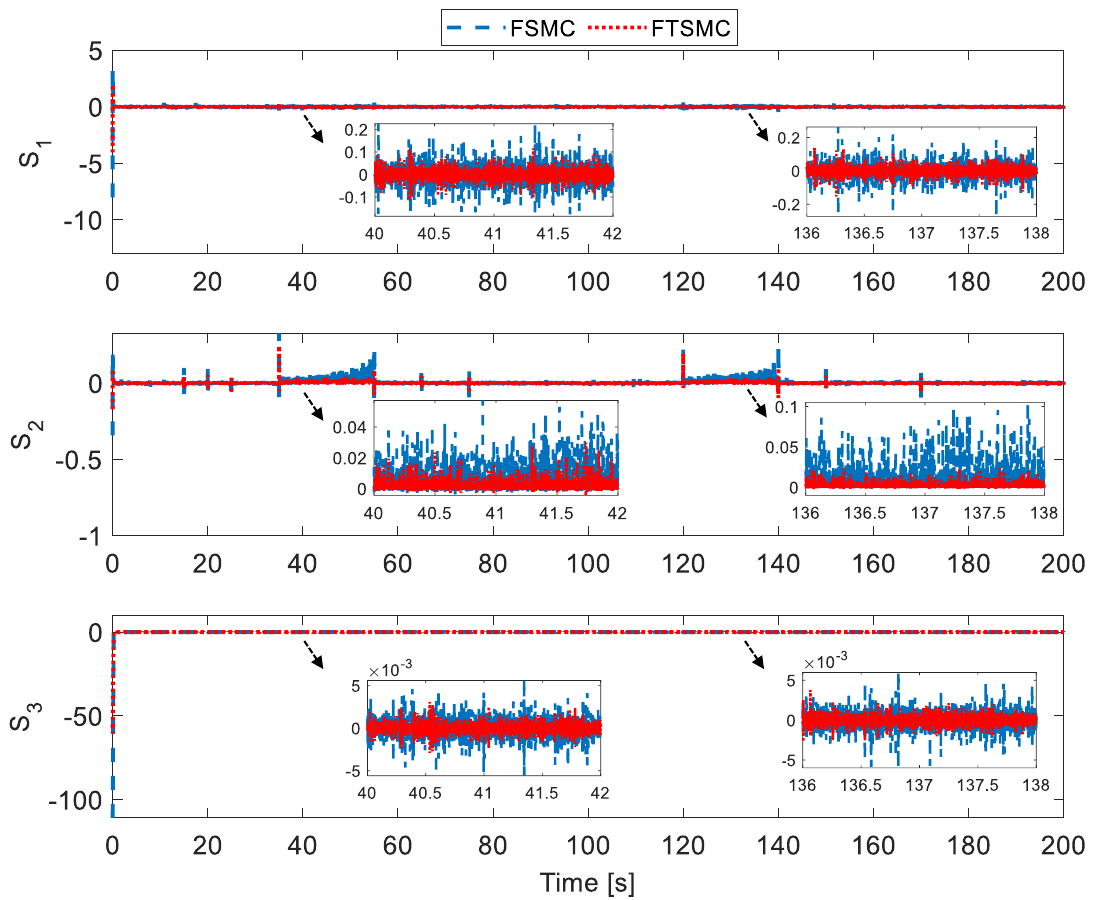}
\caption{Comparative current and speed controllers' sliding surfaces.}
\label{fig:66-7}
\end{figure}

\section{Conclusions}
\label{sec:6.6}
This chapter proposed two active fault-tolerant control schemes for rotor side converter of DFIG-based WECSs subjected to model uncertainties and sensor faults. The control schemes comprised two fractional-order nonsingular terminal sliding mode controllers with alleviated chattering and finite-time reaching time for rotor current regulation and speed trajectory tracking. An algebraic state observer was first incorporated with the proposed controllers; where benefitting from the augmented state observer, the precise fault reconstruction allowed to achieve remarkable fault-tolerant control in the presence of sensor faults, with similar behavior as the fault-free case. Then a sliding mode observer was developed to deal with the same problem, demonstrating a superior estimation performance to that of the algebraic state observer. The proposed control techniques' performance under faulty situations has been testified in comparison with the fractional-order SMC approach. According to the simulation results, the proposed fault-tolerant control schemes demonstrated impressive speed and power tracking performance with reduced chattering.

\chapter{Conclusions and Future Works} 

\label{Chapter7} 
In this chapter, first, a summary of the contributions is presented, and then recommendations for the future works are discussed.

\section{Conclusions and Contributions}
\label{sec:7.1}
The contributions of this thesis are twofold. On the one hand, various novel control approaches have been proposed to tackle some control problems of wind turbine systems, such as blade pitch control, rotor speed control and maximum power extraction, and doubly-fed induction generator (DFIG) -driven wind energy conversion systems (WECSs) power control in the presence of various fault scenarios. On the other hand, the main deficiencies of the conventional sliding mode control (SMC) approach have been tackled by developing novel higher-order and fractional-order SMCs to mitigate the chattering problem, increase the convergence speed, achieve better tracking precision, and prevent unnecessarily large control signals. The main contributions can be expressed as follows:

In Chapter \ref{Chapter4}, a new fault-tolerant pitch control strategy was proposed to effectively adjust the blade pitch angle of wind turbines (WTs) and maintain improvement in power generation performance of WT in the presence of sensor, actuator, and system faults. In this respect, taking advantage of the memory characteristics of fractional calculus, a fractional-order PID (FOPID) controller with extended memory of pitch angles was proposed. To tune the developed controller's parameters, a novel optimization algorithm, namely the dynamic weighted parallel firefly algorithm (DWPFA), was also proposed, where its optimization performance was extensively investigated and validated. Comparative simulations in fault-free and faulty conditions were conducted, and the effectiveness of the proposed fault-tolerant pitch control scheme in mitigating the fault effects and WT power generation improvements were demonstrated. A limitation for the realization of the proposed control scheme is a slight increase in the computational complexity due to the memory requirements based on the fractional-order operators and the higher number of parameters that must be tuned compared to the conventional controllers. Since approximations must be considered to implement such controllers, fractional-order operators' implementations are relatively complex and costly compared to their integer-order counterparts. However, although the conventional PI and the optimal FOPID approaches provide lower computational complexities with respect to the proposed FOPID with extended memory, prioritizing more power generation and better fault-tolerant performances will make the proposed method a more preferred candidate providing a viable solution that can be implemented with ease without the needs of costly or complex hardware installments.

In Chapter \ref{Chapter5}, a fractional-order nonsingular terminal sliding mode controller (FNTSMC) was proposed to address the maximum power extraction problem of WECSs subjected to actuator fault in the form of a partial loss on the generator's torque. The proposed FNTSMC demonstrated its remarkable performance in alleviating the chattering problem of conventional SMCs and enhanced the finite-time convergence speed of system states. Comparative performance investigations with conventional SMC and second-order fast SMC approaches were provided, and the notable optimal rotor speed tracking performance and power maximization of the proposed FNTSMC with fewer fluctuations and faster transient response were demonstrated.

In Chapter \ref{Chapter6}, active fault-tolerant nonlinear control schemes were proposed for rotor-side converter (RSC) control of DFIG-based WECSs in the presence of model uncertainties and sensor faults. The control schemes comprised an FNTSMC to regulate the rotor current and another FNTSMC for the speed trajectory tracking task, where the proposed controllers demonstrated desirable chattering mitigation performance. In addition, the control scheme was first augmented with an algebraic state estimator and then a robust SMO for rotor current estimation and reconstruction during sensor faults. It was found that the developed SMO can estimate and reconstruct the rotor current during sensors' faults with a high level of reliability. Performance investigations with respect to fractional-order SMC have been provided, and the remarkable fault-tolerant performance of the proposed control scheme was revealed with similar behavior as the fault-free case.

\section{Recommendations for Future Works}
\label{sec:7.2}

Following is a list of possible directions for future research identified in this thesis:

\begin{enumerate}
  \item Achieving the desired WECS control performance does not entirely depend on the developed controller structure. Proper tuning of the controller parameters and appropriate and realistic WECS modelling also plays a critical role in the performance validation of the developed control schemes. Furthermore, the controller parameters are often precomputed offline by trial-and-error, thus preventing them from consistently delivering the best performance, especially for WECS with varying operating conditions. Hence, optimal tuning procedures using optimization algorithms and adaptive soft computing-based methods such as fuzzy, neural networks, and learning approaches can result in more precise parameters and better control performances. Furthermore, barrier function (BF) -based SMC approaches have proven their remarkable performance in forcing the state trajectories to converge to a predefined neighborhood of zero in finite time without knowing the upper bound of disturbances \cite{mobayen2022barrier,laghrouche2021barrier}. The BF-based SMC approaches have been found to deliver effective approximations of the external disturbances, yielding a more stable closed-loop system. Accordingly, considering the WECS exposure to various disturbances, incorporating BFs with other modified SMCs can yield even better performances. This open research area is still to be investigated through WECS control.
  \item Despite the remarkable performance of SMCs in dealing with WECS control problems, these model-based methodologies require \textit{a priori} knowledge of the system dynamics. The importance of this issue will be even more salient when the WECS is dealing with faults and external disturbances, which are utterly inevitable in WECS. According to this thesis's investigations, only 8\% of the current literature studies have utilized observers, which demands more investigations in this field. In this respect, various observers such as higher-order observers, extended state observers, and sliding mode observers can be employed with the SMC-based controllers to estimate and reconstruct the system dynamics, especially when dealing with faults.
  \item Offshore wind farms are becoming increasingly popular since wind conditions are less turbulent at sea, and it is more economical to install them there. However, due to the extra dynamics introduced by floating platforms employed by the turbines in the deep sea, WT control has been a more challenging task to deal with. In this regard, model predictive control approaches can be counted as efficient solutions since they can handle the constraint requirements of the floating platform in real-time.
  \item Data-driven methods can serve as efficient ways to attain control objectives, as they do not rely on the mathematical model of the system and depend merely on the input/output (I/O) data. As subsets of data-driven approaches, reinforcement learning and deep learning methods can be incorporated with SMCs. However, artificial neural networks that are deep learning algorithms have already been implemented with SMCs and have shown their capabilities to increase the controller's performance. It is worth noting that deep learning methods learn from a training set and then apply the learning to a new data set, while reinforcement learning approaches dynamically learn by adjusting actions based on continuous feedback from the system. Accordingly, reinforcement learning approaches can be a more efficient breakthrough to enhance the SMC-based controllers' performance in dealing with WECS control problems with varying operating conditions.
  \item Data-driven approaches perform real-time analytics, learning, and decision-making based on a massive amount of data, including various sets of sensor data, information about the surrounding environment and weather conditions, and wind speed estimations. Hence, having access to powerful processing resources to run predictive or learning methods can effectively enhance the data processing, reduce the computational burden, and eventually improve the controller's performance. This opens a new window to the cloud-based (internet of things) data-driven sliding mode control of wind energy conversion systems, which is welcomed to be investigated in the future.
\end{enumerate}


\appendix 




\printbibliography[heading=bibintoc]

@article{taher2018new,
  title={A new approach using combination of sliding mode control and feedback linearization for enhancing fault ride through capability of DFIG-based WT},
  author={Taher, Seyed Abbas and Dehghani Arani, Zahra and Rahimi, Mohsen and Shahidehpour, Mohammad},
  journal={International Transactions on Electrical Energy Systems},
  volume={28},
  number={10},
  pages={e2613},
  year={2018},
  publisher={Wiley Online Library}
}

@article{deraz2013new,
  title={A new control strategy for a stand-alone self-excited induction generator driven by a variable speed wind turbine},
  author={Deraz, SA and Kader, FE Abdel},
  journal={Renewable Energy},
  volume={51},
  pages={263--273},
  year={2013},
  publisher={Elsevier}
}

@article{patel2021nonlinear,
  title={Nonlinear rotor side converter control of DFIG based wind energy system},
  author={Patel, Ravi and Hafiz, Faizal and Swain, Akshya and Ukil, Abhisek},
  journal={Electric Power Systems Research},
  volume={198},
  pages={107358},
  year={2021},
  publisher={Elsevier}
}

@article{yang2019pcsmc,
  title={PCSMC design of permanent magnetic synchronous generator for maximum power point tracking},
  author={Yang, Bo and Zhong, Linen and Yu, Tao and Shu, Hongchun and Cao, Pulin and An, Na and Sang, Yiyan and Jiang, Lin},
  journal={IET Generation, Transmission \& Distribution},
  volume={13},
  number={14},
  pages={3115--3126},
  year={2019},
  publisher={Wiley Online Library}
}

@article{merabet2016implementation,
  title={Implementation of sliding mode control system for generator and grid sides control of wind energy conversion system},
  author={Merabet, Adel and Ahmed, Khandker Tawfique and Ibrahim, Hussein and Beguenane, Rachid},
  journal={IEEE Transactions on Sustainable Energy},
  volume={7},
  number={3},
  pages={1327--1335},
  year={2016},
  publisher={IEEE}
}

@article{yin2020recurrent,
  title={Recurrent neural network based adaptive integral sliding mode power maximization control for wind power systems},
  author={Yin, Xiuxing and Jiang, Zhansi and Pan, Li},
  journal={Renewable Energy},
  volume={145},
  pages={1149--1157},
  year={2020},
  publisher={Elsevier}
}

@article{djilali2019real,
  title={Real-time neural sliding mode field oriented control for a DFIG-based wind turbine under balanced and unbalanced grid conditions},
  author={Djilali, Larbi and Sanchez, Edgar N and Belkheiri, Mohammed},
  journal={IET Renewable Power Generation},
  volume={13},
  number={4},
  pages={618--632},
  year={2019},
  publisher={Wiley Online Library}
}

@article{hussain2019efficient,
  title={An efficient wind speed computation method using sliding mode observers in wind energy conversion system control applications},
  author={Hussain, Jakeer and Mishra, Mahesh K},
  journal={IEEE Transactions on Industry Applications},
  volume={56},
  number={1},
  pages={730--739},
  year={2019},
  publisher={IEEE}
}

@article{subramaniyam2021memory,
  title={Memory-based ISMC design of DFIG-based wind turbine model via T-S fuzzy approach},
  author={Subramaniyam, Ramasamy and Joo, Young Hoon},
  journal={IET Control Theory \& Applications},
  volume={15},
  number={3},
  pages={348--359},
  year={2021},
  publisher={Wiley Online Library}
}

@inproceedings{mousavi2022active,
  title={Active Fault-tolerant Fractional-order Terminal Sliding Mode Control for DFIG-based Wind Turbines Subjected to Sensor Faults},
  author={Mousavi, Yashar and Bevan, Geraint and Kucukdemiral, Ibrahim Beklan and Fekih, Afef},
  booktitle={IEEE IAS Global Conference on Emerging Technologies (GlobConET'22)},
  year={2022},
  organization={IEEE}
}

@article{kong2022nonlinear,
  title={Nonlinear MPC for DFIG-based wind power generation under unbalanced grid conditions},
  author={Kong, Xiaobing and Wang, Xuan and Abdelbaky, Mohamed Abdelkarim and Liu, Xiangjie and Lee, Kwang Y},
  journal={International Journal of Electrical Power \& Energy Systems},
  volume={134},
  pages={107416},
  year={2022},
  publisher={Elsevier}
}

@article{abubakar2020induction,
  title={Induction motor fault detection based on multi-sensory control and wavelet analysis},
  author={Abubakar, Ukashatu and Mekhilef, Saad and Gaeid, Khalaf S and Mokhlis, Hazlie and Al Mashhadany, Yousif I},
  journal={IET Electric Power Applications},
  volume={14},
  number={11},
  pages={2051--2061},
  year={2020},
  publisher={IET}
}

@article{sebtahmadi2016current,
  title={A current control approach for an abnormal grid supplied ultra sparse Z-Source matrix converter with a particle swarm optimization proportional-integral induction motor drive controller},
  author={Sebtahmadi, Seyed Sina and Borhan Azad, Hanieh and Islam, Didarul and Seyedmahmoudian, Mehdi and Horan, Ben and Mekhilef, Saad},
  journal={Energies},
  volume={9},
  number={11},
  pages={899},
  year={2016},
  publisher={Multidisciplinary Digital Publishing Institute}
}

@article{zaihidee2019application,
  title={Application of fractional order sliding mode control for speed control of permanent magnet synchronous motor},
  author={Zaihidee, Fardila M and Mekhilef, Saad and Mubin, Marizan},
  journal={IEEE Access},
  volume={7},
  pages={101765--101774},
  year={2019},
  publisher={IEEE}
}

@article{lv2018firefly,
  title={The firefly algorithm with Gaussian disturbance and local search},
  author={Lv, Li and Zhao, Jia},
  journal={Journal of Signal Processing Systems},
  volume={90},
  number={8},
  pages={1123--1131},
  year={2018},
  publisher={Springer}
}

@article{li2020mitigating,
  title={Mitigating subsynchronous control interaction using fractional sliding mode control of wind farm},
  author={Li, Penghan and Wang, Jie and Xiong, Linyun and Ma, Meiling and Wang, Ziqiang and Huang, Sunhua},
  journal={Journal of the Franklin Institute},
  volume={357},
  number={14},
  pages={9523--9542},
  year={2020},
  publisher={Elsevier}
}

@article{fister2013comprehensive,
  title={A comprehensive review of firefly algorithms},
  author={Fister, Iztok and Fister Jr, Iztok and Yang, Xin-She and Brest, Janez},
  journal={Swarm and Evolutionary Computation},
  volume={13},
  pages={34--46},
  year={2013},
  publisher={Elsevier}
}

@article{nayak2021hyper,
  title={Hyper-parameter tuned light gradient boosting machine using memetic firefly algorithm for hand gesture recognition},
  author={Nayak, Janmenjoy and Naik, Bighnaraj and Dash, Pandit Byomakesha and Souri, Alireza and Shanmuganathan, Vimal},
  journal={Applied Soft Computing},
  volume={107},
  pages={107478},
  year={2021},
  publisher={Elsevier}
}

@article{wang2017firefly,
  title={Firefly algorithm with neighborhood attraction},
  author={Wang, Hui and Wang, Wenjun and Zhou, Xinyu and Sun, Hui and Zhao, Jia and Yu, Xiang and Cui, Zhihua},
  journal={Information Sciences},
  volume={382},
  pages={374--387},
  year={2017},
  publisher={Elsevier}
}

@article{altabeeb2019improved,
  title={An improved hybrid firefly algorithm for capacitated vehicle routing problem},
  author={Altabeeb, Asma M and Mohsen, Abdulqader M and Ghallab, Abdullatif},
  journal={Applied Soft Computing},
  volume={84},
  pages={105728},
  year={2019},
  publisher={Elsevier}
}

@article{ma2015optimal,
  title={Optimal real-time control of wind turbine during partial load operation},
  author={Ma, Zheren and Yan, Zeyu and Shaltout, Mohamed L and Chen, Dongmei},
  journal={IEEE Transactions on Control Systems Technology},
  volume={23},
  number={6},
  pages={2216--2226},
  year={2015},
  publisher={IEEE}
  doi={10.1109/TCST.2015.2410735},
}

@article{li2017adaptive,
  title={Adaptive fault-tolerant control of wind turbines with guaranteed transient performance considering active power control of wind farms},
  author={Li, Dan-Yong and Li, Peng and Cai, Wen-Chuan and Song, Yong-Duan and Chen, Hou-Jin},
  journal={IEEE Transactions on Industrial Electronics},
  volume={65},
  number={4},
  pages={3275--3285},
  year={2017},
  publisher={IEEE}
  doi={10.1109/TIE.2017.2748036},
}

@article{madsen2020experimental,
  title={Experimental analysis of the scaled DTU10MW TLP floating wind turbine with different control strategies},
  author={Madsen, FJ and Nielsen, TRL and Kim, T and Bredmose, H and Pegalajar-Jurado, A and Mikkelsen, RF and Lomholt, AK and Borg, M and Mirzaei, M and Shin, P},
  journal={Renewable Energy},
  year={2020},
  publisher={Elsevier}
  doi={10.1016/j.renene.2020.03.145},
}

@article{ren2016nonlinear,
  title={Nonlinear PI control for variable pitch wind turbine},
  author={Ren, Yaxing and Li, Liuying and Brindley, Joseph and Jiang, Lin},
  journal={Control Engineering Practice},
  volume={50},
  pages={84--94},
  year={2016},
  publisher={Elsevier}
  doi={},
}

@article{zhang2015load,
  title={Load mitigation of unbalanced wind turbines using PI-R individual pitch control},
  author={Zhang, Yunqian and Cheng, Ming and Chen, Zhe},
  journal={IET Renewable Power Generation},
  volume={9},
  number={3},
  pages={262--271},
  year={2014},
  publisher={IET}
  doi={},
}

@article{van2015advanced,
  title={Advanced pitch angle control based on fuzzy logic for variable-speed wind turbine systems},
  author={Van, Tan Luong and Nguyen, Thanh Hai and Lee, Dong-Choon},
  journal={IEEE Transactions on Energy Conversion},
  volume={30},
  number={2},
  pages={578--587},
  year={2015},
  publisher={IEEE}
  doi={},
}

@article{venkaiah2020hydraulically,
  title={Hydraulically actuated horizontal axis wind turbine pitch control by model free adaptive controller},
  author={Venkaiah, P and Sarkar, Bikash K},
  journal={Renewable Energy},
  volume={147},
  pages={55--68},
  year={2020},
  publisher={Elsevier}
  doi={},
}

@article{asgharnia2018performance,
  title={Performance and robustness of optimal fractional fuzzy PID controllers for pitch control of a wind turbine using chaotic optimization algorithms},
  author={Asgharnia, Amirhossein and Shahnazi, Reza and Jamali, Ali},
  journal={ISA transactions},
  volume={79},
  pages={27--44},
  year={2018},
  publisher={Elsevier}
  doi={},
}

@article{badihi2020fault,
  title={Fault-Tolerant Individual Pitch Control for Load Mitigation in Wind Turbines with Actuator Faults},
  author={Badihi, Hamed and Zhang, Youmin and Pillay, Pragasen and Rakheja, Subhash},
  journal={IEEE Transactions on Industrial Electronics},
  year={2020},
  publisher={IEEE}
  doi={},
}

@article{habibi2018adaptive,
  title={Adaptive PID control of wind turbines for power regulation with unknown control direction and actuator faults},
  author={Habibi, Hamed and Nohooji, Hamed Rahimi and Howard, Ian},
  journal={IEEE Access},
  volume={6},
  pages={37464--37479},
  year={2018},
  publisher={IEEE}
  doi={},
}

@article{lan2018fault,
  title={Fault-tolerant wind turbine pitch control using adaptive sliding mode estimation},
  author={Lan, Jianglin and Patton, Ron J and Zhu, Xiaoyuan},
  journal={Renewable Energy},
  volume={116},
  pages={219--231},
  year={2018},
  publisher={Elsevier}
  doi={},
}

@article{cho2018model,
  title={Model-based fault detection, fault isolation and fault-tolerant control of a blade pitch system in floating wind turbines},
  author={Cho, Seongpil and Gao, Zhen and Moan, Torgeir},
  journal={Renewable Energy},
  volume={120},
  pages={306--321},
  year={2018},
  publisher={Elsevier}
  doi={},
}

@article{badihi2014wind,
  title={Wind turbine fault diagnosis and fault-tolerant torque load control against actuator faults},
  author={Badihi, Hamed and Zhang, Youmin and Hong, Henry},
  journal={IEEE Transactions on Control Systems Technology},
  volume={23},
  number={4},
  pages={1351--1372},
  year={2014},
  publisher={IEEE}
  doi={},
}

@article{mousavi2021robust,
  title={Robust adaptive fractional-order nonsingular terminal sliding mode stabilization of three-axis gimbal platforms},
  author={Mousavi, Yashar and Zarei, Amin and Jahromi, Zeinabosadat Sane},
  journal={ISA transactions},
  year={2021},
  publisher={Elsevier}
}

@article{mousavi2018fractional,
  title={Fractional calculus-based firefly algorithm applied to parameter estimation of chaotic systems},
  author={Mousavi, Yashar and Alfi, Alireza},
  journal={Chaos, Solitons \& Fractals},
  volume={114},
  pages={202--215},
  year={2018},
  publisher={Elsevier}
  doi={},
}

@article{angel2018fractional,
  title={Fractional order PID for tracking control of a parallel robotic manipulator type delta},
  author={Angel, L and Viola, J},
  journal={ISA transactions},
  volume={79},
  pages={172--188},
  year={2018},
  publisher={Elsevier}
  doi={},
}

@article{ren2018optimal,
  title={Optimal design of a fractional-order proportional-integer-differential controller for a pneumatic position servo system},
  author={Ren, Hai-Peng and Fan, Jun-Tao and Kaynak, Okyay},
  journal={IEEE Transactions on Industrial Electronics},
  volume={66},
  number={8},
  pages={6220--6229},
  year={2018},
  publisher={IEEE}
  doi={},
}

@article{naidu2020power,
  title={Power quality enhancement in a grid-connected hybrid system with coordinated PQ theory \& fractional order PID controller in DPFC},
  author={Naidu, R Pavan Kumar and Meikandasivam, S},
  journal={Sustainable Energy, Grids and Networks},
  pages={100317},
  year={2020},
  publisher={Elsevier}
  doi={},
}

@article{azarmi2015analytical,
  title={Analytical design of fractional order PID controllers based on the fractional set-point weighted structure: Case study in twin rotor helicopter},
  author={Azarmi, Roohallah and Tavakoli-Kakhki, Mahsan and Sedigh, Ali Khaki and Fatehi, Alireza},
  journal={Mechatronics},
  volume={31},
  pages={222--233},
  year={2015},
  publisher={Elsevier}
  doi={},
}

@article{amoura2016closed,
  title={Closed-loop step response for tuning PID-fractional-order-filter controllers},
  author={Amoura, Karima and Mansouri, Rachid and Bettayeb, Ma{\^a}mar and Al-Saggaf, Ubaid M},
  journal={ISA transactions},
  volume={64},
  pages={247--257},
  year={2016},
  publisher={Elsevier}
  doi={},
}

@article{mousavi2015memetic,
  title={A memetic algorithm applied to trajectory control by tuning of fractional order proportional-integral-derivative controllers},
  author={Mousavi, Yashar and Alfi, Alireza},
  journal={Applied Soft Computing},
  volume={36},
  pages={599--617},
  year={2015},
  publisher={Elsevier}
  doi={},
}

@article{lee2010fractional,
  title={Fractional-order PID controller optimization via improved electromagnetism-like algorithm},
  author={Lee, Ching-Hung and Chang, Fu-Kai},
  journal={Expert Systems with Applications},
  volume={37},
  number={12},
  pages={8871--8878},
  year={2010},
  publisher={Elsevier}
  doi={},
}

@INPROCEEDINGS{yang2009firefly,
  title={Firefly algorithms for multimodal optimization},
  author={Yang, Xin-She},
  booktitle={International symposium on stochastic algorithms},
  pages={169--178},
  year={2009},
  organization={Springer}
  doi={},
}

@article{wang2019novel,
  title={A novel firefly algorithm based on gender difference and its convergence},
  author={Wang, Chun-Feng and Song, Wen-Xin},
  journal={Applied Soft Computing},
  volume={80},
  pages={107--124},
  year={2019},
  publisher={Elsevier}
  doi={},
}

@article{niknam2012reserve,
  title={Reserve constrained dynamic economic dispatch: A new fast self-adaptive modified firefly algorithm},
  author={Niknam, Taher and Azizipanah-Abarghooee, Rasoul and Roosta, Alireza},
  journal={IEEE Systems Journal},
  volume={6},
  number={4},
  pages={635--646},
  year={2012},
  publisher={IEEE}
  doi={},
}

@article{pazhoohesh2017optimal,
  title={Optimal harmonic reduction approach for PWM AC--AC converter using nested memetic algorithm},
  author={Pazhoohesh, Farid and Hasanvand, Saeed and Mousavi, Yashar},
  journal={Soft Computing},
  volume={21},
  number={10},
  pages={2761--2776},
  year={2017},
  publisher={Springer}
  doi={},
}

@article{gui2019multi,
  title={A multi-role based differential evolution},
  author={Gui, Ling and Xia, Xuewen and Yu, Fei and Wu, Hongrun and Wu, Ruifeng and Wei, Bo and Zhang, Yinglong and Li, Xiong and He, Guoliang},
  journal={Swarm and Evolutionary Computation},
  volume={50},
  pages={100508},
  year={2019},
  publisher={Elsevier}
  doi={},
}

@article{johnson2006control,
  title={Control of variable-speed wind turbines: standard and adaptive techniques for maximizing energy capture},
  author={Johnson, Kathryn E and Pao, Lucy Y and Balas, Mark J and Fingersh, Lee J},
  journal={IEEE Control Systems Magazine},
  volume={26},
  number={3},
  pages={70--81},
  year={2006},
  publisher={IEEE}
  doi={},
}

@article{abolvafaei2020maximum,
  title={Maximum power extraction from wind energy system using homotopy singular perturbation and fast terminal sliding mode method},
  author={Abolvafaei, Mahnaz and Ganjefar, Soheil},
  journal={Renewable Energy},
  volume={148},
  pages={611--626},
  year={2020},
  publisher={Elsevier}
  doi={},
}

@article{odgaard2013fault,
  title={Fault-tolerant control of wind turbines: A benchmark model},
  author={Odgaard, Peter Fogh and Stoustrup, Jakob and Kinnaert, Michel},
  journal={IEEE Transactions on control systems Technology},
  volume={21},
  number={4},
  pages={1168--1182},
  year={2013},
  publisher={IEEE}
  doi={},
}

@article{azizi2019fault,
  title={Fault tolerant control of wind turbines with an adaptive output feedback sliding mode controller},
  author={Azizi, Askar and Nourisola, Hamid and Shoja-Majidabad, Sajjad},
  journal={Renewable energy},
  volume={135},
  pages={55--65},
  year={2019},
  publisher={Elsevier}
  doi={},
}

@article{tang2018active,
  title={Active power control of wind turbine generators via coordinated rotor speed and pitch angle regulation},
  author={Tang, Xuesong and Yin, Minghui and Shen, Chun and Xu, Yan and Dong, Zhao Yang and Zou, Yun},
  journal={IEEE Transactions on Sustainable Energy},
  volume={10},
  number={2},
  pages={822--832},
  year={2018},
  publisher={IEEE}
  doi={},
}

@article{luo2007strategies,
  title={Strategies to smooth wind power fluctuations of wind turbine generator},
  author={Luo, Changling and Banakar, Hadi and Shen, Baike and Ooi, Boon-Teck},
  journal={IEEE Transactions on Energy Conversion},
  volume={22},
  number={2},
  pages={341--349},
  year={2007},
  publisher={IEEE}
  doi={},
}

@article{habibi2019reliability,
  title={Reliability improvement of wind turbine power generation using model-based fault detection and fault tolerant control: A review},
  author={Habibi, Hamed and Howard, Ian and Simani, Silvio},
  journal={Renewable Energy},
  volume={135},
  pages={877--896},
  year={2019},
  publisher={Elsevier}
  doi={},
}

@article{ruiz2018wind,
  title={Wind turbine fault detection and classification by means of image texture analysis},
  author={Ruiz, Magda and Mujica, Luis E and Alferez, Santiago and Acho, Leonardo and Tutiven, Christian and Vidal, Yolanda and Rodellar, Jose and Pozo, Francesc},
  journal={Mechanical Systems and Signal Processing},
  volume={107},
  pages={149--167},
  year={2018},
  publisher={Elsevier}
  doi={},
}

@book{sabatier2007advances,
  title={Advances in fractional calculus},
  author={Sabatier, JATMJ and Agrawal, Ohm Parkash and Machado, JA Tenreiro},
  volume={4},
  number={9},
  year={2007},
  publisher={Springer}
  doi={},
}

@article{luo2019enhanced,
  title={Enhanced grey wolf optimizer with a model for dynamically estimating the location of the prey},
  author={Luo, Kaiping},
  journal={Applied Soft Computing},
  volume={77},
  pages={225--235},
  year={2019},
  publisher={Elsevier}
  doi={},
}

@article{chen2020enhanced,
  title={An enhanced Bacterial Foraging Optimization and its application for training kernel extreme learning machine},
  author={Chen, Huiling and Zhang, Qian and Luo, Jie and Xu, Yueting and Zhang, Xiaoqin},
  journal={Applied Soft Computing},
  volume={86},
  pages={105884},
  year={2020},
  publisher={Elsevier}
  doi={},
}

@article{carrasco2020recent,
  title={Recent trends in the use of statistical tests for comparing swarm and evolutionary computing algorithms: Practical guidelines and a critical review},
  author={Carrasco, Jacinto and Garc{\'\i}a, Salvador and Rueda, MM and Das, S and Herrera, Francisco},
  journal={Swarm and Evolutionary Computation},
  volume={54},
  pages={100665},
  year={2020},
  publisher={Elsevier}
  doi={},
}

@article{dolan2006simulation,
  title={Simulation model of wind turbine 3p torque oscillations due to wind shear and tower shadow},
  author={Dolan, Dale SL and Lehn, Peter W},
  journal={IEEE Transactions on energy conversion},
  volume={21},
  number={3},
  pages={717--724},
  year={2006},
  publisher={IEEE}
  doi={},
}

@inproceedings{awad2017ensemble,
  title={Ensemble sinusoidal differential covariance matrix adaptation with Euclidean neighborhood for solving CEC2017 benchmark problems},
  author={Awad, Noor H and Ali, Mostafa Z and Suganthan, Ponnuthurai N},
  booktitle={2017 IEEE Congress on Evolutionary Computation (CEC)},
  pages={372--379},
  year={2017},
  organization={IEEE}
  doi={},
}

@article{viola2017design,
  title={Design and robust performance evaluation of a fractional order PID controller applied to a {DC} motor},
  author={Viola, Jairo and Angel, L and Sebastian, Jos{\'e} M},
  journal={IEEE/CAA Journal of Automatica Sinica},
  volume={4},
  number={2},
  pages={304--314},
  year={2017},
  publisher={IEEE}
  doi={},
}

@article{mousavi2020enhanced,
  title={Enhanced fractional chaotic whale optimization algorithm for parameter identification of isolated wind-diesel power systems},
  author={Mousavi, Yashar and Alfi, Alireza and Kucukdemiral, Ibrahim Beklan},
  journal={IEEE Access},
  year={2020},
  publisher={IEEE}
  doi={},
}

@article{aissaoui2013fuzzy,
  title={A Fuzzy-PI control to extract an optimal power from wind turbine},
  author={Aissaoui, Abdel Ghani and Tahour, Ahmed and Essounbouli, Najib and Nollet, Fr{\'e}d{\'e}ric and Abid, Mohamed and Chergui, Moulay Idriss},
  journal={Energy conversion and management},
  volume={65},
  pages={688--696},
  year={2013},
  publisher={Elsevier}
}

@article{bianchi2012gain,
  title={Gain scheduled control based on high fidelity local wind turbine models},
  author={Bianchi, Fernando D and S{\'a}nchez-Pe{\~n}a, Ricardo S and Guadayol, Marc},
  journal={Renewable energy},
  volume={37},
  number={1},
  pages={233--240},
  year={2012},
  publisher={Elsevier}
}

@article{jain2018fault,
  title={Fault-tolerant economic model predictive control for wind turbines},
  author={Jain, Tushar and Yam{\'e}, Joseph Julien},
  journal={IEEE transactions on sustainable energy},
  volume={10},
  number={4},
  pages={1696--1704},
  year={2018},
  publisher={IEEE}
}

@article{colombo2020pitch,
  title={Pitch angle control of a wind turbine operating above the rated wind speed: A sliding mode control approach},
  author={Colombo, L and Corradini, ML and Ippoliti, G and Orlando, G},
  journal={ISA transactions},
  volume={96},
  pages={95--102},
  year={2020},
  publisher={Elsevier}
}

@article{xu2019event,
  title={Event-trigger-based adaptive fuzzy hierarchical sliding mode control of uncertain under-actuated switched nonlinear systems},
  author={Xu, Ning and Chen, Yun and Xue, Anke and Zong, Guangdeng and Zhao, Xudong},
  journal={ISA transactions},
  year={2019},
  publisher={Elsevier}
}

@article{asgharnia2020load,
  title={Load mitigation of a class of 5-MW wind turbine with RBF neural network based fractional-order PID controller},
  author={Asgharnia, A and Jamali, A and Shahnazi, R and Maheri, A},
  journal={ISA transactions},
  volume={96},
  pages={272--286},
  year={2020},
  publisher={Elsevier}
}

@article{shan2021distributed,
  title={A distributed parallel firefly algorithm with communication strategies and its application for the control of variable pitch wind turbine},
  author={Shan, Jie and Pan, Jeng-Shyang and Chang, Cheng-Kuo and Chu, Shu-Chuan and Zheng, Shi-Guang},
  journal={ISA transactions},
  year={2021},
  publisher={Elsevier}
}

@article{benamor2019new,
  title={A new rooted tree optimization algorithm for indirect power control of wind turbine based on a doubly-fed induction generator},
  author={Benamor, A and Benchouia, MT and Srairi, K and Benbouzid, MEH},
  journal={ISA transactions},
  volume={88},
  pages={296--306},
  year={2019},
  publisher={Elsevier}
}

@article{kumar2021power,
  title={Power system stability enhancement by damping and control of Sub-synchronous torsional oscillations using Whale optimization algorithm based Type-2 wind turbines},
  author={Kumar, Rajeev and Singh, Rajveer and Ashfaq, Haroon and Singh, Sudhir Kumar and Badoni, Manoj},
  journal={ISA transactions},
  volume={108},
  pages={240--256},
  year={2021},
  publisher={Elsevier}
}

@article{xiong2020output,
  title={Output power quality enhancement of PMSG with fractional order sliding mode control},
  author={Xiong, Linyun and Li, Penghan and Ma, Meiling and Wang, Ziqiang and Wang, Jie},
  journal={International Journal of Electrical Power \& Energy Systems},
  volume={115},
  pages={105402},
  year={2020},
  publisher={Elsevier}
}

@article{bounar2019pso,
  title={PSO--GSA based fuzzy sliding mode controller for DFIG-based wind turbine},
  author={Bounar, N and Labdai, S and Boulkroune, A},
  journal={ISA transactions},
  volume={85},
  pages={177--188},
  year={2019},
  publisher={Elsevier}
}

@article{mousavi2021maximum,
  title={Maximum power extraction from wind turbines using a fault-tolerant fractional-order nonsingular terminal sliding mode controller},
  author={Mousavi, Yashar and Bevan, Geraint P and Kucukdemiral, Ibrahim B and Fekih, Afef},
  journal={Energies},
  year={2021},
  publisher={MDPI}
}

@article{biswas2009design,
  title={Design of fractional-order PI$\lambda$D$\mu$ controllers with an improved differential evolution},
  author={Biswas, Arijit and Das, Swagatam and Abraham, Ajith and Dasgupta, Sambarta},
  journal={Engineering applications of artificial intelligence},
  volume={22},
  number={2},
  pages={343--350},
  year={2009},
  publisher={Elsevier}
}

@article{beltran2008sliding,
  title={Sliding mode power control of variable-speed wind energy conversion systems},
  author={Beltran, Brice and Ahmed-Ali, Tarek and Benbouzid, Mohamed El Hachemi},
  journal={IEEE Transactions on energy conversion},
  volume={23},
  number={2},
  pages={551--558},
  year={2008},
  publisher={IEEE}
}

@inproceedings{mamani2007algebraic,
  title={An algebraic state estimation approach for {DC} motors},
  author={Mamani, G and Becedas, J and Feliu-Batlle, V and Sira-Ram{\i}rez, H},
  booktitle={Proceedings of the World Congress on Engineering and Computer Science 2007 WCECS 2007},
  year={2007},
  organization={Citeseer}
}

@article{trachanatzi2020firefly,
  title={A firefly algorithm for the environmental prize-collecting vehicle routing problem},
  author={Trachanatzi, Dimitra and Rigakis, Manousos and Marinaki, Magdalene and Marinakis, Yannis},
  journal={Swarm and Evolutionary Computation},
  volume={57},
  pages={100712},
  year={2020},
  publisher={Elsevier}
}

@article{musarrat2021improved,
  title={An Improved Fault Ride Through Scheme and Control Strategy for {DFIG}-Based Wind Energy Systems},
  author={Musarrat, Md Nafiz and Fekih, Afef and Islam, Md Rabiul},
  journal={IEEE Transactions on Applied Superconductivity},
  volume={31},
  number={8},
  pages={1--6},
  year={2021},
  publisher={IEEE}
}

@article{mousavi2021robustt,
  title={Robust optimal higher-order-observer-based dynamic sliding mode control for {VTOL} unmanned aerial vehicles},
  author={Mousavi, Yashar and Zarei, Amin and Mousavi, Arash and Biari, Mohsen},
  journal={International Journal of Automation and Computing},
  volume={18},
  number={5},
  pages={802--813},
  year={2021},
  publisher={Springer}
}

@article{benadja2018hardware,
  title={Hardware testing of sliding mode controller for improved performance of VSC-HVDC based offshore wind farm under {DC} fault},
  author={Benadja, Mounir and Rezkallah, Miloud and Benhalima, Seghir and Hamadi, Abdelhamid and Chandra, Ambrish},
  journal={IEEE Transactions on Industry Applications},
  volume={55},
  number={2},
  pages={2053--2063},
  year={2018},
  publisher={IEEE}
}

@article{merabet2018power,
  title={Power-current controller based sliding mode control for {{DFIG}}-wind energy conversion system},
  author={Merabet, Adel and Eshaft, Hisham and Tanvir, Aman A},
  journal={IET Renewable Power Generation},
  volume={12},
  number={10},
  pages={1155--1163},
  year={2018},
  publisher={IET}
}

@article{li2019sliding,
  title={Sliding mode controller based on feedback linearization for damping of sub-synchronous control interaction in {{DFIG}}-based wind power plants},
  author={Li, Penghan and Xiong, Linyun and Wu, Fei and Ma, Meiling and Wang, Jie},
  journal={International Journal of Electrical Power \& Energy Systems},
  volume={107},
  pages={239--250},
  year={2019},
  publisher={Elsevier}
}

@article{yang2018robust,
  title={Robust sliding-mode control of wind energy conversion systems for optimal power extraction via nonlinear perturbation observers},
  author={Yang, Bo and Yu, Tao and Shu, Hongchun and Dong, Jun and Jiang, Lin},
  journal={Applied Energy},
  volume={210},
  pages={711--723},
  year={2018},
  publisher={Elsevier}
}

@article{ullah2017adaptive,
  title={Adaptive fractional order terminal sliding mode control of a doubly fed induction generator-based wind energy system},
  author={Ullah, Nasim and Ali, Muhammad Asghar and Ibeas, Asier and Herrera, Jorge},
  journal={IEEE Access},
  volume={5},
  pages={21368--21381},
  year={2017},
  publisher={IEEE}
}

@inproceedings{morshed2015comparison,
  title={A comparison study between two sliding mode based controls for voltage sag mitigation in grid connected wind turbines},
  author={Morshed, Mohammad Javad and Fekih, Afef},
  booktitle={2015 IEEE Conference on Control Applications (CCA)},
  pages={1913--1918},
  year={2015},
  organization={IEEE}
}

@article{morshed2019sliding,
  title={A sliding mode approach to enhance the power quality of wind turbines under unbalanced voltage conditions},
  author={Morshed, Mohammad Javad and Fekih, Afef},
  journal={IEEE/CAA Journal of Automatica Sinica},
  volume={6},
  number={2},
  pages={566--574},
  year={2019},
  publisher={IEEE}
}

@article{patnaik2015fast,
  title={Fast adaptive finite-time terminal sliding mode power control for the rotor side converter of the {DFIG} based wind energy conversion system},
  author={Patnaik, RK and Dash, PK},
  journal={Sustainable Energy, Grids and Networks},
  volume={1},
  pages={63--84},
  year={2015},
  publisher={Elsevier}
}

@article{gentle2007matrix,
  title={Matrix algebra},
  author={Gentle, James E},
  journal={Springer texts in statistics, Springer, New York, NY, doi},
  volume={10},
  pages={978--0},
  year={2007},
  publisher={Springer}
}

@article{mousavi2021fault,
  title={Fault-tolerant optimal pitch control of wind turbines using dynamic weighted parallel firefly algorithm},
  author={Mousavi, Yashar and Bevan, Geraint and Kucukdemiral, Ibrahim Beklan},
  journal={ISA transactions},
  year={2021},
  publisher={Elsevier}
}

@article{sharmila2019fuzzy,
  title={Fuzzy Sampled-Data Control for DFIG-Based Wind Turbine With Stochastic Actuator Failures},
  author={Sharmila, V and Rakkiyappan, R and Joo, Young Hoon},
  journal={IEEE Transactions on Systems, Man, and Cybernetics: Systems},
  year={2019},
  publisher={IEEE}
}

@article{muyeen2007comparative,
  title={Comparative study on transient stability analysis of wind turbine generator system using different drive train models},
  author={Muyeen, SM and Ali, Md Hasan and Takahashi, Rion and Murata, Toshiaki and Tamura, Junji and Tomaki, Yuichi and Sakahara, Atsushi and Sasano, Eiichi},
  journal={IET Renewable Power Generation},
  volume={1},
  number={2},
  pages={131--141},
  year={2007},
  publisher={IET}
}

@article{shi2019perturbation,
  title={Perturbation estimation based nonlinear adaptive power decoupling control for DFIG wind turbine},
  author={Shi, Kai and Yin, Xin and Jiang, Lin and Liu, Yang and Hu, Yihua and Wen, Huiqing},
  journal={IEEE Transactions on Power Electronics},
  volume={35},
  number={1},
  pages={319--333},
  year={2019},
  publisher={IEEE}
}

@article{mahmoud2019adaptive,
  title={Adaptive and predictive control strategies for wind turbine systems: a survey},
  author={Mahmoud, Magdi S and Oyedeji, Mojeed O},
  journal={IEEE/CAA Journal of Automatica Sinica},
  volume={6},
  number={2},
  pages={364--378},
  year={2019},
  publisher={IEEE}
}

@article{zhang2015discrete,
  title={A discrete-time direct torque control for direct-drive PMSG-based wind energy conversion systems},
  author={Zhang, Zhe and Zhao, Yue and Qiao, Wei and Qu, Liyan},
  journal={IEEE Transactions on Industry Applications},
  volume={51},
  number={4},
  pages={3504--3514},
  year={2015},
  publisher={IEEE}
}

@article{sun2017discrete,
  title={Discrete-time fractional order terminal sliding mode tracking control for linear motor},
  author={Sun, Guanghui and Ma, Zhiqiang and Yu, Jinyong},
  journal={IEEE Transactions on Industrial Electronics},
  volume={65},
  number={4},
  pages={3386--3394},
  year={2017},
  publisher={IEEE}
}

@article{sun2018practical,
  title={Practical tracking control of linear motor via fractional-order sliding mode},
  author={Sun, Guanghui and Wu, Ligang and Kuang, Zhian and Ma, Zhiqiang and Liu, Jianxing},
  journal={Automatica},
  volume={94},
  pages={221--235},
  year={2018},
  publisher={Elsevier}
}

@article{yin2014fractional,
  title={Fractional-order sliding mode based extremum seeking control of a class of nonlinear systems},
  author={Yin, Chun and Chen, YangQuan and Zhong, Shou-ming},
  journal={Automatica},
  volume={50},
  number={12},
  pages={3173--3181},
  year={2014},
  publisher={Elsevier}
}

@article{utkin2015discussion,
  title={Discussion aspects of high-order sliding mode control},
  author={Utkin, Vadim},
  journal={IEEE Transactions on Automatic Control},
  volume={61},
  number={3},
  pages={829--833},
  year={2015},
  publisher={IEEE}
}

@article{liu2014finite,
  title={Finite time control for MIMO nonlinear system based on higher-order sliding mode},
  author={Liu, Xiangjie and Han, Yaozhen},
  journal={ISA transactions},
  volume={53},
  number={6},
  pages={1838--1846},
  year={2014},
  publisher={Elsevier}
}

@article{mathiyalagan2020second,
  title={Second-order sliding mode control for nonlinear fractional-order systems},
  author={Mathiyalagan, Kalidass and Sangeetha, G},
  journal={Applied Mathematics and Computation},
  volume={383},
  pages={125264},
  year={2020},
  publisher={Elsevier}
}

@book{bianchi2006wind,
  title={Wind turbine control systems: principles, modelling and gain scheduling design},
  author={Bianchi, Fernando D and De Battista, Hernan and Mantz, Ricardo J},
  year={2006},
  publisher={Springer Science \& Business Media}
}

@book{anaya2011wind,
  title={Wind energy generation: modelling and control},
  author={Anaya-Lara, Olimpo and Jenkins, Nick and Ekanayake, Janaka B and Cartwright, Phill and Hughes, Michael},
  year={2011},
  publisher={John Wiley \& Sons}
}

@article{mandic2012active,
  title={Active torque control for gearbox load reduction in a variable-speed wind turbine},
  author={Mandic, Goran and Nasiri, Adel and Muljadi, Eduard and Oyague, Francisco},
  journal={IEEE Transactions on Industry Applications},
  volume={48},
  number={6},
  pages={2424--2432},
  year={2012},
  publisher={IEEE}
}

@inproceedings{xu2011influence,
  title={Influence of different flexible drive train models on the transient responses of DFIG wind turbine},
  author={Xu, Zhan and Pan, Zaiping},
  booktitle={2011 International Conference on Electrical Machines and Systems},
  pages={1--6},
  year={2011},
  organization={IEEE}
}

@article{villanueva2018grid,
  title={Grid-voltage-oriented sliding mode control for DFIG under balanced and unbalanced grid Faults},
  author={Villanueva, Iv{\'a}n and Rosales, Antonio and Ponce, Pedro and Molina, Arturo},
  journal={IEEE Transactions on Sustainable Energy},
  volume={9},
  number={3},
  pages={1090--1098},
  year={2018},
  publisher={IEEE}
}

@article{valenciaga2009geometric,
  title={A geometric approach for the design of MIMO sliding controllers. Application to a wind-driven doubly fed induction generator},
  author={Valenciaga, F and Puleston, PF and Spurgeon, Sarah K},
  journal={International Journal of Robust and Nonlinear Control: IFAC-Affiliated Journal},
  volume={19},
  number={1},
  pages={22--39},
  year={2009},
  publisher={Wiley Online Library}
}

@article{liu2018dfig,
  title={DFIG wind turbine sliding mode control with exponential reaching law under variable wind speed},
  author={Liu, Yifang and Wang, Zhijie and Xiong, Linyun and Wang, Jie and Jiang, Xiuchen and Bai, Gehao and Li, Renfu and Liu, Sanming},
  journal={International Journal of Electrical Power \& Energy Systems},
  volume={96},
  pages={253--260},
  year={2018},
  publisher={Elsevier}
}

@article{nayeh2020multivariable,
  title={Multivariable robust control of a horizontal wind turbine under various operating modes and uncertainties: A comparison on sliding mode and $H_{\infty}$ control},
  author={Nayeh, Reza Faraji and Moradi, Hamed and Vossoughi, Gholamreza},
  journal={International Journal of Electrical Power \& Energy Systems},
  volume={115},
  pages={105474},
  year={2020},
  publisher={Elsevier}
}

@article{errouissi2017novel,
  title={A novel PI-type sliding surface for PMSG-based wind turbine with improved transient performance},
  author={Errouissi, Rachid and Al-Durra, Ahmed},
  journal={IEEE Transactions on Energy Conversion},
  volume={33},
  number={2},
  pages={834--844},
  year={2017},
  publisher={IEEE}
}

@article{suleimenov2020disturbance,
  title={Disturbance observer-based integral sliding mode control for wind energy conversion systems},
  author={Suleimenov, Kanat and Sarsembayev, Bayandy and Duc Hong Phuc, Bui and Do, Ton Duc},
  journal={Wind Energy},
  volume={23},
  number={4},
  pages={1026--1047},
  year={2020},
  publisher={Wiley Online Library}
}

@article{yuan2018coordinated,
  title={The coordinated control of wind-diesel hybrid micro-grid based on sliding mode method and load estimation},
  author={Yuan, Minghan and Fu, Yang and Mi, Yang and Li, Zhenkun and Wang, Chengshan},
  journal={IEEE Access},
  volume={6},
  pages={76867--76875},
  year={2018},
  publisher={IEEE}
}

@inproceedings{qingmei2019novel,
  title={A Novel Sliding Mode Control of Doubly-fed Induction Generator for Optimal Power Extraction},
  author={Qingmei, Kong and Xiangdong, Wang and Shujiang, Li},
  booktitle={2019 IEEE Innovative Smart Grid Technologies-Asia (ISGT Asia)},
  pages={1318--1323},
  year={2019},
  organization={IEEE}
}

@article{saravanakumar2015validation,
  title={Validation of an integral sliding mode control for optimal control of a three blade variable speed variable pitch wind turbine},
  author={Saravanakumar, Rajendran and Jena, Debashisha},
  journal={International Journal of Electrical Power \& Energy Systems},
  volume={69},
  pages={421--429},
  year={2015},
  publisher={Elsevier}
}

@article{cui2017observer,
  title={Observer based robust integral sliding mode load frequency control for wind power systems},
  author={Cui, Yanliang and Xu, Lanlan and Fei, Minrui and Shen, Yubin},
  journal={Control Engineering Practice},
  volume={65},
  pages={1--10},
  year={2017},
  publisher={Elsevier}
}

@inproceedings{pan2008maximum,
  title={Maximum power point tracking of wind energy conversion systems based on sliding mode extremum seeking control},
  author={Pan, Tinglong and Ji, Zhicheng and Jiang, Zhenhua},
  booktitle={2008 IEEE Energy 2030 Conference},
  pages={1--5},
  year={2008},
  organization={IEEE}
}

@inproceedings{ihedrane2017power,
  title={Power Control of Wind Turbine System based on DFIG-Generator, using Sliding Mode Technique},
  author={Ihedrane, Yasmine and Bossoufi, Badre and others},
  booktitle={2017 International Renewable and Sustainable Energy Conference (IRSEC)},
  pages={1--6},
  year={2017},
  organization={IEEE}
}

@inproceedings{zheng2009sliding,
  title={Sliding mode MPPT control of variable speed wind power system},
  author={Zheng, Xuemei and Li, Lin and Xu, Dianguo and Platts, Jim},
  booktitle={2009 Asia-Pacific Power and Energy Engineering Conference},
  pages={1--4},
  year={2009},
  organization={IEEE}
}

@article{weng2014sliding,
  title={Sliding mode regulator for maximum power tracking and copper loss minimisation of a doubly fed induction generator},
  author={Weng, Yung-Tsai and Hsu, Yuan-Yih},
  journal={IET Renewable Power Generation},
  volume={9},
  number={4},
  pages={297--305},
  year={2014},
  publisher={IET}
}

@inproceedings{amira2020sliding,
  title={Sliding Mode Control of Doubly-fed Induction Generator in Wind Energy Conversion System},
  author={Amira, Lakhdara and Tahar, Bahi and Abdelkrim, Moussaoui},
  booktitle={2020 8th International Conference on Smart Grid (icSmartGrid)},
  pages={96--100},
  year={2020},
  organization={IEEE}
}

@inproceedings{el2016comparative,
  title={Comparative study of the sliding mode and backstepping control in power control of a doubly fed induction generator},
  author={El Azzaoui, Marouane and Mahmoudi, Hassane and Bossoufi, Badre and El Ghamrasni, Madiha},
  booktitle={2016 International Symposium on Fundamentals of Electrical Engineering (ISFEE)},
  pages={1--5},
  year={2016},
  organization={IEEE}
}

@article{bekakra2014dfig,
  title={DFIG sliding mode control fed by back-to-back PWM converter with DC-link voltage control for variable speed wind turbine},
  author={Bekakra, Youcef and Attous, Djilani Ben},
  journal={Frontiers in Energy},
  volume={8},
  number={3},
  pages={345--354},
  year={2014},
  publisher={Springer}
}

@article{hu2010direct,
  title={Direct active and reactive power regulation of DFIG using sliding-mode control approach},
  author={Hu, Jiabing and Nian, Heng and Hu, Bin and He, Yikang and Zhu, ZQ},
  journal={IEEE Transactions on energy conversion},
  volume={25},
  number={4},
  pages={1028--1039},
  year={2010},
  publisher={IEEE}
}

@article{pande2013discrete,
  title={Discrete sliding mode control strategy for direct real and reactive power regulation of wind driven DFIG},
  author={Pande, VN and Mate, UM and Kurode, Shailaja},
  journal={Electric Power Systems Research},
  volume={100},
  pages={73--81},
  year={2013},
  publisher={Elsevier}
}

@article{mahboub2017sliding,
  title={Sliding mode control of grid connected brushless doubly fed induction generator driven by wind turbine in variable speed},
  author={Mahboub, M Abdelbasset and Drid, Said and Sid, MA and Cheikh, Ridha},
  journal={International Journal of System Assurance Engineering and Management},
  volume={8},
  number={2},
  pages={788--798},
  year={2017},
  publisher={Springer}
}

@inproceedings{jeong2008sliding,
  title={A sliding-mode approach to control the active and reactive powers for a DFIG in wind turbines},
  author={Jeong, Hea-Gwang and Kim, Won-Sang and Lee, Kyo-Beum and Jeong, Byung-Chang and Song, Seung-Ho},
  booktitle={2008 IEEE Power Electronics Specialists Conference},
  pages={120--125},
  year={2008},
  organization={IEEE}
}

@inproceedings{hamane2014control,
  title={Control of wind turbine based on DFIG using Fuzzy-PI and sliding mode controllers},
  author={Hamane, B and Doumbia, ML and Bouhamida, M and Benghanem, M},
  booktitle={2014 Ninth international conference on ecological vehicles and renewable energies (EVER)},
  pages={1--8},
  year={2014},
  organization={IEEE}
}

@inproceedings{hagh2015direct,
  title={Direct power control of DFIG based wind turbine based on wind speed estimation and particle swarm optimization},
  author={Hagh, Mehrdad Tarafdar and Roozbehani, S and Najaty, F and Ghaemi, S and Tan, Yingjie and Muttaqi, KM},
  booktitle={2015 Australasian Universities Power Engineering Conference (AUPEC)},
  pages={1--6},
  year={2015},
  organization={IEEE}
}

@inproceedings{tohidi2013multivariable,
  title={Multivariable input-output linearization sliding mode control of DFIG based wind energy conversion system},
  author={Tohidi, Akbar and Shamsaddinlou, Ali and Sedigh, Ali Khaki},
  booktitle={2013 9th Asian Control Conference (ASCC)},
  pages={1--6},
  year={2013},
  organization={IEEE}
}

@article{jafari2017analysis,
  title={Analysis and simulation of a sliding mode controller for mechanical part of a doubly-fed induction generator-based wind turbine},
  author={Jafari, Ahmad and Shahgholian, Ghazanfar},
  journal={IET Generation, Transmission \& Distribution},
  volume={11},
  number={10},
  pages={2677--2688},
  year={2017},
  publisher={IET}
}

@article{djoudi2018sliding,
  title={Sliding mode control of DFIG powers in the case of unknown flux and rotor currents with reduced switching frequency},
  author={Djoudi, Abdelhak and Bacha, Seddik and Iman-Eini, Hossein and Rekioua, Toufik},
  journal={International Journal of Electrical Power \& Energy Systems},
  volume={96},
  pages={347--356},
  year={2018},
  publisher={Elsevier}
}

@article{saad2015low,
  title={Low voltage ride through of doubly-fed induction generator connected to the grid using sliding mode control strategy},
  author={Saad, Naggar H and Sattar, Ahmed A and Mansour, Abd El-Aziz M},
  journal={Renewable Energy},
  volume={80},
  pages={583--594},
  year={2015},
  publisher={Elsevier}
}

@article{shang2012sliding,
  title={Sliding-mode-based direct power control of grid-connected wind-turbine-driven doubly fed induction generators under unbalanced grid voltage conditions},
  author={Shang, Lei and Hu, Jiabing},
  journal={IEEE Transactions on Energy Conversion},
  volume={27},
  number={2},
  pages={362--373},
  year={2012},
  publisher={IEEE}
}

@article{shehata2015sliding,
  title={Sliding mode direct power control of RSC for DFIGs driven by variable speed wind turbines},
  author={Shehata, EG},
  journal={Alexandria Engineering Journal},
  volume={54},
  number={4},
  pages={1067--1075},
  year={2015},
  publisher={Elsevier}
}

@article{djilali2018real,
  title={Real-time implementation of sliding-mode field-oriented control for a DFIG-based wind turbine},
  author={Djilali, Larbi and Sanchez, Edgar N and Belkheiri, Mohammed},
  journal={International Transactions on Electrical Energy Systems},
  volume={28},
  number={5},
  pages={e2539},
  year={2018},
  publisher={Wiley Online Library}
}

@article{martinez2013sliding,
  title={Sliding-mode control of a wind turbine-driven double-fed induction generator under non-ideal grid voltages},
  author={Martinez, Miren Itsaso and Susperregui, Ana and Tapia, Gerardo and Xu, Lie},
  journal={IET Renewable Power Generation},
  volume={7},
  number={4},
  pages={370--379},
  year={2013},
  publisher={IET}
}

@article{jaladi2018dc,
  title={DC-link transient improvement of SMC-based hybrid control of DFIG-WES under asymmetrical grid faults},
  author={Jaladi, Kiran Kumar and Sandhu, Kanwarjit Singh},
  journal={International Transactions on Electrical Energy Systems},
  volume={28},
  number={12},
  pages={e2633},
  year={2018},
  publisher={Wiley Online Library}
}

@article{munteanu2008energy,
  title={Energy-reliability optimization of wind energy conversion systems by sliding mode control},
  author={Munteanu, Iulian and Bacha, Seddik and Bratcu, Antoneta Iuliana and Guiraud, Joel and Roye, Daniel},
  journal={IEEE Transactions on Energy Conversion},
  volume={23},
  number={3},
  pages={975--985},
  year={2008},
  publisher={IEEE}
}

@article{dahiya2017hybridized,
  title={Hybridized gravitational search algorithm tuned sliding mode controller design for load frequency control system with doubly fed induction generator wind turbine},
  author={Dahiya, Preeti and Sharma, Veena and Naresh, Ram},
  journal={Optimal Control Applications and Methods},
  volume={38},
  number={6},
  pages={993--1003},
  year={2017},
  publisher={Wiley Online Library}
}

@article{dahiya2019optimal,
  title={Optimal sliding mode control for frequency regulation in deregulated power systems with DFIG-based wind turbine and TCSC--SMES},
  author={Dahiya, Preeti and Sharma, Veena and Naresh, R},
  journal={Neural Computing and Applications},
  volume={31},
  number={7},
  pages={3039--3056},
  year={2019},
  publisher={Springer}
}

@article{barambones2019robust,
  title={Robust wind speed estimation and control of variable speed wind turbines},
  author={Barambones, Oscar},
  journal={Asian Journal of Control},
  volume={21},
  number={2},
  pages={856--867},
  year={2019},
  publisher={Wiley Online Library}
}

@article{martinez2011sliding,
  title={Sliding-Mode Control for a DFIG-based Wind Turbine under Unbalanced Voltage},
  author={Martinez, M Itsaso and Susperregui, Ana and Tapia, Gerardo and Camblong, Haritza},
  journal={IFAC Proceedings Volumes},
  volume={44},
  number={1},
  pages={538--543},
  year={2011},
  publisher={Elsevier}
}

@inproceedings{aghatehrani2011sliding,
  title={Sliding mode control approach for voltage regulation in microgrids with DFIG based wind generations},
  author={Aghatehrani, Rasool and Kavasseri, Rajesh},
  booktitle={2011 IEEE Power and Energy Society General Meeting},
  pages={1--8},
  year={2011},
  organization={IEEE}
}

@article{tang2019non,
  title={Non-linear extended state observer-based sliding mode control for a direct-driven wind energy conversion system with permanent magnet synchronous generator},
  author={Tang, Yongwei and Li, Juan and Li, Shengquan and Cao, QingFeng and Wu, Yuanwang},
  journal={The Journal of Engineering},
  volume={2019},
  number={15},
  pages={613--617},
  year={2019},
  publisher={IET}
}

@article{yang2018passivity,
  title={Passivity-based sliding-mode control design for optimal power extraction of a PMSG based variable speed wind turbine},
  author={Yang, Bo and Yu, Tao and Shu, Hongchun and Zhang, Yuming and Chen, Jian and Sang, Yiyan and Jiang, Lin},
  journal={Renewable Energy},
  volume={119},
  pages={577--589},
  year={2018},
  publisher={Elsevier}
}

@article{hu2019sliding,
  title={Sliding mode extremum seeking control based on improved invasive weed optimization for MPPT in wind energy conversion system},
  author={Hu, Lu and Xue, Fei and Qin, Zijian and Shi, Jiying and Qiao, Wen and Yang, Wenjing and Yang, Ting},
  journal={Applied energy},
  volume={248},
  pages={567--575},
  year={2019},
  publisher={Elsevier}
}

@article{yin2015sliding,
  title={Sliding mode voltage control strategy for capturing maximum wind energy based on fuzzy logic control},
  author={Yin, Xiu-xing and Lin, Yong-gang and Li, Wei and Gu, Ya-jing and Lei, Peng-fei and Liu, Hong-wei},
  journal={International Journal of Electrical Power \& Energy Systems},
  volume={70},
  pages={45--51},
  year={2015},
  publisher={Elsevier}
}

@inproceedings{jingfeng2015maximum,
  title={Maximum power point tracking in variable speed wind turbine system via optimal torque sliding mode control strategy},
  author={Jingfeng, Mao and Aihua, Wu and Guoqing, WU and Xudong, Zhang},
  booktitle={2015 34th Chinese Control Conference (CCC)},
  pages={7967--7971},
  year={2015},
  organization={IEEE}
}

@inproceedings{errami2013maximum,
  title={Maximum power point tracking of a wind power system based on the PMSG using sliding mode direct torque control},
  author={Errami, Youssef and Maaroufi, Mohamed and Ouassaid, Mohammed},
  booktitle={2013 International Renewable and Sustainable Energy Conference (IRSEC)},
  pages={218--223},
  year={2013},
  organization={IEEE}
}

@inproceedings{kusumawardana2019simple,
  title={Simple MPPT based on Maximum Power with Double Integral Sliding Mode Current Control for Vertical Axis Wind Turbine},
  author={Kusumawardana, Arya and Gumilar, Langlang and Prihanto, Dwi and Wicaksono, Hendro and Saputra, Syaiqun Nizar Trisna and Prasetyo, Dedi},
  booktitle={2019 IEEE Conference on Energy Conversion (CENCON)},
  pages={31--36},
  year={2019},
  organization={IEEE}
}

@article{pan2020wind,
  title={Wind energy conversion systems analysis of PMSG on offshore wind turbine using improved SMC and Extended State Observer},
  author={Pan, Lin and Shao, Chengpeng},
  journal={Renewable Energy},
  volume={161},
  pages={149--161},
  year={2020},
  publisher={Elsevier}
}

@article{valenciaga2003power,
  title={Power control of a solar/wind generation system without wind measurement: A passivity/sliding mode approach},
  author={Valenciaga, Fernando and Puleston, Pablo F and Battaiotto, Pedro E},
  journal={IEEE Transactions on Energy Conversion},
  volume={18},
  number={4},
  pages={501--507},
  year={2003},
  publisher={IEEE}
}

@inproceedings{ayadi2015sliding,
  title={Sliding mode approach for blade pitch angle control wind turbine using PMSG under DTC},
  author={Ayadi, Marwa and Salem, Fatma Ben and Derbel, Nabil},
  booktitle={2015 16th International Conference on Sciences and Techniques of Automatic Control and Computer Engineering (STA)},
  pages={758--762},
  year={2015},
  organization={IEEE}
}

@inproceedings{xin2014sliding,
  title={Sliding mode control of pitch angle for direct driven PM wind turbine},
  author={Xin, Wang and Wanli, Zhu and Bin, Qin and Pengcheng, Li},
  booktitle={The 26th Chinese Control and Decision Conference (2014 CCDC)},
  pages={2447--2452},
  year={2014},
  organization={IEEE}
}

@inproceedings{lee2010sliding,
  title={Sliding mode controller for torque and pitch control of wind power system based on PMSG},
  author={Lee, Sung-Hun and Joo, Young-Jun and Back, Juhoon and Seo, Jin Heon},
  booktitle={ICCAS 2010},
  pages={1079--1084},
  year={2010},
  organization={IEEE}
}

@article{thakur2018control,
  title={Control of a PMSG Wind-Turbine Under Asymmetrical Voltage Sags Using Sliding Mode Approach},
  author={Thakur, Devbratta and Jiang, Jin},
  journal={IEEE Power and Energy Technology Systems Journal},
  volume={5},
  number={2},
  pages={47--55},
  year={2018},
  publisher={IEEE}
}

@inproceedings{jena2017novel,
  title={A novel SMC based vector control strategy used for decoupled control of PMSG based variable speed wind turbine system},
  author={Jena, Narendra Kumar and Pradhan, Haimabati and Choudhury, Abhijeet and Mohanty, KB and Sanyal, SK},
  booktitle={2017 International Conference on Circuit, Power and Computing Technologies (ICCPCT)},
  pages={1--6},
  year={2017},
  organization={IEEE}
}

@inproceedings{gajewski2017analysis,
  title={Analysis of Sliding Mode Control of variable speed wind turbine system with PMSG},
  author={Gajewski, Piotr and Pie{\'n}kowski, Krzysztof},
  booktitle={2017 International Symposium on Electrical Machines (SME)},
  pages={1--6},
  year={2017},
  organization={IEEE}
}

@article{huang2019dc,
  title={Dc-link voltage regulation for wind power system by complementary sliding mode control},
  author={Huang, Yanwei and Zhang, Zhongyang and Huang, Wenchao and Chen, Shaobin},
  journal={IEEE Access},
  volume={7},
  pages={22773--22780},
  year={2019},
  publisher={IEEE}
}

@article{merzoug2012sliding,
  title={Sliding mode control (SMC) of permanent magnet synchronous generators (PMSG)},
  author={Merzoug, MS and Benall, H and Louze, L},
  journal={Energy Procedia},
  volume={18},
  pages={43--52},
  year={2012},
  publisher={Elsevier}
}

@article{mozayan2016sliding,
  title={Sliding mode control of PMSG wind turbine based on enhanced exponential reaching law},
  author={Mozayan, Seyed Mehdi and Saad, Maarouf and Vahedi, Hani and Fortin-Blanchette, Handy and Soltani, Mohsen},
  journal={IEEE Transactions on Industrial Electronics},
  volume={63},
  number={10},
  pages={6148--6159},
  year={2016},
  publisher={IEEE}
}

@article{soufi2016particle,
  title={Particle swarm optimization based sliding mode control of variable speed wind energy conversion system},
  author={Soufi, Youcef and Kahla, Sami and Bechouat, Mohcene},
  journal={International Journal of Hydrogen Energy},
  volume={41},
  number={45},
  pages={20956--20963},
  year={2016},
  publisher={Elsevier}
}

@article{de2000dynamical,
  title={Dynamical sliding mode power control of wind driven induction generators},
  author={De Battista, Hern{\'a}n and Mantz, Ricardo J and Christiansen, Carlos F},
  journal={IEEE transactions on energy conversion},
  volume={15},
  number={4},
  pages={451--457},
  year={2000},
  publisher={IEEE}
}

@article{matas2008feedback,
  title={Feedback linearization of direct-drive synchronous wind-turbines via a sliding mode approach},
  author={Matas, Jos{\'e} and Castilla, Miguel and Guerrero, Josep M and de Vicu{\~n}a, Luis Garc{\'\i}a and Miret, Jaume},
  journal={IEEE Transactions on Power Electronics},
  volume={23},
  number={3},
  pages={1093--1103},
  year={2008},
  publisher={IEEE}
}

@inproceedings{pati2013sliding,
  title={A sliding mode controller based DTC scheme for performance improvement of cage induction generator used for wind power applications},
  author={Pati, Swagat and Samantray, Swati and Patel, Nimai Charan},
  booktitle={2013 Annual International Conference on Emerging Research Areas and 2013 International Conference on Microelectronics, Communications and Renewable Energy},
  pages={1--6},
  year={2013},
  organization={IEEE}
}

@inproceedings{pati2012performance,
  title={Performance improvement of Indirect vector controlled Induction Generator system with sliding mode controller},
  author={Pati, Swagat and Mohanty, Kanungo Barada and Das, Debiprasanna},
  booktitle={2012 Annual IEEE India Conference (INDICON)},
  pages={444--449},
  year={2012},
  organization={IEEE}
}

@inproceedings{mi2014sliding,
  title={The sliding mode pitch angle controller design for squirrel-cage induction generator wind power generation system},
  author={Mi, Yang and Bao, Xiaowei and Yang, Yang and Zhang, Han and Wang, Peng},
  booktitle={Proceedings of the 33rd Chinese control conference},
  pages={8113--8117},
  year={2014},
  organization={IEEE}
}

@article{de2000sliding,
  title={Sliding mode control of wind energy systems with DOIG-power efficiency and torsional dynamics optimization},
  author={De Battista, Hern{\'a}n and Puleston, Pablo F and Mantz, Ricardo J and Christiansen, Carlos F},
  journal={IEEE Transactions on Power Systems},
  volume={15},
  number={2},
  pages={728--734},
  year={2000},
  publisher={IEEE}
}

@article{puleston2000sliding,
  title={Sliding mode control for efficiency optimization of wind energy systems with double output induction generator},
  author={Puleston, PF and Mantz, RJ and Battaiotto, PE and Valenciaga, F},
  journal={International journal of energy research},
  volume={24},
  number={1},
  pages={77--92},
  year={2000},
  publisher={Wiley Online Library}
}

@article{amimeur2012sliding,
  title={Sliding mode control of a dual-stator induction generator for wind energy conversion systems},
  author={Amimeur, H and Aouzellag, D and Abdessemed, R and Ghedamsi, K},
  journal={International Journal of Electrical Power \& Energy Systems},
  volume={42},
  number={1},
  pages={60--70},
  year={2012},
  publisher={Elsevier}
}

@article{ur2019disturbance,
  title={A Disturbance Observer Based Sliding Mode Control for Variable Speed Wind Turbine},
  author={ur Rehman, Ateeq and Ali, Nihad and Khan, Owais and Pervaiz, Mahmood},
  journal={IETE Journal of Research},
  pages={1--8},
  year={2019},
  publisher={Taylor \& Francis}
}

@article{faskhodi2019output,
  title={Output Feedback Robust Siding Mode Controller Design for Wind Turbine},
  author={Faskhodi, Azita Sharifi and Fakharian, Ahmad},
  journal={Journal of Electrical Engineering \& Technology},
  volume={14},
  number={6},
  pages={2477--2485},
  year={2019},
  publisher={Springer}
}

@inproceedings{gui2020complementary,
  title={Complementary Sliding Mode Control for Variable Speed Variable Pitch Wind Turbine Based on Feedback Linearization},
  author={Gui, Kai and Cen, Lihui and Liu, Fang},
  booktitle={2020 Chinese Control And Decision Conference (CCDC)},
  pages={284--289},
  year={2020},
  organization={IEEE}
}

@inproceedings{larrea2018sliding,
  title={A Sliding Mode Pitch Control for Multi-Megawatt Offshore Wind Turbines},
  author={Larrea-Le{\'o}n, Carlos and Seshagiri, Sridhar},
  booktitle={2018 Clemson University Power Systems Conference (PSC)},
  pages={1--8},
  year={2018},
  organization={IEEE}
}

@inproceedings{corradini2017sliding,
  title={A sliding mode pitch controller for wind turbines operating in high wind speeds region},
  author={Corradini, Maria Letizia and Ippoliti, Gianluca and Orlando, Giuseppe},
  booktitle={2017 4th International Conference on Control, Decision and Information Technologies (CoDIT)},
  pages={0024--0029},
  year={2017},
  organization={IEEE}
}

@inproceedings{yinzhu2016study,
  title={The study of variable speed variable pitch controller for wind power generation systems based on sliding mode control},
  author={Yinzhu, Zhu and Yang, Mi},
  booktitle={2016 IEEE 11th Conference on Industrial Electronics and Applications (ICIEA)},
  pages={415--420},
  year={2016},
  organization={IEEE}
}

@article{zaafouri2018uncertain,
  title={Uncertain saturated discrete-time sliding mode control for a wind turbine using a two-mass model},
  author={Zaafouri, Chaker and Torchani, Borhen and Sellami, Anis and Garcia, Germain},
  journal={Asian Journal of Control},
  volume={20},
  number={2},
  pages={802--818},
  year={2018},
  publisher={Wiley Online Library}
}

@article{torchani2016variable,
  title={Variable speed wind turbine control by discrete-time sliding mode approach},
  author={Torchani, Borhen and Sellami, Anis and Garcia, Germain},
  journal={Isa Transactions},
  volume={62},
  pages={81--86},
  year={2016},
  publisher={Elsevier}
}

@article{zhang2013sliding,
  title={Sliding mode control-based active power control for wind farm with variable speed wind generation system},
  author={Zhang, Zhenzhen and Xu, Hongbing and Zou, Jianxiao and Zheng, Gang},
  journal={Proceedings of the Institution of Mechanical Engineers, Part C: Journal of Mechanical Engineering Science},
  volume={227},
  number={3},
  pages={449--458},
  year={2013},
  publisher={SAGE Publications Sage UK: London, England}
}

@article{berrada2020new,
  title={New structure of sliding mode control for variable speed wind turbine},
  author={Berrada, Youssef and Boumhidi, Ismail},
  journal={IFAC Journal of Systems and Control},
  volume={14},
  pages={100113},
  year={2020},
  publisher={Elsevier}
}

@inproceedings{beltran2007sliding,
  title={Sliding mode power control of variable speed wind energy conversion systems},
  author={Beltran, Brice and Ahmed-Ali, Tarek and Benbouzid, Mohamed El Hachemi},
  booktitle={2007 IEEE International Electric Machines \& Drives Conference},
  volume={2},
  pages={943--948},
  year={2007},
  organization={IEEE}
}

@article{agarwala2019design,
  title={Design of a nonlinear multi-input--multi-output sliding mode pitch angle and plunge controller for a 5MW wind turbine blade tip},
  author={Agarwala, Ranjeet and Chin, Robert A and Malali, Praveen},
  journal={Energy Sources, Part A: Recovery, Utilization, and Environmental Effects},
  volume={41},
  number={23},
  pages={2929--2943},
  year={2019},
  publisher={Taylor \& Francis}
}

@article{prasad2019non,
  title={Non-linear sliding mode control for frequency regulation with variable-speed wind turbine systems},
  author={Prasad, Sheetla and Purwar, Shubhi and Kishor, Nand},
  journal={International Journal of Electrical Power \& Energy Systems},
  volume={107},
  pages={19--33},
  year={2019},
  publisher={Elsevier}
}

@article{hu2017active,
  title={Active structural control for load mitigation of wind turbines via adaptive sliding-mode approach},
  author={Hu, Yinlong and Chen, Michael ZQ and Li, Chanying},
  journal={Journal of the Franklin Institute},
  volume={354},
  number={11},
  pages={4311--4330},
  year={2017},
  publisher={Elsevier}
}

@article{yin2019adaptive,
  title={Adaptive robust integral sliding mode pitch angle control of an electro-hydraulic servo pitch system for wind turbine},
  author={Yin, Xiuxing and Zhang, Wencan and Jiang, Zhansi and Pan, Li},
  journal={Mechanical Systems and Signal Processing},
  volume={133},
  pages={105704},
  year={2019},
  publisher={Elsevier}
}

@article{yin2015adaptive,
  title={Adaptive sliding mode back-stepping pitch angle control of a variable-displacement pump controlled pitch system for wind turbines},
  author={Yin, Xiu-xing and Lin, Yong-gang and Li, Wei and Liu, Hong-wei and Gu, Ya-jing},
  journal={ISA transactions},
  volume={58},
  pages={629--634},
  year={2015},
  publisher={Elsevier}
}

@article{ayyarao2019modified,
  title={Modified vector controlled DFIG wind energy system based on barrier function adaptive sliding mode control},
  author={Ayyarao, Tummala SLV},
  journal={Protection and Control of Modern Power Systems},
  volume={4},
  number={1},
  pages={1--8},
  year={2019},
  publisher={SpringerOpen}
}

@article{rajendran2014variable,
  title={Variable speed wind turbine for maximum power capture using adaptive fuzzy integral sliding mode control},
  author={Rajendran, Saravanakumar and Jena, Debashisha},
  journal={Journal of Modern Power Systems and Clean Energy},
  volume={2},
  number={2},
  pages={114--125},
  year={2014},
  publisher={SGEPRI}
}

@inproceedings{ameli2019adaptive,
  title={Adaptive Integral Sliding Mode Design for the Pitch Control of a Variable Speed Wind Turbine},
  author={Ameli, Sina and Morshed, Mohammad Javad and Fekih, Afef},
  booktitle={2019 IEEE Conference on Control Technology and Applications (CCTA)},
  pages={290--295},
  year={2019},
  organization={IEEE}
}

@inproceedings{merabet2011adaptive,
  title={Adaptive sliding mode speed control for wind turbine systems},
  author={Merabet, Adel and Beguenane, Rachid and Thongam, Jogendra S and Hussein, Ibrahim},
  booktitle={IECON 2011-37th Annual Conference of the IEEE Industrial Electronics Society},
  pages={2461--2466},
  year={2011},
  organization={IEEE}
}

@inproceedings{barambones2015wind,
  title={Wind turbine control scheme based on adaptive sliding mode controller and observer},
  author={Barambones, Oscar and de Durana, Jose Maria Gonzalez},
  booktitle={2015 IEEE 20th Conference on Emerging Technologies \& Factory Automation (ETFA)},
  pages={1--7},
  year={2015},
  organization={IEEE}
}

@inproceedings{barambones2016adaptive,
  title={Adaptive sliding mode control strategy for a wind turbine systems using a HOSM wind torque observer},
  author={Barambones, Oscar and de Durana, Jose M Gonzalez},
  booktitle={2016 IEEE International Energy Conference (ENERGYCON)},
  pages={1--6},
  year={2016},
  organization={IEEE}
}

@article{dash2018adaptive,
  title={Adaptive fractional integral terminal sliding mode power control of UPFC in DFIG wind farm penetrated multimachine power system},
  author={Dash, PK and Patnaik, RK and Mishra, SP},
  journal={Protection and Control of Modern Power Systems},
  volume={3},
  number={1},
  pages={8},
  year={2018},
  publisher={Springer}
}

@article{falehi2020innovative,
  title={An innovative optimal RPO-FOSMC based on multi-objective grasshopper optimization algorithm for DFIG-based wind turbine to augment MPPT and FRT capabilities},
  author={Falehi, Ali Darvish},
  journal={Chaos, Solitons \& Fractals},
  volume={130},
  pages={109407},
  year={2020},
  publisher={Elsevier}
}

@article{rui2019fractional,
  title={Fractional-order sliding mode control for hybrid drive wind power generation system with disturbances in the grid},
  author={Rui, Xiaoming and Yin, Wenliang and Dong, Yongxing and Lin, Lanlan and Wu, Xin},
  journal={Wind Energy},
  volume={22},
  number={1},
  pages={49--64},
  year={2019},
  publisher={Wiley Online Library}
}

@article{talebi2018fractional,
  title={Fractional order sliding mode controller design for large scale variable speed wind turbine for power optimization},
  author={Talebi, Jalal and Ganjefar, Soheil},
  journal={Environmental Progress \& Sustainable Energy},
  volume={37},
  number={6},
  pages={2124--2131},
  year={2018},
  publisher={Wiley Online Library}
}

@article{talebi2019fractional,
  title={Fractional-order back-stepping sliding-mode torque control for a wind energy conversion system},
  author={Talebi, Jalal and Ganjefar, Soheil},
  journal={Environmental Progress \& Sustainable Energy},
  volume={38},
  number={4},
  pages={13110},
  year={2019},
  publisher={Wiley Online Library}
}

@article{ebrahimkhani2016robust,
  title={Robust fractional order sliding mode control of doubly-fed induction generator (DFIG)-based wind turbines},
  author={Ebrahimkhani, Sadegh},
  journal={ISA transactions},
  volume={63},
  pages={343--354},
  year={2016},
  publisher={Elsevier}
}

@article{kerrouche2019fractional,
  title={Fractional-order sliding mode control for D-STATCOM connected wind farm based DFIG under voltage unbalanced},
  author={Kerrouche, KDE and Wang, L and Mezouar, A and Boumediene, L and Van Den Bossche, A},
  journal={Arabian Journal for Science and Engineering},
  volume={44},
  number={3},
  pages={2265--2280},
  year={2019},
  publisher={Springer}
}

@article{li2020fractional,
  title={Fractional-order sliding mode control for damping of subsynchronous control interaction in DFIG-based wind farms},
  author={Li, Penghan and Xiong, Linyun and Wang, Ziqiang and Ma, Meiling and Wang, Jie},
  journal={Wind Energy},
  volume={23},
  number={3},
  pages={749--762},
  year={2020},
  publisher={Wiley Online Library}
}

@article{merida2014analysis,
  title={Analysis and synthesis of sliding mode control for large scale variable speed wind turbine for power optimization},
  author={M{\'e}rida, Jov{\'a}n and Aguilar, Luis T and D{\'a}vila, Jorge},
  journal={Renewable Energy},
  volume={71},
  pages={715--728},
  year={2014},
  publisher={Elsevier}
}

@article{golnary2018design,
  title={Design and comparison of quasi continuous sliding mode control with feedback linearization for a large scale wind turbine with wind speed estimation},
  author={Golnary, Farshad and Moradi, Hamed},
  journal={Renewable Energy},
  volume={127},
  pages={495--508},
  year={2018},
  publisher={Elsevier}
}

@article{golnary2019dynamic,
  title={Dynamic modelling and design of various robust sliding mode controls for the wind turbine with estimation of wind speed},
  author={Golnary, Farshad and Moradi, Hamed},
  journal={Applied Mathematical Modelling},
  volume={65},
  pages={566--585},
  year={2019},
  publisher={Elsevier}
}

@article{beltran2008high,
  title={High-order sliding-mode control of variable-speed wind turbines},
  author={Beltran, Brice and Ahmed-Ali, Tarek and Benbouzid, Mohamed El Hachemi},
  journal={IEEE Transactions on Industrial electronics},
  volume={56},
  number={9},
  pages={3314--3321},
  year={2008},
  publisher={IEEE}
}

@article{eddine2016comprehensive,
  title={A comprehensive review of LVRT capability and sliding mode control of grid-connected wind-turbine-driven doubly fed induction generator},
  author={Eddine, Kamel Djamel and Mezouar, Abdelkader and Boumediene, Larbi and Van Den Bossche, Alex PM},
  journal={Automatika},
  volume={57},
  number={4},
  pages={922--935},
  year={2016},
  publisher={Taylor \& Francis}
}

@article{hwang2013adaptive,
  title={Adaptive fuzzy hierarchical sliding-mode control for the trajectory tracking of uncertain underactuated nonlinear dynamic systems},
  author={Hwang, Chih-Lyang and Chiang, Chiang-Cheng and Yeh, Yao-Wei},
  journal={IEEE Transactions on Fuzzy Systems},
  volume={22},
  number={2},
  pages={286--299},
  year={2013},
  publisher={IEEE}
}

@article{wu2018adaptive,
  title={Adaptive fuzzy sliding mode control for translational oscillator with rotating actuator: A fuzzy model},
  author={Wu, Tiebin and Gui, Weihua and Hu, Dong and Du, Chenglong},
  journal={IEEE Access},
  volume={6},
  pages={55861--55869},
  year={2018},
  publisher={IEEE}
}

@article{ma2018cooperative,
  title={Cooperative fault diagnosis for uncertain nonlinear multiagent systems based on adaptive distributed fuzzy estimators},
  author={Ma, Hong-Jun and Xu, Linxing},
  journal={IEEE transactions on cybernetics},
  year={2018},
  publisher={IEEE}
}

@article{xiong2019coordinated,
  title={A coordinated high-order sliding mode control of DFIG wind turbine for power optimization and grid synchronization},
  author={Xiong, Linyun and Li, Penghan and Wu, Fei and Ma, Meiling and Khan, Muhammad Waseem and Wang, Jie},
  journal={International Journal of Electrical Power \& Energy Systems},
  volume={105},
  pages={679--689},
  year={2019},
  publisher={Elsevier}
}

@inproceedings{benbouzid2012high,
  title={A high-order sliding mode observer for sensorless control of DFIG-based wind turbines},
  author={Benbouzid, Mohamed and Beltran, Brice and Mangel, Herv{\'e} and Mamoune, Abdeslam},
  booktitle={IECON 2012-38th Annual Conference on IEEE Industrial Electronics Society},
  pages={4288--4292},
  year={2012},
  organization={IEEE}
}

@article{pratap2018robust,
  title={Robust control of variable speed wind turbine using quasi-sliding mode approach},
  author={Pratap, Bhanu and Singh, Navdeep and Kumar, Vineet},
  journal={Procedia Computer Science},
  volume={125},
  pages={398--404},
  year={2018},
  publisher={Elsevier}
}

@inproceedings{zhang2009new,
  title={A new pitch control strategy for wind turbines base on quasi-sliding mode control},
  author={Zhang, Lei and Chunliang, E and Li, Haidong and Xu, Honghua},
  booktitle={2009 International Conference on Sustainable Power Generation and Supply},
  pages={1--4},
  year={2009},
  organization={IEEE}
}

@article{evangelista2012lyapunov,
  title={Lyapunov-designed super-twisting sliding mode control for wind energy conversion optimization},
  author={Evangelista, Carolina and Puleston, P and Valenciaga, Fernando and Fridman, Leonid M},
  journal={IEEE Transactions on industrial electronics},
  volume={60},
  number={2},
  pages={538--545},
  year={2012},
  publisher={IEEE}
}

@article{evangelista2013active,
  title={Active and reactive power control for wind turbine based on a MIMO 2-sliding mode algorithm with variable gains},
  author={Evangelista, Carolina and Valenciaga, Fernando and Puleston, Paul},
  journal={IEEE Transactions on Energy Conversion},
  volume={28},
  number={3},
  pages={682--689},
  year={2013},
  publisher={IEEE}
}

@article{beltran2012second,
  title={Second-order sliding mode control of a doubly fed induction generator driven wind turbine},
  author={Beltran, Brice and Benbouzid, Mohamed El Hachemi and Ahmed-Ali, Tarek},
  journal={IEEE Transactions on Energy Conversion},
  volume={27},
  number={2},
  pages={261--269},
  year={2012},
  publisher={IEEE}
}

@article{benbouzid2014second,
  title={Second-order sliding mode control for DFIG-based wind turbines fault ride-through capability enhancement},
  author={Benbouzid, Mohamed and Beltran, Brice and Amirat, Yassine and Yao, Gang and Han, Jingang and Mangel, Herv{\'e}},
  journal={ISA transactions},
  volume={53},
  number={3},
  pages={827--833},
  year={2014},
  publisher={Elsevier}
}

@article{evangelista2016receding,
  title={Receding horizon adaptive second-order sliding mode control for doubly-fed induction generator based wind turbine},
  author={Evangelista, Carolina A and Pisano, Alessandro and Puleston, Paul and Usai, Elio},
  journal={IEEE Transactions on control systems technology},
  volume={25},
  number={1},
  pages={73--84},
  year={2016},
  publisher={IEEE}
}

@article{martinez2017second,
  title={Second-order sliding-mode-based global control scheme for wind turbine-driven DFIGs subject to unbalanced and distorted grid voltage},
  author={Martinez, Miren Itsaso and Susperregui, Ana and Tapia, Gerardo},
  journal={IET Electric Power Applications},
  volume={11},
  number={6},
  pages={1013--1022},
  year={2017},
  publisher={IET}
}

@article{susperregui2013second,
  title={Second-order sliding-mode controller design and tuning for grid synchronisation and power control of a wind turbine-driven doubly fed induction generator},
  author={Susperregui, Ana and Martinez, Miren Itsaso and Tapia, Gerardo and Vechiu, Ionel},
  journal={IET Renewable Power Generation},
  volume={7},
  number={5},
  pages={540--551},
  year={2013},
  publisher={IET}
}

@article{meghni2017second,
  title={A second-order sliding mode and fuzzy logic control to optimal energy management in wind turbine with battery storage},
  author={Meghni, Billel and Dib, Djalel and Azar, Ahmad Taher},
  journal={Neural Computing and Applications},
  volume={28},
  number={6},
  pages={1417--1434},
  year={2017},
  publisher={Springer}
}

@article{valenciaga2015multiple,
  title={Multiple-input--multiple-output high-order sliding mode control for a permanent magnet synchronous generator wind-based system with grid support capabilities},
  author={Valenciaga, Fernando and Fernandez, Roberto Daniel},
  journal={IET Renewable Power Generation},
  volume={9},
  number={8},
  pages={925--934},
  year={2015},
  publisher={IET}
}

@article{abolvafaei2019maximum,
  title={Maximum power extraction from a wind turbine using second-order fast terminal sliding mode control},
  author={Abolvafaei, Mahnaz and Ganjefar, Soheil},
  journal={Renewable Energy},
  volume={139},
  pages={1437--1446},
  year={2019},
  publisher={Elsevier}
}

@inproceedings{rajendran2015adaptive,
  title={Adaptive nonsingular terminal sliding mode control for variable speed wind turbine},
  author={Rajendran, Saravanakumar and Jena, Debashisha},
  booktitle={2015 IEEE 28th Canadian Conference on Electrical and Computer Engineering (CCECE)},
  pages={937--942},
  year={2015},
  organization={IEEE}
}

@inproceedings{zheng2017full,
  title={Full-order nonsingular terminal sliding mode control for variable pitch wind turbine},
  author={Zheng, Xuemei and Song, Rui and Pang, Songnan and Li, Qiuming},
  booktitle={2017 12th IEEE Conference on Industrial Electronics and Applications (ICIEA)},
  pages={2072--2077},
  year={2017},
  organization={IEEE}
}

@article{patnaik2016adaptive,
  title={Adaptive terminal sliding mode power control of DFIG based wind energy conversion system for stability enhancement},
  author={Patnaik, RK and Dash, PK and Mahapatra, Kaveri},
  journal={International Transactions on Electrical Energy Systems},
  volume={26},
  number={4},
  pages={750--782},
  year={2016},
  publisher={Wiley Online Library}
}

@article{patnaik2020adaptive,
  title={Adaptive third order terminal sliding mode power control of DFIG based wind farm for power system stabilisation},
  author={Patnaik, RK and Dash, PK and Mishra, SP},
  journal={International Journal of Dynamics and Control},
  volume={8},
  number={2},
  pages={629--643},
  year={2020},
  publisher={Springer}
}

@article{wu2018maximal,
  title={Maximal wind energy capture fuzzy terminal sliding mode control for DFIG with speed sensorless},
  author={Wu, Zhongqiang and Wang, Xinyi},
  journal={IEEJ Transactions on Electrical and Electronic Engineering},
  volume={13},
  number={7},
  pages={953--962},
  year={2018},
  publisher={Wiley Online Library}
}

@article{zheng2018integral,
  title={Integral-type terminal sliding-mode control for grid-side converter in wind energy conversion systems},
  author={Zheng, Xuemei and Feng, Yong and Han, Fengling and Yu, Xinghuo},
  journal={IEEE Transactions on Industrial Electronics},
  volume={66},
  number={5},
  pages={3702--3711},
  year={2018},
  publisher={IEEE}
}

@article{pradhan2018composite,
  title={A Composite Sliding Mode Controller for Wind Power Extraction in Remotely Located Solar PV--Wind Hybrid System},
  author={Pradhan, Subarni and Singh, Bhim and Panigrahi, Bijaya Ketan and Murshid, Shadab},
  journal={IEEE Transactions on Industrial Electronics},
  volume={66},
  number={7},
  pages={5321--5331},
  year={2018},
  publisher={IEEE}
}

@article{karabacak2019new,
  title={A new perturb and observe based higher order sliding mode MPPT control of wind turbines eliminating the rotor inertial effect},
  author={Karabacak, Murat},
  journal={Renewable Energy},
  volume={133},
  pages={807--827},
  year={2019},
  publisher={Elsevier}
}

@article{benamor2019novel,
  title={A novel rooted tree optimization apply in the high order sliding mode control using super-twisting algorithm based on DTC scheme for DFIG},
  author={Benamor, A and Benchouia, MT and Srairi, K and Benbouzid, MEH},
  journal={International Journal of Electrical Power \& Energy Systems},
  volume={108},
  pages={293--302},
  year={2019},
  publisher={Elsevier}
}

@article{xiong2020high,
  title={High-order sliding mode control of DFIG under unbalanced grid voltage conditions},
  author={Xiong, Linyun and Li, Penghan and Wang, Jie},
  journal={International Journal of Electrical Power \& Energy Systems},
  volume={117},
  pages={105608},
  year={2020},
  publisher={Elsevier}
}

@article{belabbas2019comparative,
  title={Comparative study of back-stepping controller and super twisting sliding mode controller for indirect power control of wind generator},
  author={Belabbas, Belkacem and Allaoui, Tayeb and Tadjine, Mohamed and Denai, Mouloud},
  journal={International Journal of System Assurance Engineering and Management},
  volume={10},
  number={6},
  pages={1555--1566},
  year={2019},
  publisher={Springer}
}

@article{evangelista2012multivariable,
  title={Multivariable 2-sliding mode control for a wind energy system based on a double fed induction generator},
  author={Evangelista, Carolina Alejandra and Valenciaga, Fernando and Puleston, P},
  journal={International journal of hydrogen energy},
  volume={37},
  number={13},
  pages={10070--10075},
  year={2012},
  publisher={Elsevier}
}

@article{moussa2019super,
  title={Super-twisting sliding mode control for brushless doubly fed induction generator based on WECS},
  author={Moussa, Oussama and Abdessemed, Rachid and Benaggoune, Said},
  journal={International Journal of System Assurance Engineering and Management},
  volume={10},
  number={5},
  pages={1145--1157},
  year={2019},
  publisher={Springer}
}

@article{liu2016second,
  title={Second-order sliding mode control for power optimisation of DFIG-based variable speed wind turbine},
  author={Liu, Xiangjie and Han, Yaozhen and Wang, Chengcheng},
  journal={IET Renewable Power Generation},
  volume={11},
  number={2},
  pages={408--418},
  year={2016},
  publisher={IET}
}

@article{krim2018power,
  title={Power management and second-order sliding mode control for standalone hybrid wind energy with battery energy storage system},
  author={Krim, Youssef and Abbes, Dhaker and Krim, Saber and Faouzi Mimouni, Mohamed},
  journal={Proceedings of the Institution of Mechanical Engineers, Part I: Journal of Systems and Control Engineering},
  volume={232},
  number={10},
  pages={1389--1411},
  year={2018},
  publisher={SAGE Publications Sage UK: London, England}
}

@inproceedings{morshed2018sliding,
  title={A Sliding mode-based approach to Enhance power quality in grid connected wind turbines},
  author={Morshed, Mohammad Javad and Fekih, Afef},
  booktitle={2018 Annual American Control Conference (ACC)},
  pages={6132--6137},
  year={2018},
  organization={IEEE}
}

@article{evangelista2010wind,
  title={Wind turbine efficiency optimization. Comparative study of controllers based on second order sliding modes},
  author={Evangelista, C and Puleston, P and Valenciaga, F},
  journal={international journal of hydrogen energy},
  volume={35},
  number={11},
  pages={5934--5939},
  year={2010},
  publisher={Elsevier}
}

@article{krim2019second,
  title={A second-order sliding-mode control for a real time emulator of a wind power system synchronized with electrical network},
  author={Krim, Youssef and Abbes, Dhaker and Krim, Saber and Mimouni, Mohamed Faouzi},
  journal={International Transactions on Electrical Energy Systems},
  volume={29},
  number={9},
  pages={e12051},
  year={2019},
  publisher={Wiley Online Library}
}

@article{krim2018classical,
  title={Classical vector, first-order sliding-mode and high-order sliding-mode control for a grid-connected variable-speed wind energy conversion system: A comparative study},
  author={Krim, Youssef and Abbes, Dhaker and Krim, Saber and Mimouni, Mohamed Faouzi},
  journal={Wind Engineering},
  volume={42},
  number={1},
  pages={16--37},
  year={2018},
  publisher={SAGE Publications Sage UK: London, England}
}

@article{kelkoul2020stability,
  title={Stability analysis and study between Classical Sliding Mode Control (SMC) and Super Twisting Algorithm (STA) for Doubly Fed Induction Generator (DFIG) under Wind turbine},
  author={Kelkoul, Bahia and Boumediene, Abdelmadjid},
  journal={Energy},
  pages={118871},
  year={2020},
  publisher={Elsevier}
}

@inproceedings{liu2014sliding,
  title={Sliding mode control for DFIG-based wind energy conversion optimization with switching gain adjustment},
  author={Liu, Xiangjie and Han, Yaozhen},
  booktitle={Proceeding of the 11th World Congress on Intelligent Control and Automation},
  pages={1213--1218},
  year={2014},
  organization={IEEE}
}

@inproceedings{yao2007adaptive,
  title={Adaptive fuzzy sliding-mode control in variable speed adjustable pitch wind turbine},
  author={Yao, Xingjia and Liu, Yingming and Guo, Changchun},
  booktitle={2007 IEEE International Conference on Automation and Logistics},
  pages={313--318},
  year={2007},
  organization={IEEE}
}

@article{subramaniam2019passivity,
  title={Passivity-based fuzzy ISMC for wind energy conversion systems with PMSG},
  author={Subramaniam, Ramasamy and Joo, Young Hoon},
  journal={IEEE Transactions on Systems, Man, and Cybernetics: Systems},
  year={2019},
  publisher={IEEE}
}

@article{benchabane2012improved,
  title={An improved efficiency of fuzzy sliding mode control of permanent magnet synchronous motor for wind turbine generator pumping system},
  author={Benchabane, Fateh and Titaouine, Abdenacer and Bennis, Ouafae and Guettaf, Abderazak and Yahia, Khaled and Taibi, Djamel},
  journal={Applied Solar Energy},
  volume={48},
  number={2},
  pages={112--117},
  year={2012},
  publisher={Springer}
}

@article{hwang2019disturbance,
  title={Disturbance observer-based integral fuzzy sliding-mode control and its application to wind turbine system},
  author={Hwang, Sounghwan and Park, Jin Bae and Joo, Young Hoon},
  journal={IET Control Theory \& Applications},
  volume={13},
  number={12},
  pages={1891--1900},
  year={2019},
  publisher={IET}
}

@article{do2016disturbance,
  title={Disturbance observer-based fuzzy SMC of WECSs without wind speed measurement},
  author={Do, Ton Duc},
  journal={IEEE access},
  volume={5},
  pages={147--155},
  year={2016},
  publisher={IEEE}
}

@article{lahlou2019sliding,
  title={Sliding mode controller based on type-2 fuzzy logic PID for a variable speed wind turbine},
  author={Lahlou, Zineb and Meziane, Khaddouj Ben and Boumhidi, Ismail},
  journal={International Journal of System Assurance Engineering and Management},
  volume={10},
  number={4},
  pages={543--551},
  year={2019},
  publisher={Springer}
}

@article{tahir2018new,
  title={A new sliding mode control strategy for variable-speed wind turbine power maximization},
  author={Tahir, Khalfallah and Belfedal, Cheikh and Allaoui, Tayeb and Denai, Mouloud and Doumi, M'hamed},
  journal={International Transactions on Electrical Energy Systems},
  volume={28},
  number={4},
  pages={e2513},
  year={2018},
  publisher={Wiley Online Library}
}

@article{qu2018neural,
  title={Neural network-based $H_{\infty}$ sliding mode control for nonlinear systems with actuator faults and unmatched disturbances},
  author={Qu, Qiuxia and Zhang, Huaguang and Yu, Rui and Liu, Yang},
  journal={Neurocomputing},
  volume={275},
  pages={2009--2018},
  year={2018},
  publisher={Elsevier}
}

@article{van2019adaptive,
  title={Adaptive neural network-based backstepping sliding mode control approach for dual-arm robots},
  author={Van Nguyen, Thai and Thai, Nguyen Huu and Pham, Hai Tuan and Phan, Tuan Anh and Nguyen, Linh and Le, Hai Xuan and Nguyen, Hiep Duc},
  journal={Journal of Control, Automation and Electrical Systems},
  volume={30},
  number={4},
  pages={512--521},
  year={2019},
  publisher={Springer}
}

@article{ling150robust,
  title={Robust adaptive motion tracking of piezoelectric actuated stages using online neural-network-based sliding mode control},
  author={Ling, Jie and Feng, Zhao and Zheng, Dongdong and Yang, Jun and Yu, Haoyong and Xiao, Xiaohui},
  journal={Mechanical Systems and Signal Processing},
  volume={150},
  pages={107235},
  publisher={Elsevier}
}

@article{yin2015novel,
  title={A novel fuzzy integral sliding mode current control strategy for maximizing wind power extraction and eliminating voltage harmonics},
  author={Yin, Xiu-xing and Lin, Yong-gang and Li, Wei and Gu, Ya-jing and Liu, Hong-wei and Lei, Peng-fei},
  journal={Energy},
  volume={85},
  pages={677--686},
  year={2015},
  publisher={Elsevier}
}

@article{kenne2017new,
  title={A new adaptive control strategy for a class of nonlinear system using RBF neuro-sliding-mode technique: application to SEIG wind turbine control system},
  author={Kenn{\'e}, Godpromesse and Fotso, Armel Simo and Lamnabhi-Lagarrigue, Fran{\c{c}}oise},
  journal={International Journal of Control},
  volume={90},
  number={4},
  pages={855--872},
  year={2017},
  publisher={Taylor \& Francis}
}

@inproceedings{boulouma2016rbf,
  title={RBF neural network sliding mode control of a PMSG based wind energy conversion system},
  author={Boulouma, Sabri and Belmili, Hocine},
  booktitle={2016 International Renewable and Sustainable Energy Conference (IRSEC)},
  pages={438--443},
  year={2016},
  organization={IEEE}
}

@inproceedings{lamzouri2018robust,
  title={A Robust Double Integral Sliding Mode Control Based Neural Network of a Large Wind Turbine},
  author={Lamzouri, Fatima Ez-zahra and Boufounas, El-Mahjoub and El Amrani, Aumeur},
  booktitle={2018 International Conference on Electronics, Control, Optimization and Computer Science (ICECOCS)},
  pages={1--6},
  year={2018},
  organization={IEEE}
}

@article{haq2020maximum,
  title={Maximum power extraction strategy for variable speed wind turbine system via neuro-adaptive generalized global sliding mode controller},
  author={Haq, Izhar Ul and Khan, Qudrat and Khan, Ilyas and Akmeliawati, Rini and Nisar, Kottakkaran Soopy and Khan, Imran},
  journal={IEEE Access},
  volume={8},
  pages={128536--128547},
  year={2020},
  publisher={IEEE}
}

@article{hong2019enhanced,
  title={Enhanced radial fuzzy wavelet neural network with sliding mode control for a switched reluctance wind turbine distributed generation system},
  author={Hong, Chih-Ming and Chen, Chiung-Hsing},
  journal={Engineering Optimization},
  volume={51},
  number={7},
  pages={1133--1151},
  year={2019},
  publisher={Taylor \& Francis}
}

@inproceedings{djilali2017neural,
  title={Neural sliding mode field oriented control for DFIG based wind turbine},
  author={Djilali, Larbi and Sanchez, Edgar N and Belkheiri, Mohamed},
  booktitle={2017 IEEE International Conference on Systems, Man, and Cybernetics (SMC)},
  pages={2087--2092},
  year={2017},
  organization={IEEE}
}

@inproceedings{boufounas2014optimal,
  title={Optimal neural network sliding mode control for a variable speed wind turbine based on APSO algorithm},
  author={Boufounas, El-mahjoub and Boumhidi, Jaouad and Boumhidi, Ismail},
  booktitle={2014 Second World Conference on Complex Systems (WCCS)},
  pages={419--424},
  year={2014},
  organization={IEEE}
}

@inproceedings{berrada2015optimal,
  title={Optimal neural network sliding mode control without reaching phase using genetic algorithm for a wind turbine},
  author={Berrada, Youssef and Boufounas, El-mahjoub and Boumhidi, Ismail},
  booktitle={2015 10th International Conference on Intelligent Systems: Theories and Applications (SITA)},
  pages={1--6},
  year={2015},
  organization={IEEE}
}

@article{djilali2018neural,
  title={Neural Sliding Mode Control of a DFIG Based Wind Turbine with Measurement Delay},
  author={Djilali, Larbi and Caballero-Barrag{\'a}n, H and Osuna-Ibarra, LP and Sanchez, Edgar N and Loukianov, AG},
  journal={IFAC-PapersOnLine},
  volume={51},
  number={13},
  pages={456--461},
  year={2018},
  publisher={Elsevier}
}

@article{xin2014chattering,
  title={Chattering Free Sliding Mode Pitch Control of PMSG Wind Turbine},
  author={Xin, Wang and Wan-li, Zhu and Ceng, Song and Bin, Qin},
  journal={IFAC Proceedings Volumes},
  volume={47},
  number={3},
  pages={6758--6763},
  year={2014},
  publisher={Elsevier}
}

@article{schulte2015fault,
  title={Fault-tolerant control of wind turbines with hydrostatic transmission using Takagi--Sugeno and sliding mode techniques},
  author={Schulte, Horst and Gauterin, Eckhard},
  journal={Annual Reviews in Control},
  volume={40},
  pages={82--92},
  year={2015},
  publisher={Elsevier}
}

@inproceedings{beltran2009high,
  title={High-order sliding mode control of a DFIG-based wind turbine for power maximization and grid fault tolerance},
  author={Beltran, Brice and Benbouzid, MEH and Ahmed-Ali, Tarek},
  booktitle={2009 IEEE International Electric Machines and Drives Conference},
  pages={183--189},
  year={2009},
  organization={IEEE}
}

@article{sami2012fault,
  title={Fault tolerant adaptive sliding mode controller for wind turbine power maximisation},
  author={Sami, Montadher and Patton, Ron J},
  journal={IFAC Proceedings Volumes},
  volume={45},
  number={13},
  pages={499--504},
  year={2012},
  publisher={Elsevier}
}

@inproceedings{sami2012wind,
  title={Wind turbine power maximisation based on adaptive sensor fault tolerant sliding mode control},
  author={Sami, Montadher and Patton, Ron J},
  booktitle={2012 20th Mediterranean Conference on Control \& Automation (MED)},
  pages={1183--1188},
  year={2012},
  organization={IEEE}
}

@article{dursun2020second,
  title={Second-order sliding mode voltage-regulator for improving MPPT efficiency of PMSG-based WECS},
  author={Dursun, Emre Hasan and Kulaksiz, Ahmet Afsin},
  journal={International Journal of Electrical Power \& Energy Systems},
  volume={121},
  pages={106149},
  year={2020},
  publisher={Elsevier}
}

@article{morshed2020design,
  title={Design of a chattering-free integral terminal sliding mode approach for DFIG-based wind energy systems},
  author={Morshed, Mohammad Javad and Fekih, Afef},
  journal={Optimal Control Applications and Methods},
  volume={41},
  number={5},
  pages={1718--1734},
  year={2020},
  publisher={Wiley Online Library}
}

@article{song2019dynamic,
  title={Dynamic event-triggered sliding mode control: Dealing with slow sampling singularly perturbed systems},
  author={Song, Jun and Niu, Yugang},
  journal={IEEE Transactions on Circuits and Systems II: Express Briefs},
  year={2019},
  publisher={IEEE}
}

@article{liu2020event,
  title={Event-triggered sliding mode control of nonlinear dynamic systems},
  author={Liu, Xinxin and Su, Xiaojie and Shi, Peng and Shen, Chao and Peng, Yan},
  journal={Automatica},
  volume={112},
  pages={108738},
  year={2020},
  publisher={Elsevier}
}

@article{jiang2019takagi,
  title={Takagi--Sugeno Model Based Event-Triggered Fuzzy Sliding-Mode Control of Networked Control Systems With Semi-Markovian Switchings},
  author={Jiang, Baoping and Karimi, Hamid Reza and Kao, Yonggui and Gao, Cunchen},
  journal={IEEE Transactions on Fuzzy Systems},
  volume={28},
  number={4},
  pages={673--683},
  year={2019},
  publisher={IEEE}
}

@article{ni2016fast,
  title={Fast fixed-time nonsingular terminal sliding mode control and its application to chaos suppression in power system},
  author={Ni, Junkang and Liu, Ling and Liu, Chongxin and Hu, Xiaoyu and Li, Shilei},
  journal={IEEE Transactions on Circuits and Systems II: Express Briefs},
  volume={64},
  number={2},
  pages={151--155},
  year={2016},
  publisher={IEEE}
}

@article{jing2019adaptive,
  title={Adaptive sliding mode disturbance rejection control with prescribed performance for robotic manipulators},
  author={Jing, Chenghu and Xu, Hongguang and Niu, Xinjian},
  journal={ISA transactions},
  volume={91},
  pages={41--51},
  year={2019},
  publisher={Elsevier}
}

@article{liu2020adaptive,
  title={Adaptive sliding mode control for uncertain active suspension systems with Prescribed performance},
  author={Liu, Yan-Jun and Chen, Hao},
  journal={IEEE Transactions on Systems, Man, and Cybernetics: Systems},
  year={2020},
  publisher={IEEE}
}

@article{hou2019discrete,
  title={Discrete-time terminal sliding-mode tracking control with alleviated chattering},
  author={Hou, Huazhou and Yu, Xinghuo and Xu, Long and Chuei, Raymond and Cao, Zhenwei},
  journal={IEEE/ASME Transactions on Mechatronics},
  volume={24},
  number={4},
  pages={1808--1817},
  year={2019},
  publisher={IEEE}
}

@article{yi2019adaptive,
  title={Adaptive second-order fast nonsingular terminal sliding mode control for robotic manipulators},
  author={Yi, Shanchao and Zhai, Junyong},
  journal={ISA transactions},
  volume={90},
  pages={41--51},
  year={2019},
  publisher={Elsevier}
}

@article{guo2020terminal,
  title={Terminal sliding mode control of mems gyroscopes with finite-time learning},
  author={Guo, Yuyan and Xu, Bin and Zhang, Rui},
  journal={IEEE Transactions on Neural Networks and Learning Systems},
  year={2020},
  publisher={IEEE}
}

@article{wang2019discrete,
  title={Discrete-Time Fast Terminal Sliding Mode Control Design for DC--DC Buck Converters with Mismatched Disturbances},
  author={Wang, Zuo and Li, Shihua and Li, Qi},
  journal={IEEE Transactions on Industrial Informatics},
  volume={16},
  number={2},
  pages={1204--1213},
  year={2019},
  publisher={IEEE}
}

@article{gonzalez2011variable,
  title={Variable gain super-twisting sliding mode control},
  author={Gonzalez, Tenoch and Moreno, Jaime A and Fridman, Leonid},
  journal={IEEE Transactions on Automatic Control},
  volume={57},
  number={8},
  pages={2100--2105},
  year={2011},
  publisher={IEEE}
}

@article{hou2020composite,
  title={Composite super-twisting sliding mode control design for PMSM speed regulation problem based on a novel disturbance observer},
  author={Hou, Qiankang and Ding, Shihong and Yu, Xinghuo},
  journal={IEEE Transactions on Energy Conversion},
  year={2020},
  publisher={IEEE}
}

@article{ullah2020variable,
  title={Variable gain high order sliding mode control approaches for PMSG based variable speed wind energy conversion system},
  author={Ullah, Ameen and Khan, Laiq and Khan, Qudrat and Ahmad, Saghir},
  journal={Turkish Journal of Electrical Engineering \& Computer Sciences},
  volume={28},
  number={5},
  pages={2997--3012},
  year={2020},
  publisher={The Scientific and Technological Research Council of Turkey}
}

@article{djilali2020first,
  title={First and High Order Sliding Mode Control of a DFIG-Based Wind Turbine},
  author={Djilali, Larbi and Sanchez, Edgar N and Belkheiri, Mohammed},
  journal={Electric Power Components and Systems},
  pages={1--12},
  year={2020},
  publisher={Taylor \& Francis}
}

@article{tria2017integral,
  title={An integral sliding mode controller with super-twisting algorithm for direct power control of wind generator based on a doubly fed induction generator},
  author={Tria, FZ and Srairi, Kamel and Benchouia, MT and Benbouzid, Mohamed El Hachemi},
  journal={International Journal of System Assurance Engineering and Management},
  volume={8},
  number={4},
  pages={762--769},
  year={2017},
  publisher={Springer}
}

@article{mazen2020modeling,
  title={Modeling and Performance Improvement of Direct Power Control of Doubly-Fed Induction Generator Based Wind Turbine through Second-Order Sliding Mode Control Approach},
  author={Mazen Alhato, Mohammed and Bouall{\`e}gue, Soufiene and Rezk, Hegazy},
  journal={Mathematics},
  volume={8},
  number={11},
  pages={2012},
  year={2020},
  publisher={Multidisciplinary Digital Publishing Institute}
}

@article{sami2020sensorless,
  title={Sensorless fractional order composite sliding mode control design for wind generation system},
  author={Sami, Irfan and Ullah, Shafaat and Ullah, Nasim and Ro, Jong-Suk},
  journal={ISA transactions},
  year={2020},
  publisher={Elsevier}
}

@article{pape2019offshore,
  title={An offshore wind farm with DC collection system featuring differential power processing},
  author={Pape, Marten and Kazerani, Mehrdad},
  journal={IEEE Transactions on Energy Conversion},
  volume={35},
  number={1},
  pages={222--236},
  year={2019},
  publisher={IEEE}
}

@article{yang2016survey,
  title={A survey of fault diagnosis for onshore grid-connected converter in wind energy conversion systems},
  author={Yang, Zhimin and Chai, Yi},
  journal={Renewable and Sustainable Energy Reviews},
  volume={66},
  pages={345--359},
  year={2016},
  publisher={Elsevier}
}

@article{njiri2016state,
  title={State-of-the-art in wind turbine control: Trends and challenges},
  author={Njiri, Jackson G and S{\"o}ffker, Dirk},
  journal={Renewable and Sustainable Energy Reviews},
  volume={60},
  pages={377--393},
  year={2016},
  publisher={Elsevier}
}

@article{soliman2020novel,
  title={A Novel Adaptive Control Method for Performance Enhancement of Grid-Connected Variable-Speed Wind Generators},
  author={Soliman, Mahmoud A and Hasanien, Hany M and Al-Durra, Ahmed and Alsaidan, Ibrahim},
  journal={IEEE Access},
  volume={8},
  pages={82617--82629},
  year={2020},
  publisher={IEEE}
}

@inproceedings{wang2007survey,
  title={A survey on wind power technologies in power systems},
  author={Wang, Chen and Wang, Liming and Shi, Libao and Ni, Yixin},
  booktitle={2007 IEEE Power Engineering Society General Meeting},
  pages={1--6},
  year={2007},
  organization={IEEE}
}

@article{yin2016turbine,
  title={Turbine stability-constrained available wind power of variable speed wind turbines for active power control},
  author={Yin, Minghui and Xu, Yan and Shen, Chun and Liu, Jiankun and Dong, Zhao Yang and Zou, Yun},
  journal={IEEE Transactions on Power Systems},
  volume={32},
  number={3},
  pages={2487--2488},
  year={2016},
  publisher={IEEE}
}

@article{ghanbarpour2020dependable,
  title={Dependable power extraction in wind turbines using model predictive fault tolerant control},
  author={Ghanbarpour, Kamyar and Bayat, Farhad and Jalilvand, Abolfazl},
  journal={International Journal of Electrical Power \& Energy Systems},
  volume={118},
  pages={105802},
  year={2020},
  publisher={Elsevier}
}

@article{kuhne2018fault,
  title={Fault estimation and fault-tolerant control of the FAST NREL 5-MW reference wind turbine using a proportional multi-integral observer},
  author={K{\"u}hne, Patrick and P{\"o}schke, Florian and Schulte, Horst},
  journal={International Journal of Adaptive Control and Signal Processing},
  volume={32},
  number={4},
  pages={568--585},
  year={2018},
  publisher={Wiley Online Library}
}

@article{shaker2014active,
  title={Active sensor fault tolerant output feedback tracking control for wind turbine systems via T--S model},
  author={Shaker, Montadher Sami and Patton, Ron J},
  journal={Engineering Applications of Artificial Intelligence},
  volume={34},
  pages={1--12},
  year={2014},
  publisher={Elsevier}
}

@article{kamal2013fuzzy,
  title={Fuzzy fault-tolerant control of wind-diesel hybrid systems subject to sensor faults},
  author={Kamal, Elkhatib and Aitouche, Abdel and Oueidat, Mohamad},
  journal={IEEE Transactions on Sustainable Energy},
  volume={4},
  number={4},
  pages={857--866},
  year={2013},
  publisher={IEEE}
}

@article{shahbazi2012fpga,
  title={FPGA-based reconfigurable control for fault-tolerant back-to-back converter without redundancy},
  author={Shahbazi, Mahmoud and Poure, Philippe and Saadate, Shahrokh and Zolghadri, Mohammad Reza},
  journal={IEEE Transactions on Industrial Electronics},
  volume={60},
  number={8},
  pages={3360--3371},
  year={2012},
  publisher={IEEE}
}

@article{zhang2020adaptive,
  title={Adaptive sliding mode-based lateral stability control of steer-by-wire vehicles with experimental validations},
  author={Zhang, Jie and Wang, Hai and Zheng, Jinchuan and Cao, Zhenwei and Man, Zhihong and Yu, Ming and Chen, Long},
  journal={IEEE Transactions on Vehicular Technology},
  volume={69},
  number={9},
  pages={9589--9600},
  year={2020},
  publisher={IEEE}
}


\end{document}